\documentclass[a4paper, onecolumn]{quantumarticle}
\pdfoutput=1
\usepackage[utf8]{inputenc}
\usepackage[english]{babel}
\usepackage[T1]{fontenc}
\usepackage{csquotes}
\usepackage{bm}

\usepackage[numbers,sort&compress]{natbib}
\usepackage{amsfonts,amsthm,amsmath,amssymb,mathtools}
\usepackage{dsfont}
\usepackage{graphicx, quantikz, tikz,tikz-3dplot}
\usepackage{circuitikz}
\usetikzlibrary{shapes.geometric}
\usetikzlibrary{positioning}
\usetikzlibrary{arrows}
\usetikzlibrary{decorations.pathmorphing,decorations.markings}
\usetikzlibrary{calc, fit}
\pgfdeclarelayer{background}
\pgfsetlayers{background,main}
\usepackage{subfigure}
\usepackage{cancel}
\usepackage{slashed}
\usepackage{anyfontsize}

\usepackage{booktabs,multirow}
\usepackage{physics}
\usepackage[section]{placeins}
\usepackage{float}
\usepackage{comment}
\usepackage{fullpage}
\usepackage{enumerate}

\usepackage{sistyle}
\SIthousandsep{,}
\usepackage{cancel}

\usepackage[noorphans,vskip=0.5ex]{quoting}

\definecolor{navyblue}{rgb}{0.0, 0.0, 0.5}
\definecolor{LightPink}{rgb}{0.858, 0.188, 0.478}
\usepackage[colorlinks=true, urlcolor=blue,citecolor=purple,anchorcolor=blue, linkcolor=LightPink]{hyperref}

\newcommand{\adag}{a^{\dagger}}
\newcommand{\crt}{a^{\dagger}_{\sigma}}
\newcommand{\annh}{a_{\sigma}}

\newcommand{\fNum}{N_{\sigma}}
\newcommand{\cP}{\mathcal{P}}
\newcommand{\tannh}{\tilde{a}_{\sigma}}
\newcommand{\tcrt}{\tilde{a}^{\dagger}_{\sigma}}
\newcommand{\tmu}{\tilde{\mu}}
\newcommand{\tnu}{\tilde{\nu}}

\newcommand{\pimax}{\pi_{\max}}
\newcommand{\Pimax}{\Pi_{\max}}
\newcommand{\dpi}{\delta_\pi}
\newcommand{\calP}{\mathcal{P}}

\newcommand{\reta}{\rangle_{\eta}}
\newcommand{\piI}{\pi_{I}}
\newcommand{\Hpp}{H_{\pi}}

\newcommand{\canpi}{\slashed{\pi}}
\newcommand{\Cpi}{C_{\slashed{\pi}}}
\newcommand{\Dpi}{D_{\slashed{\pi}}}
\newcommand{\dyn}{D\pi}

\newcommand{\Hfree}{H_{\text{free}}}
\newcommand{\Heff}{H_{\text{eff}}}
\newcommand{\CI}{C_{I^2}}
\newcommand{\Dcost}{D_{\rm cost}}
\newcommand{\epsprod}{\epsilon_{\rm prod}}
\newcommand{\epscut}{\epsilon_{\rm cut}}
\newcommand{\epstrunc}{\epsilon_{\rm trunc}}

\newcommand{\epssyn}{\epsilon_{\rm syn}}

\newcommand{\hop}{\Delta}

\usepackage[capitalise,nameinlink]{cleveref}
\Crefname{lemma}{Lemma}{Lemmas}
\Crefname{proposition}{Proposition}{Propositions}
\Crefname{definition}{Definition}{Definitions}
\Crefname{theorem}{Theorem}{Theorems}
\Crefname{conjecture}{Conjecture}{Conjectures}
\Crefname{corollary}{Corollary}{Corollaries}
\Crefname{example}{Example}{Examples}
\Crefname{section}{Section}{Sections}
\Crefname{appendix}{Appendix}{Appendices}
\crefname{figure}{Fig.}{Figs.}
\Crefname{figure}{Figure}{Figures}
\crefname{equation}{Eq.}{Eqs.}
\Crefname{equation}{Equation}{Equations}
\Crefname{table}{Table}{Tables}
\Crefname{item}{Property}{Properties}
\Crefname{remark}{Remark}{Remarks}

\newtheorem{theorem}{Theorem}
\newtheorem{definition}[theorem]{Definition}
\newtheorem{corollary}[theorem]{Corollary}

\newtheorem{lemma}[theorem]{Lemma}

\usepackage{ltablex}

\begin{document}

\title{Quantum Algorithms for Simulating Nuclear Effective Field Theories}
\author{James~D.~Watson}
\affiliation{Joint Center for Quantum Information and Computer Science, NIST/University of Maryland, College Park, Maryland 20742, USA}
\affiliation{Department of Computer Science and Institute for Advanced Computer Studies, University of Maryland, College Park, MD 20742, USA}
\author{Jacob~Bringewatt}
\affiliation{Joint Center for Quantum Information and Computer Science, NIST/University of Maryland, College Park, Maryland 20742, USA}
\affiliation{Joint Quantum Institute, NIST/University of Maryland, College Park, Maryland 20742, USA}
\affiliation{Department of Physics, Harvard University, Cambridge, MA 02138 USA}
\affiliation{Department of Physics, University of Maryland, College Park, MD 20742, USA}
\author{Alexander~F.~Shaw}
\affiliation{Joint Center for Quantum Information and Computer Science, NIST/University of Maryland, College Park, Maryland 20742, USA}
\affiliation{Department of Physics, University of Maryland, College Park, MD 20742, USA}
\author{Andrew~M.~Childs}
\affiliation{Joint Center for Quantum Information and Computer Science, NIST/University of Maryland, College Park, Maryland 20742, USA}
\affiliation{Department of Computer Science and Institute for Advanced Computer Studies, University of Maryland, College Park, MD 20742, USA}
\author{Alexey~V.~Gorshkov}
\affiliation{Joint Center for Quantum Information and Computer Science, NIST/University of Maryland, College Park, Maryland 20742, USA}
\affiliation{Joint Quantum Institute, NIST/University of Maryland, College Park, Maryland 20742, USA}
\author{Zohreh~Davoudi}
\affiliation{Joint Center for Quantum Information and Computer Science, NIST/University of Maryland, College Park, Maryland 20742, USA}
\affiliation{Department of Physics, University of Maryland, College Park, MD 20742, USA}
\affiliation{Maryland Center for Fundamental Physics, University of Maryland, College Park, MD 20742, USA}

\date{May 18, 2026}

\begin{abstract}
Quantum computers offer the potential to simulate nuclear processes that are classically intractable. With the goal of understanding the necessary quantum resources to realize this potential, we employ state-of-the-art Hamiltonian-simulation methods, and conduct a thorough algorithmic analysis, to estimate the qubit and gate costs to simulate low-energy effective field theories (EFTs) of nuclear physics. In particular, within the framework of nuclear lattice EFT, we obtain simulation costs for the leading-order pionless and  pionful EFTs. For the latter, we consider both static pions represented by a one-pion-exchange potential between the nucleons, and dynamical pions represented by relativistic bosonic fields coupled to non-relativistic nucleons. Within these models, we examine the resource costs for the tasks of time evolution and energy estimation for physically relevant scales. We account for model errors associated with truncating either long-range interactions in the one-pion-exchange EFT or the pionic Hilbert space in the dynamical-pion EFT, and for algorithmic errors associated with product-formula approximations and quantum phase estimation. We find that the pionless EFT is the least costly to simulate, followed by the one-pion-exchange theory, then the dynamical-pion theory. We demonstrate how symmetries of the low-energy nuclear Hamiltonians can be utilized to obtain tighter error bounds on the simulation algorithm. Furthermore, by retaining the locality of nucleonic interactions when mapped to qubits (using Verstraete-Cirac and cubic-compact encodings), we achieve reduced circuit depth and substantial parallelization. In the process, we develop new methods to bound the algorithmic error for classes of fermionic Hamiltonians that preserve the number of fermions, and demonstrate that reasonably tight Trotter error bounds can be achieved by explicitly computing nested commutators of Hamiltonian terms. Compared to previous estimates for simulating the pionless EFT, our results represent an improvement by several orders of magnitude. This work highlights the importance of combining physics insights and algorithmic advancement in reducing the cost of quantum simulation.

\end{abstract}
\maketitle 

\clearpage
\tableofcontents
\clearpage

\section{Introduction}
A successful computational nuclear-physics program is crucial for a range of endeavors, including enhancing our
understanding of the densest forms of matter (such as in the cores of nuclei and the interiors of neutron stars); making
reliable predictions for fission and fusion processes of interest in energy research and astrophysics; and determining
the response of nuclear targets in tests of the Standard Model and searches for new physics, such as in neutrino-nucleus
scattering, neutrinoless double-beta decay, and dark-matter direct detection. Solving quantum many-body problems, in and out of equilibrium, is at the heart of this program and continues to benefit from computational advances in high-performance computing~\cite{carlson2016exascale} and from novel approaches such as machine learning~\cite{Bedaque:2021bja,boehnlein2021artificial}. 

While it is desirable to predict nuclear phenomena from the underlying Standard-Model interactions via the method of lattice quantum chromodynamics (QCD), such first-principles simulations have only been feasible so far for single nucleons~\cite{aoki2022flag,davoudi2022report}, and for light nuclei with unphysical Standard-Model parameters to tame the computational cost~\cite{davoudi2021nuclear,drischler2021towards}. An alternative and more realistic route toward simulating large nuclear systems is to consider nucleons, i.e., protons and neutrons (instead of quarks and gluons), as the fundamental degrees of freedom, where the interactions are deduced from experiment or by matching to analytical and numerical predictions of QCD.
The benefit is that nucleons behave as non-relativistic fermions for most phenomenological scenarios of interest. Thus, the problem reduces to solving a non-relativistic many-body Schr\"odinger equation, for which approaches such as 
 quantum Monte Carlo~\cite{carlson2015quantum}, no-core shell model~\cite{barrett2013ab}, coupled-cluster~\cite{hagen2014coupled}, self-consistent Green's function~\cite{carbone2013self}, in-medium similarity
renormalization group~\cite{hergert2016medium}, and nuclear lattice~\cite{lahde2019nuclear} methods have been developed and applied successfully.

Unfortunately, even such an effective approach is computationally intractable with current computing resources for certain nuclei, particularly those beyond medium-mass isotopes~\cite{carlson2016exascale,tews2020new}. The difficulty arises from the exponential increase in the size of the Hilbert space as a function of the number of nucleons, along with an intrinsic fermionic sign problem plaguing current methods. 
In particular, nuclear dynamics, relevant to studying nuclear reactions and nuclear responses to experimental probes, is a much less explored territory, except for lighter nuclei or in limited scenarios~\cite{navratil2022ab,ruso2022theoretical,launey2021nuclear}. For such problems, one has to rely on phenomenological models, as well as semi-classical, mean-field, or truncated Hilbert-space approaches~\cite{faessler2008quasiparticle,arima1981interacting,nikvsic2011relativistic,yu2003energy,bulgac2020nuclear,negele1982mean,caurier2005shell,OTSUKA2001319}, to be able to describe the physics of heavy nuclei and nuclear matter, often at the cost of unquantified uncertainties. Thus, it is important to seek feasible strategies for performing accurate nuclear-physics computations.

A reliable first-principles route in the long run may be to employ quantum computation. 
This prospect, along with recent advances in both algorithms and hardware technology, has inspired extensive research into applications of quantum computing to many computationally-oriented disciplines such as materials science and quantum chemistry~\cite{mcardle2020quantum, cao2019quantum, babbush2017low, bauer2020quantum, von2021quantum, clinton2022towards}, and more recently, high-energy and nuclear physics~\cite{davoudi2022quantum, catterall2022report, humble2022snowmass, NSAC-QIS-2019-QuantumInformationScience, beck2023quantum,Bauer:2023qgm, rubin2023quantum}.
By storing the state of a quantum system in a register of qubits (or higher-dimensional subsystems), quantum computers can represent and evolve a quantum model much more efficiently than classical computers.
A common trend in algorithmic research, particularly in materials science and quantum chemistry, has been to adopt generic quantum-simulation algorithms as a first attempt, and then to develop algorithms with improved performance through various strategies, such as extensive optimizations at the circuit level~\cite{wecker2014gate,hastings2014improving,wecker2015solving,babbush2018low}.
Such applications can benefit from advancement in generic quantum-simulation algorithms, but they can also inspire new algorithms.
For example, hybrid classical-quantum algorithms such as variational methods were developed and improved in response to the need for extremely precise energy spectra in quantum chemistry using near-term quantum computing~\cite{peruzzo2014variational,mcclean2016theory, o2016scalable, cerezo2021variational}. It is conceivable that applications in nuclear physics will provide another avenue for further development of quantum-simulation algorithms, given the peculiarities of the quantum many-body problem in nuclear physics and the diversity of phenomena to be simulated. 

The nuclear potential involves short-, intermediate-, and long-range interactions, two- and higher-body interactions, and becomes increasingly complex as the energy and density grow. Furthermore, both static and dynamical quantities are intensely studied in nuclear physics. The first adoption of quantum algorithms for the quantum many-body problem was reported in the pioneering work of Ovrum and Hjorth-Jensen~\cite{ovrum2007quantum}, followed by that of Dumitrescu et al.~\cite{dumitrescu2018cloud}, in which the deuteron binding energy was quantum computed using a variational quantum eigensolver, and of Roggero et al.~\cite{roggero2018linear,roggero2020quantum} concerning nuclear response in electron- and neutrino-nucleus scattering, stimulating a growing body of work in similar problems~\cite{roggero2020preparation,baroni2022nuclear,choi2021rodeo,qian2021demonstration,holland2020optimal,turro2023quantum,turro2022imaginary,stetcu2022projection,du2021ab, du2021quantum, qian2021solving,shehab2019toward,perez2023nuclear,robin2023quantum}. As with quantum-chemistry simulations that employ a variety of representations for the Hamiltonian, e.g., in first- or second-quantized forms~\cite{kivlichan2017bounding, su2021fault, babbush2016exponentially, jorgensen2012second,  moll2016optimizing, babbush2017exponentially} and momentum- or position-space bases~\cite{stetina2022simulating}, the nuclear many-body problem can be cast in various representations adapted to the many-body method of choice. Each approach has its own systematic uncertainties associated with ways the degrees of freedom are truncated to fit the problem within the computational resources available. For example, the aforementioned work of Roggero et al.~\cite{roggero2020quantum} adopts a (spatial) lattice formulation with the leading-order chiral EFT Hamiltonian with contact two- and three-body interactions~\cite{lee2009lattice,lahde2019nuclear}, and performs a thorough algorithmic analysis to estimate the resources used to compute time evolution of the system within given accuracy, using first- and second-order product formulae~\cite{suzuki1991general, wiebe2010higher, childs2021theory}. Subsequently, there has been more progress in bounding the errors in digitized time dynamics using product formulae. For example, it is known that information about properties of the state under evolution, such as its symmetries and energy domain, can greatly tighten the bounds~\cite{su2021nearly,yi2022spectral,csahinouglu2021hamiltonian,zhao2022hamiltonian,zhao2024entanglement}. Such improved bounds are crucial for accurately estimating simulation costs.

We should also investigate the algorithmic cost of more realistic nuclear Hamiltonians, given that more complex effective interactions are in play when larger nuclei or denser environments are concerned---systems that are prime candidates for quantum-computing applications. For example, pion exchanges and, eventually, pion production become kinematically relevant as atomic numbers and momentum transfers increase, making the use of pionless EFT~\cite{kaplan1996nucleon, kaplan1998new, van1999effective, bedaque1999renormalization, chen1999nucleon, bedaque2002effective} insufficient. A primary question is how to efficiently simulate a system described by a pionful Hamiltonian, and whether it is computationally advantageous to treat pions as dynamical degrees of freedom, 
or---as is standard in the framework of chiral nuclear forces---to integrate them out to obtain long-range potentials such as one-, two-, and multi-pion exchange potentials~\cite{weinberg1990nuclear, weinberg1991effective, machleidt2011chiral, epelbaum2009modern}. In other words, is it beneficial to work with pion potentials, resulting in a non-local formulation, or to restore locality at the cost of introducing pions explicitly? This question has parallels in lattice-gauge-theory simulations and has been recently investigated for one-dimensional theories~\cite{nguyen2022digital,shaw2020quantum}. Additionally, nucleons are fermions in a three-dimensional space, and mapping them to qubit degrees of freedom introduces a gate overhead in non-local mappings such as the Jordan-Wigner transformation~\cite{Jordan_Wigner_1928}, or both qubit and gate overhead in local mappings such as the Verstraete-Cirac encoding~\cite{Verstraete_Cirac_2005}. The interplay between the (non-)locality of interactions and the (non-)locality of the fermion-to-qubit mapping is also a key feature to investigate.

This paper provides the first steps toward addressing the questions posed above, taking algorithmic analysis for quantum simulation of nuclear lattice EFTs to the next level.
In particular, we leverage properties of the nuclear Hamiltonians that allow us to use local fermion-to-qubit mappings in combination with carefully chosen Hamiltonian decompositions for product-formula algorithms.
This allows for much greater parallelization of the simulation.
We combine this with state-of-the-art error-bound analysis for product-formula simulations, including symmetry considerations, to obtain substantially improved cost estimates for simulating time evolution and estimating the energy spectrum of nuclei with leading-order chiral EFT Hamiltonians. In particular, we provide the first cost estimates for simulations beyond pionless EFTs, including theories involving pions.
These cost estimates are given in terms of 2-qubit circuit depths and $T$-gate counts.

The rest of the paper is organized as follows. In \cref{sec:prelim}, we review the nuclear EFTs of relevance to this study, their representation on a discretized finite spatial cubic lattice, and the fermion-to-qubit mappings that we consider. 
Our methodology and results are summarized in \cref{sec:results-summary}, before complete discussions and derivations are offered in the subsequent sections. 
In \cref{sec:encoding-Hamiltonians}, we introduce the mapping of both pionless and pionful EFT interactions to Pauli operators. In \cref{sec:circuits}, we present the circuit decomposition of each step of the Trotterized time evolution in all the theories considered and estimate resource requirements for both near- and far-term quantum computing. In \cref{Sec:product_formulae}, we derive a new bound on the accuracy of the $p$th-order product formula using an improvement arising from fermion-number conservation, and apply this result to the pionless EFT. In \cref{sec:resource}, we analyze the full cost of the simulation, including time evolution and energy-spectroscopy costs. In \cref{sec:results}, we present a summary of our conclusions, along with remarks on further improvements and future directions. \crefrange{Sec:algorithms}{Sec:Analytic_Trotter_Bounds_Proof} supplement various discussions in the paper and provide detailed derivations of a number of results introduced in the main text.

\section{Preliminaries
\label{sec:prelim}
}

The goal of this section is to review basic aspects of nuclear physics and nuclear EFTs, as well as quantum-simulation theory of relevance to this work.
In particular, we motivate the set of nuclear Hamiltonians studied in this work, and give a brief overview of fermion-to-qubit mappings, product-formula methods for quantum simulation, and quantum phase estimation for energy spectroscopy. While parts of this section will likely be elementary to experts in the respective fields, the Section serves to set up the problem and introduce our notation.

\subsection{Nuclear Effective Field Theories
\label{sec:NEFT}}
The underlying interactions governing all nuclear phenomena are those of the Standard Model of particle physics: the strong and electroweak interactions. To calculate properties of nuclei from the Standard Model---in particular via QCD, the theory of the strong force---non-perturbative methods such as numerical Monte Carlo simulations are required. However, certain features of QCD allow for a more computationally tractable organization of hadronic and nuclear systems. The most consequential feature is perhaps confinement, the notion that the low-energy degrees of freedom in QCD are not quarks and gluons, but rather confined composite states of those constituents, called hadrons. The other significant feature of interactions in QCD is an approximate chiral symmetry, the property that the nearly massless left- and right-handed light quarks transform independently under a non-Abelian SU(2) quark-flavor symmetry. This symmetry breaks spontaneously in the vacuum, generating a set of (pseudo-)Goldstone bosons called pions, with masses much smaller than those of the other hadrons.\footnote{Pions would be strictly massless if the original chiral symmetry were exact. Due to a non-zero but small mass for the $u$ and $d$ quarks, chiral symmetry is explicitly broken, but only by a small perturbative amount.} The interactions of pions with themselves and with the other hadrons are greatly constrained because of this chiral symmetry breaking, since Goldstone bosons interact only with derivative couplings, so at low energies they become almost non-interacting. Chiral symmetry also relates various interactions' strengths as well as the couplings to external electromagnetic and weak currents of the Standard Model.

At low energies, a systematic expansion in powers of $Q/\Lambda_\chi$ and $m_\pi/\Lambda_\chi$, called chiral perturbation theory~\cite{gasser1984chiral, leutwyler1994foundations}, describes the interactions of pions among themselves and with other hadrons, including with the nucleons. Here, $Q$ is any intrinsic momentum in the process, e.g., the relative momentum of hadrons in a scattering process or the momentum transfer in the response of the hadron to an external probe, $m_\pi$ is the mass of the pion, and $\Lambda_\chi$ is the energy scale above which chiral perturbation theory breaks down, estimated to be around the mass of the $\rho$ resonance ($m_\rho \approx 770$~[MeV]~\cite{particle2022review}). Chiral perturbation theory is one of the most successful theories of hadronic physics. Once its interaction couplings, called low-energy constants, were constrained by experimental data in a few processes, the theory was used to make many non-trivial predictions for other processes, order-by-order in a momentum expansion~\cite{ecker1995chiral,scherer2011primer,scherer2003introduction}. 

However, for systems composed of two or more nucleons, chiral perturbation theory does not apply. In contrast to pions, whose interactions are governed by their (pseudo) Goldstone-boson nature, interactions of nucleons do not vanish at low energies. Furthermore, nucleons can interact strongly, hence the formation of atomic nuclei, which are bound states of protons and neutrons. Such features cannot be described by perturbation theory. Weinberg, nevertheless, developed an EFT  
that combines a perturbative nuclear potential with a non-perturbative solution to the corresponding Schr\"odinger equation to generate non-perturbative features such as bound states, and to compute scattering amplitudes~\cite{weinberg1990nuclear, weinberg1991effective}. Unfortunately, the convergence rate of the Weinberg scheme in some two-nucleon scattering channels is poor~\cite{beane2002towards}. Furthermore, due to the mixing of different perturbative orders in solving the Schr\"odinger equation, scattering observables computed within this scheme show sensitivity to the ultraviolet cutoff of the effective description~\cite{nogga2005renormalization,griesshammer2020consistency,van2020problem}. 
Kaplan, Savage, and Wise (KSW) came up with a strictly perturbative approach to compute observables, after non-perturbatively summing up the leading-order contact interactions of two nucleons~\cite{kaplan1996nucleon, kaplan1998new}. This approach fixes the convergence issue of the Weinberg approach in some channels, but fails to converge in channels in which the Weinberg scheme works well~\cite{fleming2000nnlo}. Despite the success of both the Weinberg and KSW schemes and their descendants in \emph{ab initio} nuclear many-body studies (i.e., those based on nucleonic degrees of freedom), 
and an enhanced understanding of their limitations, the search for the most reliable EFT of nuclear forces with pions continues~\cite{beane2002towards, nogga2005renormalization, birse2006power, epelbaum2013renormalization, valderrama2017power, wu2019perturbative, kaplan2020convergence,griesshammer2020consistency,van2020problem,hammer2020nuclear}.

At momenta much lower than the pion mass, another EFT, called pionless EFT, is applicable. In the pionless EFT, pions are integrated out and their effects are included only implicitly in the interactions between nucleons~\cite{kaplan1996nucleon, kaplan1998new, van1999effective, bedaque1999renormalization, chen1999nucleon, bedaque2002effective}. Pionless EFT has shown more robust convergence properties for a range of observables in two- and multi-nucleon systems~\cite{hammer2020nuclear}, but its range of validity is limited to rather small momenta.

Quantum computing has the potential to simulate nuclear systems that are out of reach of classical numerical methods. Capitalizing on the success of classical computing in addressing increasingly large nuclear isotopes using pionless and pionful chiral EFTs~\cite{hergert2020guided,tews2020new}, it is natural to develop quantum methods based on the same effective descriptions. Therefore, we adopt the pionless and pionful chiral effective field theories of nuclear forces as the starting point for our algorithmic analysis. This analysis is limited to leading-order interactions in the Weinberg power counting of the potential, in which both contact interactions of nucleons and the one-pion exchange contribution to the potential come at leading order. For the pionless EFT, beside the leading two-nucleon contact interactions, the three-nucleon contact interaction is further included, since the latter is necessary to properly renormalize the theory at leading order~\cite{bedaque1999renormalization, bedaque2002effective}. The interactions are then discretized on a spatial lattice of finite size to form a lattice nuclear EFT. While the continuum limit of such a theory is not well-defined (consistent with non-renormalizability of EFTs in general), the bulk limit can be taken at reasonably small lattice spacings, and discretization effects can be properly quantified and controlled~\cite{klein2015regularization, klein2018lattice}. Quantum algorithms for other formulations of the same problem, such as continuum and momentum-space methods~\cite{carlson2015quantum,barrett2013ab,hagen2014coupled,carbone2013self,hergert2016medium}, should also be developed to investigate their resource requirements, but we leave this to future work.

The leading-order Hamiltonian does not distinguish between neutrons and protons, nor between the three species of pions. This limit is called the isospin-symmetric point.\footnote{This arises from nearly equal masses of $u$ and $d$ quarks, and the blindness of strong interactions to the (different) electric charge of the two quarks.} While the explicit forms of the leading-order Hamiltonian on the lattice for both pionless and pionful effective field theories are provided in \cref{sec:encoding-Hamiltonians}, it is instructive to qualitatively introduce the various interaction terms that are in play. First, in the isospin-symmetric limit, the nucleons can be represented as a doublet in the so-called isospin space---a two-dimensional vector space associated with the internal flavor space of the nucleon, such that the upper isospin component of a nucleon is a proton and its lower component is a neutron. In other words, $N=\binom{p}{n}$, where for the proton $I=\frac{1}{2}$ and $I_3=\frac{1}{2}$, and for the neutron $I=\frac{1}{2}$ and $I_3=-\frac{1}{2}$, with $I$ and $I_3$ being the total isospin and its third Cartesian component, respectively.
Each proton and neutron, furthermore, is a doublet in the spin space, giving the spin-up proton $S=\frac{1}{2}$ and $S_3=\frac{1}{2}$, and the spin-down proton $S=\frac{1}{2}$ and $S_3=-\frac{1}{2}$, and similarly for the neutrons. Here, $S$ and $S_3$ are the total spin and its third Cartesian component, respectively. While the system of two nucleons at low energies, corresponding to an s-wave orbital angular momentum, can naively constitute four distinct states, corresponding to total isospin and spin ($I=0,S=0$), ($I=0,S=1$), ($I=1,S=0$), and ($I=1,S=1$), only the so-called isosinglet ($I=0,S=1$) and isotriplet ($I=1,S=0$) channels are allowed. 
This is due to the fact that nucleons are fermions, and by the Pauli exclusion principle, their total wavefunction must be antisymmetric under the exchange of the two nucleons. This results in only two independent two-nucleon low-energy constants,
denoted $C$ and $C_{I^2}$ in the Hamiltonians in \cref{eq:HC,eq:HCI2}, and depicted in \cref{fig:diagrams}a. 
In the pionless EFT, the three-nucleon interaction is given in \cref{Eq:H_contact_2} and depicted in \cref{fig:diagrams}b. Here, a single low-energy constant, $D$, is sufficient to ensure renormalizability at leading order~\cite{bedaque1999renormalization, bedaque2002effective}.
\begin{figure}[t!]
\centering
\includegraphics[scale=0.65]{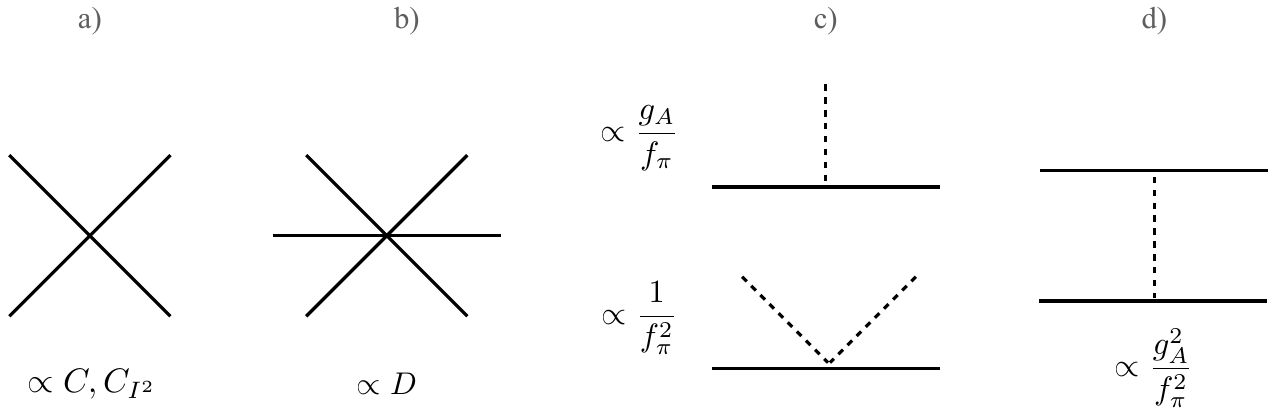}
\caption{Feynman diagrams that schematically represent the various interactions encountered at leading order in the pionless and pionful chiral EFTs of nuclear forces. Solid lines denote nucleons (of various spin and isospin flavors) and dashed lines are pions (of various isospin flavors). The full momentum (position), spin, and isospin dependence of the interactions should be deduced from the Hamiltonians given in the main text. a) and b) depict contact two- and three-nucleon forces, respectively. c) displays the axial-vector and the Weinberg-Tomozawa pion-nucleon couplings. d) depicts the OPE potential.}
\label{fig:diagrams}
\end{figure}

Besides the contact interactions in the pionful chiral EFT, nucleons interact with pions with a form that is constrained by chiral symmetry. As a result of this interaction, two nucleons can also interact by exchanging a pion at leading order in the chiral EFT, and by exchanging multiple pions at higher orders~\cite{machleidt2011chiral}. 
This interaction introduces non-trivial spin and isospin dependence into the nuclear force. At low energies, where the dynamics of pions can be neglected, a static pion potential can be considered, with a dependence on the distance between nucleons, $r$, that is Yukawa-like: $V^{\rm OPE}(r) \propto e^{-m_\pi r}/r$, where OPE stands for one-pion exchange. 
The form of this potential on the lattice is given in \cref{Eq:OPE_Long-Range_Terms} and is diagrammatically represented in \cref{fig:diagrams}d.
Alternatively, the pions can be included dynamically to keep the interactions local. Since pions are neutral and charged scalar fields, this case involves simultaneously simulating a coupled theory of bosonic and fermionic fields, as in \crefrange{Eq:Pion_Only_Hamiltonian}{Eq:H_WT}, and \cref{fig:diagrams}c). Pions can self-interact, but contributions from self-interactions of pions come at higher orders in the EFT and do not need to be simulated at leading order. In this work, we study both approaches to the inclusion of pions in nuclear EFT simulations.

The numerical values of the various constants in the nuclear EFT Hamiltonians of this work are summarized in \cref{Table:Table_of_Constants}. The values of low-energy constants are (energy) scale-dependent, and the relevant values, along with the volume and lattice-spacing values, are quoted in the corresponding sections for numerical cost evaluations. We work in the unit system in which $\hbar=c=1$, where $\hbar$ is the reduced Planck constant and $c$ is the speed of light.
\begin{table}[ht]
    \centering
    \begin{tabular}{c|c|c}
    \hline
    \textbf{Quantity} & \textbf{Symbol} & \textbf{Value} \\
     \hline Nucleon mass    & $M$  & 938~MeV \\
    Pion mass    & $m_{\pi}$ & 135~MeV \\
     Nucleon
axial charge  & $g_A$  & 1.26 \\
      Pion decay constant  & $f_{\pi}$ & 93~MeV
       \\ 
       \hline
    \end{tabular}
    \caption{The constant quantities used in this study and their (approximate) values. Although protons and neutrons, as well as the different species of pions, have slightly different masses, these differences can be neglected at leading order in the chiral nuclear forces.
    }
    \label{Table:Table_of_Constants}
\end{table}

\subsection{Setting up the Problem on a Spatial Lattice
\label{sec:lattice}}
One way to study nuclear EFTs using a digital computer is to discretize them on a spatial lattice~\cite{lahde2019nuclear,lee2009lattice}, as discussed in the preceding section. In this section, we introduce our lattice setup and the discretized degrees of freedom. Explicit Hamiltonians and their encodings into qubit systems will be presented in \cref{sec:encoding-Hamiltonians}.

Throughout this work, $D$ denotes the spatial dimension, where $D=3$ for the nuclear EFT Hamiltonians. The $L\times L \times L$ cubic lattice is denoted $\Lambda(L)$. The lattice spacing $a_L$ is typically in the range $1$--$2$~fm.
Where convenient, we use standard Cartesian coordinates $\bm{x}=(x_1,x_2,x_3)$ to denote a position on the lattice.

At each lattice site, operations can occur on four distinct degrees of freedom, corresponding to the spin-$1/2$ and isospin-$1/2$ internal space of the nucleon. We let $\sigma_S$ with $S\in \{1,2,3\}$ denote Pauli matrices acting on the spin space, and $\tau_I$ with $I\in \{1,2,3\}$ denote Pauli matrices acting on the isospin space,\footnote{These Pauli matrices should not be confused with the Pauli matrices acting on the Hilbert space associated with qubits, which are introduced later. The former Pauli matrices are introduced to organize the interactions among nucleons and pions in the spin and isospin spaces in a compact form, while the latter are the operators that act on the qubits in the quantum algorithms of this work.} where $\sigma_1=\tau_1=X$, $\sigma_2=\tau_2=Y$, and $\sigma_3=\tau_3=Z$. Furthermore, $[\sigma_S]_{\alpha \beta}$ denotes the $(\alpha,\beta)^{\textrm{th}}$ entry of the matrix $\sigma_S$, and $\sigma \cdot\sigma = \sum_{S=1}^3 \sigma_S\sigma_S$. Similar relations hold for $\tau_I$. The totally anti-symmetric tensor in both spaces is denoted $\epsilon_{\alpha\beta\gamma}$.

The fermionic annihilation and creation operators at site $\bm{x}\in \Lambda(L)$ for species $\sigma$ are denoted by $a_\sigma(\bm{x})$ and $a^{\dagger}_\sigma(\bm{x})$, respectively, where $\sigma$ runs over protons, neutrons, and their spin states: $\sigma \in\{\uparrow, \downarrow \}\times \{\text{proton}, \text{neutron}\} $. 
In other words,
\begin{align}
    a_{00}(\bm{x}) = a_{\uparrow p}(\bm{x}), \quad a_{01}(\bm{x}) = a_{\uparrow n}(\bm{x}), \quad
    a_{10}(\bm{x}) = a_{\downarrow p}(\bm{x}), \quad a_{11}(\bm{x}) = a_{\downarrow n}(\bm{x}).
\end{align}
Occasionally, the position argument $x$ may be left implicit. The hat notation on the operators will not be used, and the operator nature of symbols should be deduced from their context. The fermionic creation and annihilation operators satisfy
\begin{align} \label{Eq:Fermion_Anticommute}
    &\{a_\sigma(\bm{x}), a^{\dagger}_{\sigma'}(\bm{y})\} = \delta_{\sigma,\sigma'}\delta_{\bm{x},\bm{y}}, \\
    &\{a_\sigma(\bm{x}), a_{\sigma'}(\bm{y})\} = 0,  \\
    &\{\adag_\sigma(\bm{x}), \adag_{\sigma'}(\bm{y})\} = 0, 
\end{align}
where $\bm{x},\bm{y} \in \Lambda(L)$ and $\delta_{\bm{x},\bm{y}} \coloneqq \delta_{x_1,y_1} \delta_{x_2,y_2}  \delta_{x_3,y_3}$.
The number operator at site $\bm{x}$ is denoted by $N_{\sigma}(\bm{x}) = \crt(\bm{x})\annh(\bm{x})$.
The following ferminoic bilinear operators will also be used throughout: 
\begin{align}
    &\rho(\bm{x}) = \sum_{\alpha} \sum_{\beta } a^{\dagger}_{\alpha \beta}(\bm{x})a_{\alpha \beta}(\bm{x}) \label{Eq:Particle_Density_1}, \\
    &\rho_S(\bm{x}) = \sum_{\alpha,\gamma} \sum_\beta a^{\dagger}_{\alpha\beta}(\bm{x})[\sigma_S]_{\alpha \gamma}a_{\gamma\beta}(\bm{x}),
    \label{Eq:Particle_Density_2}\\
    &\rho_I(\bm{x}) = \sum_\alpha \sum_{\beta,\delta} a^{\dagger}_{\alpha\beta}(\bm{x}) [\tau_I]_{\beta\delta}a_{\alpha\delta}(\bm{x}),
    \label{Eq:Particle_Density_3}\\
    &\rho_{S,I}(\bm{x}) = \sum_{\alpha,\gamma} \sum_{\beta,\delta} a^{\dagger}_{\alpha\beta}(\bm{x}) [\sigma_S]_{\alpha\gamma}[\tau_I]_{\beta\delta}a_{\gamma\delta}(\bm{x}),
    \label{Eq:Particle_Density_4}
\end{align}
where $S$ and $I$ are run over the spin and isospin indices, respectively.

An operator $F$ is \emph{number preserving} if it is a sum of products of creation and annihilation operators, where each product has an equal number of creation and annihilation operators. For an operator $A$, $:A:$ is the normal-ordering operation, which places creation operators to the left of annihilation operators. We will also find it useful to define the following semi-norm, which is just the operator norm restricted to a subspace of a specified fixed number of fermions: 

\begin{definition}[Fermionic Semi-Norm (Section 2.3 of Ref.~\cite{su2021nearly})]\label{Def:Fermionic_Seminorm}
Let $X$ be a number-preserving operator and let $\ket{\psi_\eta}$ and $\ket{\phi_\eta}$ be normalized states with exactly $\eta$ fermions.
Then, the fermionic semi-norm of $X$ is
\begin{align}
    \norm{X}_{\eta} = \max_{\ket{\psi_\eta}, \ket{\phi_\eta}} |\bra{\psi_\eta}X \ket{\phi_\eta}|.
\end{align}
Furthermore, if $X$ is Hermitian, then
\begin{align}
    \norm{X}_{\eta} = \max_{\ket{\psi_\eta}} |\bra{\psi_\eta}X \ket{\psi_\eta}|.
    \label{Eq:norm-eta-def}
\end{align}
\end{definition}

Finally, we need a lattice-discretized formulation of pions. The pion field at lattice site $\bm{x}$ is represented by $\pi_I(\bm{x})$ for the isospin indices
$I\in \{1,2,3 \}$, and the corresponding conjugate-momentum field is denoted $\Pi_I(x)$.
The bosonic field operators satisfy the standard commutation relations
\begin{align}
    &[\pi_I(\bm{x}),\Pi_{I'}(\bm{y})]=\frac{1}{a_L^D}\delta_{I,I'}\delta_{\bm{x},\bm{y}}\mathds{1}, \label{eq:commutations-pi-1} \\
   &[\pi_I(\bm{x}),\pi_{I'}(\bm{y})]=0, \\
   &[\Pi_I(\bm{x}),\Pi_{I'}(\bm{y})]=0.
\label{eq:commutations-pi}
\end{align}
Note that, in this work, the pion and its conjugate field are treated as dimensionful quantities.

To realize the fermionic and bosonic operators with operators acting on qubits, one needs to find a mapping that preserves the relevant (anti)commutation relations. The encoding schemes used in this work are introduced in the following section.

\subsection{Encoding Fermions and Bosons in Qubits
\label{Sec:intro-to-mappings}}
Fermionic and bosonic Hamiltonians can be represented on a quantum computer by defining operators acting on qubits that maintain the necessary commutation or anticommutation relations.
This section outlines the encodings that we use in this work.

\subsubsection{Fermionic Field Encodings} \label{Sec:Fermionic_Encoding_Introduction}
A wide variety of fermionic encodings, i.e., methods of replicating the fermionic anticommutation relations with Pauli operators, have been developed for both classical- and quantum-computing applications. Formally, given a fermionic Hamiltonian $H$, an encoding corresponds to an isometry $V$ that defines a qubit Hamiltonian
\begin{align}
    \tilde{H} = VHV^\dagger. 
\end{align}
Then, perhaps restricting to an appropriate subspace, the Hamiltonians $H$ and $\tilde{H}$ are equivalent up to a unitary transformation, and the simulator Hamiltonian $\tilde{H}$ replicates the physics of $H$ within this subspace.
Different encodings accomplish this task by mapping fermionic operators to different Pauli operators and using various numbers of qubits per fermionic mode, yielding different simulation costs. For a Pauli operator $S=\bigotimes_{i=1}^N P_{i}$, where $P_i\in\{I, X, Y, X\}$, the weight of $S$ is defined as the number of non-identity $P_i$ operators. Lower-weight operators lead to shorter-depth simulation circuits, whereas a low number of qubits per fermionic mode reduces the overhead in the number of qubits needed for the simulation. Typically, these features must be balanced since optimizing for one of them may negatively impact the other~\cite{chien2020custom, Bringewatt_Davoudi_2022, chien2022optimizing}.

For instance,  the well-known Jordan-Wigner encoding~\cite{Jordan_Wigner_1928} uses only one qubit per fermionic mode, but requires Pauli operators of weight $O(L^{D-1})$ (the so-called Jordan-Wigner strings) on a $D$-dimensional spatial lattice.
Lower-weight schemes include the Bravyi-Kitaev scheme, with operators of weight $O(\log(L))$~\cite{Bravyi_Kitaev_2002}, its generalizations~\cite{Havicek_Troyer_Whitfield_2017, chien2020custom,Bringewatt_Davoudi_2022}, and others~\cite{Steudtner_Wehner_2018, Jiang_Kalev_Mruczkiewicz_Neven_2020}. If the interactions described by the fermionic Hamiltonian are physically local, there exist encodings that produce terms with only constant-weight interactions. Examples include the Bravyi-Kitaev superfast encoding~\cite{Bravyi_Kitaev_2002} and its generalizations~\cite{Kanav_Bravyi_Mezzacapo_Whitfield_2019}, the Verstraete-Cirac encoding~\cite{Verstraete_Cirac_2005}, Majorana loop stabilizer codes~\cite{Jiang_McClean_Babbush_Neven_2019}, the compact encoding~\cite{Derby_et_al_2021, Derby_Klassen_2021}, and bosonization schemes \cite{chen2020exact,chen2023equivalence}.
Finally, recent work has shown how to achieve essentially optimal fermionic encodings (relative to a certain cost function) for translation-invariant systems by searching over the space of possible encodings~\cite{chien2022optimizing}. 
Determining what this cost function should be for a particular task depends on multiple factors including, but not limited to, the locality of the Hamiltonian's interactions, whether the physical interactions preserve the number of fermions, the dimensionality of the physical space, the number of species of fermions, the architecture of the quantum computer, and constraints on circuit depth and qubit numbers.

In this paper, we choose an encoding that provides significant advantage over the common Jordan-Wigner encoding. 
In particular, we work primarily with the Verstraete-Cirac (VC) encoding~\cite{Verstraete_Cirac_2005}.
As previously mentioned, this encoding has the advantage that, for local interactions on a 3D lattice (e.g., hopping between neighboring lattice sites), the qubit Hamiltonian has constant-weight terms. 
Furthermore, the number of ancillary qubits introduced by the encoding depends only on the dimensionality of the lattice, not the number of species of fermions.
We also show that compact encoding of Ref.~\cite{Derby_et_al_2021} provides similar advantages for the pionless EFT.
The question of whether these fermionic encodings are ``optimal'' depends on multiple factors, including device architecture, availability of qubits versus available circuit depth, and many other factors.

\paragraph{The Verstraete-Cirac encoding.}
As mentioned before, the Jordan-Wigner strings are inherently non-local.
The fundamental idea underlying the VC encoding is that, by adding unphysical ``auxiliary fermions'' to all physical sites, one can modify the Hamiltonian such that the Jordan-Wigner strings associated with these auxiliary fermions cancel out the Jordan-Wigner strings of the physical fermions on the same site, making many of the previously non-local interactions local.
When restricting to a particular subspace of these auxiliary fermions on which the new, modified Hamiltonian acts on invariantly, one obtains a qubit Hamiltonian that acts on the physical part of the Hilbert space in a way that preserves the relevant physics (see Refs.~\cite{Verstraete_Cirac_2005, Bausch_2020} for detailed reviews). The simulation must be initiated in the proper subspace, adding additional state-preparation cost, which can be performed in $O(1)$ circuit depth (see \cref{Sec:VC_Stabilizers}). We do not analyze this cost further, instead focusing on resource estimates for time evolution and spectroscopy.

Nuclear-EFT Hamiltonians involve four distinct fermion species, corresponding to (non-relativistic) protons and neutrons, with two spin states each. On a cubic lattice in $D=3$ spatial dimensions, the VC encoding uses $\lceil D/2 \rceil=2$ additional qubits per spatial lattice site to maintain locality.
This gives a total of $4+2=6$ qubits per physical site on the lattice, i.e., 1.5 qubits per fermionic mode per site.
The two auxiliary fermionic modes corresponding to physical spatial site $i$ will be labeled as $\mu$ and $\nu$, where $i$ runs from $1$ to $L^3$. The fermion species in the Jordan-Wigner string are then labeled at a particular lattice site as $\uparrow p, \downarrow p, \uparrow n, \downarrow n, \mu, \nu$, as shown in \cref{Fig:JW_String}.

\begin{figure}[ht]
\centering
\includegraphics[width=0.8\textwidth]{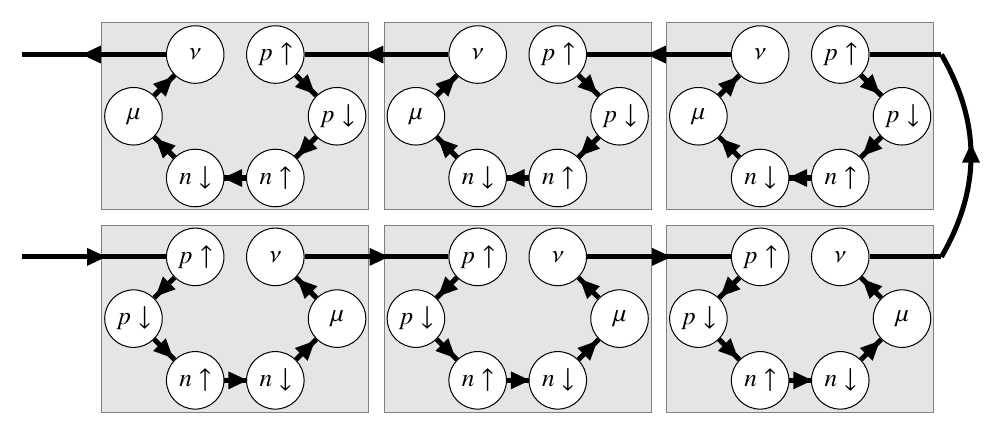}
\caption{A sketch of how the Jordan-Wigner string appears in the $x$-$y$ plane for the VC encoding.
Shaded rectangles denote spatial sites on the lattice, while circles represent qubits used to encode fermions. The physical qubits are labeled by the fermions they represent, and the auxiliary fermions are denoted by $\mu$ and $\nu$.
The string is purely illustrative of the order in which the sites appear in the fermionic operators and has no physical interpretation. A similar pattern is assumed once all sites at a given sheet in the $x$-$y$ plane are traversed and the string connects to a site in the neighboring sheet along the $z$ direction, as shown in \cref{Fig:VC_Auxiliary_Terms_3D} in \cref{Sec:VC_Stabilizers}.}
\label{Fig:JW_String}
\end{figure}

We denote the Pauli $Z$ operators at site $i$ on the respective qubits as $Z_i^{\uparrow p},Z_i^{\downarrow p},Z_i^{\uparrow n},Z_i^{\downarrow n},Z_i^\mu,Z_i^\nu$, with similar notation for the Pauli $X$ and $Y$ operators. The qubit representations $\tilde{a}_\sigma$ of the fermionic annihilation operators $a_\sigma$ for species $\sigma$ are
\begin{align}
    \label{eq:a-tilde-up-p}
    \tilde{a}_{\uparrow p}(j) &= \frac{1}{2} \left(\prod_{i<j} Z_i^{\uparrow p}Z_i^{\downarrow p}Z_i^{\uparrow n}Z_i^{\downarrow n}Z_{i'}Z_{i''}\right) (X_j^{\uparrow p}+iY_j^{\uparrow p}), \\
    \label{eq:a-tilde-down-p}
    \tilde{a}_{\downarrow p}(j) &= \frac{1}{2}\left(\prod_{i<j} Z_i^{\uparrow p}Z_i^{\downarrow p}Z_i^{\uparrow n}Z_i^{\downarrow n}Z_{i'}Z_{i''}\right) Z_j^{\uparrow p}(X_j^{\downarrow p}+iY_j^{\downarrow p}), \\
    \label{eq:a-tilde-up-n}
    \tilde{a}_{\uparrow n}(j) &=\frac{1}{2} \left(\prod_{i<j} Z_i^{\uparrow p}Z_i^{\downarrow p}Z_i^{\uparrow n}Z_i^{\downarrow n}Z_{i'}Z_{i''}\right) Z_j^{\uparrow p}Z_j^{\downarrow p}(X_j^{\uparrow n}+iY_j^{\uparrow n}), \\
    \label{eq:a-tilde-down-n}
    \tilde{a}_{\downarrow n}(j) &= \frac{1}{2}\left(\prod_{i<j} Z_i^{\uparrow p}Z_i^{\downarrow p}Z_i^{\uparrow n}Z_i^{\downarrow n}Z_{i'}Z_{i''}\right) Z_j^{\uparrow p}Z_j^{\downarrow p}Z_j^{\uparrow n}(X_j^{\downarrow n}+iY_j^{\downarrow n}),
\end{align}
with creation operators defined straightforwardly by Hermitian conjugation.

We now introduce the Majorana operators of the two auxiliary fermions associated with site $j$, defined as
\begin{align}
\label{eq:mu-def-I}
    &\tmu(j) = \left(\prod_{i<j} Z_i^{\uparrow p}Z_i^{\downarrow p}Z_i^{\uparrow n}Z_i^{\downarrow n}Z_{i'}Z_{i''}\right)Z_j^{\uparrow p}Z_j^{\downarrow p}Z_j^{\uparrow n}Z_j^{\downarrow n}X_j^\mu,
    \\
\label{eq:mu-def-II}
    &\Bar{\tmu}(j)=\left(\prod_{i<j} Z_i^{\uparrow p}Z_i^{\downarrow p}Z_i^{\uparrow n}Z_i^{\downarrow n}Z_{i'}Z_{i''}\right)Z_j^{\uparrow p}Z_j^{\downarrow p}Z_j^{\uparrow n}Z_j^{\downarrow n}Y_j^\mu, \\
\label{eq:nu-def-I}
    &\tnu(j) = \left(\prod_{i<j} Z_i^{\uparrow p}Z_i^{\downarrow p}Z_i^{\uparrow n}Z_i^{\downarrow n}Z_{i'}Z_{i''}\right)Z_j^{\uparrow p}Z_j^{\downarrow p}Z_j^{\uparrow n}Z_j^{\downarrow n}Z_j^\mu X_j^\nu,
    \\
\label{eq:nu-def-II}
    &\Bar{\tnu}(j)=\left(\prod_{i<j} Z_i^{\uparrow p}Z_i^{\downarrow p}Z_i^{\uparrow n}Z_i^{\downarrow n}Z_{i'}Z_{i''}\right)Z_j^{\uparrow p}Z_j^{\downarrow p}Z_j^{\uparrow n}Z_j^{\downarrow n}Z_j^\mu Y_j^\nu.
\end{align}
The goal is to then use these auxiliary Majorana fermions to keep terms of the form $\adag_{\sigma}(i)a_{\sigma}(j) + \adag_{\sigma}(j)a_{\sigma}(i)$ when mapped to a qubit Hamiltonian.
It is possible to restrict to a subspace in which, along certain paths on the lattice, the relation $i\tmu(i)\Bar{\tmu}(j)\ket{\psi} = \ket{\psi} $ holds for nearest-neighbor pairs $(i,j)$. 
It can then be checked that the following mapping of hopping terms becomes local while preserving the operator on the physical fermions:
\begin{align}
    \adag_{\sigma}(i)a_{\sigma}(j) + \adag_{\sigma}(j)a_{\sigma}(i) &\rightarrow \left(\tcrt(i)\tannh(j) + \tcrt(j)\tannh(i)\right)i\tmu(i)\Bar{\tmu}(j).
\end{align}
Provided one is in the relevant restricted subspace, the action of the right-hand side can be seen to be unchanged by the inclusion of $i\tmu(i)\Bar{\tmu}(j)$.
A similar construction is then utilized for the $i\tnu(i)\Bar{\tnu}(j)$ along a different set of paths on the lattice.
By properly choosing these paths, this construction allows mapping all nearest-neighbor fermionic terms in the physical Hamiltonian to nearest-neighbor interactions in the qubit Hamiltonian. 
See \cref{Sec:VC_Stabilizers} for more details and examples concerning the mapping itself, and \cref{sec:encoding-Hamiltonians} for its use in simulating nuclear EFTs.

\paragraph{The cubic compact encoding.}
For the pionless EFT, we investigate the three-dimensional version of the compact encoding \cite[Sec.~7]{Derby_et_al_2021} and show that it can somewhat reduce circuit depths, at the expense of using more qubits.
The underlying architecture is a cubic lattice, where some faces of the cubic unit cells have an additional ``auxiliary'' qubit embedded in them.

Working with Majorana operators $\gamma(\bm{x}) = a(\bm{x}) + \adag(\bm{x})$ and $i\bar{\gamma}(\bm{x}) = a(\bm{x}) - \adag(\bm{x})$, one can construct edge and vertex operators, where the edge operators act between nearest-neighbor vertices $(\bm{x},\bm{y})$. These satisfy
\begin{align}
    E(\bm{x},\bm{y}) = -i\gamma(\bm{x}) \gamma(\bm{y}), \quad \quad V(\bm{x}) = -i \gamma(\bm{x}) \bar{\gamma}(\bm{x}).
\end{align}
Furthermore, the edge operators satisfy a non-local condition  for any closed loop of fermionic modes $\ell = (\bm{x}_1, \bm{x}_2, \dots , \bm{x}_{|\ell|})$, with $\bm{x}_1 = \bm{x}_{|\ell|}$:
\begin{align}
\label{Eq:Loop_Condition}
    i^{|\ell|} \prod_{\bm{x}=\bm{x}_1}^{\bm{x}_{|\ell|-1}}E(\bm{x},\bm{x}+1)=\mathds{1}.
\end{align}

The edge and vertex operators suffice to construct the operators appearing in a fermionic parity-preserving Hamiltonian.
To represent these operators on qubits, one needs to choose a qubit representation which satisfies all the (anti)commutation relations between $E(\bm{x},\bm{y})$ and $V(\bm{x})$, as well as the loop condition in \cref{Eq:Loop_Condition}.
A possible choice in terms of Pauli operators is
\begin{align}
    &\tilde{V}(i) = Z_i,   \label{Eq:V_j_Paulis} \\ 
    &\tilde{E}(i,j) = X_iY_jP(i,j).\label{Eq:E_ij_Paulis}
\end{align}
Here, $i$ ($j$) is the qubit index associated with site $\bm{x}$ ($\bm{y}$), and $P(i,j)$ is a Pauli string of weight at most 2 acting on auxiliary qubits, which depends on the direction of the edge operator $\tilde{E}(i,j)$.
We leave detailed definition of $P(i,j)$ to Ref.~\cite{Derby_Klassen_2021} and simply note that $\tilde{E}(i,j)$ is a Pauli string of weight at most 4.
The compact encoding on a cubit lattice uses at most 2.5 qubits per fermionic mode.
As will be shown later, we will only be interested in compact encoding for a \emph{single species} of fermions.
Hence, unlike the VC encoding, one does not need to consider how to embed multiple species simultaneously.

\subsubsection{Bosonic Field Encodings
\label{Sec:bosonic-field-encodings}}
\paragraph{Field representation.}

When simulating EFTs with explicit bosonic degrees of freedom, i.e., the spin-0 pion fields, we need an encoding for the bosons.
Following similar schemes as in Refs.~\cite{Jordan_Lee_Preskill_2012, Jordan_Lee_Preskill_Scattering_2014,Jordan_Hari_Lee_Preskill_2018, Klco_Savage_2019} for digitizing scalar field theories, we primarily work with Hamiltonians for which the pion fields $\pi_I(\bm{x})$ are represented in the field basis in position space (sometimes called the JLP basis).
This basis choice is motivated since i) all of the interactions in the Hamiltonians of this work are spatially local, and ii) many of them depend on $\pi_I(\bm{x})$ or $\pi_I^2(\bm{x})$, which are diagonal in this basis, resulting in circuits with lower gate complexity.

Since only a finite number of degrees of freedom can be encoded digitally, one must impose a cutoff on bosonic Hilbert spaces. Specifically, we put an upper bound $\pimax$ on the pion-field strength, so that $-\pi_{\rm max} \leq \pi_I(\bm{x}) \leq \pi_{\rm max}$, and introduce a digitization scale denoted by $\delta_{\pi}$, such that $2\pi_{\rm max}$ is an odd multiple of $\delta_\pi$.
Explicitly, for every lattice site $\bm{x}$, we introduce an operator which can be written in a diagonal basis as
\begin{align}
\label{Eq:Field_Operator}
    \pi_I = \sum_{k=0}^{2\pimax/\delta_\pi}
    \lambda_k \ket{k}\bra{k},
\end{align}
where $\lambda_k = -\pimax +\dpi k $ increases in increments of $\dpi$ so that there are $2\pimax/\delta_{\pi}+1$ distinct eigenvalues.
The errors introduced by the digitization and the cutoff are characterized in \cref{Sec:Pion_Cutoff}.

The digitized field operator is encoded by representing the eigenstates $\ket{k}$ with integers in either a unary or binary encoding.
We focus on the latter, which reduces qubit counts without significantly increasing gate counts. In particular, a binary choice uses $n_b=
\log_2(2\pimax/\delta_\pi+1)
$ qubits to encode the field strength, where we have assumed that $\delta_\pi$ is chosen such that $n_b$ is an integer, that is $\delta_\pi = 2\pi_\text{max}/(2^{n_b}-1)$~\cite{Klco_Savage_2019} (see also e.g., the encoding of the electric field used in Ref.~\cite[Sec 3.2]{shaw2020quantum} and in Ref.~\cite[Eq.~(56)]{davoudi2022general}). Explicitly, for a spatial site $\bm{x}$, and a pion field of species $I$,
\begin{align}
\label{Eq:Field_Operator_2}
 \pi_I(\bm{x})=-\pi_{\rm max}\mathds{1}+\frac{\delta_\pi}{2} \left ( (2^{n_b}-1)\mathds{1}-\sum_{m=0}^{n_b-1} 2^m Z^{(m)}_{\bm{x}} \right),
\end{align}
where $Z^{(m)}_{\bm{x}}$ is the Pauli $Z$ operator acting on the $m^{\textrm{th}}$ qubit at site $\bm{x}$, and $m=0$ represents the least significant bit of a positive integer.
This gives an explicit realization of \cref{Eq:Field_Operator} where computational basis states with binary representation of $k$ have an eigenvalue $-\pimax + \delta_\pi k$.
This scheme avoids dealing explicitly with encoding the sign and with performing signed arithmetic in binary. 

\paragraph{Conjugate-momentum representation.}
The conjugate momentum $\Pi_I(\bm{x})$ to a pion field $\pi_I(\bm{x})$ satisfies the canonical commutation relations in \cref{eq:commutations-pi}. As a result, the basis in which $a_L^3\Pi_I(\bm{x})$ is diagonal is the Fourier transform of the basis in which $\pi_I(\bm{x})$ is diagonal~\cite[Proposition 1]{Jordan_Lee_Preskill_2012}.
Let $U_{\rm QFT}^{(I)}$ be the unitary representing the quantum Fourier transform (QFT) acting on  $\mathbb{Z}_{2\pimax/\delta_\pi+1}$. 
Then,
\begin{align}
    \tilde{\Pi}_I(\bm{x}) = U_{\rm QFT}^{(I)} \, \Pi_I(\bm{x}) \, {U_{\rm QFT}^{(I)}}^\dagger,
\end{align}
and the momentum operator $\tilde{\Pi}_I$ is diagonal in this basis:
\begin{align}\label{Eq:Field_Operator_3}
    \tilde{\Pi}_I = \sum_{k=0}^{2\Pi_{\max{}}/\delta_\Pi } \tilde{\lambda}_k \ket{k}\bra{k},
\end{align}
where $\tilde{\lambda}_k = -\Pimax + \delta_{\Pi}k$ 
and $\delta_\Pi$ and $\Pi_{\max{}}$ are defined as~\footnote{Throughout this manuscript, $\pi$ without subscript and argument is simply the mathematical constant Pi.} 
\begin{align}
   &\delta_\Pi = \frac{2\pi}{a_L^3\dpi 
    (2\pi_{\rm max}/\delta_\pi+1)}, \qquad
    \Pi_{\max{}} = \frac{\pi }{a^3_L\dpi},
\end{align}
respectively, as per Eq. (10) of Ref.~\cite{Klco_Savage_2019} or Proposition 1 of Ref.~\cite{Jordan_Lee_Preskill_2012}. Finally, the momentum operator can be written as a sum of Pauli operators in a binary encoding as
\begin{align}
 \tilde{\Pi}_I(\bm{x}) = 
 -\frac{\pi }{a^3_L\dpi}+ \frac{\pi}{a_L^3 \delta_\pi (2\pi_{\rm max}/\delta_\pi+1)}\left ((2^{n_b}-1)\mathds{1}-\sum_{m=0}^{n_b-1} 2^m Z^{(m)}_{\bm{x}} \right).
\label{Eq:Pi-tilde}
\end{align}

\section{Summary of Methods and Results
\label{sec:results-summary}}
\noindent
To aid navigating the main strategies and results of this paper, we briefly outline the techniques we use to simulate nuclear EFTs and the various tradeoffs made. We also describe several general results we derive concerning non-relativistic fermionic simulations, and summarize the scaling of our quantum algorithms for simulating various formulations of nuclear-EFT Hamiltonians. The subsequent sections present these results in detail, and the full resource costs for certain tasks are presented in \cref{sec:resource}.

\paragraph{Representation of the Hamiltonian.}
Throughout, we use real-space representations of nuclear EFT Hamiltonians on a discretized lattice in the second-quantization formulation.
The fact that both the kinetic and interaction terms in all the EFTs have some notion of spatial locality suggests that gate counts are likely to be smaller for the real-space representation.
Our choice to work in second quantization rather than first quantization is motivated by similarity to the Fermi-Hubbard model, for which much optimization has been done in the second-quantization formulation \cite{Clinton_Bausch_Cubitt2021}.
Furthermore, there is an additional cost associated with antisymmetrization in the first quantization approach, which must be analyzed carefully to compare the cost of first and second quantization for nuclear-EFT simulations. 
That being said, there may be benefits to the first-quantization approach, as studied in the context of quantum-chemistry simulations~\cite{su2021fault}.

\paragraph{Fermionic encodings.}
To simulate fermions on a quantum computer, one must implement fermionic exchange statistics using qubits.
The commonly used Jordan-Wigner and Bravyi-Kitaev encodings incur large gate overheads, particularly as the number of fermionic modes increases.
Here, we exploit the locality of nuclear EFT Hamiltonians and the fact that they preserve the total number of fermions to apply the Verstraete-Cirac encoding.
This allows the hopping terms to be implemented in $O(1)$ depth, independent of the system's volume.
For the pionless EFT, the fact that the Hamiltonian does not mix different species allows us to apply the compact encoding in a way that improves the circuit depths even further.

\paragraph{Parallelizable circuit implementation.}
Our use of local fermionic encodings means that, for all the nuclear EFTs we study, the interactions are spatially local in the sense that their qubit representation only acts on the qubits representing the fermionic modes acted on by the Hamiltonian.
This allows all of the interaction terms to be highly parallelized, i.e., implemented on disjoint sets of qubits simultaneously, giving circuit depths independent of the number of fermionic modes.

\paragraph{Truncating long-range interactions.}
One of the EFT Hamiltonians we consider, the one-pion exchange EFT, has long-range interactions that decay exponentially with distance.
Not only are these interactions complicated by the presence of different spin, isospin, and orbital angular-momentum structures, but there are also many such interactions to implement.
We take advantage of the rapid decay of the interactions, characterized by the Compton wavelength of the pion, to truncate the interaction range and bound the associated error. This is a source of systematic error in the algorithm and is taken into account in assessing the final simulation cost.

\paragraph{Truncating pions' Hilbert spaces.}
For EFTs with explicit pions, one must choose a representation for the bosonic field and introduce a finite cutoff to the (otherwise unbounded) Hilbert space. We choose the discretized-space field basis used by Jordan, Lee, and Preskill~\cite{Jordan_Lee_Preskill_2012} for a scalar field theory, as each of the three isospin components of the pion field can be naturally expressed in this basis. 
We employ techniques similar to Ref.~\cite{Jordan_Lee_Preskill_2012} to truncate the pion-field strength, where we now account for the presence of nucleons.
These are energy-based constraints that are used to determine the finite cutoff on the field strength, and the corresponding digitization scale, bounding the energy expectation values in any state in the theory.
This cutoff and the digitization scale then impact the resource requirements of the simulation, and have been taken into account in our analysis.

\paragraph{Error bounds for product formulae.}
We show how one can exploit properties of fermionic Hamiltonians to reduce the Trotter error, and in the process prove the first Trotter bounds in terms of fermion number for general Hamiltonians. 
More specifically, we prove general upper bounds on the error associated with a class of fermionic Hamiltonians, including nuclear EFT Hamiltonians.
Notably, by taking advantage of translation invariance, locality, and particle-number conservation, we prove the following:
\begin{quoting}
\vspace{-\baselineskip}
\begin{theorem}[Informal Statement of \cref{Theorem:General_Order_Fermionic_Error_2}: Trotter Error for Fermion-Only Systems] \label{Theorem:General_Error_Main_Results}
    Let $H=\sum_\gamma H_\gamma$, where $H_\gamma$ are translation-invariant terms that act only on fermionic modes, preserving the total number of fermions, and suppose $e^{-itH_\gamma}$ can be implemented with an efficient circuit.
    Consider the evolution of a state with exactly $\eta$ fermions.
    Then, the error from the $p$th-order product formula, $\mathcal{P}_p(t)$, is
    \begin{align}
        \norm{\mathcal{P}_p(t) - e^{-itH}}_\eta \leq C_pt^{p+1}\eta,
    \end{align}
    where $C_p$ is a factor depending on $p$, the spectral norms of the local terms, and the degree of locality of the Hamiltonian.
\end{theorem}
\end{quoting}
This gives $O(\eta)$ scaling in the error, or $O(\eta^{1/p})$ scaling in the number of Trotter steps to implement the formula for a fixed evolution time and allowed error, which is \emph{independent of the lattice size and the number of fermionic modes}.
Although such a result has been achieved for the Fermi-Hubbard model in Refs.~\cite{su2021nearly, Clinton_Bausch_Cubitt2021, schubert2023trotter}, our results simultaneously i) apply to a more general class of Hamiltonians than the Fermi-Hubbard models studied in these works, and ii) explicitly compute the prefactor $C_p$. 

Additionally, for EFTs that explicitly include bosons, we have the following.
\begin{quoting}
\vspace{-\baselineskip}
\begin{theorem}[Informal Statement of \cref{Theorem:Pionful_Asymptotics}: General Error Bound] \label{Theorem:General_Error_Main_Results_Bosons}
    Let $H=\sum_\gamma H_\gamma$, where $H_\gamma$ are terms of the form of the dynamical-pion EFT considered in this work, which acts on both fermionic and bosonic modes, preserving the number of fermions. 
    Consider the evolution of a state with exactly $\eta$ fermions.
    Then, the error from the $p$th-order product formula, $\mathcal{P}_p(t)$, is
    \begin{align}
        \norm{\mathcal{P}_p(t) - e^{-itH}}_\eta \leq C_p\pimax^{p+1}\Pimax^{p+1}L^{3}t^{p+1},
    \end{align}
    where $C_p$ is a factor depending on $p$ and the degree of locality of the Hamiltonian, and $\pimax$ and $\Pimax$ are the cutoff values for the bosonic field strength and its conjugate momentum, respectively. Furthermore, $\pimax$ and $\Pimax$ scale as
    \begin{align}
        \pimax = O\left(\sqrt{\frac{\eta L^3E}{\epsilon}}  \right), \quad \quad 
        \Pimax =O\left(\sqrt{\frac{\eta L^3E}{\epsilon}}  \right),
    \end{align}
    where $\epsilon$ is the target precision and $E$ is the energy scale.

\end{theorem}
\end{quoting}

\paragraph{Explicit error-bound calculation.}
The above theorems give bounds for general Hamiltonians satisfying the relevant conditions, and the scaling in the number of fermions is likely optimal.
However, the combinatorial prefactor $C_p$ may be very loose as the formula only accounts for the locality of the interaction rather than its explicit form. 
By explicitly considering the Hamiltonians for the nuclear EFTs of this work and computing the commutators of the terms in the sum, we achieve much tighter bounds. In particular, for the $p=1$ product-formula simulation of the pionless EFT, we improve the bound by a factor of about $10^6$ by evaluating and bounding the relevant commutators explicitly.

\paragraph{Asymptotic simulation cost of nuclear EFTs.}
\Cref{Table:Asymptotic_Resources} gives the asymptotic scaling of the 2-qubit circuit depth, $T$-gate count, and number of qubits to simulate constant-time evolution of the different EFTs for a certain set of parameters, collected from the analyses of this work.

\begin{table}[ht]
    \centering
  \resizebox{\textwidth}{!}{  \begin{tabular}{ | c | c| c | c |} 
\hline 
& \multicolumn{3}{c |}{\textbf{Scaling of Resources for Fixed-Time Evolution}}  \\ 
  \hline
   & \textbf{Circuit Depth} & \textbf{$T$-Gate Count} & \textbf{Number of Qubits}  \\
  \hline
  \textbf{Pionless EFT }& $O\bigg(\frac{\eta^{1/p}}{\epsilon^{1/p}}\bigg)$ & $O\bigg(\frac{\eta^{1/p}L^3\log(\eta^{1/p}L^3/\epsilon^{1+1/p})}{\epsilon^{1/p}}\bigg)$ & $O(L^3)$ \\ 
  \hline
  \textbf{One-Pion Exchange} & $O\bigg(\frac{\eta^{1/p}\log^3(\eta/\epsilon)}{\epsilon^{1/p}}\bigg)$ & $O\bigg(\frac{\eta^{1/p}L^3\log^3(\eta/\epsilon)\log\left(\eta^{1/p}L^3\log^3(\eta/\epsilon)/\epsilon^{1+1/p}\right)}{\epsilon^{1/p}}\bigg)$ & $O(L^3)$ \\ 
  \hline
  \textbf{Dynamical Pions} & $
  O\bigg( \frac{E^3 L^{12}\eta^{9/2}
  n_b^2
  }{\epsilon^{{3+1/p}}}  \bigg)$ & $
  O\bigg( \frac{E^3 L^{15}\eta^{9/2}
  n_b^2 \log\left(E^3 L^{15}\eta^{9/2}n_b^2/\epsilon^{4+1/p}\right)
  }{\epsilon^{{3+1/p}}}  \bigg)$  & $O\left(L^3
  n_b\right)$ \\ 
  \hline
\end{tabular}}
\caption{Scaling of resources for constant-time simulation using the $p$th-order product formula in terms of fermion number $\eta$, energy scale $E$, precision $\epsilon$, and lattice size $L$. Here,
$n_b$ is the size of each bosonic qubit register, which asymptotically scales as $\log( \eta L^3 E/\epsilon)$.
In all cases, the resources scale as $t^{1+1/p}$, where $t$ is the total evolution time. }
\label{Table:Asymptotic_Resources}
\end{table}

\section{Nuclear EFT Hamiltonians and their Qubit Encodings}

 \label{sec:encoding-Hamiltonians}
In this section, the discretized EFT Hamiltonians of this work are explicitly represented in terms of Pauli operators on qubits, using the mappings introduced in \cref{Sec:intro-to-mappings}. The interactions on a 2D representative plane of the 3D lattice are depicted schematically in \cref{Fig:Interactions_Diagram}.

\begin{figure}[h!t]
\centering
\includegraphics[scale=0.6]{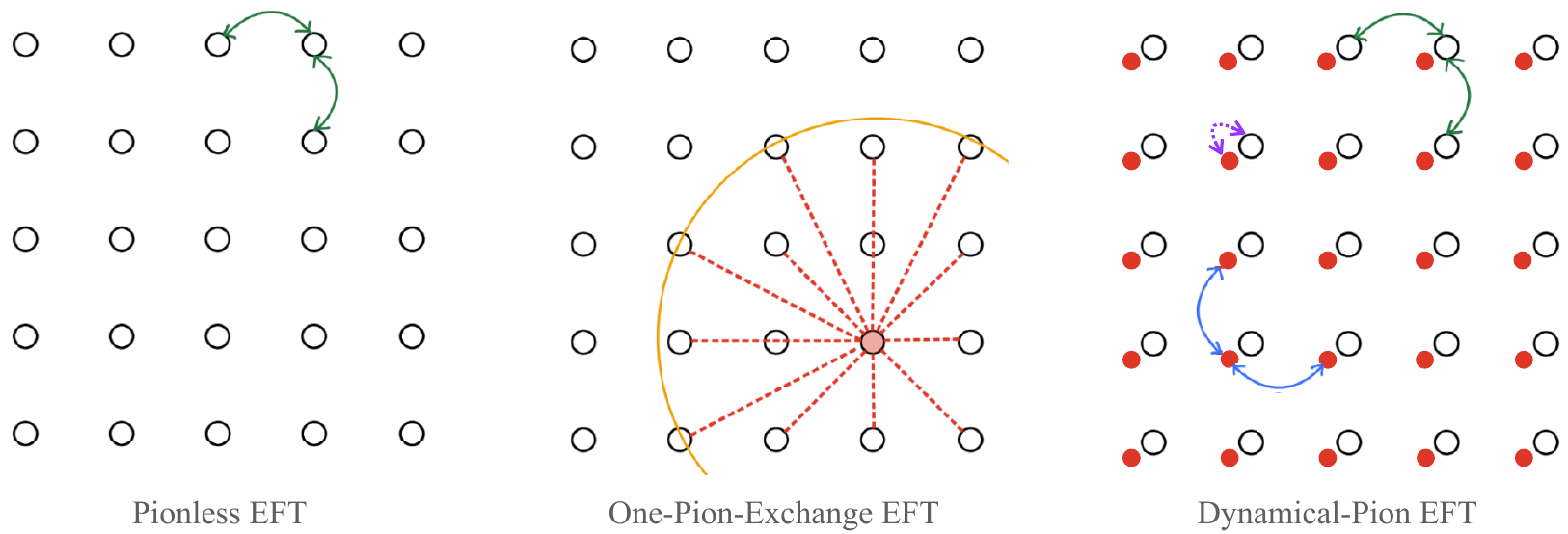}
\caption{
Schematic representation of various interactions in the different EFTs on a representative 2D plane of the 3D lattice. Hollow black circles denote fermionic sites. Red dots indicate the pion sites (if present in the theory). The pion lattice points are shifted slightly compared to nucleon lattice points for visual aid.
Pions are only able to interact with the nucleons on the same spatial site, denoted by the purple dashed line. In all the EFTs, the nucleons are able to move between lattice sites. Only in the dynamical-pion EFT are the pions both present and able to move. In the one-pion exchange EFT, the interactions are denoted by dashed red lines and the interaction-range cutoff is denoted by an orange circle centered around any given nucleon (here a representative nucleon site is denoted by a filled pink circle).
}
\label{Fig:Interactions_Diagram}
\end{figure}

\subsection{The Pionless-EFT Hamiltonian}\label{Sec:Pionless_EFT_Hamiltonian}
We start with the simplest Hamiltonian representing interactions among the nucleons.
At low energies, the lattice-EFT Hamiltonian involves only the propagation of nucleons on the spatial lattice plus short-range (contact) interactions between the nucleons~\cite
{hammer2020nuclear}. We use the form of the Hamiltonian given in Ref.~\cite{roggero2020quantum}, which assumes that the isotriplet and isosinglet scattering lengths are the same, hence only a single leading-order low-energy constant is sufficient for each of the two and three-nucleon contact interactions. Then, the Hamiltonian consists of three contributions:
\begin{align}\label{Eq:pionless_Hamiltonian}
    H_{\canpi} =  \Hfree + H_{\Cpi} + H_{\Dpi}, 
\end{align}
where $H_\mathrm{free}$ describes free fermions and $H_{\Cpi}$ and $H_{\Dpi} $ are on-site interaction terms:
\begin{align}
     \Hfree &= -h\sum_{\langle \bm{x},\bm{y} \rangle} \sum_{\sigma}\left( \crt(\bm{x})\annh(\bm{y}) + \crt(\bm{y})\annh(\bm{x}) \right) + 6h\sum_{\bm{x}}\sum_\sigma \fNum(\bm{x}), \label{Eq:H_free_Hamiltonian} \\
     H_{\Cpi} &= \frac{\Cpi}{2}\sum_{\bm{x}}\sum_{\sigma\neq \sigma'} N_\sigma(\bm{x})N_{\sigma'}(\bm{x}),\label{Eq:H_contact_1} \\
      H_{\Dpi} &=\frac{\Dpi}{6}\sum_{\bm{x}}\sum_{\sigma\neq \sigma'\neq \sigma''} N_\sigma(\bm{x})N_{\sigma'}(\bm{x})N_{\sigma''}(\bm{x})\label{Eq:H_contact_2}.
\end{align}
Here, $h \coloneqq \frac{1}{2M {a_L^2}}$, $\braket{\bm{x},\bm{y}}$ denotes nearest-neighbor points on the three-dimensional lattice, and $\Cpi$ and $\Dpi$ are low-energy constants which are constrained by fitting to scattering data, with values for the two different lattice spacings we consider in this paper given in \cref{Table:Pionless_Parameters}.

\begin{table}[h!]
    \centering
\begin{tabular}{c|c|c|c}
  &  $h$ [MeV] & $\Cpi$ [MeV] & $\Dpi$ [MeV]\\
    \hline 
   $a_L=1.4$~fm & 10.58 & -98.23 & 127.84 \\
   \hline 
   $a_L=2.2$~fm & 4.29 & -40.19 & 42.51
\end{tabular}
    \caption{Parameter values for the pionless-EFT Hamiltonian with lattice spacing $a_L=1.4$~fm and $a_L=2.2$~fm, taken from Ref.~\cite{roggero2020quantum} and Ref. \cite{rokash2013finite}, respectively. 
    (The values of these low-energy constants vary with the lattice scale.)}
    \label{Table:Pionless_Parameters}
\end{table}

\subsubsection{Encoding the Free-Fermion Terms (Verstraete-Cirac Encoding)} \label{Sec:Free_Hamiltonian_Encoding}
Each point $\bm{x}$ ($\bm{y}$, etc.) on the lattice is first mapped to a qubit index $i$ ($j$, etc.) along a Jordan-Wigner path. 
We then specify the VC paths along which the composite auxiliary Majorana operators $i\tmu(i)\Bar{\tmu}(j)$ or $i\tnu(i)\Bar{\tnu}(j)$ act, as depicted in \cref{Fig:VC_Auxiliary_Terms_3D} of \cref{Sec:VC_Stabilizers}. The form of the hopping terms, therefore, depends on which axis they are along.
Given our choice, the hopping terms in $\Hfree$ in \cref{Eq:H_free_Hamiltonian} can be shown to map to the following operators:
\begin{align}
    &\text{hopping along the } x \text{ axis:} \nonumber\\
    &\qquad h_{\sigma}^{x}(i,j) =    \adag_{\sigma}(i)a_{\sigma}(j) + \adag_{\sigma}(j)\adag_{\sigma}(i) \rightarrow \,  \,  \tcrt(i)\tannh(j) + \tcrt(j)\tannh(i), \\
    &\text{hopping along the } y \text{ axis:} \nonumber\\
    &\qquad h_{\sigma}^{y}(i,j)= \adag_{\sigma}(i)a_{\sigma}(j) + \adag_{\sigma}(j)\adag_{\sigma}(i) \rightarrow \left(\tcrt(i)\tannh(j) + \tcrt(j)\tannh(i)\right)i\tmu(i)\Bar{\tmu}(j), \\
    &\text{hopping along the } z \text{ axis:} \nonumber\\
    &\qquad h_{\sigma}^{z}(i,j) = \adag_{\sigma}(i)a_{\sigma}(j) + \adag_{\sigma}(j)\adag_{\sigma}(i) \rightarrow \left(\tcrt(i)\tannh(j) + \tcrt(j)\tannh(i)\right)i\tnu(i)\Bar{\tnu}(j).
\end{align}
As an example, the hopping operators for the spin-down proton are as follows: 
\begin{align}
     \tilde{h}_{\downarrow p}^{x}(i,j) 
    &= \frac{1}{2} \left(X_{i}^{\downarrow p}Y_{j}^{\downarrow p} +  Y_{i}^{\downarrow p}X_{j}^{\downarrow p} \right) Z_i^{\uparrow n}Z_i^{\downarrow n}Z_{i'}Z_{i''}Z_j^{\uparrow p}, 
    \label{Eq:h_x_Hopping_Paulis}\\
     \tilde{h}_{\downarrow p}^{y}(i,j)
    &= \frac{1}{2}(Y_i^{\downarrow p} X_j^{\downarrow p} - X_i^{\downarrow p}Y_j^{\downarrow p}) Z_i^{\uparrow n} Z_i^{\downarrow n} X_{i'} Z_j^{\uparrow n} Z_j^{\downarrow n}Y_{j'},
\label{Eq:h_y_Hopping_Paulis}\\
    \tilde{h}_{\downarrow p}^{z}(i,j)
    &= \frac{1}{2}(Y_i^{\downarrow p}X_j^{\downarrow p} - X_i^{\downarrow p}Y_j^{\downarrow p})Z_i^{\uparrow n}Z_i^{\downarrow n}Z_{i'}X_{i''}
    Z_j^{\uparrow n}Z_j^{\downarrow n}Z_{j'}Y_{j''} \label{Eq:h_z_Hopping_Paulis}.
\end{align}
Here, it is assumed that $j>i$. 
The terms for all other species can be obtained similarly by noting the definitions in \crefrange{eq:a-tilde-up-p}{eq:nu-def-II}, the Jordan-Wigner and associated VC paths in \cref{Sec:VC_Stabilizers}, and the chosen ordering of the physical and auxiliary degrees of freedom on each lattice site as shown in \cref{Fig:JW_String}.
In general, the highest Pauli weight for hopping terms is 12, which appears in $\tilde h_{\uparrow p}^{z}(i,j)$. 

Finally, the terms proportional to the number operator in $\Hfree$ in \cref{Eq:H_free_Hamiltonian} take a simple form when converted to a qubit Hamiltonian, as the number operator $N_\sigma(x)=\crt(x)\annh(x) $ becomes
\begin{align}
 \tilde{N}_{\sigma}(i) &= \frac{1}{2}(\mathds{1} - Z_{i}^{\sigma}),
 \label{Eq:Number_Operators}
\end{align}
where $i$ is the qubit index associated with lattice site $x$.

\subsubsection{Encoding the Contact Terms
\label{Sec:Pionless_Contact_Terms}}
Using \cref{Eq:Number_Operators}, the contact terms in \cref{Eq:H_contact_1,Eq:H_contact_2} can be straightforwardly written in terms of Pauli operators:
\begin{align}
    \tilde{H}_{\Cpi} &= \frac{\Cpi}{8}\sum_i\sum_{\sigma\neq \sigma'} (\mathds{1}-Z_{i}^{\sigma})(\mathds{1}-Z_{i}^{\sigma'}),
    \\
    \tilde{H}_{\Dpi} &= \frac{\Dpi}{48}\sum_i\sum_{\sigma\neq \sigma'\neq \sigma''} (\mathds{1}-Z_{i}^{\sigma})(\mathds{1}-Z_{i}^{\sigma'})(\mathds{1}-Z_{i}^{\sigma''}).
\end{align}

\subsubsection{Pionless EFT and the Compact Encoding} \label{Sec:Pionless_EFT_Compact_Encoding}

Here, we consider an alternative to the VC encoding.
To motivate this, first note that for the pionless-EFT Hamiltonian, the operators acting on different species of the nucleon are essentially independent in the following sense. 
Not only do 
$\adag_{\sigma}(\bm{x})a_{\sigma}(j) + \adag_{\sigma}(j)a_{\sigma}(\bm{x})$, $N_\sigma(\bm{x})N_{\sigma'}(\bm{x})$, and $N_\sigma(\bm{x})N_{\sigma'}(\bm{x})N_{\sigma''}(\bm{x})$ operators commute with the total number operator $\sum_{\bm{x},\sigma}N_\sigma(\bm{x})$, but they also commute with the number operator for each species $\sum_{\bm{x}} N_\sigma(\bm{x}) 
$ individually.
As a result, $\adag_{\sigma}(\bm{x})a_{\sigma}(\bm{y}) + \adag_{\sigma}(\bm{y})a_{\sigma}(\bm{x})$, $N_\sigma(\bm{x})N_{\sigma'}(\bm{x})$, and $N_\sigma(\bm{x})N_{\sigma'}(\bm{x})N_{\sigma''}(\bm{x})$ are each block diagonal across the occupation basis for each respective species of the nucleon. 
That is, the subspace with only $\eta$ nucleons can be decomposed as 
\begin{align}
\bigoplus_{n_1+n_2+n_3+n_4=\eta}\mathcal{H}^{n_1}_{\uparrow p}\otimes\mathcal{H}^{n_2}_{\downarrow p}\otimes \mathcal{H}^{n_3}_{\uparrow n}\otimes \mathcal{H}^{n_4}_{\downarrow n},
\end{align}
where $\mathcal{H}^{n}_\sigma$ denotes the subspace of $n$ nucleons of species $\sigma$.
The terms $\adag_{\sigma}(\bm{x})a_{\sigma}(\bm{y}) + \adag_{\sigma}(\bm{y})a_{\sigma}(\bm{x})$, $N_\sigma(\bm{x})N_{\sigma'}$, and $N_\sigma(\bm{x})N_{\sigma'}(\bm{x})N_{\sigma''}(\bm{x})$ are then block diagonal according to this decomposition.
Since a number-preserving fermionic encoding must be able to represent the entire algebra of number-preserving operators (i.e., the algebra acting on the subspace with a fixed fermion number), one can represent the terms in $\Hfree$, $H_{\Cpi}$, and $H_{\Dpi}$ as sums of tensor products of operators from \emph{separate} fermionic encodings (i.e., each fermion species $\sigma$ is encoded with a separate encoding). In other words, there is no need to implement the fermionic anticommutation relations between different nucleon species when implementing Hamiltonian evolution via a product formula that uses a Hamiltonian decomposition in which all terms of different species commute. Thus, each fermionic Hilbert space can be encoded \emph{independently}. 

One can use this observation to encode all four species of the nucleons separately and then ``stack'' these encodings together as a tensor product of separate fermionic encodings.
In particular, we will use the compact encoding of Refs.~\cite{Derby_et_al_2021, Derby_Klassen_2021} to implement the evolution of the pionless EFT with reduced circuit depths, but with a larger qubit overhead. Using the compact encoding on a cubit lattice, the encoded number and hopping operators are
\begin{align}
    &\tilde{N}_{\sigma}(i) = \frac{1}{2}(\mathds{1} - Z_{i}^{\sigma}), \\
    &\tilde{h}_\sigma(i,j) = -\frac{i}{2}\tilde{E}_\sigma(i,j)\left( \tilde{V}_\sigma(j) - \tilde{V}_\sigma(i) \right),
\end{align}
where $i$ denotes the qubit index of site $\bm{x}$, and $\tilde{V}_\sigma(j)$ and $\tilde{E}_\sigma(i,j)$ are given in \cref{Eq:V_j_Paulis,Eq:E_ij_Paulis}, respectively, with $\sigma$ subscript denoting each encoded species.
Here, the hopping term $\tilde{h}_\sigma(i,j)$ is given by the sum of two Pauli strings each with a weight of at most 4. This should be compared with the weight-12 operators for the VC encoding as per \cref{Eq:h_z_Hopping_Paulis}, which demonstrates how a lower circuit depth is achieved with a stacked compact encoding.
Conveniently, the expression for the contact terms in the stacked compact encoding is identical to that in \cref{Sec:Pionless_Contact_Terms}. Note that one could also ``stack'' VC encodings of the pionless-EFT Hamiltonian, but this gives slightly larger overheads. 
This is because now two auxiliary qubits need to be allocated to each nucleon species, and in turn the hopping interactions will have a Pauli weight of at most 6 instead of 12. This means there will need to be 3 qubits per fermion in a stacked VC encoding in contrast to 1.5 qubits per fermion in the regular VC encoding, with only a moderate gain in circuit depth. Hence, we will not consider the stacked VC encoding further.

For Hamiltonians that are not block diagonal in the occupation basis for each fermionic species individually, it is not generally possible to represent the terms in the Hamiltonian as sums of tensor products of individually encoded operators. Notably, the idea of stacking different fermionic encodings will not work for the other models studied in this work as they do not satisfy this condition.\footnote{For example, the terms comprising $ H_{\rm LR}$ in \cref{Eq:OPE_Long-Range_Terms} do not satisfy this condition as they mix nucleons of different species.}
We emphasize that the concept of stacking encodings is an established technique (e.g., it was used in Ref.~\cite{Clinton_Bausch_Cubitt2021}).

\subsection{The One-Pion-Exchange Hamiltonian}
\label{Sec:OPE_Hamiltonian}
Next, let us examine how a Hamiltonian involving the effect of pions is encoded. Explicitly encoding the pion fields and their interactions adds to the number of qubits needed for the simulation and further increases the circuit depth.
As discussed in \cref{sec:NEFT}, instead of explicitly including pions, it may be useful to integrate them out.
This leads to the generation of a long-range Yukawa-type interaction among the nucleons, i.e., a static potential corresponding to the exchange of one or more static pions among the nucleons.
Since pions are not massless, the effective range of pion-exchange potentials drops off exponentially as a function of the distance among the nucleons, with a length scale set by the Compton wavelength of the pions, i.e., proportional to the inverse pion mass.

Within this formulation of nuclear EFTs, the effective Hamiltonian at leading order in Weinberg's organizational scheme of interactions is given by~\cite{weinberg1990nuclear, weinberg1991effective}
\begin{align}
    H_{\rm OPE} = \Hfree + H_C + H_{C_{I^2}} + H_{\rm LR},
    \label{Eq:H-LR}
\end{align}
where $\Hfree$ is as in \cref{Eq:H_free_Hamiltonian}, $H_C, H_{\CI}$ are on-site contact interactions, and $H_{\rm LR}$ is a long-range interaction that accounts for the OPE contribution: 
\begin{align}
     &H_{C} = \frac{C}{2}\sum_{\bm{x}} :\rho^2(\bm{x}):\,,
     \label{eq:HC}\\
     &H_{C_{I^2}} = 
     \frac{C_{I^2}}{2}\sum_{\bm{x}}\sum_{I} :\rho_I^2(\bm{x}): \,,
     \label{eq:HCI2}\\
     &H_{\rm LR}= \sum_{\alpha,\beta,\gamma, \delta} \sum_{\alpha',\beta',\gamma', \delta'} 
      \sum_{\bm{x},\bm{y}}[G(|\bm{x}-\bm{y}|)]_{\alpha'\beta'\alpha\beta\gamma'\delta'\gamma\delta} :\adag_{\alpha' \beta'}(\bm{x})a_{\gamma' \delta'}(\bm{x})\adag_{\alpha\beta}(\bm{y})a_{\gamma\delta}(\bm{y}):. 
      \label{Eq:OPE_Long-Range_Terms}
\end{align}
Here, $\bm{x},\bm{y} \in \Lambda(L)$ as before, $I,J \in \{1,2,3\}$, and $\alpha,\beta,\gamma,\delta,\alpha',\beta',\gamma', \delta' \in \{1,2\}$.
The nucleonic bilinear operators $\rho(\bm{x})$ and $\rho_I(\bm{x})$ are defined in \cref{Eq:Particle_Density_1,Eq:Particle_Density_3}, respectively, and the values of $C$ and $\CI$ (and the parameters they are calculated from) at a sample lattice spacing are given in \cref{Table:OPE_Parameters}.
The function $G(|\bm{x}-\bm{y}|)$ in \cref{Eq:OPE_Long-Range_Terms} is defined as
\begin{align}
    [G(|\bm{x}-\bm{y}|)]_{\alpha'\beta'\alpha\beta\gamma'\delta'\gamma\delta} \coloneqq
    &\frac{1}{12\pi} \left(\frac{g_A}{2f_\pi}\right)^2 \sum_{I} [\tau_I(\bm{x})]_{\beta' \delta'}[\tau_I(\bm{y}) ]_{\beta \delta}\nonumber \\
    &\times\bigg\{ m_{\pi}^2\frac{e^{-m_{\pi}r}}{r}\bigg[[S_{12}]_{\alpha'\gamma'\alpha\gamma} \left(1 + \frac{3}{m_{\pi}r} + \frac{3}{m_{\pi}^2r^2} \right) 
    +\sum_S[\bm{\sigma}_S(\bm{x})_{\alpha'\gamma'}[\bm{\sigma}_S(\bm{y})]_{\alpha\gamma} \bigg]\nonumber \\
    &\qquad- \frac{4\pi}{3}{a_L^{-3}}\sum_S[\bm{\sigma}_S(\bm{x})_{\alpha'\gamma'}[\bm{\sigma}_S(\bm{y})]_{\alpha\gamma}\delta_{\bm{x},\bm{y}} \bigg\},
    \label{Eq:G-def}
\end{align}
where $r= |\bm{x}-\bm{y}|$, and $r \neq 0$ is assumed in all but the last term in the curly brackets. Finally, $S_{12}$ is defined as
\begin{align}
[S_{12}]_{\alpha' \gamma' \alpha \gamma} \coloneqq 3[\hat{\bm{x}}\cdot\bm{\sigma}(\bm{x})]_{\alpha' \gamma'}[\hat{\bm{y}}\cdot\bm{\sigma}(\bm{y})]_{\alpha\gamma}-\sum_S [\sigma_S(\bm{x})]_{\alpha' \gamma'}[\sigma_S(\bm{y})]_{\alpha\gamma}.
\label{Eq:S12-def}
\end{align}
For future convenience, we define $H_{\rm LR}(r)$ to be the subset of terms of $H_{\rm LR}$ in which the sum over spatial sites runs over only the points that satisfy $r=|\bm{x}-\bm{y}|$, i.e., the sum over $\bm{x}$ and $\bm{y}$  in \cref{Eq:OPE_Long-Range_Terms} is replaced by $\sum_{|\bm{x}-\bm{y}|=r}$.

The free Hamiltonian can be encoded using the VC encoding in the same way as in \cref{Sec:Free_Hamiltonian_Encoding}, hence we focus on the contact and long-range interactions.

\begin{table}[t!]
    \centering
\begin{tabular}{c|c|c|c}
    $\tilde{C}^{I=1}$ [MeV$^{-2}$] & $\tilde{C}^{I=0}$ [MeV$^{-2}$] & $C$ [MeV] & $\CI$ [MeV]\\
    \hline 
    $-5.021\times 10^{-5}$ & $-5.714\times 10^{-5}$ & $\frac{1}{4a_L^{3}}(3\tilde{C}^{I=1}+\tilde{C}^{I=0})$ & $\frac{1}{4a_L^{3}}(\tilde{C}^{I=1}-\tilde{C}^{I=0})$
\end{tabular}
    \caption{The values of low-energy constants for the OPE EFT Hamiltonian. 
    The values of $C$ and $\CI$ are calculated from $\tilde{C}^{I=1}$ and $\tilde{C}^{I=0}$, which are taken from Eqs.~(43) and (44) of Ref.~\cite{borasoy2008chiral} using a lattice spacing of $a^{-1} = (100~\mathrm{MeV})^{-1}\approx 2$~fm.
    }
    \label{Table:OPE_Parameters}
\end{table}

\subsubsection{Encoding Contact Terms}
Given the definition of $\rho$ in \cref{Eq:Particle_Density_1}, $H_C$ in \cref{eq:HC} can be written as a sum of number operators, giving
\begin{align}
    \tilde{H}_C&= \frac{C}{8}\sum_i \sum_{\sigma, \sigma'}(\mathds{1}_i^\sigma - Z_{i}^{\sigma})(\mathds{1}_i^{\sigma'} - Z_{i}^{\sigma'}), \label{Eq:H_C_OPE_Paulis}
\end{align}
with $i$ being the qubit index of lattice site $\bm{x}$ for each nucleon species $\sigma$ (or $\sigma'$), which runs from $1$ to $L^3$. The $H_{\CI}$ term in \cref{eq:HCI2} is slightly more complex as it mixes the creation and annihilation operators of different species of fermions on the same site [see the definition of $\rho_I$ in \cref{Eq:Particle_Density_3}]. 
Explicitly,
\begin{align}
&H_{C_{I^2}} = 
\frac{C_{I^2}}{2}\sum_{\bm{x}} :\left[ N_{\uparrow p}^2+N_{\downarrow p}^2+N_{\uparrow n}^2+N_{\downarrow n}^2-6N_{\uparrow p}N_{\uparrow n}+2N_{\uparrow p}N_{\downarrow p}-2N_{\uparrow p}N_{\downarrow n}-2N_{\downarrow p}N_{\uparrow n}+2N_{\uparrow n}N_{\downarrow n}\right . \nonumber\\
& \hspace{8.0 cm} \left . -6N_{\downarrow p}N_{\downarrow n} -\, 4\left(a^\dagger_{\uparrow p}a_{\downarrow p}a^\dagger_{\downarrow n}a_{\uparrow n}+{\rm h.c.}\right) \right]: \,,
\label{Eq:explicit-HCI2}
\end{align}
where all the operators have an implicit $\bm{x}$ dependence. To keep the presentation compact, we will not write out the encoded Hamiltonian for all these terms in full, but rather demonstrate how the term with the highest Pauli weight arises: 
\begin{align}
    \adag_{\uparrow p}(i)a_{\uparrow n}(i)\adag_{\downarrow n}(i)a_{\downarrow p}(i) + \text{h.c.} &\rightarrow \frac{1}{16}(X_i^{\uparrow p} - iY_i^{\uparrow p})
    (X_i^{\uparrow n}+iY_i^{\uparrow n}) (X_i^{\downarrow n} - iY_i^{\downarrow n})(X_i^{\downarrow p}+iY_i^{\downarrow p})Z_i^{\uparrow n}
    Z_i^{\uparrow p} \nonumber +\text{h.c.}\\
    &=  - 
    \frac{1}{8}X_i^{\uparrow p}Y_i^{\uparrow n}X_i^{\downarrow n}Y_i^{\downarrow p} + \text{(7 other terms)}, \label{Eq:H_CI2_Paulis}
\end{align}
with $i$ being the qubit index of lattice site $\bm{x}$ as before. The seven terms not shown are those including other possibilities with zero, two, and four $X$ (or $Y$) Pauli matrices. Thus, such a term in $H_{\CI}$ consists of 8 strings, each with Pauli weight 4. 
All these Pauli strings commute. The rest of the terms in \cref{Eq:explicit-HCI2} depend on number operators, which map trivially according to \cref{Eq:Number_Operators}. These will end up in strings with Pauli weights of at most two.

\subsubsection{Encoding Long-Range Terms}\label{Sec:Long_Ranged_Pauli}
The long-range Hamiltonian in \cref{Eq:OPE_Long-Range_Terms} contains terms of the general form
\begin{align}
\label{Eq:rhoSI-rhoSI}
    :\rho_{S_1,I}(\bm{x})\rho_{S_2,I}(\bm{y}): \, = -&\left[\sum_{\alpha',\beta',\gamma',\delta'} a^{\dagger}_{\alpha'\beta'}(\bm{x}) [\sigma_{S_1}]_{\alpha'\gamma'}[\tau_I]_{\beta'\delta'}a_{\gamma'\delta'}(\bm{x})\right] \nonumber \\
    &\qquad\qquad\times\left[\sum_{\alpha,\beta,\gamma,\delta} a^{\dagger}_{\alpha\beta}(\bm{y}) [\sigma_{S_2}]_{\alpha\gamma}[\tau_I]_{\beta\delta}a_{\gamma\delta}(\bm{y})\right],
\end{align}
so it is a sum of fermionic terms of the form $\adag_{\alpha',\beta'}(\bm{x})a_{\gamma'\delta'}(\bm{x})\adag_{\alpha,\beta}(\bm{y})a_{\gamma\delta}(\bm{y}) +\text{h.c.}$
These can be expressed as hopping-like terms, where the hopping occurs between nucleons on any pair of sites. 
As an example, the highest-weight terms arise from
\begin{align}
    &\adag_{\uparrow p}(i)a_{\downarrow n}(i)\adag_{\uparrow p}(j)a_{\downarrow n}(j) + \text{h.c.}\qquad\nonumber \\
    &\quad\longrightarrow \frac{1}{16}\bigg[(X_i^{\uparrow p} - iY_i^{\uparrow p})(X_i^{\downarrow n} + iY_i^{\downarrow n}) Z_i^{\uparrow p} Z_i^{\downarrow p}Z_i^{\uparrow n} 
    (X_j^{\uparrow p} - iY_j^{\uparrow p})(X_j^{\downarrow n} + iY_j^{\downarrow n}) Z_j^{\uparrow p} Z_j^{\downarrow p}Z_j^{\uparrow n} \nonumber\\
    &\qquad\qquad+Z_j^{\uparrow n} Z_j^{\downarrow p}Z_j^{\uparrow p}(X_j^{\downarrow n} - iY_j^{\downarrow n})(X_j^{\uparrow p} + iY_j^{\uparrow p})Z_i^{\uparrow n}Z_i^{\downarrow p}Z_i^{\uparrow p}(X_i^{\downarrow n} - iY_i^{\downarrow n})(X_i^{\uparrow p} + iY_i^{\uparrow p})\bigg]
    \nonumber \\
    &\quad= -\frac{1}{8}Y_i^{\uparrow p}X_i^{\downarrow n}Z_i^{\downarrow p}Z_i^{\uparrow n} Y_j^{\uparrow p}X_j^{\downarrow n}Z_j^{\downarrow p}Z_j^{\uparrow n}
    + \text{(7 other terms).} \label{Eq:Long_Ranged_Pauli_String}
\end{align}
All of the 8 Pauli strings have an even number of $X$ operators, so they all commute. Each string has Pauli weight 8.

\subsubsection{Simulation with a Truncated Long-Range Hamiltonian} \label{Sec:Cutoff_OPE}
Since the long-range terms decay exponentially with the distance between the nucleons, we can simplify the Hamiltonian by introducing a cutoff, beyond which the interactions are weak enough to be neglected.
This reduces the number of terms that need to be simulated at the cost of introducing some additional error.
As we will show, provided the cutoff is sufficiently large, this error can be negligible. Similar analyses are performed for bounding the error in simulating power-law interactions in, e.g., Ref.~\cite[Appendix B]{Tran_2019}).

\begin{lemma}[Long-Range Cutoff Length]\label{Lemma:Cutoff_Length}
Let $H_\ell$ be the same interaction as $H_{\rm LR}$ but with the long-range interaction truncated at length $\ell = |\bm{x}-\bm{y}|$, where $\bm{x}$ and $\bm{y}$ are the positions of the two interacting nucleons on the lattice.
Then, 
\begin{align}
   \norm{e^{-iH_{\ell}t} - e^{-iH_{\rm LR}t}   }_{\eta}\leq t \min &\bigg\{  \eta^2   \left[  \left( 72g_1(\ell+a_L) + 
   648g_2(\ell+a_L)\right) \right], 
   \nonumber\\
   & \frac{4\pi\eta}{m_\pi^2 a_L^3} (\ell +a_L) g_1(\ell +a_L)
   \left[ 720(m_\pi\ell + m_\pi a_L+ 1) + 
   3888\right] \bigg\},\label{eq:lem2}
\end{align}
where $\norm{\cdots}_\eta$ denotes the spectral norm of the enclosed operator in a sector with a fixed number of nucleons, $\eta$, as in \cref{Eq:norm-eta-def}, and
\begin{align}
     g_1(r) \coloneqq \frac{1}{12\pi}\left(\frac{g_A}{2f_\pi}\right)^2m_{\pi}^2 \frac{e^{-m_{\pi}r}}{r}, \quad \quad \quad  g_2(r)& \coloneqq g_1(r)\left(1 + \frac{3}{m_{\pi}r} + \frac{3}{m_{\pi}^2r^2} \right).
\end{align}
\end{lemma}

\noindent The proof is presented in \cref{Sec:Cutoff_Proof}.

\subsection{Dynamical-Pion EFT Hamiltonian}
\label{Sec:Dyn_Pions}
Instead of introducing a static OPE potential, the pions can be retained in the model as explicit dynamical degrees of freedom that mediate interactions among the nucleons. In this model, the relativistic pions interact with non-relativistic nucleons. Expressing pions as complex scalar fields, the discretized Hamiltonian can be written as~\cite{Lee_Borasoy_Schaefer_2004, Lucas_2018}
\begin{align}
    H_{D\pi} = \Hfree + H_C + H_{C_{I^2}} + H_{\pi} + H_{N\pi },
    \label{Eq:H-dyn}
\end{align}
where $\Hfree$ is the free nucleon Hamiltonian as in \cref{Eq:H_free_Hamiltonian} and $H_C $, $H_{C_{I^2}}$ are nucleon-nucleon contact terms as per \cref{eq:HC,eq:HCI2} in the previous section. The free pion Hamiltonian in \cref{Eq:H-dyn} is
\begin{align}
    H_{\pi} &= \frac{a_L^3}{2}\sum_{\bm{x}} \sum_{I} \left[ \Pi^2_I(\bm{x}) + (\nabla \pi_I(\bm{x}))^2 + m_{\pi}^2\pi_I^2 (\bm{x})\right],
    \label{Eq:Pion_Only_Hamiltonian}
\end{align}
with $\nabla$ being the finite-difference derivative (see \cref{Sec:Encode_Free_Pion}). The self interaction of pions can be ignored at this order in the chiral EFT expansion.  Finally, $H_{N\pi }$ is the pion-nucleon interaction Hamiltonian, which can be split into the axial-vector term, $H_{\rm AV}$, and the Weinberg-Tomozawa term, $H_{\rm WT}$:\footnote{In \crefrange{Eq:Pion_Only_Hamiltonian}{Eq:H_WT}, we have made an explicit choice to associate each bosonic site $\bm{x}\in \Lambda$ with the corresponding fermionic site, and then allow the fermions and bosons to interact at that site.  An alternative choice would be to place the bosonic sites halfway between the fermionic sites, and allow each bosonic site to interact only with its nearest neighbors. Such choices affect the form of discrete derivatives and hence the discretization effects, which are not the focus of this study.}
\begin{align}
    H_{N\pi } &= H_{\rm AV} + H_{\rm WT}, 
\end{align}
with
\begin{align}
    H_{\rm AV} &= \frac{g_A}{2f_\pi}\sum_{\bm{x}} \sum_{\alpha,\beta,\gamma,\delta}\sum_{I,S} \adag_{\alpha\beta}(\bm{x}) [\tau_I]_{\beta\delta}[\sigma_S]_{\alpha\gamma} \partial_S \pi_I(\bm{x}) a_{\gamma\delta}(\bm{x}), \label{Eq:H-AV} \\
     H_{\rm WT} &= \frac{1}{4f^2_{\pi}}\sum_{\bm{x}}\sum_{I_1,I_2,I_3}\sum_{\alpha,\beta,\delta} \epsilon_{I_1I_2I_3} \pi_{I_2}(\bm{x})\Pi_{I_3}(\bm{x}) \adag_{\alpha\beta}(\bm{x}) [\tau_{I_1}]_{\beta \delta}a_{\alpha\delta}(\bm{x}). \label{Eq:H_WT}
\end{align}

\subsubsection{Encoding the Free Pion Hamiltonian}
\label{Sec:Encode_Free_Pion}
As mentioned in \cref{Sec:bosonic-field-encodings}, to encode the dynamical pions, we choose to work with the field and conjugate-momentum basis in position space. This retains the locality of the interaction terms and reduces the circuit depth required to implement these interactions.

Part of the free pion Hamiltonian $H_\pi$ in \cref{Eq:Pion_Only_Hamiltonian} involves the $\pi_I^2$ operator, which becomes
\begin{align}
\label{Eq:piI-Squared-Decomposition}
\pi_I^2(\bm{x}) &\rightarrow 
\left[P\mathds{1}+Q\sum_{m=0}^{n_b-1} 2^m Z^{(m)}_{I,\bm{x}}\right]^2
=P^2\mathds{1}+2PQ\sum_{m=0}^{n_b-1} 2^m Z^{(m)}_{I,\bm{x}}+Q^2\sum_{m,m'=0}^{n_b-1} 2^{m+m'} Z^{(m)}_{I,\bm{x}}Z^{(m')}_{I,\bm{x}},
\end{align}
where $P \coloneqq -\pi_{\rm max}+\frac{\delta_\pi}{2}(2^{n_b}-1)$ and $Q \coloneqq -\frac{\delta_\pi}{2}$ [see \cref{Eq:Field_Operator_2}].
The term involving $\Pi_I^2(\bm{x})$ can be encoded similarly since ${\Pi}_I(\bm{x})$ is diagonal in the Fourier basis, as discussed in \cref{Sec:bosonic-field-encodings}.
Finally, the terms proportional to the square of the pion derivative operator can be encoded as
\begin{align}
\label{Eq:Del-piI-Squared-Decomposition}
&(\nabla \pi_I(\bm{x}))^2=
\sum_{j=1,2,3}\left[\frac{\pi_I(\bm{x}+a_L\hat{\bm{n}}_j)-\pi_I(\bm{x})}{a_L}\right]^2 \nonumber \\
&\rightarrow \frac{Q^2}{a_L^2}\sum_{j=1,2,3} \bigg[\sum_{m,m'=0}^{n_b-1} 2^{m+m'} Z^{(m)}_{I,\bm{x}}Z^{(m')}_{I,\bm{x}}+
\sum_{n,n'=0}^{n_b-1} 2^{n+n'} Z^{(n)}_{I,\bm{x}+a_L\hat{\bm{n}}_j}Z^{(n')}_{I,\bm{x}+a_L\hat{\bm{n}}_j}-\sum_{m,n=0}^{n_b-1} 2^{m+n+1} Z^{(m)}_{I,\bm{x}}Z^{(n)}_{I,\bm{x}+a_L\hat{\bm{n}}_j}\bigg],
\end{align}
where $\hat{\bm{n}}_j$ is the unit vector along the Cartesian coordinate $j$, and $\{m,m'\}$ ($\{n,n'\}$) are indices associated with the qubit register of size $n_b$ used to encode the field in binary at position $\bm{x}$ ($\bm{x}+\hat{\bm{n}}_j$). Generalization to symmetric or other improved lattice derivatives is straightforward. In summary, the free pion Hamiltonian generates operators with Pauli weight of at most two.

\subsubsection{Encoding the Axial-Vector Term}
\label{Sec:Encode_HAV}
Using the discrete-derivative relation as in the free Hamiltonian, the axial-vector Hamiltonian $H_{\rm AV}$ in \cref{Eq:H-AV} can be expressed as a qubit Hamiltonian as well. Explicitly, the highest-weight term in the summation over lattice sites and spin and isospin components becomes
    \begin{align}\label{Eq:AV_Paulis}
        \frac{\pi_1(\bm{x}+a_L\hat{\bm{n}}_1)-\pi_1(\bm{x})}{a_L} &\left[\adag_{\uparrow p}(\bm{x}) a_{\downarrow n}(\bm{x}) + \adag_{\downarrow n}(\bm{x})a_{\uparrow p}(\bm{x})\right] \nonumber \\ &\rightarrow \frac{Q}{2a_L}\left( \sum_{n=0}^{n_b-1} 2^n Z^{(n)}_{1,\bm{x}+a_L\hat{\bm{n}}_j} - \sum_{m=0}^{n_b-1} 2^m Z^{(m)}_{1,\bm{x}} \right) \left(X_{i}^{\uparrow p}X_{i}^{\downarrow n}+ Y_{i}^{\uparrow p}Y_{i}^{\downarrow n} \right) Z_i^{\downarrow p}Z_i^{\uparrow n},
    \end{align}
where $i$ denotes the qubit index associated with the fermionic site $\bm{x}$. There are $4n_b$ strings in this summation with Pauli weight 5, and all strings in the sum commute. All other operators in $H_{\rm AV}$ have Pauli weight 5 or less.
    
    \subsubsection{Encoding the Weinberg-Tomozawa Term} \label{Sec:Encode_HWT}

In order to map the Weinberg-Tomozawa Hamiltonian $H_{\rm WT}$ in \cref{Eq:H_WT} to Pauli operators, first note that $\Pi_{I_2}(x)$ and $\pi_{I_3}(x)$, $I_2\neq I_3$, act on different Hilbert spaces, so they can be diagonalized simultaneously using the quantum Fourier transform.
Recalling that $\tilde{\Pi}_Is$ is the Fourier-transformed conjugate-momentum operator (i.e., in the basis for which it is diagonal), one of the Weinberg-Tomozawa terms containing the highest Pauli-weight in the summation becomes
\begin{align}\label{Eq:WT_Pauli_Strings}
    \pi_2(\bm{x})\tilde{\Pi}_3(\bm{x})
    &\left[\adag_{\uparrow p}(\bm{x}) a_{\uparrow n}(\bm{x}) + \adag_{\uparrow n}(\bm{x})a_{\uparrow p}(\bm{x})\right] 
    \nonumber\\
    &\longrightarrow 
    \left(P\mathds{1}+Q\sum_{m=0}^{n_b-1} 2^m Z^{(m)}_{2,\bm{x}}\right)\left(P'\mathds{1}+Q'\sum_{l=0}^{n_b-1} 2^l Z^{(l)}_{3,\bm{x}}\right)
    \left(X_{i}^{\uparrow p}X_{i}^{\uparrow n}+ Y_{i}^{\uparrow p}Y_{i}^{\uparrow n} \right) Z_i^{\downarrow p}.
\end{align}
Here $P' \coloneqq -\Pi_{\rm max}+\frac{\delta \Pi}{2}(2^{n_b}-1)$ and $Q' \coloneqq -\frac{\delta \Pi}{2}$ [see \cref{Eq:Pi-tilde}]. The right-hand side of \cref{Eq:WT_Pauli_Strings} can be decomposed as a summation of $2(n_b+1)^2 $ Pauli string operators, each of which has a highest Pauli weight of 5. All of these Pauli strings commute.

\subsubsection{Simulation in the Truncated Field Space}
\label{Sec:Truncated-Field-Space}
As previously mentioned, in order to keep track of a finite number of bosonic degrees of freedom, one must impose a cutoff and a digitization scale for the field strength of the pion.
We follow the methods introduced in Ref.~\cite{Jordan_Lee_Preskill_2012} and show that if the evolution is restricted to states with a given energy $E$, then a high-fidelity representation of the exact state is achievable with particular digitization and cutoff scales.
In this section, we simply state the bounds to be used in our simulation-cost analysis of nuclear EFTs and refer the reader to \cref{Sec:Pion_Cutoff} for the details of the proof. Our bounds in this section are not entirely general as, for technical reasons, we make some assumptions on the relationship between the lattice spacing and the constants $m_\pi$ and $f_\pi$, as explained in the appendix.

\begin{lemma}[Dynamical-Pion Cutoff]
\label{Lem:Dyn-Pions}
Let $\ket{\psi_{\rm cut}}$ be a state with the field cutoff $\pimax$, conjugate-momentum field cutoff $\Pimax$, and total nucleon number $\eta$, such that $\langle \psi |H| \psi \rangle_\eta \leq E$. To achieve $|\braket{\psi|\psi_{\rm cut}}|\geq 1 -\epscut$ with $3L^3$ bosonic fields (three types of pion fields at $L^3$ lattice sites), it is sufficient to choose  
\begin{align}
    \label{eq:pimax-expression}
    \pimax &=\left(  \sqrt{\frac{3L^3}{\epscut}}+1\right)\left( \frac{3g_A}{f_\pi a_L A} + \sqrt{  \frac{E+8\eta|C|+4\eta|C_{I^2}|}{A}+3\eta 
    \left(\frac{3g_A}{f_\pi a_L A}\right)^2 + \frac{9\eta m_{\pi}^2a_L^3}{A}\left( \frac{6g_A}{m_\pi^2f_\pi a_L^4} \right)^2 } \right) , \\
    \Pimax &=\left(  \sqrt{\frac{3L^3}{\epscut}}+1\right)
    \sqrt{ \frac{E+8\eta|C|+4\eta|C_{I^2}|}{B}+ \frac{3 \eta}{AB}  \left(\frac{3g_A}{f_\pi a_L}\right)^2 + \frac{9\eta m_{\pi}^2a_L^3}{B} \left( \frac{6g_A}{m_\pi^2f_\pi a_L^4} \right)^2  },
    \label{eq:Pimax-expression}
\end{align}
where 
\begin{align}
\label{Eq:A-B-Def}
   A \coloneqq \frac{m_\pi^2a_L^3}{2}-\frac{1}{2f^2_{\pi}a_L}, \quad B\coloneqq \frac{a_L^3}{2}-\frac{a_L}{2f^2_{\pi}},
\end{align} 
for lattice spacings $a_L$ such that $A,B>0$.
\end{lemma}

The proof is presented in \cref{Sec:Pion_Cutoff}. This result sets the number of qubits used to represent each pion field. Recalling the relations $n_b= \log_2(2\pimax/\delta_\pi+1)$ and $\Pi_{\max{}} = \pi/(a^3_L\dpi)$ from \cref{Sec:bosonic-field-encodings} gives
\begin{align}
    n_b = \log_2\left(\frac{2a_L^3}{\pi} \Pimax\pimax+1\right).
    \label{Eq:nb-for-dyn-and-ins}
\end{align}
Crucially, since $n_b$ is the number of qubits used to encode the pion field, it must be an integer. Thus, in practice we do not exactly substitute the bounds for $\pimax$ and $\Pimax$ into \cref{Eq:nb-for-dyn-and-ins}. Rather, we choose the nearest cutoffs above these bounds to ensure $n_b$ is an integer.
 
An alternative method of truncating the bosonic Hilbert space, proposed in Ref.~\cite{Tong_Albert_Mcclean_Preskill_Su_2022}, cuts off the bosonic occupation number (see also Refs.~\cite{somma2015quantum,macridin2018electron,Klco_Savage_2019}), and introduces exponentially small error in the occupation-number cutoff at any fixed lattice spacing, improving over the polynomial energy-based bound of Ref.~\cite{Jordan_Lee_Preskill_2012}.
However, this bound only applies to Hamiltonians of a particular form.
Unfortunately for our purposes, the Weinberg-Tomozawa term in the pionful Hamiltonian violates the necessary assumptions for the improved bound to apply. 
We note, however, that the Weinberg-Tomozawa term is often comparatively small (and identically zero in the static-pion limit), so in practice one may be able to achieve better bounds using the work of Ref.~\cite{Tong_Albert_Mcclean_Preskill_Su_2022}.
There are other works bounding the error associated with a cutoff on the bosonic space~\cite{Maskara_et_al_2022, Kuwahara_Van_vu_Saito_2022}; however, the assumptions in these works do not apply here either.
In particular, the result of Ref.~\cite{Kuwahara_Van_vu_Saito_2022} only applies when the potential term in the Hamiltonian is a function of number operators, and the result of Ref.~\cite{Maskara_et_al_2022} only applies for number-preserving bosonic Hamiltonians.

\section{Circuit Implementation of Trotter Steps
\label{sec:circuits}}

To characterize the resources for time evolution via Trotterization, we evaluate the cost of implementing each of the unitaries $e^{-iH_{\gamma}\delta t}$, as well as their controlled versions.
The uncontrolled unitaries will be used for time evolution, while the controlled versions are necessary for QPE.
We consider two metrics: the 2-qubit gate depth of the circuit, $\Dcost$, and the $T$-gate count, $T_{\rm cost}$.
The latter is a relevant metric for fault-tolerant algorithms, whereas the former is mostly relevant for non-error corrected computations prevalent in the near-term era of quantum computing (although the circuit depth is not completely irrelevant in the fault-tolerant setting).
When working in the circuit-depth model, we assume that the 2-qubit CNOT gates and 1-qubit $H$, $T$, and $Z$ rotations are available operations (although, in the near-term era, one can assume arbitrary 1-qubit rotations are available). 
We put no constraints on the qubit connectivity.
Generally, we resort to the most straightforward optimization of the circuits to parallelize 2-qubit gates and reduce the circuit depths, but we make no further attempt to improve this optimization in many instances. Additional improvements will not change the scaling of the total circuit depth with the parameters of the simulation, although those will likely be important for any near-term implementation of the algorithms presented in this work.

The only gates that require $T$ gates to be synthesized fault-tolerantly are 1-qubit $Z$ rotations, $R_z(\theta)=e^{-i\theta Z/2}$. 
Clifford operations (i.e., those that can be written in terms of CNOT, $H$, and $T^2$ gates) are essentially ``free'' operations in the fault-tolerant setting. To evaluate the $T$-gate cost, we use the following result from Ref.~\cite{Paetznick_Svore_2014}: for a 1-qubit $Z$ rotation $R_z$, a 1-qubit $\Tilde{R}_z$ gate can be implemented using the repeat-until-success method such that
\begin{align}
    \norm{R_z - \Tilde{R}_z}\leq \epssyn,
    \label{Eq:Def-epssyn}
\end{align}
with
\begin{align}
1.15\log(2/\epssyn) + 9.2
\label{eq:rus-cost}
\end{align}
$T$ gates in expectation.

The time evolution of all Hamiltonian terms is performed by decomposing them into Pauli strings. 
In the controlled-gate setting, each Pauli string takes two $Z$ rotations to implement, and in the non-controlled setting each takes only one $Z$ rotation~\cite{Nielsen_Chuang_2010}.
Thus, the number of $T$ gates primarily depends on the number of Pauli strings.

\subsection{Pionless-EFT Simulation Costs}
Here, we consider the resource costs for the pionless-EFT Hamiltonian for both the VC and compact encodings.
The analysis is split into two parts: the kinetic (or hopping) term and the contact-interaction terms. We also report both the 2-qubit circuit depth and $T$-gate counts, where the latter is fully determined from the $R_z$ gate counts.

\subsubsection{Hopping Operators}\label{Sec:Hopping_Depth}
We first consider the costs associated with implementing the kinetic terms in the VC encoding. The hopping operators in terms of Pauli operators are given in \crefrange{Eq:h_x_Hopping_Paulis}{Eq:h_z_Hopping_Paulis} for the spin-down proton and can be similarly deduced for other species of the nucleons.

\begin{lemma}[Kinetic-Energy Circuit Depth in the VC Encoding]\label{Lemma:Pionless_Hopping_Term_VC}
There is a circuit implementing the kinetic terms $e^{-it\tilde{h}_\sigma^x(i,j)},e^{-it\tilde{h}_\sigma^y(i,j)}, e^{-it\tilde{h}_\sigma^z(i,j)}$ with 2-qubit circuit depths of at most $16$, $22$, and $26$, respectively.
The controlled evolutions can be implemented with circuit depths of at most $20$, $26$, and $30$, respectively.
\end{lemma}

\begin{proof}
To implement the hopping terms, we appeal to a standard gate decomposition: the evolution of a $k$-local Pauli operator can be implemented by $2(k-1)$ CNOT gates (see e.g., Ref.~\cite[
Sec. 4.7]{Nielsen_Chuang_2010}), and controlled $k$-local Pauli-operator evolution takes $2k$ CNOTs.
If two Pauli strings are applied successively such that they have the same Pauli operators on all but $m$ qubits, then the CNOT gates cancel on all but the $m$ qubits, giving a total of $2(k-1)+2m$ CNOT gates. 
In \cref{Sec:Free_Hamiltonian_Encoding}, it was shown that the hopping interactions along the $z$ direction generate two Pauli strings with at most weight 12 (associated with proton-up hopping).
Each of the two Pauli strings share all but two different Pauli operations. It can be similarly shown that hopping terms along the $x$- and $y$-directions are two Pauli strings of at most Pauli weight $7$ and $10$, respectively (associated with proton-up hopping). Therefore, for a generic species $\sigma$,
\begin{align}
    \Dcost(e^{-it\tilde{h}_\sigma^x(i,j)}) &\leq 2(7-1)+4 =16, \\
    \Dcost(e^{-it\tilde{h}_\sigma^y(i,j)}) &\leq 2(10-1)+4 = 22, \\
    \Dcost(e^{-it\tilde{h}_\sigma^z(i,j)}) &\leq 2(12-1)+4 = 26.
\end{align}
When applying controlled implementations of these, two additional CNOT gates for each $Z$ rotation are required.
There are two Pauli strings per hopping term, giving $\Dcost(\text{C}[e^{-it\tilde{h}_\sigma^x(i,j)}])\leq 20$, $\Dcost(\text{C}[e^{-it\tilde{h}_\sigma^y(i,j)}])\leq 26$, $\Dcost(\text{C}[e^{-it\tilde{h}_\sigma^z(i,j)}]) \leq 30$. Here and throughout, the notation $\text{C}[\cdot]$ denotes a controlled operation with respect of to the state of a single qubit.
\end{proof}

\begin{lemma}[Kinetic-Energy Circuit Depth in the Stacked Compact Encoding]\label{Lemma:Pionless_Hopping_Term_Compact}
There is a circuit implementing the kinetic terms $e^{-it\tilde{h}_\sigma(i,j)}$ with circuit depth $\Dcost(e^{-it\tilde{h}_\sigma(i,j)}) \leq 10$. The controlled version can be implemented with circuit depth $\Dcost(C[e^{-it\tilde{h}_\sigma(i,j)}]) \leq 14.$
\end{lemma}

\begin{proof}
The hopping interactions are composed of two Pauli strings of at most weight 4 with the same Pauli operators on all but $2$ qubits, which gives
$\Dcost(e^{-it\tilde{h}_\sigma(i,j})) \leq 2(4-1) + 4 =10$ for hopping operators along any direction. For controlled implementations, $ \Dcost(C[e^{-it\tilde{h}_\sigma(i,j)}])\leq 14$ since each string comes with a $Z$ rotation that can be controlled with two additional CNOTs per $Z$ rotation.
\end{proof}

Crucially, the kinetic Hamiltonian can be implemented with depth $O(1)$ in both the VC and compact encodings.
In the Jordan-Wigner encoding, implementing this term would take depth $O(L^2)$, and other implementations involving the fermionic Fourier transform, fermionic SWAP networks, or Givens rotations all have circuit depths that scale with the number of fermionic modes~\cite{Kivlichan_et_al_2018, Kivlichan_et_al_2020}. 

\subsubsection{Contact Operators}

\begin{lemma}[Contact-Term Circuit Depth in the VC Encoding~\cite{roggero2020quantum}] \label{Lemma:Pionless_Contact_Term_VC}
The circuit in \cref{Fig:Contact_Pionless_Circuit} exactly implements the term $e^{-it\left(H_{\Cpi}(i)+H_{\Dpi}(i)\right)}$ and has circuit depth 8.
The controlled circuit has depth 22.
\end{lemma}

\begin{proof}
Despite the use of the VC encoding, the on-site contact terms we wish to implement have the  same representation when using the Jordan-Wigner encoding.
This allows us to use the optimized circuit developed in Ref.~\cite[Table III and Eq.~(B47)]{roggero2020quantum} (shown in \cref{Fig:Contact_Pionless_Circuit}) to implement the contact interactions.
The circuit implements contact-term time evolution at each site $i$, comprised of operators $e^{-it\left(H_{\Cpi}(i)+H_{\Dpi}(i)\right)} = e^{-it\left(\theta_1\sum_\sigma Z_i^\sigma+\theta_2\sum_{\sigma < \sigma'}Z_i^\sigma Z_i^{\sigma'}+\theta_3\sum_{\sigma < \sigma' < \sigma''}Z_i^\sigma Z_i^{\sigma'}Z_i^{\sigma''}\right)}$. 
Concatenating circuits for all sites does not change the 2-qubit gate depth, so the overall 2-qubit gate depth is
\begin{align}
    \Dcost(e^{-it(H_{\Cpi}(i)+H_{\Dpi}(i))}) \leq 8.
\end{align} 
When performing the controlled evolution,  each 1-qubit rotation is implemented in a controlled manner, giving a 2-qubit gate depth of $14$. On the other hand, only 8 of the 16 CNOT gates need to be controlled, since one can take advantage of the relation $C[UAU^\dagger]=UC[A]U^\dagger$ for any unitary operator $U$ to eliminate the need for control on 4 pairs of CNOT gates. In total, this gives $\Dcost(\text{C}[e^{-it(H_{\Cpi}(i)+H_{\Dpi}(i)))}])\leq 22$.
\end{proof}

\begin{lemma}[Contact-Term Circuit Depth in the Stacked Compact Encoding]\label{Lemma:Pionless_Contact_Term_Compact}
There is a circuit implementing $e^{-it(H_{\Cpi}+H_{\Dpi})}$ in the pionless-EFT Hamiltonian with $\Dcost(e^{-it(H_{\Cpi}(i)+H_{\Dpi}(i))}) \leq 8$ and the controlled version $\Dcost(C[e^{-it(H_{\Cpi}(i)+H_{\Dpi}(i))}]) \leq 22$.
\end{lemma}
\begin{proof}
Since $\tilde{N}_{\sigma}(i) = Z^{\sigma}_i$ in both the CV and compact encoding, the circuit from \cref{Fig:Contact_Pionless_Circuit} can be used again to give circuit depths of 8 and 22 in the non-controlled and controlled cases, respectively.
\end{proof}

\begin{figure}[t]
\centering
\includegraphics[width=0.85\textwidth]{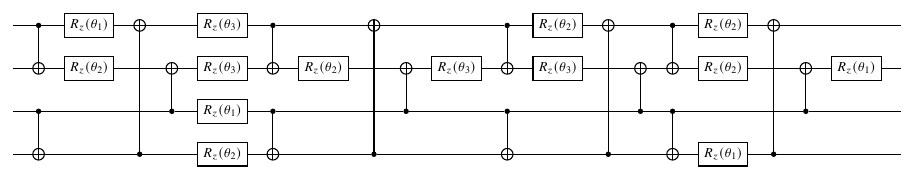}
\caption{The circuit used to implement the time evolution of the contact interaction for the pionless EFT, taken from Ref.~\cite{roggero2020quantum}. 
$R_z(\cdot)$ denotes a $Z$ rotation through the specified angle. }
\label{Fig:Contact_Pionless_Circuit}
\end{figure}

\subsubsection{Total Pionless-EFT Circuit Depth
\label{Sec:Total-Pionless-EFT-Circuit-Depth}}

Here, we examine the costs of simulating the time evolution of pionless EFT for different orders of product formulae.

\begin{lemma}[Pionless-EFT Trotter-Step Circuit Depth in the VC Encoding] 
\label{Lemma:Total_Pionless_EFT_Depth_VC}
The time evolution of the pionless-EFT Hamiltonian in the VC encoding using the $p=1$ Trotter formula can be implemented in circuit depth $\Dcost(\calP^{(\canpi)}_1(t)) \leq 520$ and $\Dcost(C[\calP_1^{(\canpi)}(t))] \leq 630$, where $\calP_1^{(\canpi)}(t)$ is defined in \cref{eq:P-first-order}.
\end{lemma}

\begin{proof}
To run the simulation, the Hamiltonian can be split into 6 layers $H_\gamma$ with $\gamma=1,\ldots,6$. Two layers correspond to two sets of hopping terms along the $x$ direction as depicted in \cref{Fig:Kinetic_Lattice_Decomposition}, in such a way that within each set, the hopping terms commute so that their evolution can be implemented simultaneously. Similarly, hopping along $y$ and $z$ directions each are split into two sets such that within each set, the hopping terms can be simulated simultaneously. Finally, the contact interactions at all sites can be implemented simultaneously. Consequently, the total circuit depth is
\begin{align}
    \Dcost(\calP^{(\canpi)}
    _1(t)) &\leq 2\times 4 (16 +  22 +  26) + 8 
    = 520
\end{align}
independent of the system size,
where we have used the circuit depths $16$, $22$, $26$, and $8$ for simulating hopping terms associated with each of the four nucleon species, $\tilde{h}_\sigma^x(i,j)$, $\tilde{h}_\sigma^y(i,j)$, and $\tilde{h}_\sigma^y(i,j)$, and the contact terms, respectively, as per \cref{Lemma:Pionless_Hopping_Term_VC,Lemma:Pionless_Contact_Term_VC}. Similarly, the controlled evolution takes circuit depth $\Dcost(C[\calP_1^{(\canpi)}(t)]) \leq 2\times 4 (20 +  26 +  30) + 22 = 630$.
\end{proof}

\begin{figure}[t!]
    \centering
    \includegraphics[scale=0.675]{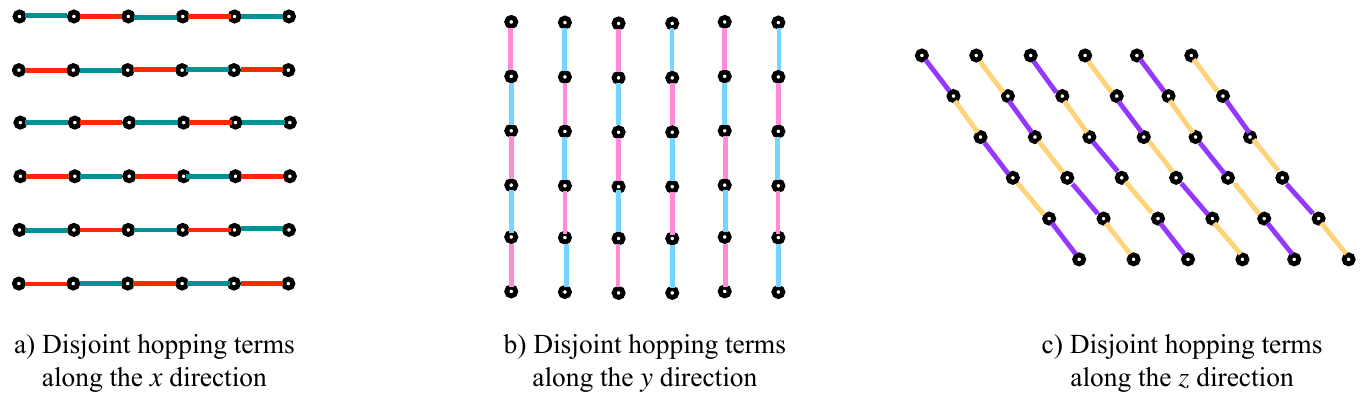}
    \caption{A 2D cross section of the 3D lattice showing how the kinetic hopping terms along a) $x$, b) $y$, and c) $z$ directions are grouped together for each $H_\gamma$ (shown by different colors). The lines connecting the circles denote kinetic hopping terms.
    Terms with the same color can be implemented simultaneously.
    }
\label{Fig:Kinetic_Lattice_Decomposition}
\end{figure}

\begin{lemma}[Pionless-EFT Trotter-Step Circuit Depth in the Stacked Compact Encoding]
\label{Lemma:Total_Pionless_EFT_Depth_Compact}
The time evolution of the pionless-EFT Hamiltonian in the stacked compact encoding using the $p=1$ Trotter formula can be implemented in circuit depth $\Dcost(\mathcal{P}^{(\canpi)}_1(t)) \leq 68$ and $\Dcost(C[\mathcal{P}^{(\canpi)}_1(t)]) \leq 106 $.
\end{lemma}
\begin{proof}
As with the VC encoding, the kinetic terms of each of the fermion species can be split into 6 disjoint sets and all terms within a set can be implemented simultaneously. Considering \cref{Lemma:Pionless_Hopping_Term_Compact,Lemma:Pionless_Contact_Term_Compact}, one arrives at the circuit depths stated in the Lemma.
\end{proof}

Extending these circuits to simulate second-order formulae, we find
the 2-qubit depth costs given in \cref{Table:D_Cost_Pionless_EFT}. Note that according to \cref{Eq:Second-Order-PF-Def} for the second-order formula, the last non-commuting layer ($H_\Gamma$) evolved for time $t/2$ can be combined with the first non-commuting layer of the next evolution for time $t/2$. In other words, only one implementation of $e^{-itH_\Gamma}$ is required, while for the other terms, two separate half-time evolutions occur. Here, we take $H_\Gamma$ to be the layer with the highest circuit depth in each encoding so as to minimize the overall second-order product-formula circuit depth.

\begin{table}
\begin{center}
\resizebox{\textwidth}{!}{\begin{tabular}{ c | c |c | c |c }
\hline \multicolumn{5}{c}{\textbf{Pionless EFT Circuit Depths}}  \\
\hline \begin{tabular}{@{}c@{}}\textbf{Trotter Formula }  \\\textbf{Order} $p$\end{tabular} & \begin{tabular}{@{}c@{}}\textbf{2-Qubit Gate     } \\ \textbf{Circuit Depth (VC)}\end{tabular} &  \begin{tabular}{@{}c@{}}\textbf{Controlled 2-Qubit} \\ \textbf{Gate  Circuit} \\ \textbf{Depth (VC)}\end{tabular} &    \begin{tabular}{@{}c@{}}\textbf{2-Qubit Gate} \\ \textbf{Circuit Depth } \\ \textbf{(Compact)}\end{tabular}    & \begin{tabular}{@{}c@{}}\textbf{Controlled 2-Qubit} \\ \textbf{Gate Circuit} \\ \textbf{Depth (Compact)}\end{tabular}\\
\hline 
     1   & 520 & 630 & 68 & 106 \\ 
        \hline
    2 & 1024 & 1230 & 126 & 190  \\
    \hline
\end{tabular}}
\caption{ Upper bounds on the 2-qubit gate depth for a single step of time evolution under the pionless-EFT Hamiltonian for both the VC encoding and compact encoding.}
\label{Table:D_Cost_Pionless_EFT}
\end{center}
\end{table}

\subsubsection{Total Pionless-EFT \texorpdfstring{$T$}{T}-Gate Cost}

Here, we derive the number of $T$ gates to implement a single $p=1$ Trotter step for the pionless EFT.

\begin{lemma}[Pionless-EFT Trotter-Step $T$-gate Costs in both the VC and Compact Encodings] \label{Lemma:Pionless_T-Gate_Count}
Let $\mathcal{P}_1^{(\canpi)}(t)$ be the $p=1$ product formula for the pionless EFT with the VC encoding. For any $t\in \mathbb{R}$ and $\delta>0$, there exists a circuit that implements a unitary operator $\tilde{V}(t)$ with $\norm{\tilde{V}(t) - \mathcal{P}_1^{(\canpi)}(t)} \leq \delta$, where $\tilde{V}(t)$ has an expected $T$-gate count of $ 42L^3[1.15\log(84L^3/\delta )+9.2]$, with $L$ being the number of lattice sites in each Cartesian direction.
Furthermore, the controlled unitary $C[\tilde{V}(t)]$ has an expected $T$-gate count of $84L^3[1.15\log(168L^3/\delta )+9.2]$.
The same bounds when using the compact encoding.
\end{lemma}

\begin{proof}
    To implement $\mathcal{P}_1^{(\canpi)}
    (t)$ fault-tolerantly, we consider using the repeat-until-success method to synthesize $Z$ rotations. Therefore, we count the number of $R_z$ gates, or in turn the Pauli strings, to obtain the required number of $T$ gates.
    For a single site on the lattice, the number of Pauli strings  that need to be implemented can be obtained by noting that there are 4 species of nucleons, each requiring 3 sets of kinetic terms along each Cartesian direction, with 2 Pauli strings per term. There is an on-site contribution to $\Hfree$ in \cref{Eq:H_free_Hamiltonian} which did not matter in the circuit-depth analysis but involves $4$ $R_z$ gates at each site. Adding to this a total of $14$ $R_z$-rotations involved in each contact-term evolution per site, the overall number of $Z$ rotations is $(4\times 3 \times 2 + 4 + 14) L^3 = 42 L^3$. Thus, each rotation must be done to precision $\delta= 42 L^3$, requiring
    $1.15\log(84L^3/\delta )+9.2$ $T$ gates per rotation by \cref{eq:rus-cost}. The overall expected $T$-gate cost is
    \begin{align}
        42L^3(1.15\log(84L^3/\delta )+9.2).
    \end{align}
In the controlled case, each Pauli string takes twice as many $Z$ rotations, thus requiring $84L^3$ $R_z$ gates, giving an expected $T$-gate cost of $84L^3(1.15\log(168L^3/\delta )+9.2)$.

The same bound holds for the compact encoding as the kinetic term takes the same number of Pauli strings to implement, and the circuit for the contact terms has the same number of Pauli strings.
\end{proof}

\subsection{One-Pion-Exchange EFT Simulation Costs
\label{Sec:OPE-circuits}}

We now turn to the discussion of simulation costs for the OPE EFT. The hopping terms are the same for the OPE Hamiltonian as for the pionless EFT, hence the circuit depths quoted in \cref{Sec:Hopping_Depth} apply equally to this model.
We thus proceed with analyzing the simulation cost of the on-site contact terms and the long-range interactions.

\subsubsection{Contact Operators}
\begin{lemma}
\label{Lemma:Contact-Depth-OPE}
There exists a circuit implementing $e^{-itH_C(i)}$ and $e^{-itH_{C_{I^2}}(i)}$ exactly with circuit depths $\Dcost(e^{-itH_C(i)}) \leq 6$ and $\Dcost(e^{-itH_{C_{I^2}}(i)}) \leq 54$.
The controlled versions can be implemented exactly with circuit depths $\Dcost(C[e^{-itH_C(i)}]) \leq 26$ and $\Dcost(C[e^{-itH_{C_{I^2}}(i)}])\leq 98$.
\end{lemma}
\begin{proof}

$H_C(i)$ as given in \cref{Eq:H_C_OPE_Paulis} acts on pairs of nucleons on each spatial lattice site, of which there are 6, but each 2 pairs with non-shared qubits can be implemented simultaneously, leading to only 3 non-commuting layers. Each term has Pauli weight 2, hence each requiring CNOT-gate depth 2 to implement. This gives a total circuit depth of
\begin{align}
    \Dcost(e^{-itH_C(i)}) \leq 6.
\end{align}
For the controlled operation, besides the 3 sets of non-commuting layers of CNOT gates, each with depth 2, each layer also contains 2 1-qubit $Z$ rotations that need to be controlled, giving an additional 2-qubit depth of 4 in each layer. Finally, the product of $N_{\sigma}(i)N_{\sigma'}(i)$ for $\sigma \neq \sigma'$ creates 4 1-qubit $Z$ rotations associated with each of the 4 qubits representing the fermionic species at a site, which when controlled give a total of 8 CNOT gates. Therefore, $\Dcost(C[e^{-itH_C(i)}]) \leq 3\times(2+4)+8=26$.

For $H_{C_{I^2}}(i)$ given in \cref{Eq:explicit-HCI2}, there are two types of term: ones consisting of only number operators and a term of the form $\adag_{\uparrow p}(i)a_{\uparrow n}(i)\adag_{\downarrow n}(i)a_{\downarrow p}(i) + \text{h.c.}$ The latter contribution is decomposed in \cref{Eq:H_CI2_Paulis} into 8 non-commuting Pauli operators of weight 4, hence each requiring a CNOT-gate depth 6 to implement. Out of the terms consisting of number operators, $N_{\uparrow p}^2(i)+N_{\downarrow p}^2(i)+N_{\uparrow n}^2(i)+N_{\downarrow n}^2(i)$ contains no 2-qubit rotations, while the remainder of the terms in  \cref{Eq:explicit-HCI2} consist of all $N_{\sigma}(i)N_{\sigma'}(i)$ operators with $\sigma \neq \sigma'$, hence exhibiting the same structure as $H_C(i)$ above. This means that these terms can be implemented in a total circuit depth 6. Overall,
\begin{align}
    \Dcost(e^{-itH_{C_{I^2}}(i)}) \leq (8 \times 6)+6=54.
    \label{Eq:H_CI^2_Decomposition}
\end{align}
The controlled operation of the term $\adag_{\uparrow p}(i)a_{\uparrow n}(i)\adag_{\downarrow n}(i)a_{\downarrow p}(i) + \text{h.c.}$ demands a circuit depth of $8 \times 8=64$ since each of the eight Pauli strings now needs 8 CNOT gates. The controlled $N_{\uparrow p}^2(i)+N_{\downarrow p}^2(i)+N_{\uparrow n}^2(i)+N_{\downarrow n}^2(i)$ operator results in 4 controlled $Z$ rotations on each of the qubits, hence a CNOT-gate depth 8. The remaining terms require the same 2-qubit gate depth as the controlled simulation of $H_C(i)$, which takes a circuit depth of 26. Putting this all together gives $\Dcost(C[e^{-itH_{C_{I^2}}(i)}]) \leq 64+8+26=98$.
\end{proof}

\subsubsection{Long-Range Operators}

We now consider the circuit depths to implement the long-range terms. 
Recall that these terms are truncated such that only those acting between sites within certain distance from each other are included.

\begin{lemma}
\label{Lemma:Long-Range-Depth-OPE}
There exists a circuit that implements $e^{-itH_{\rm LR}(i,j)}$ between all pairs of points $(i,j)$ at distance $|\bm{x}-\bm{y}| \leq \ell$, where $i$ ($j$) denotes the qubit index of spatial site $\bm{x}$ ($\bm{y}$), respectively. The circuit has a 2-qubit gate depth $\Dcost(e^{-itH_{\rm LR}(i,j)}) \leq 14336$.
The controlled version has a circuit depth $\Dcost(C[e^{-itH_{\rm LR}(i,j)}]) \leq 
16384$.
\end{lemma}
\begin{proof}

The long-range terms, as given in \cref{Eq:OPE_Long-Range_Terms}, consist of pairs of creation and annihilation operators acting on different sites.
As per \cref{Sec:Long_Ranged_Pauli}, the terms decompose into a set of at most 8 weight-8 Pauli strings.
Each term requires a CNOT-gate depth of $2\times(8-1)=14$ to simulate.
For a given pair of sites, one needs to determine the number of terms coupling nucleons on those sites, which can be obtained by counting all possible combinations of terms. 
At site $\bm{x}$, the creation operator can act on 4 possible terms, as can the annihilation operator, giving a total of $16=2^4$ possible terms.
The same is true at site $\bm{y}$, so for a pair of sites, there are $2^8$ possible combinations.
Nonetheless, each of the 16 terms at each site consists of 4 number operators and 6 Hermitian-conjugate pairs. These generate 16 combinations of the form $N_\sigma(\bm{x})N_{\sigma'}(\bm{y})$, $4 \times 6 = 24$ operators of the form $N_\sigma(\bm{x}) \adag_{\alpha,\beta}(\bm{y})a_{\gamma\delta}(\bm{y}) +\text{h.c.}$ with $\alpha \neq \gamma, \beta \neq \delta$ (and 24 operators with $\bm{x} \leftrightarrow \bm{y}$), and $(12 \times 12)/2=72$ combinations that involve no number operators, for a total of $136$ terms. To simplify the circuit-depth analysis, we skip such a refined analysis and simply assume the $2^8$ possible terms can be reduced to $2^7=128$ pairs, where each pair is composed of 8 Pauli strings with a Pauli weight of at most 8, as in \cref{Eq:Long_Ranged_Pauli_String}. This still leads to a rigorous upper bound on the circuit depth since a number operator has a Pauli weight of half or less compared with the $\adag_{\alpha,\beta}(\bm{y})a_{\gamma\delta}(\bm{y})$ operator with $\alpha \neq \gamma, \beta \neq \delta$, hence justifying the division of the total number of terms by 2 and using the highest-weight term to obtain an upper bound.
Thus, one can simulate a pair of terms between given sites with circuit depth\footnote{An improved circuit depth with a more refined counting of the terms and their contributions can be obtained as follows. As mentioned, there are 16 combinations of the form $N_\sigma(\bm{x})N_{\sigma'}(\bm{y})$, which are at most of Pauli weight 2. Next, there are $4 \times 6 = 24$ operators of the form $N_\sigma(\bm{x}) \adag_{\alpha,\beta}(\bm{y})a_{\gamma\delta}(\bm{y}) +\text{h.c.}$ with $\alpha \neq \gamma, \beta \neq \delta$ (and 24 operators with $\bm{x} \leftrightarrow \bm{y}$). 
Each of these operators can be decomposed into 8 Pauli strings of weight at most 5, and all these Pauli strings commute. Finally, there are $(12 \times 12)/2=72$ combinations that involve no number operators. They can consist of a maximum of 8 Pauli strings with a weight of at most 8, as shown in \cref{Eq:Long_Ranged_Pauli_String}. Therefore, each long-range Hamiltonian $H_{\rm LR}(\bm{x},\bm{y})$ for $\bm{x} \neq \bm{y}$, when expressed in terms of Pauli operators, consists of 16 strings with Pauli weight 2, $48 \times 8= 384$ strings of Pauli weight of at most 5, and $72 \times 8=576$ strings of Pauli weight of at most 8, giving a total circuit depth $16 \times 2+ 384 \times 8 +576 \times 14 = 11168$. This is lower than the circuit depth assuming a total of 128 combinations of pairs of fermionic operators but each with the highest number of Pauli strings and the highest weight.  This number can be further improved by noting that not all the terms within each set have the maximum weight. In general, we do not account for such fine-grained resource counting unless the benefit is significant.}
\begin{align}
    \Dcost(e^{-itH_{\rm LR}(i,j)}) \leq  2^7\times 14 \times 8 = 14336.
\end{align}
For the controlled version, one obtains $\Dcost(C[e^{-itH_{\rm LR}(i,j)}])\leq  2^7\times 16 \times 8 = 16384$.
\end{proof}

\subsubsection{Total OPE-EFT Circuit Depth
\label{Sec:Total-OPE-EFT-Circuit-Depth}}
Here, we examine the costs of simulating a single time step of evolution of the OPE EFT for different orders of product formulae.

\begin{figure}
    \centering
    \includegraphics[scale=0.675]{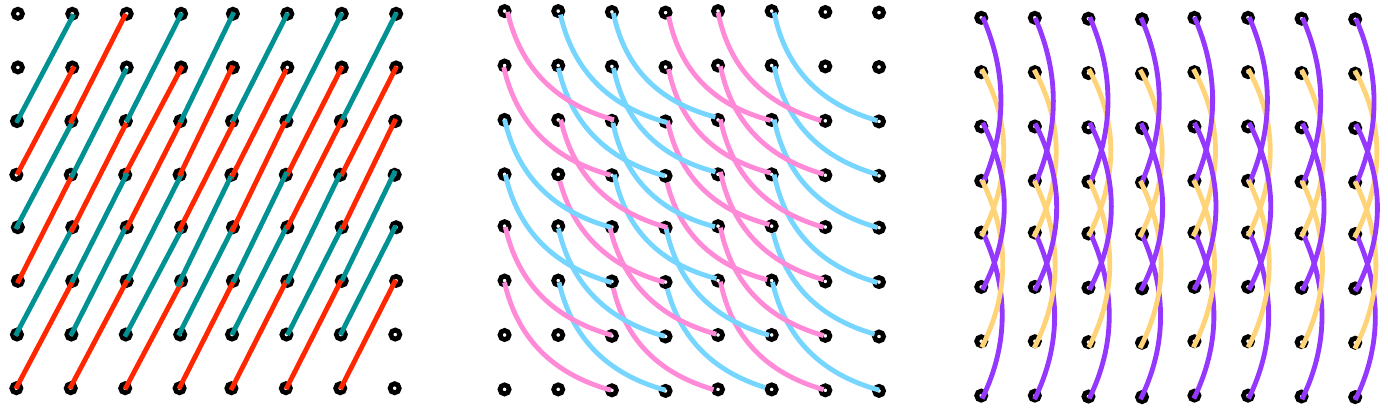}
    \caption{A sample of interaction terms present in the long-range Hamiltonian. In each figure, the interactions denoted by the same color are commuting and can be simulated in parallel, while the two sets with different colors in each figure are non-commuting and must be applied in series. Similar interactions can be pictured along other 3D lattice directions and with various lengths up to cutoff $\ell$. All these interaction types, nonetheless, can be separated into two disjoint sets in a similar way.
    }
    \label{Fig:LR_Layers}
\end{figure}

\begin{lemma}[OPE-EFT Trotter-Step Circuit Depth]
\label{Lemma:Total-Depth-OPE}
The time evolution of the OPE-EFT Hamiltonian using the $p=1$ Trotter formula can be implemented in circuit depth $\Dcost(\calP^{\rm(OPE)}
_1(t)) \leq 572+ 14336 \ R(\ell)$ and $\Dcost(C[\calP_1^{\rm(OPE)}(t))] \leq 732+ 16384 \ R(\ell)$, where $R_\ell \coloneqq \lceil 4\pi (\ell+1)^3/3 \rceil$ and $\calP_1(t)$ is defined in \cref{eq:P-first-order}.
\end{lemma}

\begin{proof}
First, the free Hamiltonian, $\Hfree$, can be implemented as already discussed in \cref{Sec:Total-Pionless-EFT-Circuit-Depth} for the pionless-EFT case, with a circuit depth of 512 (and 608 for the controlled case). The contact terms $H_C$ and $H_{C_{I^2}}$ for all lattice sites can be simulated with a circuit depth $60$ (and $124$ for the controlled case), independent of system size. 

Now to simulate $H_{\rm LR}$, more than just a single pair of lattice sites must be implemented, i.e., one needs to consider all possible pairs of interacting terms with interaction length less than the cutoff $\ell$, while taking advantage of possible parallelizations to reduce the circuit depth. 
For each interaction type, i.e., with given directionality and range, the interactions can be divided into two non-commuting layers, where within each layer all interactions commute and can be applied in parallel (see \cref{Fig:LR_Layers} for a few examples). This is because each site participates in only two interaction bonds of a given type, so by walking along bonds from site to site, the colors alternate. Therefore, to obtain the total number of interaction layers to be applied in series, it suffices to find the number of all possible interaction types.
The number of sites in a cubic lattice within distance $\ell$ of the origin is upper bounded by\footnote{The exact expression is not known except asymptotically. Exact values for small instances can be found in Ref.~\cite[A117609]{oeis}.} $\frac{4\pi}{3}(\ell+\frac{\sqrt{3}}{2})^3$.
This is twice the number of interaction types that need to be simulated on the 3D lattice.
Hence, the circuit depth satisfies
\begin{align}
    \Dcost(e^{-itH_{\rm LR}}) &\leq 2\times 14336 \times \frac{1}{2}\left\lceil \frac{4\pi }{3}\left(\ell+1\right)^3 \right\rceil
    = 14336 \ R(\ell),
\end{align}
where the factor of $2$ arises from the two disjoint sets of interactions associated with each type. Furthermore, the controlled evolution takes depth $\Dcost(C[e^{-itH_{\rm LR}}]) \leq 16384 \ R(\ell)$.

Adding the circuit depth of the hopping, contact, and long-range terms gives the claimed costs.
\end{proof}

\subsubsection{Total OPE EFT \texorpdfstring{$T$}{T}-Gate Cost}

We now examine the number of $T$ gates to implement a single $p=1$ Trotter step for the OPE EFT.

\begin{lemma}[OPE-EFT Trotter Step $T$-gate Costs]\label{Lemma:T-gate-OPE}
For any $t\in \mathbb{R}$ and $\delta>0$, there exists a circuit that implements a unitary operator $\tilde{V}(t)$ such that $\norm{\tilde{V}(t) - \mathcal{P}_1^{\rm (OPE)}(t) } \leq \delta$, where $\tilde{V}(t)$ has an expected $T$-gate count of $g(L,\ell)(1.15\log(2g(L,\ell)/\delta )+9.2)$, 
where $g(L,\ell)\coloneqq \left(52+1024\left\lceil \frac{4\pi}{3} \left(\ell+1\right)^3 \right\rceil \right)L^3$. Here $L$ is the total number of lattice sites in each Cartesian direction, and $\ell$ is the cutoff length introduced in \cref{Sec:Cutoff_OPE}.
The expected number of $T$ gates for the controlled unitary $C[\tilde{V}(t)]$ is $2g(L,\ell)(1.15\log(4g(L,\ell)/\delta )+9.2)$.
\end{lemma}

\begin{proof}
    To implement $\mathcal{P}_1^{\rm (OPE)}(t)$ fault-tolerantly, we use the repeat-until-success method to synthesize the $R_z$ gates, and assume an equal error for each of the $T$ gates in the synthesis. The number of $T$ gates is determined from the total number of 1-qubit $Z$ rotations, which can be counted as follows.
    $\Hfree$ is implemented with $28L^3$ total $Z$ rotations as with the pionless EFT.
    For $H_C$, there are 6 possible pairs of $Z_i^\sigma Z_i^{\sigma'}$ rotations with $\sigma \neq \sigma'$, as well as four possible $Z_i^\sigma$ rotations, all generated out of $N_\sigma(i)N_{\sigma'}(j)$ terms. Each of these requires one $Z$ rotation to implement, giving $10L^3$ $Z$ rotations in total. However, 4 of these rotations can be combined with 1-qubit rotations from implementing $e^{-it\Hfree}$.
    For $H_{C_{I^2}}$, the $\adag_{\uparrow p}(i)a_{\uparrow n}(i)\adag_{\downarrow n}(i)a_{\downarrow p}(i) + \text{h.c.}$ term contains 8 Pauli strings, leading to 8 $Z$ rotations. Out of the terms consisting of number operators, $N_{\uparrow p}^2(i)+N_{\downarrow p}^2(i)+N_{\uparrow n}^2(i)+N_{\downarrow n}^2(i)$ consists of 4 $Z$ rotations, while the remainder of the terms have the same structure as $H_C(i)$, leading to 10 $Z$ rotations. However, 4 of these can be combined with the $Z$ rotations from the $N_{\uparrow p}^2(i)+N_{\downarrow p}^2(i)+N_{\uparrow n}^2(i)+N_{\downarrow n}^2(i)$ operator. So in total, simulating $H_{C_{I^2}}$ requires $18L^3$ $Z$ rotations. 
    For $H_{\rm LR}$, between any two lattice sites, there are $2^7$ Hermitian terms to implement, and each decomposes into $8$ Pauli strings.
    For each lattice point, there are up to $\left\lceil \frac{4\pi }{3}\left(\ell+1\right)^3 \right\rceil$ points within distance $\ell$ asymptotically, which determines the number of interaction terms to be simulated at each site.
    Hence, the total number of $Z$ rotations is $8 \times 2^7 \times \left\lceil \frac{4\pi }{3}\left(\ell+1\right)^3 \right\rceil L^3  =  1024 \left\lceil \frac{4\pi }{3}\left(\ell+1\right)^3 \right\rceil L^3 $.

Now defining the function 
\begin{align}
    g(L,\ell) \coloneqq \left(  28+10-4+18+1024\left\lceil \frac{4\pi }{3}\left(\ell+1\right)^3 \right\rceil \right)L^3,
\end{align}
each rotation should be done to precision $\delta/g(L,\ell)$, giving $1.15\log(2g(L,\ell)/\delta )+9.2$ $T$ gates per rotation, on average. Thus, the expected overall $T$-gate cost is
\begin{align}
    g(L,\ell)(1.15\log(4g(L,\ell)/\delta )+9.2).
\end{align}
In the controlled case, each Pauli string takes twice as many $Z$ rotations, giving $2g(L,\ell)$ $Z$ rotations, and an overall requirement of
$
        2g(L,\ell)(1.15\log(2g(L,\ell)/\delta )+9.2)
$ expected $T$ gates.
\end{proof}

\subsection{Dynamical-Pion EFT Simulation Costs}
\label{Sec:Pionful_Circuit_Costs}

We now discuss the simulation costs for the dynamical-pion EFT. The costs of simulating the free-nucleon Hamiltonian are essentially the same as in the previous models, so here we focus on the pion and pion-nucleon terms in the Hamiltonians.

\subsubsection{The operator proportional to \texorpdfstring{$\pi_I^2$}{(πI)²}}\label{Sec:Pion_Field_Strength_Depth}

\begin{lemma} \label{Lemma:pi^2_Circuit_Depth}
\label{Lemma:piI-Squared-Depth}
There exists a circuit that implements the term $e^{-itH_{\pi^2}}$ with $H_{\pi^2} \coloneqq \frac{m_\pi^2a_L^3}{2}\sum_{\bm{x},I}\pi_I(x)^2$ on $n_b$ qubits with circuit depth $\Dcost(e^{-itH_{\pi^2}}) \leq 2\left\lceil n_b/2 \right\rceil+2n_b-4$.
The controlled version can be implemented in $\Dcost(C[e^{-itH_{\pi^2}}) \leq n_b^2+2\left\lceil n_b/2 \right\rceil+3n_b-4$.
\end{lemma}
\begin{proof}

One can use the decomposition introduced earlier in \cref{Eq:piI-Squared-Decomposition} to write
\begin{align}
\label{Eq:Exp_piI_Squared}
 e^{-it\frac{m_\pi^2a_L^3}{2}\pi_I^2(\bm{x})} = e^{-it\frac{m_\pi^2a_L^3}{2}\left[P^2\mathds{1}+2PQ\sum_{m=0}^{n_b-1} 2^m Z^{(m)}_{I,\bm{x}}+Q^2\sum_{m,m'=0}^{n_b-1} 2^{m+m'} Z^{(m)}_{I,\bm{x}}Z^{(m')}_{I,\bm{x}}\right]},
\end{align}
where $P$ and $Q$ are constants defined after \cref{Eq:piI-Squared-Decomposition}. 
This implementation uses $n_b(n_b-1)/2$ $ZZ$ rotations, or $n_b(n_b-1)$ CNOT gates, between pairs of qubits. Nonetheless, it can be shown that the operations can be parallelized, improving the circuit depth. Consider an $n_b$-qubit circuit which involves all possible $ZZ$ rotations among pairs of qubits, and let $d$ denote the distance between the qubits. The distance $d$ takes values between 1 and $n_b-1$. As is clear from the examples shown in \cref{Fig:pi_Squared_Layers}, all pairs of interactions with a fixed value of $d$ can be either implemented all in parallel (when $d \geq \lceil n_b/2 \rceil$) or can be split into two sets (when $d < \lceil n_b/2 \rceil$) where interactions within each set can all be implemented in parallel. This means that there are
\begin{align}
\label{Eq:ZZ-parallelization}
    2\times\left(\left\lceil \frac{n_b}{2} \right\rceil-1\right)+1\times\left(n_b-\left\lceil \frac{n_b}{2} \right\rceil\right)
\end{align}
separate layers of $ZZ$ rotations, or twice this value for the layers of CNOT gates, which should be implemented in series. Hence, the 2-qubit circuit depth of the circuit is twice that in \cref{Eq:ZZ-parallelization}. Note that this depth scales as $n_b$, which is an improvement over the $n_b^2$ scaling of the naive implementation.\footnote{Our strategy is different from that of Ref.~\cite{shaw2020quantum}, which gives a CNOT-gate count of $(n_b+2)(n_b-1)/2$. Here, we focus on optimizing the circuit depth rather than the CNOT count.} Thus, the circuit depth of the $e^{-itH_{\pi^2}}$ operator is
\begin{align}
\label{Eq:Dcost-pI-Squared}
    \Dcost(e^{-itH_{\pi^2}}) \leq 4\left(\left\lceil \frac{n_b}{2} \right\rceil-1\right)+2\left(n_b-\left\lceil \frac{n_b}{2} \right\rceil\right)=2\left\lceil \frac{n_b}{2} \right\rceil+2n_b-4,
\end{align}
where we have taken into account the fact that each $\pi_I(\bm{x})$ acts on a distinct set of qubits, so the full evolution can be done in a circuit depth independent of the system size.

The controlled-unitary circuit depth can be obtained by considering that, first of all, there are $n_b$ 1-qubit $Z$ rotations associated with the term proportional to $PQ$ in \cref{Eq:Exp_piI_Squared}, that once controlled, lead to 2-qubit circuit depth $2n_b$. Then, there are operators proportional to $Q^2$, which involve $n_b(n_b-1)/2$ $ZZ$ rotations, leading to the same number of 1-qubit $Z$ rotations when decomposed into CNOT gates. 
When controlled, each of these produces 2 CNOT gates, which must be added to the circuit depth of uncontrolled evolution in \cref{Eq:Dcost-pI-Squared}. Finally, if the control is performed upon separate ancilla qubits for each $\pi_I^2(\bm{x})$ term, evolution of each can be performed in parallel with the rest, keeping the circuit depth system-size independent. Therefore, in total, we arrive at a circuit depth $D[C[e^{-itH_{\pi^2}})]] \leq 2n_b+n_b(n_b-1)+\left(2\left\lceil n_b/2 \right\rceil+2n_b-4\right)=n_b^2+2\left\lceil n_b/2 \right\rceil+3n_b-4$.
\end{proof}

\begin{figure}[t!]
    \centering
    \includegraphics[scale=0.685]{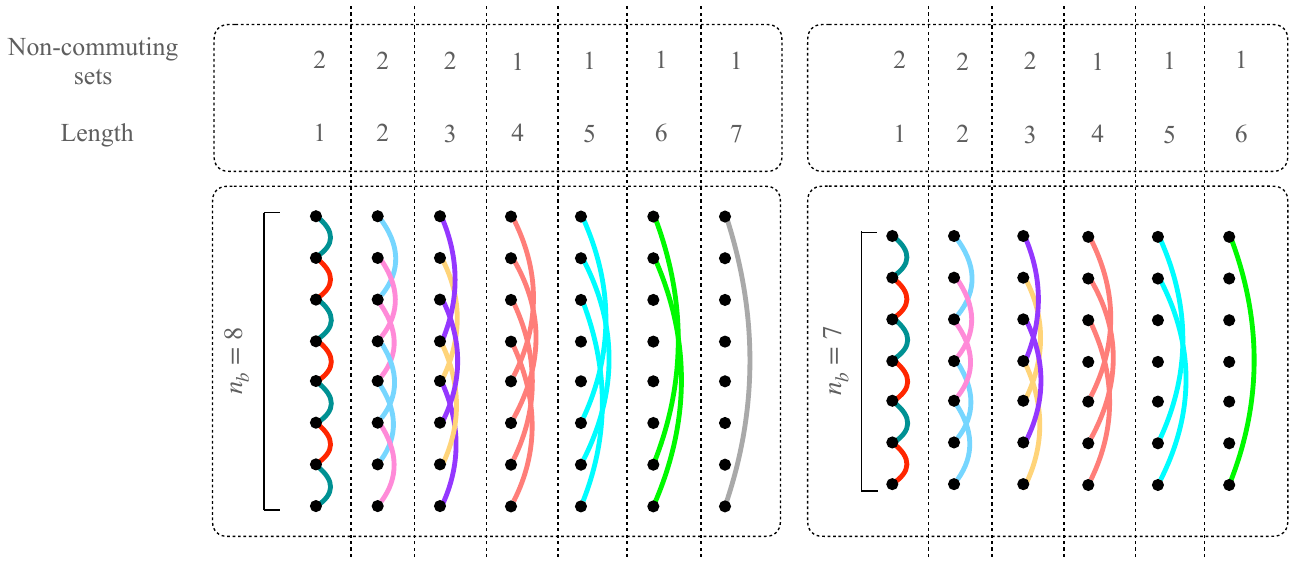}
    \caption{Examples of 8-qubit and 7-qubit systems, where colored lines between the filled circles represent entangling $ZZ$ rotations between each pair of qubits. Interactions with the same color commute and can be applied simultaneously. The total number of interactions in each case is $n_b(n_b-1)/2$, but these can be applied in $2\times\left(\left\lceil n_b/2 \right\rceil-1\right)+1\times\left(n_b-\left\lceil n_b/2 \right\rceil\right)$ separate layers. Note that further optimization is possible to reduce the number of interaction sets that are applied in series, but the circuit depth will still scale with $n_b$.}
    \label{Fig:pi_Squared_Layers}
\end{figure}

\subsubsection{The operator proportional to \texorpdfstring{$(\nabla\pi_I)^2$}{(∇πI)²}
}\label{Sec:Pion_Field_Derivative_Depth}

\begin{lemma} \label{Lemma:Del-piI-Squared-Depth}
There is a circuit that implements the term $e^{-itH_{(\nabla \pi)^2}}$  with $H_{(\nabla \pi)^2} \coloneqq \frac{a_L}{2}\sum_{\langle \bm{x},\bm{y}\rangle,I}(\pi_I(\bm{x})-\pi_I(\bm{y}))^2$ on $n_b$ qubits, where $\langle \bm{x},\bm{y}\rangle$ denotes nearest-neighbor sites, with circuit depth $\Dcost(e^{-itH_{(\nabla \pi)^2}})\leq 12 \left\lceil n_b/2 \right\rceil+24n_b-24$.
The controlled version can be implemented in circuit depth $\Dcost(C[e^{-itH_{(\nabla \pi)^2}}])\leq 24n_b^2+12 \left\lceil n_b/2 \right\rceil+36n_b-24$.
\end{lemma}

\begin{proof}
The decomposition in \cref{Eq:Del-piI-Squared-Decomposition} can be used to write
\begin{align}
\label{Eq:pix-minus-piy-Decomposion}
&e^{-it\frac{a_L}{2}(\pi_I(\bm{x})-\pi_I(\bm{y}))^2}
=e^{-it\frac{a_L}{2}\pi_I(\bm{x})^2}e^{-it\frac{a_L}{2}\pi_I(\bm{y})^2}e^{it a_L\pi_I(\bm{x})\pi_I(\bm{y})}
\nonumber\\
&\hspace{2em}=e^{-it\frac{a_LQ^2}{2} \sum_{m,m'=0}^{n_b-1} 2^{m+m'} Z^{(m)}_{I,\bm{x}}Z^{(m')}_{I,\bm{x}}}
e^{-it\frac{a_LQ^2}{2} \sum_{n,n'=0}^{n_b-1} 2^{n+n'} Z^{(n)}_{I,\bm{y}}Z^{(n')}_{I,\bm{y}}}e^{it a_LQ^2\sum_{m,n=0}^{n_b-1} 2^{m+n} Z^{(m)}_{I,\bm{x}}Z^{(n)}_{I,\bm{y}}}.
\end{align}
Each of the first two exponentials has a circuit depth of $2\left\lceil n_b/2 \right\rceil+2n_b-4$ according to \cref{Lemma:pi^2_Circuit_Depth}. 
These two can be simulated simultaneously as they act on distinct sites. 
The last exponential in \cref{Eq:pix-minus-piy-Decomposion} consists of $n_b^2$ $ZZ$ rotations on distinct qubits, giving a CNOT-gate depth of $2n_b^2$. Nonetheless, the same parallelization strategy as in \cref{Lemma:pi^2_Circuit_Depth} can be applied to improve this depth. In particular, the interactions involved are a special case of the general circuit considered before, in which now instead of all $ZZ$ rotations among the $2n_b$ qubits, only interactions with length $d \geq \left \lceil 2n_b/2 \right \rceil$ are allowed. This means that only 
\begin{align}
1 \times \left( 2n_b -\left \lceil \frac{2n_b}{2} \right \rceil \right)=n_b
\end{align}
separate layers of $ZZ$ rotations, or twice this value for the layers of CNOT gates, need to be implemented in series. Therefore, each $e^{it a_L\pi_I(\bm{x})\pi_I(\bm{y})}$ has a 2-qubit circuit depth of $2n_b$.

Now for the full time-evolution operator, observe that $\sum_{\langle \bm{x},\bm{y}\rangle}e^{-it \frac{a_L}{2}(\pi_I(\bm{x})-\pi_I(\bm{y}))^2}$ acts on adjacent sites. We apply the same strategy used for the fermionic-hopping simulation to separate the terms into two disjoint sets along each of the three Cartesian directions, where within each set, all terms can be applied together. Furthermore, kinetic operators associated with each isospin component of the pion act on distinct sets of qubit registers and can be all applied at once. Putting everything together gives
\begin{align}
\Dcost(e^{-itH_{(\nabla \pi)^2}})\leq 6 \times \left(2\left\lceil \frac{n_b}{2} \right\rceil+2n_b-4+2n_b\right)=12 \left\lceil \frac{n_b}{2} \right\rceil+24n_b-24.
\end{align}

For the controlled evolution, besides the circuit depth for the uncontrolled evolution, we count and add the circuit depth associated with controlled 1-qubit $Z$ rotations. There are $n_b+n_b(n_b-1)/2$ 1-qubit $Z$ rotations associated with each of the $e^{-it\frac{a_L}{2}\pi_I(\bm{x})^2}$ and $e^{-it\frac{a_L}{2}\pi_I(\bm{y})^2}$ operators, and $n_b^2$ 1-qubit $Z$ rotations for the $e^{it a_L\pi_I(\bm{x})\pi_I(\bm{y})}$ operator. Again, introducing separate ancilla qubits for each $\langle \bm{x},\bm{y}\rangle$ keeps the circuit depth independent of system size. Overall, $\Dcost(C[e^{-itH_{(\nabla \pi)^2}}])\leq 12 \left\lceil n_b/2 \right\rceil+24n_b -24 + 6 \times 2 \left(2n_b+n_b(n_b-1)\right) + 6 \times 2 n_b^2=24n_b^2+12 \left\lceil n_b/2 \right\rceil+36n_b-24$.
\end{proof}

\subsubsection{The operator proportional to \texorpdfstring{$\Pi_I^2$}{(ΠI)²}}

\begin{lemma}
\label{Lemma:Conjugate-PiI-Squared-Depth}
The operation $e^{-itH_{\Pi^2}}$ with $H_{\Pi^2} \coloneqq \frac{a_L^3}{2}\sum_{\bm{x},I}\Pi_I(x)^2$, acting on $n_b$ qubits, can be implemented with circuit depth $\Dcost(e^{-itH_{\Pi^2}}) \leq 2n_b^2+2\left\lceil n_b/2 \right\rceil-4$.
The controlled version can be implemented with circuit depth $\Dcost(C[e^{-itH_{\Pi^2}}]) \leq 3n_b^2+2\left\lceil n_b/2 \right\rceil+n_b-4$.
\end{lemma}

\begin{proof}
To implement the operator composed of the conjugate-momentum field while working in the field basis, we transform $\Pi$ via a QFT to $\tilde{\Pi}$, which has a diagonal representation in the field basis (see \cref{Sec:bosonic-field-encodings}). Then, the operator to be implemented is
\begin{align}
\label{Eq:Pi_Squared_with_QFTs}
    e^{-it\frac{a_L^3}{2}\Pi_I(x)^2}={U_{\rm QFT}^{(I)}}^\dagger e^{-it\frac{a_L^3}{2}\tilde{\Pi}_I(x)^2} U_{\rm QFT}^{(I)},
\end{align}
where $U_{\rm QFT}^{(I)}$ is the unitary implementing the QFT on an $n_b$-qubit register encoding $\pi_I$, which has 2-qubit circuit depth $n_b(n_b-1)$~\cite[Sec.~5.1]{Nielsen_Chuang_2010}. The $e^{-it\frac{a_L^3}{2}\tilde{\Pi}_I(x)^2}$ operator can be implemented in the same way as the $e^{-it\frac{m_\pi^2a_L^3}{2}\pi_I(x)^2}$ operator, with circuit depth $2\left\lceil n_b/2 \right\rceil+2n_b-4$ according to \cref{Lemma:pi^2_Circuit_Depth}. Finally, terms associated with different $I$ and $\bm{x}$ can be implemented simultaneously. Therefore, in total,
\begin{align}
    \Dcost(e^{-itH_{\Pi^2}}) = 2\Dcost(U_{\rm QFT}^{(I)})+\Dcost(e^{-itH_{\tilde{\Pi}^2}})
    &\leq 2\times n_b(n_b-1)+ 2\left\lceil \frac{n_b}{2} \right\rceil+2n_b-4 = 2n_b^2+2\left\lceil \frac{n_b}{2} \right\rceil-4.
\end{align}
For the controlled version, since the QFT unitaries do not have to be controlled, the circuit depth of the controlled evolution is equal to twice that of $U_{\rm QFT}$, plus that of $C[e^{-it\frac{a_L^3}{2}\tilde{\Pi}_I(x)^2}]$, which has a circuit depth analyzed in \cref{Lemma:pi^2_Circuit_Depth}. Putting these together gives $\Dcost(C[e^{-itH_{\Pi^2}}]) 
    \leq 2\times n_b(n_b-1)+ n_b^2+2\left\lceil n_b/2 \right\rceil+3n_b-4 
    = 3n_b^2+2\left\lceil n_b/2 \right\rceil+n_b-4$, where again we have assumed that one ancilla qubit is available per $I$ and $\bm{x}$.
\end{proof}

\subsubsection{The Axial-Vector Hamiltonian}

\begin{lemma}
\label{Lemma:HAV_Circuit_Depth}
Let $H_{\rm AV}$ be the pion-nucleon axial-vector interaction in \cref{Eq:H-AV}. There exists a circuit implementing $e^{-itH_{\rm AV}}$ on $n_b$ qubits with circuit depth $\Dcost(e^{-itH_{\rm AV}}) \leq 1296+864n_b$. 
The controlled version takes circuit depth $\Dcost(C[e^{-itH_{\rm AV}}]) \leq 1296+1728n_b$.
\end{lemma}

\begin{proof}
$H_{\rm AV}$ given in \cref{Eq:H-AV} involves a summation over spin $S$ and isospin $I$, giving $3\times3=9$ different combinations. Each term with given $S$ and $I$ is composed of at most 4 combinations of creation and annihilation operators, which can be considered as 2 Hermitian-conjugate pairs for simplicity. Further, each of the $9 \times 2=18$ terms consists of $4n_b$ Pauli strings with Pauli weight of at most 5 [see \cref{Eq:AV_Paulis}]. These $4n_b$ terms can be divided into 2 sets of terms each containing $2n_b$ $Z$ operators, so that within each set, the strings share four fermionic Pauli operations. It is then easy to see that each set can be implemented with 2-qubit circuit depth $2(4-1)+4n_b$. A representative circuit of a highest-weight term implementation is shown in \cref{Fig:AV_Circuit}.
Finally, note that $H_{\rm AV}$ couples nucleons on nearest-neighbor sites, hence introducing the familiar factor of 6 into the overall circuit depth (i.e., two sets of disjoint interactions along each of the three Cartesian coordinates).
Putting everything together gives
\begin{align} 
    \Dcost(e^{-itH_{\rm AV}}) &\leq 6\times 18\times 2 \times (6+4n_b)=1296+864n_b.
\end{align}
For the controlled version, one should account for extra $2\times 4n_b$ 1-qubit $Z$ rotations to be controlled within each term. Therefore, $\Dcost(C[e^{-itH_{\rm AV}}]) \leq (1296+864n_b)+6\times18\times (2\times 4n_b) =1296+1728n_b$.
\end{proof}

\begin{figure}[t!]
\centering
\includegraphics[scale=0.625]{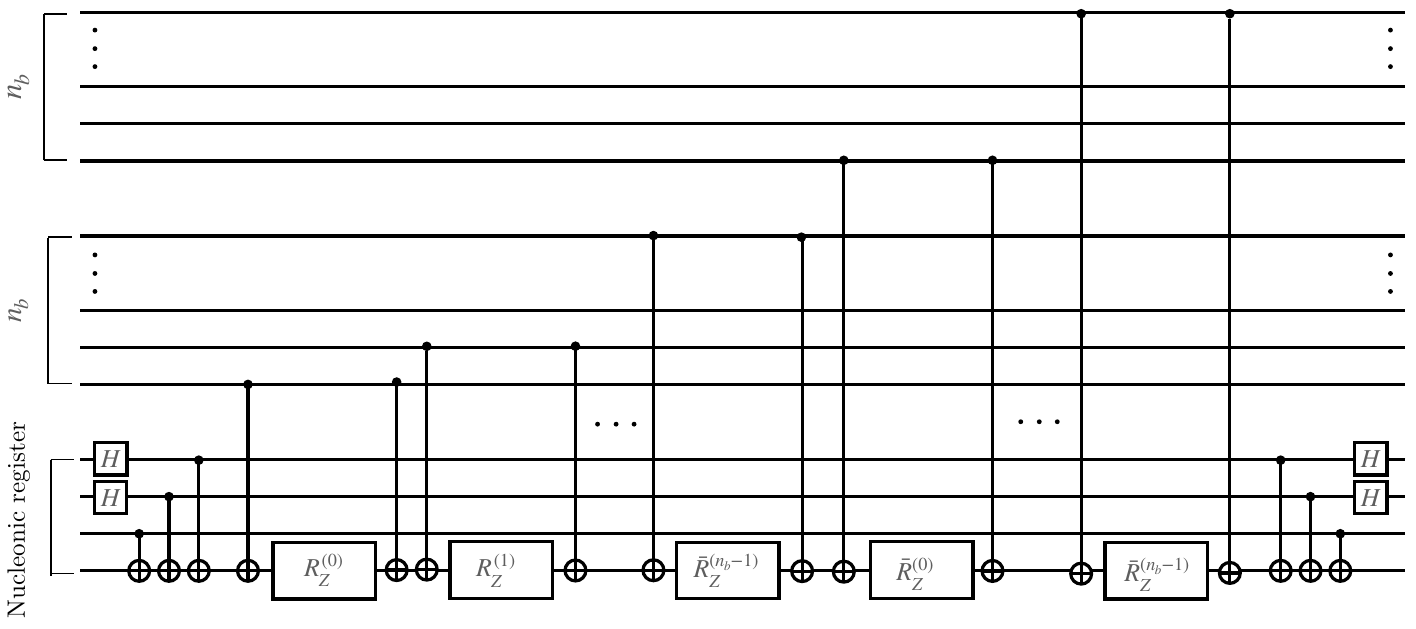}
\caption{The circuit used to implement evolution under the term $\frac{g_AQ}{2f_\pi a_L}\left( \sum_{n=0}^{n_b-1} 2^n Z^{(n)}_{1,\bm{y}} - \sum_{m=0}^{n_b-1} 2^m Z^{(m)}_{1,\bm{x}} \right) X_{i}^{\uparrow p}X_{i}^{\downarrow n} Z_i^{\downarrow p}Z_i^{\uparrow n}$ appearing in \cref{Eq:AV_Paulis}. $H$ denotes a Hadamard gate, $R_z^{(k)}$ is a $Z$ rotation with angle $-\frac{tg_AQ}{2f_\pi a_L}2^{k}$, and $\bar{R}_z^{(k)}$ is a $Z$ rotation with angle $\frac{tg_AQ}{2f_\pi a_L}2^{k}
$. The superscript $i$ denotes the qubit index of site $\bm{x}$, and $\bm{y}$ is a nearest-neighbor site to $\bm{x}$. The two bosonic registers encoding $\pi_1(\bm{x})$ and $\pi_1(\bm{y})$ each involve $n_b$ qubits. A similar circuit can be used to evolve under the other Pauli string in \cref{Eq:AV_Paulis}.
}
\label{Fig:AV_Circuit}
\end{figure}

\subsubsection{The Weinberg-Tomozama Hamiltonian}

\begin{lemma}
\label{Lemma:HWT_Circuit_Depth}
Let $H_{\rm WT}$ be the pion-nucleon axial-vector interaction in \cref{Eq:H_WT}. There exists a circuit implementing $e^{-itH_{\rm WT}}$ on $n_b$ qubits with circuit depth $\Dcost(e^{-itH_{\rm WT}}) \leq 98n_b^2+94n_b+96$.
The controlled version takes circuit depth $\Dcost(e^{-itH_{\rm WT}}) \leq 146n_b^2+190n_b+144$.
\end{lemma}

\begin{proof}
Implementing $e^{-itH_{\rm WT}}$ presents a small difficulty as it involves both the $\Pi_{I_2}(x)$ and $\pi_{I_3}(x)$ operators for $I_2\neq I_3$  simultaneously [see \cref{Eq:H_WT}].
Thus to implement this term, we must ensure that the relevant registers are in the proper basis. Define
\begin{align}
e^{-itH_{\rm WT}^{(I_2,I_3)}(\bm{x})}=e^{-\frac{it}{4f^2_{\pi}}\sum_{\alpha,\beta,\delta} \epsilon_{I_1I_2I_3} \pi_{I_2}(\bm{x})\Pi_{I_3}(\bm{x}) \adag_{\alpha\beta}(\bm{x}) [\tau_{I_1}]_{\beta \delta}a_{\alpha\delta}(\bm{x})},
\end{align}
where $I_1$ is fixed for given $I_2$ and $I_3$ because of the Levi-Civita tensor. Then, at each site $\bm{x}$, one may decompose the operator $e^{-itH_{\rm WT}(\bm{x})}$ as
\begin{align}
e^{-itH_{\rm WT}(\bm{x})} \approx e^{-itH_{\rm WT}^{(1,2)}(\bm{x})}e^{-itH_{\rm WT}^{(3,2)}(\bm{x})}e^{-itH_{\rm WT}^{(1,3)}(\bm{x})}e^{-itH_{\rm WT}^{(2,3)}(\bm{x})}e^{-itH_{\rm WT}^{(2,1)}(\bm{x})}e^{-itH_{\rm WT}^{(3,1)}(\bm{x})},
\end{align}
up to a Trotter error that is calculated in \cref{Sec:Analytic_Trotter_Bounds_Proof}. This decomposition lets us implement the evolution in the basis of $\pi_I(\bm{x})$ fields using only 6 QFT unitaries. Explicitly, denoting the QFT acting on the qubit register associated with the isospin index $I$ by $U_{\rm QFT}^{(I)}$, $e^{-itH_{\rm WT}(\bm{x})}$ can be implemented as
\begin{align}
e^{-itH_{\rm WT}(\bm{x})}=U_{\rm QFT}^{(2)\dagger} &e^{-it\tilde{H}_{\rm WT}^{(1,2)}(\bm{x})} e^{-it\tilde{H}_{\rm WT}^{(3,2)}(\bm{x})} U_{\rm QFT}^{(2)} U_{\rm QFT}^{(3)\dagger} e^{-it\tilde{H}_{\rm WT}^{(1,3)}(\bm{x})}\nonumber \\
&\times e^{-it\tilde{H}_{\rm WT}^{(2,3)}(\bm{x})} U_{\rm QFT}^{(3)} U_{\rm QFT}^{(1)\dagger} e^{-it\tilde{H}_{\rm WT}^{(2,1)}(\bm{x})}e^{-it\tilde{H}_{\rm WT}^{(3,1)}(\bm{x})} U_{\rm QFT}^{(1)},
\label{Eq:H_WT_QFT_Decom}
\end{align}
where $\tilde{H}_{\rm WT}^{(I_2,I_3)}$ contains the QFT-transformed field $\tilde{\Pi}_{I_3}$ in place of $\Pi_{I_3}$. The circuit shown in \cref{Fig:H_WT_Circuit} implements this operator in such a way that four of the QFT operations can be implemented in parallel with four of the $e^{-itH_{\rm WT}^{(I_2,I_3)}(\bm{x})}$ operators. Therefore, the circuit depth satisfies
\begin{align}
\Dcost[e^{-itH_{\rm WT}(\bm{x})}] \leq & 2\max_{(I_2,I_3)}\left[\Dcost(e^{-it\tilde{H}_{\rm WT}^{(I_2,I_3)}(\bm{x})})\right]+2\Dcost( U_{\rm QFT}^{(I)})\nonumber\\
&+
4\max \left[\max_{(I_2,I_3)}\left[\Dcost(e^{-it\tilde{H}_{\rm WT}^{(I_2,I_3)}(\bm{x})})\right],\Dcost( U_{\rm QFT}^{(I_1)})\right].
\label{Eq:H_WT_Depth_Decomp}
\end{align}

\begin{figure}[t!]
\centering
\includegraphics[scale=0.675]{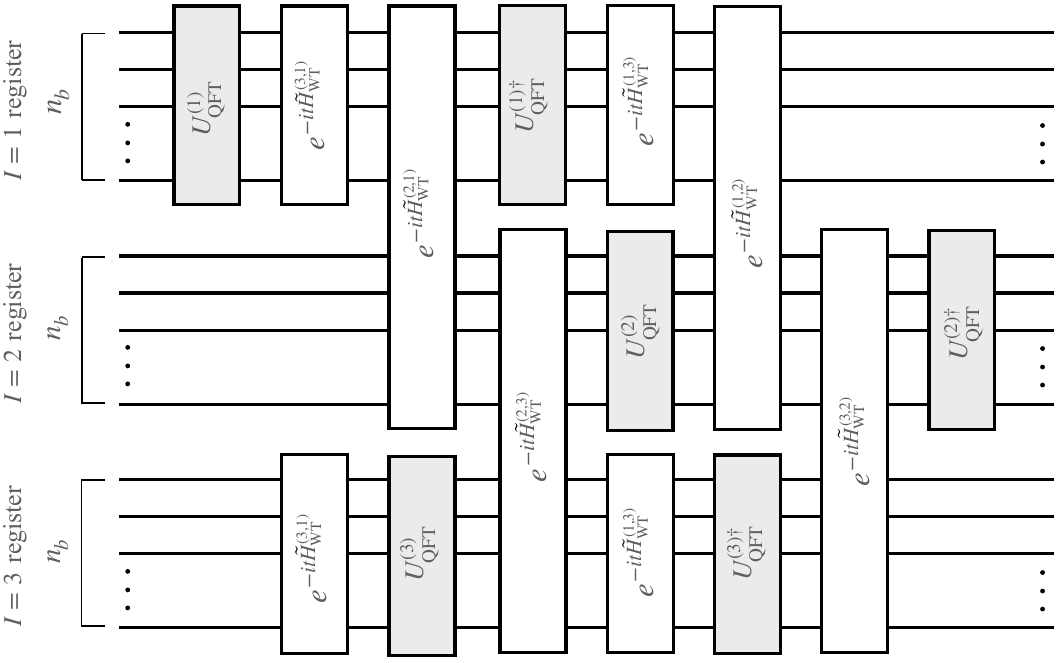}
\caption{The circuit used to implement $e^{-itH_{\rm WT}(\bm{x})}$ according to the decomposition proposed in \cref{Eq:H_WT_QFT_Decom}.}
\label{Fig:H_WT_Circuit}
\end{figure}

We now bound the 2-qubit circuit depths for various terms in \cref{Eq:H_WT_Depth_Decomp}. First, each QFT unitary is implemented on $n_b$ qubits with circuit depth $n_b(n_b-1)$. Second, each $\tilde{H}_{\rm WT}^{(I_2,I_3)}$ involves an $a^\dagger \tau_{I_1} a$ operator, which consists of at most four types of nucleonic operators or $2$ Hermitian-conjugate pairs for simplicity. Each of those pairs can be encoded into up to $2(n_b+1)^2$ Pauli strings, with the highest Pauli weight equal to 5, as demonstrated in \cref{Eq:WT_Pauli_Strings}. The 2-qubit gate depth can be bounded by dividing the $2(n_b+1)^2$ terms into 2 sets of terms each containing one of the two different fermionic strings. Within each set, the strings share 3 fermionic Pauli operations, which accompany either $n_b^2$ terms of $ZZ$ type, $2n_b$ terms of $Z$ type, or unity, hence these can be implemented with a circuit depth of up to $2(3-1)+4n_b^2+4n_b$. A representative circuit of a highest-weight term implementation is shown in \cref{Fig:WT_2_3__Circuit}. Therefore, the circuit depth for each $e^{-itH^{(I_2,I_3)}_{\rm WT}(\bm{x})}$ is upper bounded by $2 \times 2 \times \left(4+4n_b^2+4n_b\right)=16(n_b^2+n_b+1)$. Finally, note that $e^{-itH_{\rm WT}(\bm{x})}$ for all $\bm{x}$ can be applied simultaneously. Putting everything together, the full circuit depth for simulating $H_{\rm WT}$ is
\begin{align}
\Dcost[e^{-itH_{\rm WT}}] 
&\leq 2 \times 16(n_b^2+n_b+1) +2 \times n_b(n_b-1)+4 \max\left[16(
n_b^2+n_b+1),n_b(n_b-1)\right]\nonumber\\
&=98n_b^2+94n_b+96.
\end{align}

In order to apply the controlled version, we introduce one ancilla qubit per $I$ and $\bm{x}$. Then, none of the QFT unitaries need to be controlled, and 4 of those can still be implemented in parallel with  $e^{-it\tilde{H}_{\rm WT}^{(I_2,I_3)}(\bm{x})}$. Each of the 6 $e^{-it\tilde{H}_{\rm WT}^{(I_2,I_3)}(\bm{x})}$ operators, on the other hand, need to be controlled, which in addition to the circuit depth of the uncontrolled version of each, $2\times2\times(n_b+1)^2$ controlled 1-qubit $Z$ rotations should be counted, giving  $8(n_b+1)^2$ additional CNOT gates. This gives an overall circuit depth $\Dcost(C[e^{-itH_{\rm WT}}]) \leq (98n_b^2+94n_b+96)+6\times8(n_b+1)^2=146n_b^2+190n_b+144$.

\begin{figure}[t]
\centering
\includegraphics[scale=0.59]{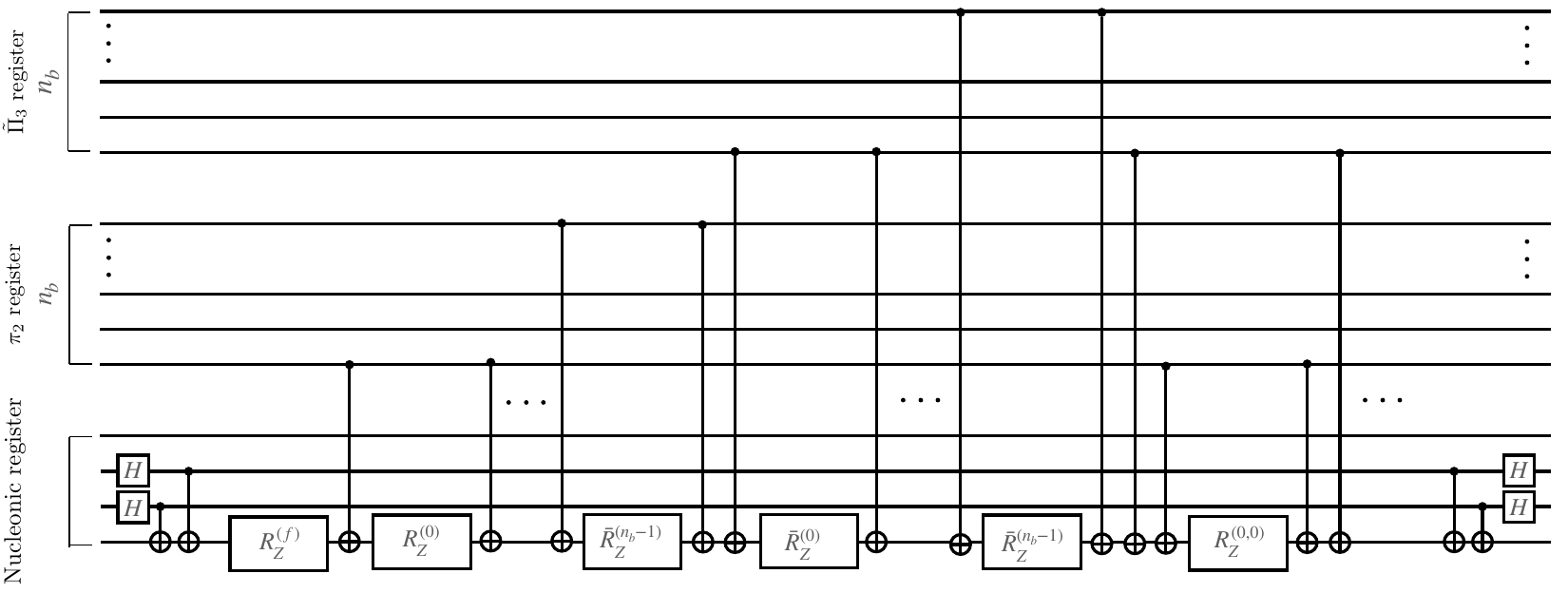}
\caption{Circuit implementing evolution under $\frac{1}{4f_\pi^2}\left(P\mathds{1}+Q\sum_{m=0}^{n_b-1} 2^m Z^{(m)}_{2,\bm{x}}\right)\left(P'\mathds{1}+Q'\sum_{l=0}^{n_b-1} 2^l Z^{(l)}_{3,\bm{x}}\right)
X_{i}^{\uparrow p}X_{i}^{\uparrow n} Z_i^{\downarrow p}$ [see \cref{Eq:WT_Pauli_Strings}]. $H$ denotes a Hadamard gate, $R_z^{(f)}$ is a $Z$ rotation with angle $\frac{tPP'}{4f_\pi^2}$, $R_z^{(k)}$ is a $Z$ rotation with angle $\frac{tQP'}{4f_\pi^2}2^{k}$, $\bar{R}_z^{(k)}$ is a $Z$ rotation with angle $\frac{tQ'P}{4f_\pi^2}2^{k}$, and $R_z^{(k,l)}$ is a $Z$ rotation with angle $\frac{tQQ'}{4f_\pi^2}2^{k+l}$. The superscript $i$ denotes the qubit index of site $\bm{x}$. The two bosonic registers encoding $\pi_2(\bm{x})$ and $\tilde{\Pi}_3(\bm{x})$ each involve $n_b$ qubits. A similar circuit can be used to evolve under the other Pauli string in \cref{Eq:WT_Pauli_Strings}.
}
\label{Fig:WT_2_3__Circuit}
\end{figure}
\end{proof}

\subsubsection{Total Dynamical-Pion EFT Circuit-Depth Costs}

Now we examine the cost of simulating a single time step of evolution of the dynamical-pion EFT Hamiltonian for different orders of product formulae.

\begin{lemma}[Dynamical-pion EFT Trotter-Step Circuit Depth]
\label{Lemma:Total-Depth-Dyn-Pions}
The time evolution of the EFT Hamiltonian with dynamical pions using the $p=1$ Trotter formula can be implemented in circuit depths $\Dcost(\calP_1^{\dyn})\leq 97n_b^2+959n_b+1392+\max\left\{572,n_b^2+16 \left\lceil \frac{n_b}{2} \right\rceil+27n_b -32\right\}$ and $\Dcost(C[\calP_1^{\dyn}]) \leq  
145n_b^2+1919n_b+1440+\max\{732,27n_b^2+16 \left\lceil n_b/2 \right\rceil+41n_b-32\}$.
\end{lemma}
\begin{proof}
    
A single Trotter step of evolution with the dynamical-pion Hamiltonian can be implemented as
\begin{align}
    \calP_1^{\dyn} &= e^{-it\Hfree} e^{-itH_{C} }e^{-itH_{C_{I^2}}} e^{-itH_{\pi^2}}e^{-itH_{(\nabla \pi)^2}}e^{-itH_{\Pi^2}}e^{-itH_{\rm AV}}e^{-itH_{\rm WT}}.
\end{align}
    The $\Hfree$ and pion-only terms, i.e., $H_{\pi^2}$, $H_{(\nabla \pi)^2}$, and $H_{\Pi^2}$, can be implemented simultaneously as they act on different sets of qubits. Thus,
\begin{align}
    \Dcost(\calP_1^{\dyn}) &= \max\bigg\{ \Dcost(e^{-it\Hfree})+\Dcost(e^{-itH_{C}}) + \Dcost (e^{-itH_{C_{I^2}}}), \Dcost(e^{-itH_{\pi^2}})\nonumber\\ 
    &\hspace{1 cm}+
    \Dcost(e^{-itH_{(\nabla \pi)^2} })+
    \Dcost(e^{-itH_{\Pi^2}}) \bigg\} +\Dcost(e^{-itH_{\rm AV}})+\Dcost(e^{-itH_{\rm WT}})\nonumber \\
    & \leq \max\left\{572,2n_b^2+16 \left\lceil \frac{n_b}{2} \right\rceil+26n_b-32\right\}+98n_b^2+958n_b+1392.
\end{align}

Proceeding similarly for the controlled implementation of the product formula,  and assuming that one ancilla qubit is allocated to each fermionic register at site $\bm{x}$ and another ancilla qubit is allocated to each bosonic register at site $\bm{x}$, the circuit depth is $\Dcost(C[\calP_1^{\dyn}]) =\max\{732,28n_b^2+16 \left\lceil n_b/2 \right\rceil+40n_b-32\}+146n_b^2+1918n_b+1440$.
\end{proof}

\subsubsection{Total Dynamical-Pion EFT \texorpdfstring{$T$}{T}-Gate Costs}

Here, we obtain the number of $T$ gates to implement a single $p=1$ Trotter step for the dynamical-pion EFT.

\begin{lemma}[Dynamical-Pion EFT Trotter-Step $T$-gate Costs]
\label{Lemma:Total-T-Gate-Dyn-Pions}
For any $t\in \mathbb{R}$ and $\delta>0$, there exists a circuit that implements $\tilde{V}(t)$ such that $\norm{\tilde{V}(t) - \mathcal{P}_1^{\dyn}(t) } \leq \delta$
with an expected $T$-gate count of $ g(L,n_b)(1.15\log(2g(L,n_b)/\delta )+9.2)$, where $g(L,n_b)=(45n_b^2+114n_b+76)L^3$. Here, $n_b$ is the number of qubits encoding each on-site pion field and $L$ is the total number of lattice sites in each Cartesian direction.
The controlled unitary $C[\tilde{V}(t)]$ can be implemented with $2g(L,n_b)(1.15\log(4g(L,n_b)/\delta )+9.2)$ $T$ gates in expectation.
\end{lemma}

\begin{proof}
To implement $\mathcal{P}_1^{\dyn}(t)$ fault-tolerantly, we use the repeat-until-success method to synthesize $R_z$ gates, and give an equal error allowance to each of the $T$ gates in the synthesis.
We begin by determining the number of $Z$ rotations for simulating each of the Hamiltonian terms.

Let us first consider the nucleon-only Hamiltonian terms. For $e^{-it\Hfree}$, the $R_z$-gate count of $28L^3$ is given in \cref{Lemma:Pionless_T-Gate_Count}. 
For $e^{-itH_C}$ and $e^{-itH_{C_{I^2}}}$, the $R_z$-gate costs are reported in \cref{Lemma:T-gate-OPE} and are $10L^3$ and $18L^3$ $T$ gates, respectively. Note that $4L^3$ of these can be combined with those in the $e^{-it\Hfree}$ circuit.

Next, consider the pion-only Hamiltonian terms. For $e^{-itH_{\pi^2}}$ in \cref{Eq:Exp_piI_Squared}, there are $n_b$ 1-qubit $Z$ rotations and $n_b(n_b-1)/2$ $ZZ$ rotations for each species of pion at each spatial site, which, after expressing entangling rotations in terms of CNOT gates, gives $\left(n_b(n_b-1)/2+n_b\right) \times 3L^3=3n_b(n_b+1)L^3/2$ $R_z$ gates in total.
For $e^{-itH_{(\nabla \pi)^2}}$ in \cref{Eq:pix-minus-piy-Decomposion}, there are $n_b$ 1-qubit $Z$ rotations on the register encoding $\pi_I(\bm{x})$ and $n_b$ 1-qubit $Z$ rotations encoding $\pi_I(\bm{y})$ for each $I$ and nearest-neighbor sites $\bm{x}$ and $\bm{y}$. Then, there are $n_b(n_b-1)/2$ $ZZ$ rotations on the $\pi_I(\bm{x})$ register and the same number on the $\pi_I(\bm{y})$ register. Finally, there are $n_b^2$ $ZZ$ rotations entangling the two registers. After expressing the entangling operations in terms of CNOT gates, this gives $\left(2n_b+2\times n_b(n_b-1)/2+n_b^2\right) \times 3L^3=3(2n_b^2+n_b)L^3$ $R_z$ gates in total.
For $e^{-itH_{\Pi^2}}$ in \cref{Eq:Pi_Squared_with_QFTs}, we apply a QFT and an inverse QFT, each using $n_b(n_b-1)$ 1-qubit $Z$ rotations~\cite[Sec. 5.1]{Nielsen_Chuang_2010}. The rest is essentially the same circuit as for $e^{-itH_{\pi^2}}$, using $n_b(n_b+1)/2$ $R_z$ gates. So, in total, for all pion species throughout the lattice, implementing $e^{-itH_{\Pi^2}}$ involves applying $\left(2n_b(n_b-1)+n_b(n_b+1)/2\right) \times 3L^3=3(5n_b^2-3n_b)L^3/2$ $R_z$ gates. 

Next, consider the pion-nucleon Hamiltonian terms. For $e^{-itH_{\rm AV}}$ as discussed in \cref{Lemma:HAV_Circuit_Depth}, there are 18 terms with different spin and isospin structure, each composed of at most $4n_b$ Pauli strings, giving in total $72n_b$ 1-qubit $Z$ rotations. So implementing $e^{-itH_{\rm AV}}$ amounts to applying $72n_bL^3$ $R_z$ gates in total. For $e^{-itH_{\rm WT}}$ in \cref{Eq:H_WT_QFT_Decom}, we implement 6 QFT unitaries or their inverses, each using $n_b(n_b-1)$ 1-qubit $Z$ rotations. Then, each of the 6 $e^{-itH^{(I_2,I_3)}_{\rm WT}}$ involves $2(n_b+1)^2$ 1-qubit $Z$ rotations. So, in total, implementing this term amounts to applying $6\left(n_b(n_b-1)+2(n_b+1)^2\right)L^3=6(3n_b^2+3n_b+2)L^3$ $R_z$ gates.

We now define $g(L,n_b)$ to be the total number of $Z$ rotations required for the full Hamiltonian, that is,
\begin{align}
g( L,n_b)=(33n_b^2+90n_b+64)L^3.
\end{align}
Each rotation is implemented to precision $\delta/g(L,n_b)$ using
$1.15\log(2g(L,n_b)/\delta )+9.2$ $T$ gates per rotation, on average.
The overall cost is therefore
\begin{align}
    g(L,n_b)(1.15\log(2g(L,n_b)/\delta)+9.2).
\end{align}

In the controlled case, each Pauli string takes twice as many $Z$ rotations, giving $2g(L,n_b)$ $Z$ rotations, and an overall requirement of
$
    2g(L,n_b)(1.15\log(4g(L,n_b)/\delta )+9.2)
$ $T$ gates in expectation.
\end{proof}

\section{Errors from Product-Formula Simulations and Beyond}\label{Sec:product_formulae}

In this section, we consider the errors associated with the product-formulae algorithm, as well as other sources of error that have been introduced, such as by truncating the Hamiltonians or during the circuit synthesis.

\subsection{General Trotter Error Bounds for Number-Preserving Hamiltonians}

An key aspect of our attempts to minimize the gate counts of the simulation routine is putting better upper bounds on the Trotter error (i.e., the error associated with implementing time evolution via product formulae).
With this in mind, we first consider the general case of product-formula-based simulations in which the Hamiltonian i) preserves the number of fermions and ii) can be Trotterized using \emph{local} terms that also preserve the number of fermions. Our bounds are derived from those in Ref.~\cite{childs2021theory} which characterize the Trotter error in terms of the commutators of the terms in the Hamiltonian.

As already presented in \cref{Sec:Product_Formulae_Summary}, the key quantity appearing in the $p$th-order product-formula error bound is $\tilde{\alpha}_{\mathrm{comm}}$, defined in \cref{Eq:tilde-alpha-def}.
While this quantity can be trivially bounded by $\tilde{\alpha}_{\mathrm{comm}}\leq \sum_{\gamma_{p+1},\gamma_{p},\dots, \gamma_1}^\Gamma \max_{\alpha}\norm{H_\alpha}^{p+1}  $, we would like to find a bound with better scaling. To begin, observe that each fermionic operator can be written as a sum of terms of the following form.

\begin{definition}[Number-Preserving Fermionic Operators] \label{Def:NPFO_Def}
Let $\vec{i} \coloneqq \{i_1,i_2,\dots, i_{k_{\vec i}}\}$ label fermionic modes.
A \emph{number-preserving fermionic operator} (NPFO) $h_{\vec{i}}$ can be expressed as
\begin{align}
    h_{\vec{i}} \coloneqq \adag(i_1)\adag(i_2)\dots \adag(i_{m})a(i_{m+1})\dots a(i_{2m})N(i_{2m+1})\dots N(i_{k_{\vec i}}),
\end{align}
where $i_1\neq i_2 \neq \dots \neq i_{k_{\vec i}}$, and there are an equal number of creation and annihilation operators.
If there is a constant prefactor to the NPFO, we call this the \emph{weight} of the NPFO.
\end{definition}

With this definition in mind, we obtain a general upper bound on $\tilde{\alpha}_{\mathrm{comm}}$ for a class of fermionic Hamiltonians that includes the ones we study.
In particular, we use the fact that, if all the local terms are number-preserving, the simulation remains in the subspace with a given number of fermions.
By projecting into the $\eta$-fermion space, the bound on $\tilde{\alpha}_{\rm comm}$ can be greatly reduced.

\begin{theorem} \label{Theorem:General_Order_Fermionic_Error_2}
Let $\{H_{\gamma_i}\}
$ be a set of translation-invariant Hamiltonians with disjoint support such that
\begin{align}
    H_{\gamma_i} = J^{(\gamma_i)}\sum_{\vec{j}} h_{\vec{j}}^{(\gamma_i)},
\end{align}
where each $h_{\vec{j}}^{(\gamma_i)}$ is an NPFO with locality $k^{(\gamma_i)}$.
Here, $\vec{j}$ denotes a subset of
fermionic modes on a lattice, and locality $k^{(\gamma_i)}$ is the number of modes $h_{\vec{j}}^{(\gamma_i)}$ acts on.
Then,
\begin{align}
    &\norm{\left[ H_{\gamma_{p+1}},\dots ,[H_{\gamma_2},H_{\gamma_1}] \right]}_\eta \leq \left(\prod_{n=1}^{p+1} \big| J^{(\gamma_n)} \big| \right) 
    \prod_{m=2}^{p+1}\bigg[ 2k^{(\gamma_m)} (k^{(\gamma_m)} -1) \left(\sum_{n=1}^{m-1}k^{(\gamma_n)} - (m-2) \right) 
    \nonumber\\
    & \hspace{4.25 cm} \times\left(\sum_{n=1}^{m-1}k^{(\gamma_n)} - (m-1) \right)  2^{1+\min\{ k^{(\gamma_m)}, \sum_{n=1}^{m-1}k^{(\gamma_n)} - (m-2) \}/2  }\bigg] \left\lceil \frac{\eta}{\left\lceil k_{\min}/2\right\rceil} \right\rceil ,
\end{align}
where $k_{\min}=\min_{1\leq i\leq p+1}\left\{ k^{(\gamma_i)} \right\}$.
\end{theorem}

The full proof is presented in \cref{Sec:Higher_Order_Trotter_Error}.
The main idea is to observe that
commutators of NPFOs can be written as sums of NPFOs.
Further, when the NPFO is normal ordered, it will only be non-zero when acting on states with fermions present. Having rewritten the nested commutator as a sum of NFPOs, we then decompose it into subsets of NPFOs that do not have intersecting support.
The fermionic semi-norm of these subsets must be $O(\eta)$ as each NPFO in the subset is only non-zero when fermions are present, but they also have disjoint support. 
The prefactor then depends on how many subsets the nested commutator needs to be separated into, which in turn depends on the locality of the NPFOs that occur in the Hamiltonian.

\begin{corollary}\label{Corollary:Error_Scaling}
    For Hamiltonians of the form given in \cref{Theorem:General_Order_Fermionic_Error_2}, the error in the $p$th-order product formula is
    \begin{align}
    \norm{e^{-itH} - \calP_p(t)}_{\eta} = O\left( t^{p+1}\eta \right).
\end{align}
\end{corollary}

Notably, \cref{Corollary:Error_Scaling} means that the number of Trotter steps to reach a certain error scales as $O(\eta^{1/p})$ and, consequently, is independent of the lattice size.
There are other bounds in the literature for fermionic Hamiltonians that are also independent of system size (e.g., Refs.~\cite{Clinton_Bausch_Cubitt2021, su2021nearly}). 
However, these results consider a more restricted form of Hamiltonian, do not give explicit numerical prefactors of the error bounds, or have worse scaling in $\eta$.

\subsection{Low-Order Trotter Error for Nuclear EFTs}\label{s:trotter_error_bnds}

We now focus on computing the quantity $\tilde{\alpha}_{\mathrm{comm}}$ for low-order product formulae applied to the nuclear EFTs that we consider.
Although it is possible, in principle, to calculate the nested commutators for $p>2$ in a similar manner, the calculation is quite involved and we do not perform it here.

\subsubsection{Analytical \texorpdfstring{$p=1$}{p=1} and \texorpdfstring{$p=2$}{p=2} Bounds for Pionless EFT} \label{Sec:Analytical_Pionless_Bounds}

The following theorems obtain bounds on the error in first- and second-order product-formula simulations of the time evolution of the pionless-EFT Hamiltonian defined in \cref{Sec:Pionless_EFT_Hamiltonian}. These bounds are derived using the improved commutator error-bound relations of Ref.~\cite{childs2021theory} that are summarized in \cref{Sec:Product_Formulae_Summary} [see in particular \cref{Eq:p-1-commutator-bound,Eq:p-2-commutator-bound}].

\begin{theorem}[$p=1$ Pionless-EFT Trotter Error]\label{thm:pionless_trotter_error_p1}
For the pionless-EFT Hamiltonian described in \cref{Sec:Pionless_EFT_Hamiltonian},
\begin{align}
    \norm{e^{-itH_\canpi} - \calP_1^{(\canpi)}(t)}_\eta \leq t^2\left( 15h^2\eta + 6h
    \left(A_1\left\lfloor\frac{\eta}{2}\right\rfloor 
    +A_2\left\lfloor\frac{\eta}{3}\right\rfloor+ A_3\left\lfloor\frac{\eta}{4}\right\rfloor \right)  
    \right),
\end{align}
where $h=\frac{1}{2M{a_L^2}}$ is the coefficient of the hopping term, and
\begin{align}
    A_1 = 2|\Cpi|, \quad \quad A_2= 2|3\Cpi+\Dpi|+|\Dpi|, \quad \quad  A_3=2|6\Cpi + 4\Dpi| + 4|\Dpi|,
\end{align}
Here, $\Cpi$ and $\Dpi$ are the low-energy constants of two- and three-nucleon contact terms.
\end{theorem}

The proof is presented in \cref{Sec:p=1_Pionless_Trotter_Error}.
The fundamental idea is to decompose the Hamiltonian into 7 sets of terms: 6 sets corresponding to the kinetic hopping terms on the lattice, and one corresponding to the on-site interaction term acting between fermions of different species, as described in \cref{Sec:Total-Pionless-EFT-Circuit-Depth}.
Within each set, all terms commute with each other, but they do not necessarily commute with terms in the other sets. We then compute the commutators for each of these pairs of sets.
The resulting terms can be written as sets of disjoint normal-ordered fermionic creation and annihilation operators.
Since normal-ordered fermionic operators are zero if the annihilation operators act on fermionic modes with no fermions present, the
disjoint sets of fermionic Hamiltonians have the property that their fermionic semi-norm does not scale with the lattice size, but instead scales as the number of fermions (see \cref{Theorem:NPFO_Norm} for the proof).

The $p=1$ error bound in \cref{thm:pionless_trotter_error_p1}, which is computed using the pionless-EFT Hamiltonian by evaluating commutators manually, can be compared against the general bound obtained from \cref{Theorem:General_Order_Fermionic_Error_2}.
As can be seen from \cref{Table:By-Hand_Comparison}, the ``manual'' method is better by a factor of $10^2$ for the first-order product formula. This indicates that the number of Trotter steps needed to reach a given accuracy is lower by a factor of $10^2$ than if the analysis was based on the loose bound of \cref{Theorem:General_Order_Fermionic_Error_2}. In other words, one can get significant gains from taking into account the actual structure of the Hamiltonian. However, evaluating explicit nested commutators for a general high-order formula may be impractical for complicated Hamiltonians.

\begin{table}
    \centering
    \begin{tabular}{c|c}
    \hline
    \multicolumn{2}{c}{\textbf{Comparison of $\tilde{\alpha}_{\mathrm{comm}}$ Upper Bound for the $p=1$ formula}}  \\ 
    \hline
      \textbf{General Bound from \cref{Theorem:General_Order_Fermionic_Error_2}}  &  \textbf{Manual Computation in \cref{thm:pionless_trotter_error_p1}} \\  \hline 
        $2.7\times 10^{6}$ & $1.1\times 10^4$
        \\
        \hline
    \end{tabular}
    \caption{Upper bounds on $\Tilde{\alpha}_{\mathrm{comm}}$ with $p=1$ for two fermions/nucleons as computed by the general formula in \cref{Theorem:General_Order_Fermionic_Error_2} compared to the manual computation done directly as given in \cref{thm:pionless_trotter_error_p1}. } 
    \label{Table:By-Hand_Comparison}
\end{table}

\begin{theorem}[$p=2$ Pionless-EFT Trotter Error]\label{thm:pionless_trotter_error_p2}
For the pionless-EFT Hamiltonian described in \cref{Sec:Pionless_EFT_Hamiltonian}:
\begin{align}
    \norm{e^{-iH_\canpi t} - \calP_2^{(\canpi)}(t)}_\eta &\leq \frac{t^3}{12} \bigg( 125h^3\eta + 216h^2\left( n_2 +n_3 + n_4 + c_3+c_4\right) \\ &+60h^2(w_1+w_2+w_3) + 12h\left(2 (q_2+q_3+q_4) +   q_3'+q_4' \right) \bigg),
\end{align}
where $h=\frac{1}{2M{a_L^2}}$ is the coefficient of the hopping term, and
\begin{align}
    n_2 &=|\Cpi|\left\lfloor\frac{\eta}{2}\right\rfloor, \quad n_3 =|3\Cpi+\Dpi|\left\lfloor\frac{\eta}{3}\right\rfloor, \quad n_4 =|6\Cpi+4\Dpi|\left\lfloor\frac{\eta}{4}\right\rfloor, \\
    c_3 &= |\Dpi|\left\lfloor\frac{\eta}{3}\right\rfloor, \quad c_4 =4|\Dpi|\left\lfloor\frac{\eta}{4}\right\rfloor, \\
    w_2 &=2|\Cpi|\left\lfloor\frac{\eta}{2}\right\rfloor, \quad  
    w_3 =\left(|\Dpi|+ 2|3\Cpi+\Dpi|\right)\left\lfloor\frac{\eta}{3}\right\rfloor, \quad w_4 =\left(4|\Dpi|+ 2|6\Cpi+4\Dpi|\right)\left\lfloor\frac{\eta}{4}\right\rfloor, \\
    q_2&=2|\Cpi|^2\left\lfloor\frac{\eta}{2}\right\rfloor, \quad
    q_3=4\left|\frac{\Cpi}{2}+\frac{\Dpi}{6}\right|\left(12 \left |\frac{\Cpi}{2}+\frac{\Dpi}{6}\right|+\left|\Dpi\right|\right)\left\lfloor\frac{\eta}{3}\right\rfloor, \\
    q_4&=24 \left |\frac{\Cpi}{2}+\frac{\Dpi}{3}\right|\left(6 \left |\frac{\Cpi}{2}+\frac{\Dpi}{3}\right| +\left|\Dpi\right|\right)\left\lfloor\frac{\eta}{4}\right\rfloor, \\
    q_3'&=\left(8|\Dpi|\left|\frac{\Cpi}{2}+\frac{\Dpi}{6}\right|+\frac{2}{3}\Dpi^2\right)\left\lfloor\frac{\eta}{3}\right\rfloor, \qquad
    q_4'=8|D_\pi|\left(6\left|\frac{\Cpi}{2}+\frac{\Dpi}{3}\right|+|\Dpi|\right)\left\lfloor\frac{\eta}{4}\right\rfloor,
\end{align}
Here, $\Cpi$ and $\Dpi$ are the low-energy constants of two- and three-nucleon contact terms. 
\end{theorem}

The proof is presented in \cref{Sec:Pionless_EFT_p2}.

\subsubsection{Analytical \texorpdfstring{$p=1$}{p=1} Bounds for OPE and Dynamical-Pion EFTs Trotterization Error} \label{Sec:OPE_Dyn_Trot_Error}

Similar product-formula error bounds can be derived for the OPE and dynamical-pion EFTs. 
Due to the significantly more complex interactions, we refrain from writing down the corresponding error expressions and defer both the statements and proofs to \cref{Sec:Trotter_Error_OPE_p1,Sec:Trotter_Error_Dyn_p1}, where we consider the $p=1$ case.
Here instead, we simply describe the scaling of the simulation error in terms of evolution time and system parameters. 
An additional error arises from the truncated interaction range for the OPE Hamiltonian (see \cref{Sec:Cutoff_OPE,Sec:Cutoff_Proof}) and the truncated digitized-field Hilbert space for the dynamical-pion Hamiltonian (see \cref{Sec:Truncated-Field-Space,Sec:Pion_Cutoff}), which introduce truncation errors $\epstrunc$ and $\epscut$, respectively---we save this discussion for \cref{Sec:Beyond_Prod_Errors}.
For the moment, we only consider the Trotterization error from simulating the truncated Hamiltonians using product formulae, i.e. the error associated with Hamiltonians explicitly constructed in \cref{Sec:OPE_Hamiltonian,Sec:Dyn_Pions}.

The OPE EFT can be simulated using a $p$th-order product formula with an error that scales as
\begin{align}
    \norm{\mathcal{P}_p(t) - e^{-itH_{\rm OPE}}}_\eta = O\left( t^{p+1}\eta\right).
    \label{eq:ope_error}
\end{align}
For the dynamical-pion EFT with $p=1$,
\begin{align}
    \norm{\mathcal{P}_1(t) - e^{-itH_{D\pi}}}_\eta = O\left( \pimax^{2}\Pimax^{2}t^{2}L^3\right),
\end{align}
where $\pimax, \Pimax$ are defined in \cref{Lem:Dyn-Pions}.
For a $p$th-order product formula, we have
\begin{align}
    \norm{\mathcal{P}_p(t) - e^{-itH_{D\pi}}}_\eta = O\left( \pimax^{p+1}\Pimax^{p+1}t^{p+1}L^{3}\right).
\end{align}
While the prefactor for the $p=1$ case is calculated in \cref{Sec:Trotter_Error_Dyn_p1}, the explicit calculation for the $p=2$ case is cumbersome and is not reported in this work.

\subsection{Errors Beyond Product-Formula Error
\label{Sec:Beyond_Prod_Errors}}

In general, the product-formula error discussed previously in \cref{Sec:OPE_Dyn_Trot_Error} is not the only source of error. 
The Hamiltonians themselves are approximated, as already discussed, and this introduces an additional error in time evolution that we will account for in this section. 
Recall that in the OPE case, the long-range OPE interactions are cut off, while in the dynamical-pion case, the pion field and its conjugate momentum are truncated and digitized. 
Here, we present the dependence of the full error on both product-formula and truncation errors for each model. 
The pionless-EFT simulation error only arises from product-formula error and was presented in \cref{Sec:Analytical_Pionless_Bounds}.

\begin{lemma} \label{Lemma:Trotter_Cutoff_Error}
    Let $H_{\rm OPE}$ be the full OPE-EFT Hamiltonian as defined in \cref{Sec:OPE_Hamiltonian}, and let $\tilde{H}_{\rm OPE}$ be the OPE Hamiltonian with the long-range terms truncated to only include terms in which nucleons interact up to a maximum distance $\ell$.
    Then,
    \begin{align}
      \norm{ e^{-itH_{\rm OPE}} - \mathcal{P}^r(t/r) } \leq r\epsprod(t/r) + \epstrunc,
    \end{align}
    where $\epsprod(t/r)$ is the standard product-formula error and $\epstrunc = t \norm{H_{\rm OPE} - \tilde{H}_{\rm OPE}}$.
\end{lemma}
\begin{proof}
    The result can be deduced by a straightforward application of the triangle inequality:
    \begin{align}
        \norm{ e^{-itH} - \mathcal{P}^r(t/r) } &\leq \norm{ e^{-itH} - e^{-it\tilde{H}} } + \norm{ e^{-it\tilde{H}} - \mathcal{P}^r(t/r) } 
        \leq \epstrunc(t)  + r\epsprod(t/r),
    \end{align}
    as claimed.
\end{proof}

For the pionful Hamiltonians, as discussed previously, we introduce a finite cutoff scale for the strength of the pion field and its conjugate momentum.
This introduces an associated error denoted by $\epscut$.

\begin{lemma}\label{Lemma:Trotter_Truncation_Error}
    Let $H_{\dyn}$ be the full dynamical -pion EFT Hamiltonian as per \cref{Sec:Dyn_Pions}, and let $\rho=\ket{\psi}\bra{\psi}$ and $\rho_{\rm cut}=\ket{\psi_{\rm cut}}\bra{\psi_{\rm cut}}$ be the density matrices associated with the states of the untruncated ($\ket{\psi}$) and truncated ($\ket{\psi_{\rm cut}}$) bosonic fields, respectively.
    Then,
    \begin{align}
        \norm{ e^{-itH_{\dyn}}\rho e^{iH_{\dyn}t} - \mathcal{P}^r(t/r)\rho_{\rm cut} \mathcal{P}^{r\dagger }(t/r) }_1 \leq r\epsprod(t/r) + 2\sqrt{2\epscut
        },
    \end{align}
    where $\|\cdot\|_1$ denotes the Schatten 1-norm, $F(\rho,\rho_{\rm cut}) = |\braket{\psi|\psi_{\rm cut}}|^2 = (1 - \epscut)^2$, and $\epsprod(t/r)$ is the product-formula error for time $t/r$.
\end{lemma}

\begin{proof}
    Once again, the triangle inequality can be used to derive this result:
    \begin{align}
        \norm{ e^{-itH}\rho e^{iHt} - \mathcal{P}^r(t/r)\rho_{\rm cut} \mathcal{P}^{r\dagger}(t/r) }_1 &\leq \norm{ e^{-itH}\rho e^{itH} - e^{-itH}\rho_{\rm cut} e^{itH} }_1\nonumber \\
        &\hspace{2cm} + \norm{ e^{-itH}\rho_{\rm cut} e^{itH} - \mathcal{P}(t)\rho_{\rm cut}  \mathcal{P}(t) }_1 \\
        &\leq \norm{ \rho  - \rho_{\rm cut}  }_1 + \norm{ e^{-itH}\rho_{\rm cut} e^{itH} - \mathcal{P}^r(t/r)\rho_{\rm cut}  \mathcal{P}^{r\dagger}(t/r) }_1 \\
        &\leq 2\sqrt{1-F(\rho,\rho_{\rm cut})} + r\epsprod(t/r) \\
        &\leq 2\sqrt{2\epscut
        } + r\epsprod(t/r),
    \end{align}
    where we have used the fact that, for a positive semi-definite matrix $A$, $\|A\|_1={\rm tr}(A)$, and that, for Hermitian $A-B$, ${\rm tr}(A-B)={\rm tr}\sqrt{(A-B)^\dagger(A-B)}=\|A-B\|_{\rm tr}\leq 2\sqrt{1-F(A,B)}$, with $\|A-B\|_{\rm tr}$ being the trace distance between $A$ and $B$ and $F(A,B)$ being the fidelity between $A$ and $B$.
\end{proof}

Besides $\epsilon_{\rm prod}$, $\epsilon_{\rm trunc}$, and $\epsilon_{\rm cut}$, there is also an error associated with imperfect synthesis of $R_z$ gates using $T$ gates. This latter error source is only relevant for simulations in the fault-tolerant setting, and is bounded by the total number of $R_z$ gates in each operator's implementation times the synthesis error $\epsilon_{\rm syn}$ introduced in \cref{Eq:Def-epssyn}. In both the near- and fault-tolerant cases, our strategy in the following section is to split the total error on the time-evolution operator equally among the applicable sources of error. This choice is summarized in \cref{Table:Error_Budget}.

\begin{table}
    \centering
    \resizebox{\textwidth}{!}{\begin{tabular}{c|c|c} 
    \hline
    \multicolumn{1}{c}{} & \textbf{Near-Term Evolution Error Budget} & \textbf{Fault-Tolerant Evolution Error Budget}\tabularnewline
    \hline 
    \textbf{Pionless EFT} &$r\epsprod=\epsilon$ & $r\epsprod=rN_{R_z}\epssyn=\epsilon/2$ \tabularnewline
    \hline 
    \textbf{One-Pion Exchange} &$r\epsprod=\epstrunc=\epsilon/2$ &  $r\epsprod=\epstrunc=rN_{R_z}\epssyn=\epsilon/3$ \tabularnewline
    \hline 
    \textbf{Dynamical Pions} & $r\epsprod=2\sqrt{2\epscut
    }=\epsilon/2$ & $r\epsprod=2\sqrt{2\epscut
    }=rN_{R_z}\epssyn=\epsilon/3$  \tabularnewline
\hline
\end{tabular}}
\caption{Error budget for the time-evolution task in different models. Here, near term refers to non-error-corrected circuits which do not require $T$ gates to be synthesized,
$r$ is the number of Trotter steps, and $N_{R_z}$ is the number of 1-qubit $R_z$ gates for each Trotter step. 
The total error on the time-evolution operator is denoted by $\epsilon$.}
\label{Table:Error_Budget}
\end{table}

For the phase estimation task described in \cref{Sec:QPE}, we must account for additional error in the measurement of the eigenvalue, besides the Trotter, truncation, and gate-synthesis errors (where the latter, as mentioned, is only relevant for fault-tolerant simulations). For the purposes of this work, we ignore the error incurred in the eigenstate-preparation task. As in \cref{Sec:QPE}, we follow the error analysis from Refs.~\cite{babbush2018encoding, stetina2022simulating} and add the phase estimation error to the rest of incurred errors in quadrature. For the OPE EFT,
\begin{align}
\label{Eq:RMS_Error}
    \frac{\Delta E}{\norm{H}} = \sqrt{ \left( \frac{1}{2^{m+1}}\right)^2 + \left(\frac{r\epsprod +  \epstrunc 
    +rN_{R_z}\epssyn}{2\pi}\right)^2},
\end{align}
where $m$ is the bit accuracy of the energy eigenvalue, $r$ is the number of Trotter steps performed in the QPE algorithm, and $N_{R_z}$ is the number of $R_z$ gates per Trotter step. For the pionless EFT, we set $\epstrunc=0$ in \cref{Eq:RMS_Error}, while for the dynamical-pion EFT, we replace $\epstrunc$ with $2\sqrt{2\epscut
}$. Finally, for near-term implementations, we set $\epssyn=0$. Note that $\epsprod$ and $\epstrunc$ should be evaluated at time $t=2\pi/\norm{H}$.
Generally, we choose to split the error budget equally between all the error terms in the parentheses on the right in \cref{Eq:RMS_Error}.
Our choice is summarized in \cref{Table:Error_Budget_QPE} for the various EFTs of this work.

\begin{table}
    \centering
    \resizebox{\textwidth}{!}{\begin{tabular}{ c | c| c } 
  \hline
   \multicolumn{1}{c}{} & \textbf{Near-Term QPE Error Budget} & \textbf{Fault-Tolerant QPE Error Budget} \tabularnewline
  \hline
  \textbf{Pionless EFT }& $r\epsprod=\sqrt{3}\pi/2^{m}$ & $r\epsprod=rN_{R_z}\epssyn = \sqrt{3}\pi/2^{m+1}$  \tabularnewline 
  \hline
  \textbf{One-Pion Exchange} & $r\epsprod=\epstrunc=\sqrt{3}\pi/2^{m+1}$ & $r\epsprod=\epstrunc=rN_{R_z}\epssyn =\sqrt{3}\pi/(3\times2^{m})$ \tabularnewline
  \hline
  \textbf{Dynamical Pions} & $r\epsprod=2\sqrt{2\epscut}=\sqrt{3}\pi/2^{m+1}$ & $r\epsprod=2\sqrt{2\epscut}=rN_{R_z}\epssyn =\sqrt{3}\pi/(3\times2^{m})$    \tabularnewline 
  \hline
\end{tabular}}
\caption{
Error budget for the spectroscopy task using QPE in different models. Here, near term refers to non-error-corrected circuits which do not require $T$ gates to by synthesized,
$m$ is the bit accuracy of the energy eigenvalue, $r$ is the number of Trotter steps, and $N_{R_z}$ is the number of 1-qubit $R_z$ gates for each Trotter step. 
The total error on the time-evolution operator is denoted by $\epsilon$, and $\epsprod$ and $\epstrunc$ should be evaluated at time $t=2\pi/\norm{H}$.}
\label{Table:Error_Budget_QPE}
\end{table}

\section{Resource Estimates for the Full Simulation
\label{sec:resource}}

Given the circuit and error-bound analyses of the previous sections, we are ready to combine all the results to assess resource requirements for simulating nuclear EFTs. This section focuses on two simulation tasks: time-evolving the nucleons across the lattice and energy spectroscopy via a quantum-phase-estimation algorithm. We assume that state preparation can be done with separate resources and with high fidelity.

\subsection{Time Evolution}

Here, we estimate the resources to simulate time evolution. We consider a characteristic time for the nucleons to cross the lattice, defined as
\begin{align}
    T_{\text{cross}} \coloneqq \frac{a_LLM_N}{P} = a_LL\sqrt{\frac{M_N}{2E_{\rm kin}}},
\end{align}
where $P$ is the total momentum of a single nucleon, $E_{\rm kin}$ is the single-nucleon kinetic energy, $M_N$ is the mass of a single nucleon, $a_L$ is the lattice spacing, and $L$ is the unitless lattice dimension (i.e., the number of lattice points along each Cartesian axis). 
This is (approximately) the relevant timescale for events such as scattering experiments where particles are fired at each other across the lattice.

To be concrete, let us set $a_L=2.2$~fm, $L=10$, $E_{\rm kin}=10$ MeV, and further allow a total error of at most $0.1$ on the spectral norm of the time-evolution operator. 
This value of lattice spacing ensures that the bounds in \cref{Lem:Dyn-Pions} are valid. The values of the coefficients $C$ and $C_{I^2}$ at $a_L=2.2$~fm are not provided in the literature, so we use the values given in \cref{Table:OPE_Parameters}, which are valid for $a_L=2.0$~fm for the OPE and dynamical-pion Hamiltonians~\footnote{We have varied the coefficients $C$ and $C_{I^2}$ by $\pm 50\%$ to confirm that the difference in our results for coefficients within this range is entirely negligible. This indicates, as expected, that the slight inaccuracy in the values of the low-energy constants we use is irrelevant to the cost estimates presented here.}.
The scalings of the circuit depths and $T$ gates in terms of the number of fermions are plotted in \cref{Fig:Depth_TGate_Crossing_Time} for the crossing time for the three EFT models considered in this work for $p=1$. The cost increases with the number of fermions. The theory with dynamical pions is the most costly, while the pionless EFT is the least costly. For the dynamical-pion theory and the chosen parameters and error thresholds, $n_b=33-39$ qubits are required to encode each dynamical pion per isospin component per site. The exact value of $n_b$ depends on the fermion number, see the expression for the cutoffs $\pi_\mathrm{max}$ and $\Pi_\mathrm{max}$ in \cref{eq:pimax-expression,eq:Pimax-expression}, which determine $n_b$ via \cref{Eq:nb-for-dyn-and-ins}.
\begin{figure}
\begin{minipage}[c]{0.475\linewidth}
\includegraphics[width=\textwidth]{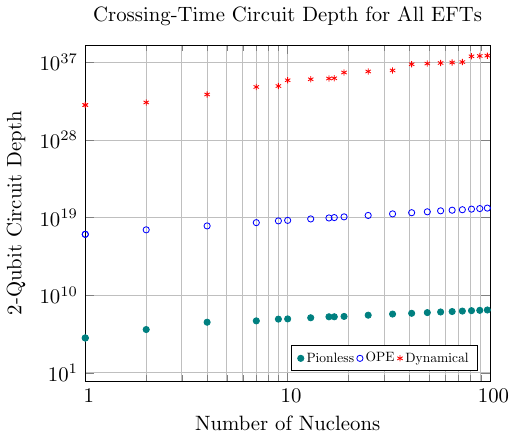}
\end{minipage}
\hfill
\begin{minipage}[c]{0.475\linewidth}
\includegraphics[width=\textwidth]{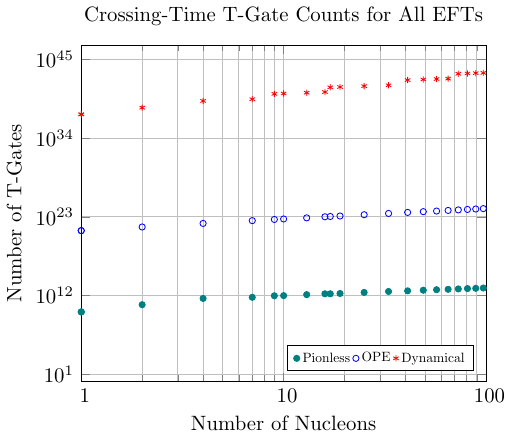}
\end{minipage}%
\caption{
Plots showing the 2-qubit circuit depth (left) and $T$-gate count (right) as a function of the number of nucleons for simulating the evolution according to different EFTs for the crossing time with a total error of at most 0.1. All costs are for the $p=1$ product formula and assume a $10\times 10 \times 10$ lattice with $a_L = 2.2$~fm, with a kinetic energy per nucleon $E_{\rm kin}=10$~MeV. }
\label{Fig:Depth_TGate_Crossing_Time}
\end{figure}
\begin{table}
    \centering
\begin{tabular}{ | c | c| c | c |} 
\hline 
   & \textbf{Circuit Depth} & \textbf{$T$-Gate Count} & \textbf{Number of Qubits}  \\
  \hline
  \textbf{Pionless EFT (VC)}& $6.0\times 10^8$ & $4.7\times 10^{12}$ & 6,000 \\ 
  \hline
  \textbf{Pionless EFT (Compact)}& $7.9\times 10^7$ & $4.7\times 10^{12}$ & 10,000 \\ 
  \hline
  \textbf{One-Pion Exchange} & $3.5\times 10^{19}$ & $5.9\times 10^{23}$ & 6,000 \\
  \hline
  \textbf{Dynamical Pions} &$6.0\times 10^{36}$ & $1.3\times 10^{42}$ & 99,000 \\ 
  \hline
\end{tabular}
\caption{Simulation costs for the crossing time for different EFTs to a total error of 0.1 in the time-evolution operator with 40 nucleons present. All costs are for the $p=1$ product formula and assume a $10\times 10 \times 10$ lattice with $a_L = 2.2$~fm, with a kinetic energy per nucleon of $E_{\rm kin}=10$~MeV.
}
    \label{Table:Resource_Costs}
\end{table}
\begin{figure}
\begin{minipage}[c]{0.475\linewidth}
\includegraphics[width=\textwidth]{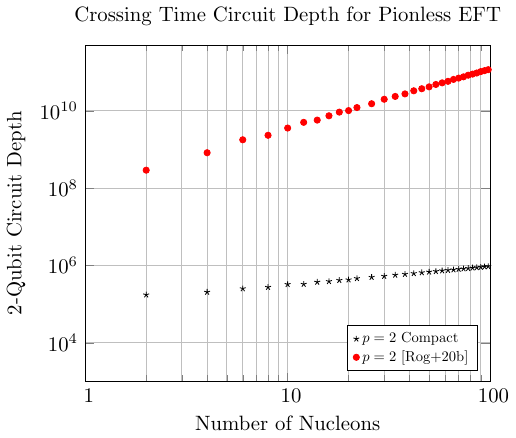}
\caption{Comparison of the pionless EFT 2-qubit circuit depths for this work and Ref.~\cite{roggero2020quantum} for simulating evolution for the crossing time to a total error of 0.1 using a $p=2$ product formula. Here, we assume a $10\times 10 \times 10$ lattice with $a_L = 1.4$~fm (to be consistent with the lattice spacing choice in Ref.~\cite{roggero2020quantum}), with a kinetic energy per nucleon of $E_{\rm kin}=10$~MeV.}
\label{Fig:roggero_Comparison}
\end{minipage}
\hfill
\begin{minipage}[c]{0.475\linewidth}
\includegraphics[width=\textwidth]{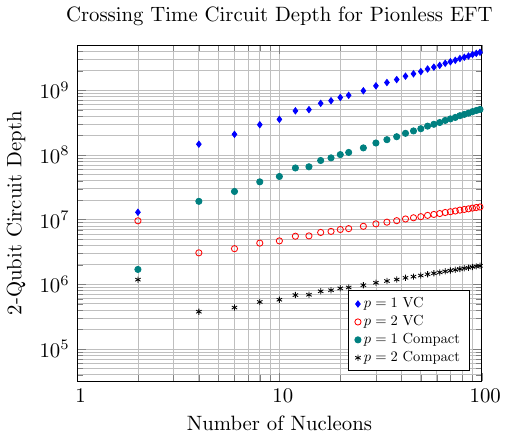}
\caption{Comparison of the pionless EFT 2-qubit circuit depths for VC and compact encodings and $p=1,2$ product formulae, for simulating evolution for the crossing time to a total error of 0.1. We assume a $10\times 10 \times 10$ lattice with $a_L = 2.2$~fm, with a kinetic energy per nucleon of $E_{\rm kin}=10$~MeV.
}
\label{Fig:Pionless_p_Encoding_Comparison}
\end{minipage}%
\end{figure}

Our work provides significant improvement over previous simulation algorithms for pionless EFTs~\cite{roggero2020quantum}.
As shown in \cref{Fig:roggero_Comparison}, for the $p=2$ product formula, our circuit depths can be a factor of about $10^5$ smaller for around 40 fermions on a $10\times 10\times 10$ lattice.
The majority of these gains come from two sources: i) a significantly smaller error bound for product formulae, obtained by direct computation of pertinent commutators for $p=2$~\footnote{For the $p=1$ product formula, Ref.~\cite{roggero2020quantum} considers a more detailed analysis similar to the one we consider for both $p=1$ and $p=2$. Nonetheless, even in the $p=1$ case, our analysis is slightly different than that in Ref.~\cite{roggero2020quantum}, in that we sum contributions to the product-formula error bound from different possibilities for particle numbers per site, whereas Ref.~\cite{roggero2020quantum} considers the maximum of these different possibilities. Compare, for instance, \cref{thm:pionless_trotter_error_p1} to Appendix B of Ref.~\cite{roggero2020quantum}.
To our understanding, such bounds based on the maximum can be slightly violated, although the bounds are still expected to hold given the many other sources of looseness in such analyses.}, and ii) using a local fermionic encoding rather than the Jordan-Wigner encoding, which allows for significant circuit parallelization.
These both contribute roughly equally to the circuit-depth reduction.
Despite this, \cref{Table:Resource_Costs} shows that, for the current analysis and simulation regime considered, and for comparable regimes, the simulation will not be feasible on a near-term quantum computer.

Finally, we compare the efficiency of $p=1$ and $p=2$ formulae, and that of the different fermionic encodings.
First, \cref{Fig:Pionless_p_Encoding_Comparison} shows that the $p=2$ formula drastically outperforms $p=1$ for the pionless EFT. 
Second, the stacked compact encoding allows for a small but meaningful reduction in circuit depths. 
Since the compact encoding uses more qubits, whether it is worthwhile will depend on the number of qubits and circuit depths available.

\subsection{Energy Spectroscopy via QPE} \label{Sec:QPE_Resource_Estimates}
\begin{figure}
\begin{minipage}[c]{0.475\linewidth}
\includegraphics[width=\linewidth]{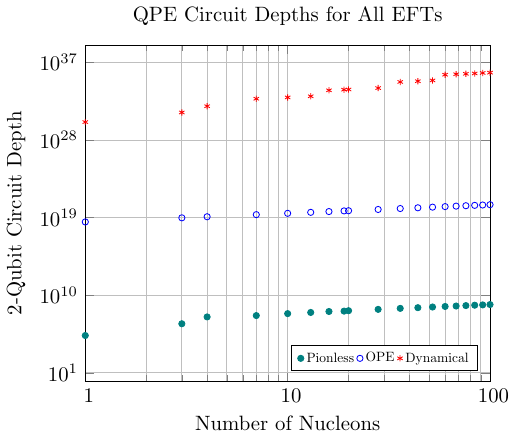}
\end{minipage}
\hfill
\begin{minipage}[c]{0.475\linewidth}
\includegraphics[width=\linewidth]{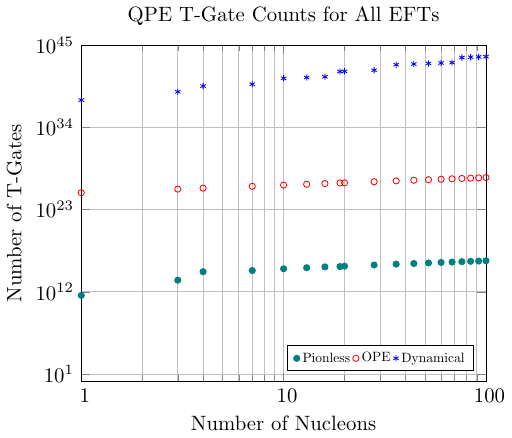}
\end{minipage}%
\caption{Circuit depth (left) and $T$-gate cost (right) as a function of the number of nucleons for quantum phase estimation to 1~MeV of precision on a $10\times 10\times 10$ lattice with $a_L=2.2$~fm with an energy cutoff of $140$~MeV, and with correctness probability $1- \delta = 0.3$. }
\label{Fig:QPE}
\end{figure}
The cost of performing QPE to determine an energy eigenvalue with a given precision is illustrated in \cref{Fig:QPE} for all EFTs, assuming that the corresponding eigenstate has already been prepared~\footnote{QPE can also be used to prepare the state we wish to learn the energy of using projective measurements. However, we leave this to future work.}.
For concreteness, we consider phase estimation with the $p=1$ product formula to a precision of $\Delta E = 1$~MeV on a $10\times 10\times 10$ lattice with $a_L=2.2$~fm with a success probability of $
0.3$. To use the analytical bounds in the case of the EFT with dynamical pions, we must set a cutoff on $\norm{H_{\dyn}}$. We assume that states are bounded by an energy of $E_{\rm max}=140$~MeV (approximately the mass of the pion) such that no dynamical pions are produced in the process. Thus, we replace $\norm{H_{D\pi}}$ with $E_{\rm max} = 140$~MeV in \cref{Eq:RMS_Error}. As observed, pionless EFT is still the cheapest and dynamical-pion EFT is still the most expensive for this task. Finally, \cref{Fig:Pionless_EFT_QPE} shows a comparison of circuit depths between different fermionic encodings and between different orders of product formula for the case of pionless EFT for the QPE task.

\begin{figure}
    \centering
    \includegraphics[width=0.475\linewidth]{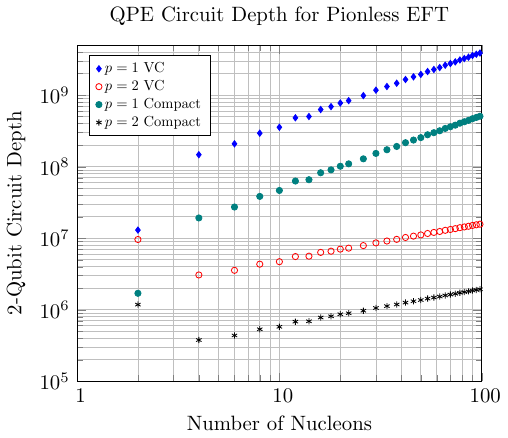}
    \caption{
    Quantum-phase-estimation circuit depth costs for the pionless EFT for $p=1$ and $p=2$ product formulae and for the VC and compact encodings on a $10\times 10 \times 10$ lattice with $a_L=2.2$~fm to 1~MeV of precision, with an energy cutoff of $140$~MeV and with correctness probability $1- \delta = 0.3$. }
    \label{Fig:Pionless_EFT_QPE}
\end{figure}

\subsection{Discussion}

Both the time evolution and QPE benchmarks described above involve computational resources that are currently unavailable.  Current hardware is limited by both noise and gate fidelities, constraining the number of gates that can be implemented before useful information can be extracted.
In particular, given the requirement of a few hundred layers of gates to be performed in parallel for nuclear-EFT simulations of this work, all but the smallest systems are unlikely to be simulatable in the near term.
Indeed, these tasks may be challenging even in the fault-tolerant era.

The comparison between the different models shows that, as expected, the pionless EFT is the least resource-intensive to simulate. 
Perhaps surprisingly, given the large number of long-range interactions which need to be implemented, the OPE EFT is drastically less expensive to simulate than the dynamical-pion EFT both in terms of circuit depth and $T$-gate costs and the number of qubits.
As such, when going beyond pionless EFT, working with the OPE EFT is advantageous, despite the cost of replacing local interactions with non-local ones. 
Part of the reason the OPE EFT is more competitive against the dynamical-pion EFT is due to the larger associated Trotter error resulting from norms of the terms involving pion fields appearing in the error bounds. 
Placing more stringent bounds on these terms using strategies beyond those used in this work may reduce the resource estimates in the future. 
On the other hand, for the OPE EFT, all-to-all connectivity is key. whereas for pionless EFT and dynamical-pion EFTs, one only needs to implement interaction terms between nearby sites.
Although we have assumed an all-to-all architecture here, when implementing simulations on certain realistic architectures, such a feature may not be available, leading to additional overhead.
This may change the comparative advantage of simulating this EFT or other local formulations.

The circuit-depth costs for the pionless EFT show that the $p=2$ product formula can offer significant savings over the $p=1$ case, particularly as $\eta$ grows larger. Higher-order product formulae beyond $p=2$ will likely be even more efficient, but bounding their errors in terms of nested commutators is a daunting task.
Furthermore, the compact encoding gives modest reductions in circuit depth over the VC encoding, by a factor of about $7$, at the expense of a modest increase in the number of qubits.
For near-term devices in particular, this trade-off may be advantageous, but the opposite may hold for fault-tolerant systems.

\section{Summary, Conclusions, and Outlook
\label{sec:results}}

In this work, we have evaluated the cost of simulating various effective field theories of low-energy nuclear physics using near-term and fault-tolerant quantum computers. 
We compared the performance of different simulation methods and investigated how the choices of the EFT formulation, fermionic and bosonic encodings, truncation and digitization of the bosonic Hilbert spaces, cutoffs for long-range interactions, product-formulae order, decomposition of the simulation unitaries into elementary gates, and bounding the error in the chosen algorithm, impact the simulation cost for performing basic tasks such as time evolution and energy spectroscopy. Along the way, we have developed new methods, applied the existing ones in new contexts, obtained new insights, and improved upon prior results. 
In this section, we summarize our findings and conclusions. Despite presenting an extensive study of quantum simulation of low-energy nuclear EFTs, our work can be expanded and complemented in several directions, as we discuss in the second part of this section.

\subsection{Conclusions and Takeaways
\label{Sec:Conclusions}}
Detailed results of this study for the full simulation costs of nuclear EFTs are presented in \cref{sec:results-summary,sec:resource}. Here, we summarize some of the main conclusions that can be taken away from the lengthy analyses of this work, potentially informing other studies of similar model Hamiltonians.

\paragraph{Leveraging the structure of the Hamiltonian and symmetries.} Product-formulae error estimates and the associated circuit depths can be drastically improved by taking advantage of the symmetries and structure of the Hamiltonian, and using fermionic-to-qubit encodings that respect them. This is the primary reason for the significant improvement in the simulation cost of the pionless EFT compared with the result of Ref.~\cite{roggero2020quantum}.

\paragraph{Local versus non-local formulation of the pionful EFT.} Going beyond pionless EFT, which is the least costly EFT considered here, the OPE EFT outperforms the dynamical-pion EFTs in the number of qubits required, circuit depth, and $T$-gate count. As such, it seems the locality of the interactions in the pionful EFTs does not significantly reduce simulation costs and requires significantly more qubits to simulate. However, at least part of the comparative advantage may be due to the fact that we only have poor bounds on the norms of the dynamical-pion Hamiltonians, and future improvements in error-bound analysis may bring the cost down considerably.

\paragraph{First-order versus higher-order product formulae.} Although we have only studied $p=1,2$ product formulae, it is clear that $p=2$ outperforms $p=1$. This is consistent with the conclusions of previous work for the case of pionless EFT, albeit with different error analysis~\cite{roggero2020quantum}.
While higher-order formulae may continue to improve the error bound, one is faced with the issue of placing tight bounds on nested commutators of Hamiltonian terms, which is challenging for complex nuclear Hamiltonians.

\paragraph{Feasibility of simulating nuclear Hamiltonians.} With current techniques and error guarantees, even small-scale quantum simulations of nuclear EFTs are unlikely to be feasible on the noisy intermediate-scale quantum (NISQ) devices. In fact, quantum simulation of nuclear EFTs is currently unlikely to compete with state-of-the-art classical methods for spectroscopy or other static properties of nuclei. This holds even without accounting for the cost of quantum-state preparation, which may be significant. Nonetheless, the case for the promise of using quantum computation in nuclear physics in the fault-tolerant era remains strong, as it is believed that \emph{ab initio} classical methods will not be able to accurately simulate large nuclear isotopes, nor can they systematically access general dynamical properties~\cite{beck2023quantum}.

\subsection{Further Work and Improvements}

While we have examined many aspects of simulating nuclear EFTs,
there is still considerable room for further improvements. Indeed, the design space for quantum simulation (e.g., formulations, encodings, algorithms) is large, so there may be many ways to further optimize the cost of time evolution and other tasks. Here, we enumerate areas that can advance the current state of the art.

\paragraph{Partial error correction via fermionic encodings.} As mentioned in \cref{Sec:Fermionic_Encoding_Introduction}, fermionic encodings work by restricting to a particular subspace of the simulator Hamiltonian.
For the VC encoding, this subspace is defined by a set of stabilizer operators, similarly to a quantum error-correcting code.
Indeed, one can use fermionic encodings to perform partial error correction. This is because at least some subset of physical errors will move the state outside of the simulating subspace. 
By measuring the stabilizers at the end of the simulation, one can detect errors.
This is a general property of many fermionic encodings, and its effectiveness may influence the choice of encoding. 
For example, Ref.~\cite{Jiang_McClean_Babbush_Neven_2019} designs an encoding that can correct all 1-qubit errors on a two-dimensional square lattice. 
The error-correction properties of the compact and VC encodings have also previously been used in optimizing simulation of the Fermi-Hubbard model~\cite{Clinton_Bausch_Cubitt2021}.
It may be worth comparing the feasibility of this partial error correction for the VC and (stacked) compact encoding for the pionless EFT.
Naturally, other encodings exist, many of which have better error correction/detection properties.
However, typically as the code distance increases, the representations of the operators become more complicated~\cite{chen2022error}.

\paragraph{Fermionic quantum computers for simulating fermionic models.} The fermionic encodings discussed all introduce some form of overhead to simulate fermions.
However, it is possible to run quantum computations on devices that are based on fermions. 
This can be used to remove any overhead associated with fermionic interactions compared to the qubit-based systems we have assumed here, see e.g.,  Ref.~\cite{gonzalez2023fermionic} for recent progress.

\paragraph{Cost reduction and circuit optimization.} A more fine-grained analysis of the cost of simulating each Hamiltonian term, which we have avoided in a number of instances, can be performed to further improve the total cost. For example, instead of assigning the highest weight to each operator in a given class (e.g., long-range nucleon-nucleon interactions), one could account for the true weight of each operator. Beyond this, other optimization strategies can be utilized to improve the circuit depths and $T$-gate costs. For example, we have used standard circuit decompositions for various unitaries.
However, these decompositions are by no means optimal. Previous work demonstrated that the circuits can be heuristically optimized using various optimization algorithms~\cite{DiasdaSilva_Prius_Kashefi_2013, childs2018toward, McKeever_Lubasch_2022, Tepaske_Hahn_Luitz_2022}.
There are also less conventional ways of decomposing the circuits. 
For example, the subcircuit model introduced in Ref.~\cite{Clinton_Bausch_Cubitt2021} is potentially more appropriate for circuit compilation than standard gate-set techniques for NISQ-era devices.

\paragraph{Better error bounds on bosonic simulations.} 
Much of the significant cost of simulating the EFT that explicitly includes pions is due to a loose bound on the error associated with introducing a cutoff of the pion-field strength.
In particular, we have used the energy-based truncation methods from Ref.~\cite{Jordan_Lee_Preskill_2012}, as other improved methods of calculating the field strength cutoff, such as that of Ref.~\cite{Tong_Albert_Mcclean_Preskill_Su_2022}, cannot be applied to the Hamiltonian in this work in their current form.
We strongly suspect that this bound can be improved. 
Another potential route for improvement is to find a way of applying a tighter large-deviation inequality than the Chebyshev bound used in \cref{Sec:Pion_Cutoff}, such as Hoeffding's inequality. This may require additional assumptions on the pion field and is left to future studies.

\paragraph{Better error analysis and empirical scaling.}
The bounds on the quantity $\tilde{\alpha}_{\mathrm{comm}}$ in \cref{Eq:tilde-alpha-def}, which determines the product-formulae error bound, are unlikely to be tight. An alternative method is to simply simulate the system classically and determine the actual error, which can be extrapolated to larger systems~\cite{childs2018toward,stetina2022simulating,nguyen2022digital}.
However, this is not a straightforward task even for rather small systems of nucleons. Recall that the simulations involve 6 qubits per site on a 3D lattice in the VC encoding, so even an unphysically small lattice of size $L=2$ requires simulating the dynamics of 48 qubits, which is at the edge of what is feasible with the most powerful classical computers. Using other encodings, including the non-local Jordan-Wigner encoding, does not improve the situation much (and can increase the number of non-commuting operators due to the induced non-locality). One can potentially resort to non-exact but efficient classical Hamiltonian-simulation methods, such as tensor networks, but even such methods are not widely applicable to 3D quantum many-body systems. In fact, one may need to await the availability of large-scale quantum computers to be able to perform simulation tasks and discover empirical scalings for the algorithms (e.g., by benchmarking against known results from experiment). Until such knowledge is available, strategies for improving and effectively calculating error bounds will be highly valuable in estimating resource requirements more accurately

\paragraph{Designing error-correction protocols.}
The resource counts we find suggest that fault-tolerant quantum computers will be required to implement the algorithms of this work. Hence, a potentially fruitful avenue is to design error-correcting codes that take advantage of the structure of the simulation. This route can reduce the overhead for fault tolerance and thus make the algorithms easier to implement in the near future. Examples include taking advantage of the inherent error-detection abilities of fermionic encodings~\cite{Jiang_McClean_Babbush_Neven_2019,Bausch_2020}, or otherwise designing these protocols with fermions in mind.

\paragraph{Beyond product formulae.}
There exist various other time-evolution algorithms with better asymptotic scaling in terms of error and evolution time~\cite{BCCKS14,berry2015hamiltonian, low2017optimal,low2019hamiltonian,LW18,haah2021quantum}. These typically involve more complex circuits that use additional ancilla qubits.
Empirical studies suggest that, for certain problems, product formulae perform better for instances of modest size~\cite{childs2018toward}, as mentioned above, but it might still be worth studying whether such approaches can be valuable for nuclear-EFT simulations in some regimes.
Alternatively, techniques such as Trotter-error extrapolation might be used to reduce the error~\cite{rendon2022improved}.  Ultimately, knowledge of the simulation's input state may improve the product-formula error bounds, as studied in Refs.~\cite{csahinouglu2021hamiltonian,yi2022spectral,zhao2022hamiltonian,zhao2024entanglement}, which should be explored further in the context of nuclear-EFT simulations. 

\paragraph{Different quantum-phase-estimation routines.}
The phase estimation routine used in \cref{Sec:QPE_Resource_Estimates} is a standard variant of QPE.
However, there are many alternative QPE methods that may improve the gate counts, and in particular, some may be more suitable for near-term devices~\cite{Svore_Hastings_Freedman_2014, OBrien_Tarasinski_Terhal_2019, somma2019quantum, Clinton_Bausch_Klassen_Cubitt_2021, lin2022heisenberg}. These can be explored in the context of quantum simulation of nuclear EFTs in the future.

\paragraph{More restricted nuclear systems.}
In this work, we have considered nuclear systems in the presence of all species of nucleons and pions. 
However, there are some use cases where one may be able to remove some species. For example, when studying neutron matter (e.g., in neutron stars), the interactions between particles can be simplified, reducing the resource requirements for simulation.

\paragraph{Boundary conditions.}
Here, we have considered simulation with open boundary conditions. 
However, periodic boundary conditions may cause less boundary distortion in the wavefunctions.
Most of the analysis of this work will remain similar for the periodic case, but with a small overhead to account for terms crossing the boundary.

\paragraph{Instantaneous-pion EFT: Combining classical and quantum routines.} The instantaneous-pion EFT is a limiting case of the dynamical-pion EFT in which the pions undergo no dynamics, and serve as a background static-field configuration in which nucleon dynamics take place. 
Such a formulation leaves local pion-nucleon interactions in the description, and is in fact equivalent to the long-range one- (and multi-) pion exchange EFT considered in this work. Simulating such a model happens to be less costly than the dynamical-pion EFT, as can be verified by taking $\Pi_I=0$ in the dynamical-pion EFT analysis. Nonetheless, the pion-field configurations need to be  sampled classically, e.g., using Monte Carlo importance-sampling methods with the static-pion action as the sampling weight. Each configuration is then used to initialize the state of the pion fields in the quantum algorithms of this work, so that the quantum dynamics of the nucleons coupled to these pion states can be studied on a quantum computer. Such a hybrid classical-quantum algorithm may be worthwhile in the near term, but concrete determination of its resource requirements necessitates an error analysis that combines both classical Monte Carlo and quantum-simulation algorithm errors. Similar hybrid approaches to quantum simulation have recently been proposed in other contexts, e.g., in Refs.~\cite{harmalkar2020quantum,gustafson2021toward}. However, such algorithms are limited when the classical calculation has a sign problem, so fully quantum approaches may be necessary in general.

\paragraph{Improved EFT Hamiltonians.} Our analysis has been limited to pionless nuclear EFT at leading order (contact two- and three-body interactions) and pionful nuclear EFTs at leading order in Weinberg power counting (contact interactions plus OPE potential, or alternatively, leading pion-nucleon couplings in the dynamical-pion theory). A clear next step is to devise simulation algorithms with bounded errors for higher-order EFTs, which would involve derivatively coupled nucleons, multi-pion exchange potentials, or in the dynamical-pion case, higher-order pion-nucleon couplings and pion self-interactions. 
Introducing these terms adds further complexity that will increase resource requirements, but they are essential in accurate and high-precision studies of medium- and large-mass nuclei. The methods in this work are broadly applicable and should allow for more complex interactions to be studied, including higher-derivative couplings between the fields. Concrete simulation costs will need similar dedicated studies.

\paragraph{Holistic uncertainty quantification.}
Considering the numerous systematic errors in the simulation, from model uncertainties (e.g., the finite EFT order, lattice-discretization effects, finite-volume effects, field truncation and digitization effects), to algorithmic approximations (e.g., product-formula order, time digitization, gate synthesis), a more holistic approach to uncertainty quantification may be needed to obtain realistic resource estimates. In particular, it may not be justified to overly suppress algorithmic errors at the cost of drastically increasing resources while accuracy will be limited by other systematic uncertainties.

\paragraph{State preparation.}
In this work, we have ignored the cost of state preparation, which may be very expensive.
In general, ground-state preparation is QMA-hard,\footnote{QMA (or Quantum Merlin Arthur) is a quantum version of the classical complexity class NP (or Non-deterministic Polynomial-time)~\cite{Kitaev_Shen_Vyalyi_2002}.} so there should not be efficient general-purpose algorithms for this task.
However, there are many provably convergent methods (which require exponential time in general)~\cite{motta2020determining,cubitt2023dissipative}, and many heuristic approaches such as the variational quantum eigensolver and the unitary coupled-cluster ansatz, that have been explored with various degrees of success~\cite{tilly2022variational}.
Alternatively, given the tremendous success of classical \emph{ab initio} quantum many-body methods in nuclear physics, it is reasonable to suppose that known nearly exact or approximate nuclear wavefunctions obtained from such methods may enable more efficient initialization of the quantum-simulation algorithms~\cite{hergert2020guided}, although more work is needed to make this approach concrete and understand its performance in detail.
Additionally, the local fermionic encodings used in this work incur state-preparation overhead to initialize the simulation in the appropriate encoded state.
However, we expect this cost to be much less than the overall cost of the simulation, and in the fault-tolerant regime, to be much less than the fault-tolerant overhead.

\paragraph{Other applications in nuclear physics.} In this work, we focused on understanding the costs of time evolution and spectroscopy. Naturally, there are numerous other relevant properties of nuclear systems, such as scattering amplitudes, reaction rates, thermal properties, and structure and response functions. Examining algorithmic resource requirements for determining these properties is left for future work. However, time evolution is a basic subroutine that should be useful for accessing all these quantities, so the circuit constructions and cost analysis of this work should be relevant. If state preparation and measurement involve different bases than those considered here (e.g, momentum- versus position-space fields), one can implement the relevant basis transformations between various stages of the simulation, as demonstrated for both bosonic and fermionic field theories in, e.g., Refs.~\cite{barata2021single,mueller2023quantum}.

\section*{Acknowledgments}

{\begingroup
\hypersetup{urlcolor=navyblue}
We thank
\href{https://orcid.org/0000-0002-5087-9346}{Toby Cubitt},
Charles Derby, and
\href{https://orcid.org/0000-0003-1438-3172}{Joel Klassen} for helpful discussions on the topic of fermionic encodings.
We also thank
\href{https://orcid.org/0000-0003-2539-271X}{Alessandro Baroni},
\href{https://orcid.org/0000-0002-4370-3297}{Thomas Cohen}, \href{https://orcid.org/0000-0002-3630-567X}{Dean Lee}, and \href{https://orcid.org/0000-0002-8334-1120}{Alessandro Roggero} for useful discussions regarding nuclear EFTs and nuclear-state preparation. We are grateful to \href{https://orcid.org/0000-0002-2020-8971}{Christopher~Kane} for his careful read of an earlier version of this manuscript and his valuable feedback in improving aspects of this work.
\endgroup}

J.D.W.~and~A.M.C.~acknowledge support from the United States Department of Energy (DoE), Office of Science, Office of Advanced Scientific Computing Research (ASCR), Accelerated Research in Quantum Computing (ARQC) program (award No.~DE-SC0020312), and from the National Science Foundation (NSF) Quantum Leap Challenge Institutes (QLCI) (award No.~OMA-2120757). 
J.B.~and A.V.G.~were supported in part by the NSF QLCI (award No.~OMA-2120757), NSF STAQ program, AFOSR, DoE ASCR Quantum Testbed Pathfinder program (award No.~DE-SC0019040), AFOSR MURI, ARL (W911NF-24-2-0107),  DARPA SAVaNT ADVENT, and NQVL:QSTD:Pilot:FTL. 
J.B.~and A.V.G.~also acknowledge support from the DOE, Office of Science, National Quantum Information Science Research Centers, Quantum Systems Accelerator.  J.B.~and A.V.G.~also acknowledge support from the DoE, Office of Science, Accelerated Research in Quantum Computing, Fundamental Algorithmic Research toward Quantum Utility (FAR-Qu).  
J.B.~also acknowledges support from the Harvard Quantum Initiative.
A.F.S.~was supported by the Department of Defense and the NSF Graduate Research Fellowship Program (GRFP). 
Z.D.~was supported by the DoE, Office of Science Early Career Award (award No.~DE-SC0020271) for theoretical and algorithmic developments for simulating nuclear effective field theories on quantum computers. 
She was further supported by the DoE, Office of Science, Office of ASCR, ARQC program (award No.~DE-SC0020312) for algorithmic developments for quantum simulation of fermionic systems, and by the DoE, Office of Science, Accelerated Research in Quantum Computing, Fundamental Algorithmic Research toward Quantum Utility (FAR-Qu) for improved resource and algorithmic-error  analysis in fermionic and bosonic systems.

\newpage
\appendix

\section{Algorithmic Overview
\label{Sec:algorithms}}
\noindent
The two main algorithms considered in this work to simulate nuclear effective field theories are product-formula algorithms for simulating digitized time evolution and quantum phase estimation to obtain energy spectra. These algorithms and their error analysis are well known and are summarized in this Appendix for completeness.

\subsection{Quantum Simulation with Product Formulae} \label{Sec:Product_Formulae_Summary}
The most straightforward approach to quantum simulation employs product formulae to write the exponential of a sum of the Hamiltonian terms as a product of exponentials of the individual terms~\cite{Lloyd_1996}.
This approach can be improved by employing higher-order approximations such as a widely used recursive construction of Suzuki~\cite{suzuki1991general}, leading to asymptotically more efficient quantum-simulation algorithms~\cite{childs2004quantum,berry2007efficient,childs2021theory}.
Product-formula simulations have been shown to perform well compared to more complex simulation algorithms~\cite{childs2018toward}, with the benefits of preserving the locality of the system being simulated and not requiring additional auxiliary qubits.
Since the error of product-formula approximations is determined by norms of nested commutators of Hamiltonian terms (rather than simply the norms of the terms), this approach can perform well in practice~\cite{Childs_Su_2019, childs2018toward,shaw2020quantum}.

The basic idea of product-formula simulation is to split the time evolution of a quantum Hamiltonian into a sequence of simpler evolutions for small time steps, each of which can be performed efficiently.
Suppose one wishes to implement the unitary $e^{-itH}$, where $H = \sum_{\gamma=1}^{\Gamma} H_\gamma$, and suppose that each $e^{-itH_\gamma}$ can be implemented exactly (or almost exactly), for any desired time $t$, by a simple quantum circuit.
Then, for the first-order product formula
\begin{align}
\mathcal{P}_1(t)\coloneqq\prod_{\gamma=1}^\Gamma e^{-it H_\gamma},
\label{eq:P-first-order}
\end{align}
it can be shown that~\cite[Proposition 9]{childs2021theory}
\begin{align}
    \label{Eq:p-1-commutator-bound}
    \norm{e^{-itH} - \mathcal{P}_1(t)}
    &\leq \frac{t^2}{2}\sum_{\gamma_1=1}^\Gamma\norm{ \left[H_{\gamma_1},\sum_{\gamma_2=\gamma_1+1}^\Gamma H_{\gamma_2} \right]  }.
\end{align}
Using the triangle inequality, the time evolution can be broken into $r$ steps of length $t/r$, such that 
\begin{align}
    \norm{e^{-itH} - \mathcal{P}_1(t)  } \leq r\norm{e^{-itH/r} - \mathcal{P}_1(t/r)}.
\end{align}
Thus, to implement the time-evolution unitary with an overall error of at most $\epsprod$, it suffices to use
\begin{align}
    r = \frac{t^2}{2\epsprod} \sum_{\gamma_1=1}^\Gamma\norm{\left[H_{\gamma_1},\sum_{\gamma_2=\gamma_1+1}^\Gamma H_{\gamma_2} \right]  }
    \label{Eq:r-vs-epsprod}
\end{align}
time steps.
By choosing a sufficiently large $r$, time evolution can be simulated to any desired precision with only polynomial overhead in the simulation time. 

As mentioned above, the asymptotic performance of this approach can be improved by using higher-order approximations.
For example, for the second-order formula
\begin{align}
    \label{Eq:Second-Order-PF-Def}
    \cP_2(t) &\coloneqq \prod_{\gamma=1}^\Gamma e^{-itH_\gamma/2} \prod_{\gamma=\Gamma}^1 e^{-itH_\gamma/2},
\end{align}
it can be shown that~\cite[Proposition 10]{childs2021theory}
\begin{align}
    \label{Eq:p-2-commutator-bound}
    \norm{e^{-itH} - \calP_2(t)} \leq \frac{t^3}{12}\sum_{\gamma_1=1}^{\Gamma} \norm{ \left[\sum_{\gamma_3=\gamma_1+1}^\Gamma H_{\gamma_3}, \left[ \sum_{\gamma_2=\gamma_1+1}^{\Gamma}  H_{\gamma_2}, H_{\gamma_1}  \right]\right] 
      }+ \frac{t^3}{24}\sum_{\gamma_1=1}^\Gamma\norm{  \left[ H_{\gamma_1}  \left[H_{\gamma_1}, \sum_{\gamma_2=\gamma_1+1}^{\Gamma}  H_{\gamma_2}   \right]\right] }.   
\end{align}

Suzuki recursively defined $p$th-order product formulae for all even $p$ as \cite{suzuki1991general}
\begin{align}
   \cP_{p+2}(t) &\coloneqq \cP_{p}^2(s_pt) \cP_{p}\big((1-4s_p)t\big) \cP_{p}^2(s_pt), \quad \quad  p\in 2\mathbb{N}, p\geq 2
\end{align}
where
\begin{align}
    s_p \coloneqq (4-4^{1/(p+1)})^{-1}.
\end{align}
The error of these higher-order formulae satisfies \cite[Theorem 6 and Appendix E]{childs2021theory}
\begin{align} \label{Eq:Product_Formulae_Error}
        \norm{e^{-itH} - \calP_p(t)} \leq 2\Upsilon^{p+1}\frac{t^{p+1}}{p+1}\tilde{\alpha}_{\mathrm{comm}}^{(p)},
\end{align}
where $\Upsilon \coloneqq 2\times 5^{p/2-1}$ and
\begin{align}
    \label{Eq:tilde-alpha-def}
    \tilde{\alpha}_{\mathrm{comm}}^{(p)} \coloneqq \sum_{\gamma_{p+1},\gamma_{p},\dots, \gamma_1=1}^\Gamma\norm{[H_{\gamma_{p+1}}[H_{\gamma_{p}},\dots [H_{\gamma_2},H_{\gamma_1}]]]  }.
\end{align}
For constant $p$, the number of Trotter steps that ensures an error of at most $\epsprod$ for a given time $t$ scales as $O(t^{1+1/p}\epsprod^{-1/p})$.
Thus, higher-order product formulae offer asymptotically better performance for large $t$ and small $\epsprod$. However, the rapid growth of $\tilde{\alpha}_{\mathrm{comm}}^{(p)}$ with $p$ restricts the order that should be used in practice for particular finite values of $t$ and $\epsprod$.
Nonetheless, it was found in Ref.~\cite{childs2018toward} that even for small systems with tens of qubits, formulae with $p=4$ or $p=6$ can outperform lower-order formulae. 

A significant portion of the current work consists of deriving good bounds on $\tilde{\alpha}_{\mathrm{comm}}^{(p)}$ for particular nuclear EFT Hamiltonians by exploiting their known structure and \emph{a priori} knowledge about the physical system. 
This allows us to minimize the gate counts needed to achieve a particular precision.
In this work, we compute $\tilde{\alpha}_{\mathrm{comm}}^{(p)}$ with $p=1,2$ for relevant Hamiltonians, and also find loose upper bounds on $\tilde{\alpha}_{\mathrm{comm}}^{(p)}$ for higher-order ($p\geq 4$) formulae for some general fermionic Hamiltonians.

\subsection{Spectroscopy via Quantum Phase Estimation} \label{Sec:QPE}
To benchmark nuclear-simulation algorithms and hardware, and to enable ab initio theoretical determinations of nuclear spectra for large atomic isotopes, nuclear spectroscopy will be a desired task for quantum computers. A common approach to determining an energy eigenvalue of a Hamiltonian is to use a Quantum Phase Estimation (QPE) routine.
The QPE algorithm assumes an oracle has prepared an eigenstate (or a state with non-vanishing overlap with the eigenstate) whose eigenvalue is to be estimated. Textbook phase estimation involves circuits for inverse QFT on $n$ ancillary qubits and controlled-$U^{2^j}$ for all $j\in [0,n-1],$ where in our case, $U$ is the unitary operator implementing (often an approximation to) $e^{-itH}$. Here, $0 < t < 2\pi/\norm{H}$. Let us first assume that $U$ implements $e^{-itH}$ exactly. We will remove this assumption shortly. 
Consider an input eigenstate $\ket{\lambda}$, where $2\pi\lambda \coloneqq -tE$ is an eigenvalue of the operator $-tH$, with $E$ being the eigenenergy to be estimated. Note that the condition on $t$ ensures that $0 \leq |\lambda| < 1$. The circuit first performs $n$ Hadamard gates on $n$ ancilla qubits that are initialized in $\ket{0}^{\otimes n}$, followed by controlled-$U^{2^j}$ on the register holding the eigenstate, where the control is upon the $j^{\rm th}$ ancilla qubit. Finally, the circuit applies an inverse QFT to the ancillary register. This gives
\begin{align}
    \ket{\lambda}\otimes\ket{0}^{\otimes n} \xrightarrow[]{{\rm QPE}} &  \ket{\lambda} \otimes \sum_{k=0}^{2^n-1}\biggl(\frac{1}{2^n}\sum_{j=0}^{2^n-1} e^{2\pi i \left(\lambda - 2^{-n}k\right)j}\biggr) \ket{k}
    \\
     = & \ket{\lambda}\otimes \sum_{k=0}^{2^n-1}\biggl( \frac{1}{2^n}\sum_{j=0}^{2^n-1} e^{2\pi i \Delta \lambda j}\biggr) e^{2\pi i \left(\widetilde{\lambda} - 2^{-n}k\right)j} \ket{k}.
     \label{Eq:QPE-2}
\end{align}
Here, $\widetilde{\lambda}$ is the closest $n$-bit approximation to $\lambda$,
that is $\lambda = \widetilde{\lambda}+\Delta \lambda$ with $0 \leq |\Delta \lambda|\leq2^{-n-1}$. Now if $\Delta \lambda = 0$, measuring the ancillary register will obtain $\ket{\widetilde{\lambda}}=\ket{\lambda}$ with probability unity. For $\Delta \lambda \neq 0$, the measurement obtains $\ket{\widetilde{\lambda}}$ with a fixed probability of at least $4/\pi^2$, which is obtained by bounding the absolute square of the geometric sum in the parentheses in \cref{Eq:QPE-2}, see e.g., Ref.~\cite{Cleve_1998}. To improve the guarantee on the probability, and for a fixed number of ancilla qubits $n$, one needs to compromise on the absolute error. Explicitly, it can be shown that, to reach an absolute error $2^{-m-1}$ on the eigenvalue with $m < n$, with a guaranteed success probability of $1-\delta$, the number of ancilla qubits required is given by
\begin{align}
    n = m + \biggl \lceil \log_2 \biggl(\frac{1}{2\delta} + \frac{1}{2} \biggr) \biggr \rceil.
\end{align}
The eigenvalue estimate, $\widetilde{\lambda}$, is obtained by rounding off the resulting $n$-bit string to its most significant $m$ bits~\cite{Cleve_1998}.

Among the variations of the standard QPE is the iterative algorithm which replaces $n$ ancilla qubits and the costly QFT routines with a single ancilla qubit and $n$ iterations of 1-qubit rotations,  measurements, and classical feedback, with the same probability of success as before, see Refs.~\cite{kitaev2002classical,Nielsen_Chuang_2010}. Other improvements to the standard QPE, as well as other phase-estimation algorithms, have also been developed~\cite{knill2007optimal,nagaj2009fast,berry2009perform,higgins2007entanglement,lin2022heisenberg,ge2019faster,lin2020near,somma2019quantum}; nonetheless, we consider only the iterative QPE in this work to keep the presentation simple.

For QPE, there are multiple sources of error in the extracted energy, but for the time being let us consider two primary sources: the error inherent to QPE due to the $m$-bit approximation of the output eigenvalue, and the error due to the approximate time evolution, which in this work amounts to using the product-formulae algorithms. One possibility is to use the error analysis from Ref.~\cite{Tran_2019} to bound the difference between the full Hamiltonian and the Hamiltonian induced by Trotterization.
However, for simplicity, we follow the error analysis from Refs.~\cite{babbush2018encoding, stetina2022simulating} and estimate the error through the root-mean squared of the two error sources: 
\begin{align}
\label{Eq:RMS_Error_Two-Errors}
        t\Delta E = \sqrt{ \left( \frac{\pi}{2^m}\right)^2 + 
    \left(t\Delta E_{\rm prod}\right)^2 }.
\end{align}
As we will see, if a model simulation involves other sources of error beyond the ones accounted for in \cref{Eq:RMS_Error_Two-Errors}, they can be simply added to the product-formula error in this equation. We come back to this point in \cref{Sec:Beyond_Prod_Errors}. Let us now proceed to bound $t\Delta E_{\rm prod}$.

\paragraph{QPE Costs with $p=1$ Product Formula.}
If $\Heff$ is the effective Hamiltonian induced by the first-order product formula, i.e,, $\mathcal{P}_1(t) = e^{-it\Heff } $, then 
\begin{align}
  t\Delta E_{\rm prod} \leq  
  t\norm{H - \Heff}  \approx \norm{e^{-itH} - \mathcal{P}_1(t)} \leq  r \norm{e^{-itH/r} - \mathcal{P}_1(t/r)} 
  =r\epsprod(t/r),
\end{align}
where $\epsprod$ is defined in \cref{Eq:r-vs-epsprod}. Recall that $t$ is upper bounded by $2\pi/\norm{H}$. Therefore the largest possible $t\Delta E_{\rm prod}$, that is $2\pi \Delta E_{\rm prod}/\norm{H}$, is bounded by $r\epsprod(2\pi/\left(\norm{H}r)\right)$.

\paragraph{QPE Costs with $p=2$ Product Formula.} For $p=2$ formulae, we have the similar bound
\begin{align}
    t\Delta E_{\rm prod} 
    \approx \norm{e^{-itH} - \mathcal{P}_2(t)} \leq r \norm{e^{-itH/r} - \mathcal{P}_2(t/r)} 
  =r\epsprod(t/r),
\end{align}
where $\epsprod=\frac{16t^3}{3r^2}\tilde{\alpha}_{\rm comm}^{(2)}$, with $\tilde{\alpha}_{\rm comm}^{(2)}$ defined in \cref{Eq:tilde-alpha-def}.

\section{Overview of the Verstraete-Cirac Encoding}\label{Sec:VC_Stabilizers}

In this appendix, we review the Verstraete-Cirac encoding in 2D and 3D, including details of the subspace in which the simulation needs to be restricted to for the encoding to work.

\paragraph{The 2D case.} Consider the case of a 2D lattice first. For the VC encoding to function correctly, the simulation should be restricted to a subspace that satisfies 
\begin{align}
P_{ij}^\mu\ket{\psi} = \ket{\psi},
\end{align}
where $P_{ij}^\mu \coloneqq i\mu_{i}\bar{\mu}_{j}$ on a set of edges defined along appropriate directed paths. Each auxiliary Majorana operator needs to appear in exactly one $P_{i'j'}$ along those paths. A possible configuration of paths on a $4 \times 4$ lattice is shown in \cref{Fig:VC_Auxiliary_Terms_2D}, corresponding to the Jordan-Wigner ordering of the physical modes chosen in the left panel. 
The consequence of this construction is that, for the hopping term linking site indices $i$ and $j$, whose indices are linked by a $P_{ij}^\mu$, the following property holds:
\begin{align}
\left(\adag(i)a(j) + \adag(j)a(i)\right) P_{ij}^\mu\ket{\psi} = \left(\adag(i)a(j) + \adag(j)a(i)\right)\ket{\psi}.
\label{eq:hop_VC_stabilizer}
\end{align}
Furthermore, it is easy to see that, while the right-hand side of \cref{eq:hop_VC_stabilizer} turns into a non-local spin interaction via the original Jordan-Wigner transform, the left-hand side is mapped to a local term with the help of the auxiliary modes in the code space, recalling the definitions in \crefrange{eq:mu-def-I}{eq:nu-def-II}. In other words, the choice of auxiliary-mode pairing and of the paths allows for the cancellation of Jordan-Wigner strings between geometrically local interactions separated by the chosen qubit indexing. Note that a hopping term linking site indices $k$ and $l$, which are not linked by a $P_{kl}^\mu$, 
does not need to be modified by adding a stabilizer, as this term is already mapped to a local qubit interaction in the original Jordan-Wigner mapping, as can be seen in the example of \cref{Fig:VC_Auxiliary_Terms_2D}.

\begin{figure}[ht]
\centering
\includegraphics[width=0.8\textwidth]{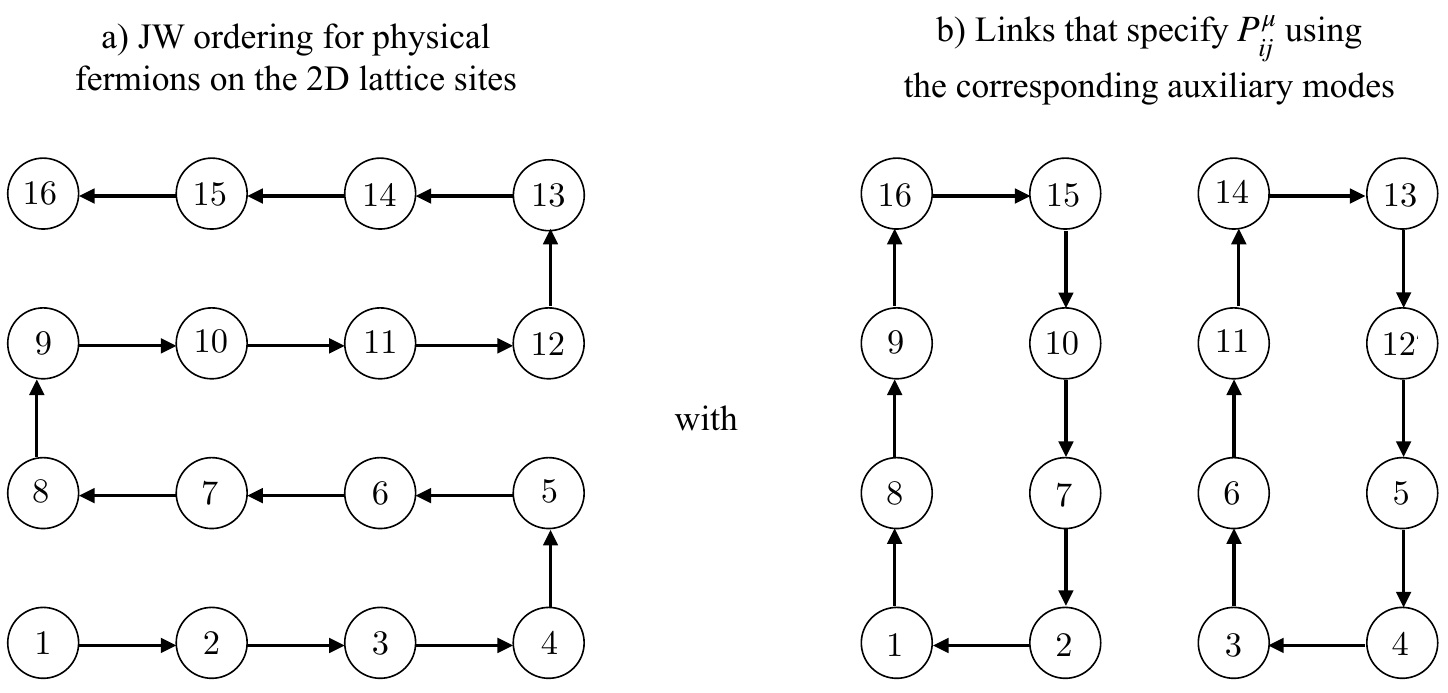}
\caption{a) A possible ordering of the physical sites (circles) on a 2D lattice for mapping to qubits via a Jordan-Wigner transformation, where the qubit index is noted inside the circles. b) The corresponding auxiliary layer of qubits each containing two Majorana modes $\mu$ and $\bar{\mu}$, along with a possible choice of a set of $P_{ij}^\mu$ operators along the arrows.
}
\label{Fig:VC_Auxiliary_Terms_2D}
\end{figure}

\paragraph{The 3D case.} To generalize to 3D, it is sufficient to introduce another set of auxiliary Majorana modes $\nu$ and $\bar{\nu}$ defined on another layer of auxiliary qubits. \Cref{Fig:VC_Auxiliary_Terms_3D} shows a choice of mapping physical fermionic modes to qubits indexed by $i$ along a Jordan-Wigner path, along with the auxiliary layers $\mu$ and $\nu$, for which a set of stabilizers $P_{ij}^\mu$ and $P_{ij}^\nu$ along given paths in the $x$-$y$ and $y$-$z$ planes are introduced, respectively. 
These configurations ensure that any geometrically nearest-neighbor hopping term in the Hamiltonian remains local after the mapping, either because it is still nearest-neighbor along the Jordan-Wigner path, or because the non-local Jordan-Wigner strings associated with unphysical separation along the Jordan-Wigner path are canceled out by the introduction of the stabilizers on the corresponding edge. 
The 3D choice described here is what we have implemented in this work to map the relevant EFT Hamiltonians to qubit Hamiltonians (except where we use the compact encoding for the pionless EFT).
Note that the presence of four distinct physical fermionic modes on each physical site in the nuclear EFT problem does not require introducing additional auxiliary layers of qubits and fermionic operators, as argued in the main text.

\begin{figure}[t!]
\centering
\includegraphics[width=0.95\textwidth]{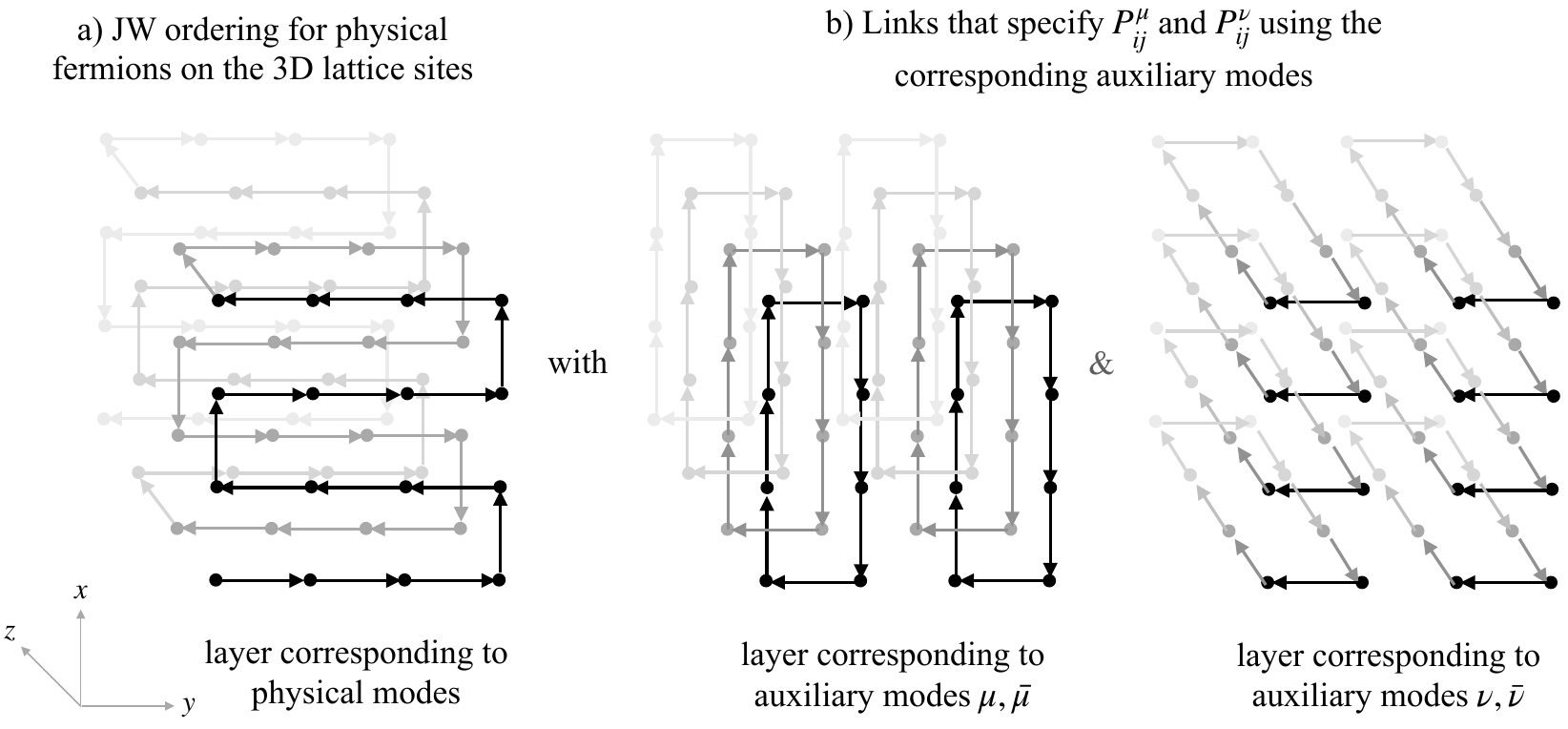}
\caption{
a) A possible ordering of the physical sites on a 3D lattice for mapping to qubits via a Jordan-Wigner transformation. Qubit indices are left implicit. b) The corresponding auxiliary layers of qubits, each layer containing two Majorana modes $\mu,\bar{\mu}$ and $\nu,\bar{\nu}$, along with a possible choice of a set of $P_{ij}^\mu$ and $P_{ij}^\nu$ stabilizers.
}\label{Fig:VC_Auxiliary_Terms_3D}
\end{figure}

\paragraph{State preparation costs.}
The restricted subspace defined by the stabilizers is equivalent to the toric code up to a constant-depth quantum circuit.
With an appropriate procedure, the 2D toric-code ground state can be prepared by an $ O(\log(L))$-depth quantum circuit~\cite{konig2009exact}, which is equivalent to the compact and VC encodings code up to finite-depth circuits~\cite{Derby_et_al_2021, bombin2012universal}.
Thus, the state preparation for the fermionic encodings should be possible in $O(\log(L))$ total depth in 2D, and should, therefore, contribute only limited quantum resources compared to the time-simulation algorithm itself. While we expect similar state-preparation cost for the local fermionic encoding in 3D, this expectation needs to be verified in future work.

\section{Bounding the Simulation with Truncated Long-Range Interactions}\label{Sec:Cutoff_Proof}
In this Appendix, we present a proof of \cref{Lemma:Cutoff_Length}.

\begin{proof}[Proof of \cref{Lemma:Cutoff_Length}]
    First note that if $H$ and $K$ are two time-independent Hamiltonians on the same Hilbert space, then it can be shown     (e.g.,~\cite[Lemma 4]{haah2021quantum}) that 
\begin{align}
    \norm{e^{iHt} - e^{iKt}   } \leq t\norm{H-K} \eqqcolon \epstrunc(t).
\end{align}
In the following, we bound the spectral-norm difference between the long-range Hamiltonian $H_{\rm LR}$ in \cref{Eq:OPE_Long-Range_Terms} and its truncated form, $H_\ell$, that only incorporates two-body interactions up to range $\ell$. 
All norms are considered in a fixed nucleon-number sector, denoted by an $\eta$ subscript. 
The difference can, therefore, be expressed as
\begin{align}
    \norm{H_{\ell} - H_{\rm LR}  }_\eta
    = \norm{\sum_{\substack{\bm{x},\bm{y} \\|\bm{x}-\bm{y}|> \ell}}\sum_{\substack{\alpha,\beta,\gamma,\delta \\ \alpha', \beta', \gamma', \delta'}} [G(|\bm{x}-\bm{y}|)]_{\alpha'\beta'\alpha\beta \gamma'\delta' \gamma \delta} \adag_{\alpha'\beta'}(\bm{x})\adag_{\alpha\beta}(\bm{y})a_{\gamma'\delta'}(\bm{x})a_{\gamma\delta}(\bm{y})}_\eta,
    \label{Eq:Hl-minus-HLR-norm}
\end{align}
where $G(|\bm{x}-\bm{y}|)$ is defined in \cref{Eq:G-def}. 
There are two ways of bounding this norm, leading to the two arguments of the minimum function in \cref{eq:lem2}:

\paragraph{Method 1.}
First note that, among $\eta$ nucleons, there can be at most $\eta(\eta-1)/2$ pairwise
interactions.
To proceed, we decompose the state in the fermion occupation basis.
That is, let $\zeta\in \{0,1\}^{|\Lambda|}$, with $|\Lambda|$ denoting the number of points on the 3D lattice. Then let $\ket{\psi^\zeta_\eta}$ be a state that has $\eta$ fermions, with $\zeta_{\bm{x}}=0$ if there is no fermion at lattice site $\bm{x}\in \Lambda$ and $\zeta_{\bm{x}}=1$ otherwise. 
Consider also $\ket{\psi_\eta} \coloneqq \sum_{\zeta}b_\zeta \ket{\psi_\eta^\zeta} $, with $\sum_\zeta |b_\zeta|^2 = 1$, which is a superposition of states with different distributions of non-zero $\zeta$ values at various sites but containing exactly $\eta$ fermions. 
The notation $\|\cdot\|_{\eta,\zeta}$ indicates the spectral norm with respect states $\ket{\psi_\eta^\zeta}$.
Since both $H_{\ell}$ and $H_{\rm LR}$ are block diagonal with respect to this decomposition, the only terms that need to be considered are those acting between occupied sites.
We further define $S_\zeta$ to be the set of all lattice points with at least one fermion present.
Then,
\begin{align}
    &\norm{H_{\ell} - H_{\rm LR}  }_\eta \nonumber \\
    &\quad\leq \max_{\{b_\zeta\}} \sum_\zeta |b_\zeta|^2  \norm{
    \sum_{\substack{\bm{x}\in S_\zeta,\bm{y}\in S_\zeta\\ |\bm{x}-\bm{y}|> \ell}}\sum_{\substack{\alpha,\beta,\gamma,\delta \\ \alpha', \beta', \gamma', \delta'}} [G(|\bm{x}-\bm{y}|)]_{\alpha'\beta'\alpha\beta \gamma'\delta' \gamma \delta} \adag_{\alpha'\beta'}(\bm{x})\adag_{\alpha\beta}(\bm{y})a_{\gamma'\delta'}(\bm{x})a_{\gamma\delta}(\bm{y})  }_{\eta,\zeta} \label{Eq:LR_CutOff_Method_1}
    \\
    & \quad\leq 
    \max_{\zeta} \norm{ \sum_{\substack{\bm{x}\in S_{\zeta},\bm{y}\in S_{\zeta}\\ |\bm{x}-\bm{y}|= \ell+a_L}} \sum_{\substack{\alpha,\beta,\gamma,\delta \\ \alpha', \beta', \gamma', \delta'}} G(\ell+a_L)_{\alpha'\beta'\alpha\beta \gamma'\delta' \gamma \delta} \adag_{\alpha'\beta'}(\bm{x})\adag_{\alpha\beta}(\bm{y})a_{\gamma'\delta'}(\bm{x})a_{\gamma\delta}(\bm{y})}_{\eta,\zeta},
    \label{Eq:method-I-bound}
\end{align}
where in the last line, we have used the fact that $|G(r_1)|<|G(r_2)|$ for $r_1>r_2$.  Note that the maximization over $\zeta$ makes the norm independent of the fermionic occupation configuration $S_\zeta$, which is why $\max_{\{b_\zeta\}} \sum_\zeta |b_\zeta|^2$ is set to one in the second inequality. 

To bound \cref{Eq:method-I-bound},  we note that the norm is maximized with regard to $\zeta'$ when all spin-isospin interactions contribute a non-zero norm, hence providing an upper bound on the full operator norm.  Each term in $[\bm{\tau}\cdot\bm{\tau}][\bm{\sigma}\cdot\bm{\sigma}]$ generates 16 products of creation or annihilation operators, or 8 pairs forming Hermitian operators.
In total, there are 9 such terms in $[\bm{\tau}\cdot\bm{\tau}][\bm{\sigma}\cdot\bm{\sigma}]$, leading to $9\times 8=72$ Hermitian operators.
We must also consider the terms weighted by $[\bm{\tau}\cdot\bm{\tau}]S_{12}$ in $G(|\bm{x}-\bm{y}|)$, which by a similar argument generate Hermitian operators with total prefactors of $27 \times 8 \times 3=648$ from the $[\bm{\tau}\cdot\bm{\tau}] [\hat{\bm{x}}\cdot\bm{\sigma}][\hat{\bm{y}}\cdot\bm{\sigma}]$ part (including a prefactor of $3$ for the operator), and $72$ from the $[\bm{\tau}\cdot\bm{\tau}][\bm{\sigma}\cdot\bm{\sigma}]$ part, as can be seen from the definition of $S_{12}$ in \cref{Eq:S12-def}.

Putting everything together gives the first argument of the minimum function in \cref{eq:lem2}:
\begin{align}
    \norm{H_{\ell} - H_{\rm LR}  }_\eta 
    &
    \leq 
    \eta^2 \left[  \left( 72g_1(\ell+a_L) +  648g_2(\ell+a_L)\right) \right],
\end{align}
where
\begin{align}
     g_1(r) \coloneqq \frac{1}{12\pi}\left(\frac{g_A}{2f_\pi}\right)^2m_{\pi}^2 \frac{e^{-m_{\pi}r}}{r}, \quad \quad \quad  g_2(r)& \coloneqq g_1(r)\left(1 + \frac{3}{m_{\pi}r} + \frac{3}{m_{\pi}^2r^2} \right).
\end{align}
Note that in the first inequality, we have used the fact that, for any $S_\zeta$, there are at most $\eta^2$ non-zero terms in the semi-norm arising from the sum over $\bm{x}$ and $\bm{y}$ in \cref{Eq:method-I-bound} for each $S_\zeta$, thus applying the triangle inequality gives the first equality.

\paragraph{Method 2.}
Starting from \cref{Eq:Hl-minus-HLR-norm}, we now instead proceed as follows: 
\begin{align}
    &\norm{H_{\ell} - H_{\rm LR}}_\eta \nonumber \\
    &\quad\leq  \max_{\{b_\zeta\}}\sum_\zeta |b_\zeta|^2 
    \norm{\sum_{\bm{x}\in S_{\zeta}  } \sum_{
    \substack{\bm{y} \\ |\bm{x}-\bm{y}|>\ell}}\sum_{\substack{\alpha,\beta,\gamma,\delta \\ \alpha', \beta', \gamma', \delta'}} [G(|\bm{x}-\bm{y}|)]_{\alpha'\beta'\alpha\beta \gamma'\delta' \gamma \delta} \adag_{\alpha'\beta'}(\bm{x})\adag_{\alpha\beta}(\bm{y})a_{\gamma'\delta'}(\bm{x})a_{\gamma\delta}(\bm{y})   }_{\eta,\zeta} \\
     &\quad\leq \eta \,  
     \max_{\zeta}\norm{ \sum_{
    \substack{\bm{y} \\ |\bm{x}-\bm{y}|>\ell}} \sum_{\substack{\alpha,\beta,\gamma,\delta \\ \alpha', \beta', \gamma', \delta'}} [G(|\bm{x}-\bm{y}|)]_{\alpha'\beta'\alpha\beta \gamma'\delta' \gamma \delta} \adag_{\alpha'\beta'}(\bm{x})\adag_{\alpha\beta}(\bm{y})a_{\gamma'\delta'}(\bm{x})a_{\gamma\delta}(\bm{y})   }_{\eta,\zeta} \\
     &\quad\leq \eta \, 
     \sum_{
    \substack{\bm{y} \\ |\bm{x}-\bm{y}|>\ell}} \norm{  \sum_{\substack{\alpha,\beta,\gamma,\delta \\ \alpha', \beta', \gamma', \delta'}} [G(|\bm{x}-\bm{y}|)]_{\alpha'\beta'\alpha\beta \gamma'\delta' \gamma \delta} \adag_{\alpha'\beta'}(\bm{x})\adag_{\alpha\beta}(\bm{y})a_{\gamma'\delta'}(\bm{x})a_{\gamma\delta}(\bm{y})}, 
\end{align}
where, in the second line, $\bm{x}$ is an arbitrary fixed site (e.g.\ the origin) and where, going from the first line to the second line, we have used the fact that the norm maximized over $\zeta$ is independent of $S_{\zeta}$, and that $\sum_{\zeta}|b_\zeta|^2=1$. The norm in the last line, therefore, is independent of both $\eta$ and $\zeta$.
This final form can then be bounded by integrating over all sites beyond $\ell$:
\begin{align}
\label{eq:Hl-Hlr}
    \norm{H_{\ell} - H_{\rm LR}  }_\eta &\leq  \frac{4\pi\eta}{a_L^2}
   \int_{\ell+a_L}^{\infty} dr r^2\left[ 72g_1(r) + 648g_2(r)\right],
\end{align}
where the prefactors for each term in parentheses are obtained in the same way as in Method 1  to arrive at a $S_\zeta$-independent bound. 
While the integral in the left-hand side of \cref{eq:Hl-Hlr} could straightforwardly be performed numerically (and doing so would give slightly tighter bounds), to obtain a simpler closed-form bound, we upper bound it as follows:
\begin{align}
    \int_{\ell +a_L}^\infty r^2g_1(r) 
    \leq \frac{1}{12\pi}\left(\frac{g_A}{2f_\pi}\right)^2
    (m_\pi\ell + m_\pi a_L+1)e^{-m_\pi(\ell+a_L)}.
\end{align}
To bound the integral over $g_2(r)$, we first upper bound $g_2(r)$ as
\begin{align}
    g_2(r) \leq  \frac{1}{12\pi}\left(\frac{g_A}{2f_\pi}\right)^2 g_1(r)\left( 1 + 2\times \frac{3}{m_\pi r} \right),
\end{align}
which gives
\begin{align}
    \int_{\ell +a_L}^\infty r^2g_2(r) 
    \leq \frac{1}{12\pi }\left(\frac{g_A}{2f_\pi}\right)^2e^{-m_\pi(\ell+a_L)}\left[ (m_\pi\ell + m_\pi a_L+1) + 6  \right].
\end{align}
Putting everything gives the second argument of minimum function in \cref{eq:lem2}:
\begin{align}
    \norm{H_{\ell} - H_{\rm LR}  }_\eta &\leq \frac{\eta}{3a_L^3} \left(\frac{g_A}{2f_\pi}\right)^2 e^{-m_\pi(\ell +a_L)}\left[ 720(m_\pi\ell + m_\pi a_L+ 1) + 3880\right].
\end{align}
\end{proof}

\section{Bounding the Simulation with the Truncated Pion-Field Strength}\label{Sec:Pion_Cutoff}

We follow Ref.~\cite{Jordan_Lee_Preskill_2012} in this Appendix to bound the pion-field strengths in $H_{D\pi}$, as defined in \cref{Eq:H-dyn}, such that the Hamiltonian expectation value with respect to any state remains below a given energy cutoff. For completeness, we first summarize the result of Ref.~\cite{Jordan_Lee_Preskill_2012} before applying it to the dynamical-pion nuclear EFT.

Let $p_{\rm out}$ be the probability that one of the $3L^3$ $\pi_I(\bm{x})$ fields is not contained in the range $[-\pimax,\pimax]$. Let $\ket{\psi}$ be a state and $\ket{\psi_{\rm cut}}$ be the same state constrained to the space with $\pi_I (\bm{x})\in [-\pimax,\pimax]$.
Then, according to Section A.4 of Ref.~\cite{Jordan_Lee_Preskill_2012},
\begin{align}
    \braket{\psi|\psi_{\rm cut}} \geq 1 - 3L^3\max_{\bm{x}}p_{\rm out}(\bm{x}).
\end{align} 
Now considering $\pi_I$ as a distribution, the cutoff value can be expressed as
\begin{align}
\pi_{\rm max}=|\braket{\pi_I}|+k\sqrt{\braket{\pi_I^2}-\braket{\pi_I}^2}
\label{Eq:pimax-in-terms-of-k}
\end{align}
for a real positive $k$. 
Chebyshev's inequality gives
\begin{align}
   p_{\rm out} =  \text{Pr}\left( |\pi_I - \langle \pi_I \rangle| > k \sqrt{\braket{\pi_I^2}-\braket{\pi_I}^2} \right) \leq \frac{1}{k^2}.
\end{align}
Hence, to get $\bra{\psi}\psi_{\rm cut}\rangle \geq 1 - \epscut$,  one can set $k=\sqrt{3L^3/\epscut}$ in \cref{Eq:pimax-in-terms-of-k}. Note that \cref{Eq:pimax-in-terms-of-k} can be simplified as $\pi_{\rm max} \leq (k+1)\sqrt{\braket{\pi_I^2}}$. Therefore, a conservative value for $\pi_{\rm max}$ is
\begin{align}
    \pimax &= \left( \sqrt{\frac{3L^3}{\epscut}}+1\right)\sqrt{\langle \pi_I^2 \rangle},
\end{align}
where we have used the inequality $\braket{\pi_I} \leq \sqrt{\braket{\pi_I^2}}$ (Proposition 2 of Ref.~\cite{Jordan_Lee_Preskill_2012}). Using the same reasoning, a conservative value for $\Pi_{\rm max}$ is
\begin{align}
    \Pimax =  \left( \sqrt{\frac{3L^3}{\epscut}}+1\right)\sqrt{\langle \Pi_I^2 \rangle}.
\end{align}

Thus, we need to bound the expectation values of the squared operators, i.e., $|\langle \pi^2_I(\bm{x}) \rangle|$ and $|\langle \Pi^2_I(\bm{x}) \rangle|$, in the dynamical-pion EFT, which we now proceed to do. Recall that the dynamical-pion Hamiltonian is
\begin{align}
    H_{D\pi} = \Hfree+H_C+H_{C_{I^2}}&+ \frac{a_L^3}{2}\sum_{\bm{x}} \sum_{I} \left( \Pi^2_I(\bm{x}) + (\nabla \pi_I(\bm{x}))^2 + m_{\pi}^2\pi_I^2 (\bm{x})\right)\nonumber\\
    &+ \frac{g_A}{2f_\pi}\sum_{\bm{x}} \sum_{\alpha,\beta,\gamma,\delta}\sum_{I,S} \adag_{\alpha\beta}(\bm{x}) [\tau_I]_{\beta\delta}[\sigma_S]_{\alpha\gamma} \partial_S \pi_I(\bm{x}) a_{\gamma\delta}(\bm{x}) \nonumber\\
    &+ \frac{1}{4f^2_{\pi}}\sum_{\bm{x}}\sum_{I_1,I_2,I_3}\sum_{\alpha,\beta,\delta} \epsilon_{I_1I_2I_3} \pi_{I_2}(\bm{x})\Pi_{I_3}(\bm{x}) \adag_{\alpha\beta}(\bm{x}) [\tau_{I_1}]_{\beta \delta}a_{\alpha\delta}(\bm{x}).
\end{align}

\begin{lemma}\label{Lemma:Dynamical_Expectation_Values}
Let $\ket{\psi_\eta}$ be any state with $\eta$ fermions such that $\langle H \reta \coloneqq \langle \psi_\eta | H | \psi_\eta \rangle \leq E$.
Then,
\begin{align}
    |\langle \pi^2_I(\bm{x}) \rangle| \leq \left[ \frac{3g_A}{f_\pi a_L A} + \sqrt{\frac{E+8\eta|C|+4\eta|C_{I^2}|}{A}+3\eta 
    \left(\frac{3g_A}{f_\pi a_L A}\right)^2 + \frac{9\eta m_{\pi}^2a_L^3}{A}\left( \frac{6g_A}{m_\pi^2f_\pi a_L^4} \right)^2} \right]^2,
\end{align}
where 
\begin{align}
   A \coloneqq \frac{m_\pi^2a_L^3}{2}-\frac{1}{2f^2_{\pi}a_L},
   \label{Eq:A}
\end{align}
for lattice spacings $a_L$ such that $A>0$.
\end{lemma}
\begin{proof}
Noting that $\langle \Hfree \rangle\geq 0$ and $ \langle (\nabla \pi_I(\bm{x}))^2 \rangle \geq 0 $, we have
\begin{align}
     E \geq \langle H_C \rangle_{\eta}+\langle H_{C_{I^2}} \rangle_{\eta}+ & \frac{a_L^3}{2}\sum_{\bm{x}} \sum_{I} \left( \braket{\Pi^2_I(\bm{x})}_\eta + m_{\pi}^2\braket{\pi_I^2 (\bm{x})}_\eta\right) \nonumber\\
    &
    + \frac{g_A}{2f_\pi}\sum_{\bm{x}} \sum_{I,S} \langle F_{I,S}(\bm{x}) \frac{\pi_I(\bm{x}+\hat{\bm{n}}_S)-\pi_I(\bm{x})}{a_L} \rangle_{\eta}
    \nonumber\\
    &+ \frac{1}{4f^2_{\pi}}\sum_{\bm{x}}\sum_{I_1,I_2,I_3}\epsilon_{I_1I_2I_3}\langle G_{I_1}(\bm{x})  \pi_{I_2}(\bm{x})\Pi_{I_3}(\bm{x}) \rangle_{\eta},
\end{align}
where for brevity, we have defined 
\begin{align}
&F_{I,S}(\bm{x}) \coloneqq \sum_{\alpha,\beta,\gamma,\delta} \adag_{\alpha\beta}(\bm{x}) [\tau_I]_{\beta\delta}[\sigma_S]_{\alpha\gamma} a_{\gamma\delta}(\bm{x}),
\\
&G_I(\bm{x})\coloneqq \sum_{\alpha,\beta,\delta} \adag_{\alpha\beta}(\bm{x}) [\tau_I]_{\beta\delta}a_{\alpha\delta}(\bm{x}).
\end{align}

To proceed, we decompose the state in which the expectation value is computed into the fermion occupation basis. That is, let $\zeta\in \{0,1\}^{|\Lambda|}$, with $|\Lambda|$ being the number of points on the 3D lattice. Then let $\ket{\psi_\eta^\zeta}$ be a state with exactly $\eta$ fermions, with $\zeta_{\bm{x}}=1$ at lattice site $\bm{x}\in \Lambda$ that has at least one fermion present, and $\zeta_{\bm{x}}=0$ otherwise. Let
$\langle\cdots\rangle_{\eta,\zeta}$ denote an expectation value with respect to such a state.
Consider also $\ket{\psi_\eta} \coloneqq \sum_{\zeta}b_\zeta \ket{\psi_\eta^\zeta} $, which is a superposition of states with different distributions of non-zero $\zeta$ values at various sites but containing exactly $\eta$ fermions.
Note that $H_{D\pi}-\Hfree$ is block diagonal with respect to this decomposition.
In the following, we define $S_\zeta$ to be the set of all lattice sites with at least one fermion present and $S_\zeta'$ to be the set of $\bm{x}\in\Lambda$ that are either in $S_\zeta$ or a neighbor of a point in $S_\zeta$. We denote the complement of a set with an overline. With this decomposition, we first bound the fermionic contact interactions:
\begin{align}
    \langle H_C \reta+\langle H_{C_{I^2}}\reta &\geq \sum_{\zeta,\zeta'}b_\zeta b_{\zeta'}^* \delta_{\zeta,\zeta'} \bigg[\frac{C}{2}\sum_{\bm{x}\in S_{\zeta}} \langle \rho^2(\bm{x})\rangle_{\eta,\zeta}+\frac{C_{I^2}}{2}\sum_{I,\bm{x}\in S_{\zeta}}\langle \rho_I^2(\bm{x})\rangle_{\eta,\zeta}\bigg]
    \\
    &\geq  \sum_{\zeta}|b_\zeta|^2 \bigg[ -\frac{|C|}{2}
    \sum_{\bm{x}\in S_{\zeta}} |\langle \rho^2(\bm{x})\rangle_{\eta,\zeta}|-\frac{|C_{I^2}|}{2}\sum_{I,\bm{x}\in S_{\zeta}}|\langle \rho_I^2(\bm{x})\rangle_{\eta,\zeta}|\bigg]
    \\\
    &\geq  \sum_{\zeta}|b_\zeta|^2 \bigg[ -4\eta\left(2|C|+|C_{I^2}|\right)\bigg],
\end{align}
where, in the second inequality, we have used the fact that $|\langle \rho^2(\bm{x})\rangle_{\eta,\zeta}| \leq 16$ and $ \sum_I|\langle \rho_I^2(\bm{x}) \rangle_{\eta,\zeta}|\leq 8$.

Next, inspecting the axial-vector term, we find
\begin{align}
  \sum_{\bm{x}} \sum_{I,S} \langle F_{I,S}(\bm{x}) &\frac{\pi_I(\bm{x}+a_L\hat{\bm{n}}_S)-\pi_I(\bm{x})}{a_L}\reta=\sum_{\zeta}|b_\zeta|^2\sum_{\bm{x}\in S_\zeta} \sum_{I,S} \langle F_{I,S}(\bm{x}) \frac{\pi_I(\bm{x}+a_L\hat{\bm{n}}_S)-\pi_I(\bm{x})}{a_L}\rangle_{\eta,\zeta} 
  \\
  &\geq -a_L^{-1} \sum_{\zeta}|b_\zeta|^2 \sum_{I,S} \bigg[\sum_{\bm{x}\in S_\zeta'}|\langle F_{I,S}(\bm{x-a_L\hat{\bm{n}}_S})\pi_I(\bm{x})\rangle_{\eta,\zeta}|+\sum_{\bm{x}\in S_\zeta}|\langle F_{I,S}(\bm{x})\pi_I(\bm{x})\rangle_{\eta,\zeta}|\bigg]
  \\
  &\geq - 2 a_L^{-1}\sum_{\zeta}|b_\zeta|^2\sum_{I,S} \bigg[\sum_{\bm{x}\in S_{\zeta}'} |\langle \pi_I(\bm{x})\rangle_\eta|+\sum_{\bm{x}\in S_\zeta} |\langle \pi_I(\bm{x})\rangle_\eta|\bigg]
  \\
  &\geq - 12a_L^{-1}\sum_{\zeta}|b_\zeta|^2\sum_{I,\bm{x}\in S_\zeta'} 
    \sqrt{\langle \pi_I^2(\bm{x})\rangle_\eta},
\end{align}
where, in the second inequality, we have used the fact that $|\langle F_{I,S}(\bm{x})\rangle_{\eta,\zeta}|\leq 2$.
The third inequality is obtained by noting that $\sum_{\bm{x}\in S_\zeta} |\langle \pi_I(\bm{x})\rangle_\eta| \leq \sum_{\bm{x}\in S_\zeta'} |\langle \pi_I(\bm{x})\rangle_\eta|$ and that $ \langle \pi(\bm{x}) \rangle_\eta^2 \leq \langle \pi^2(\bm{x}) \rangle_\eta$. Furthermore, the sum over $S$ returns a factor of $3$.

Finally, inspecting the Weinberg-Tomozawa term, we have
\begin{align}
\sum_{\bm{x}}\sum_{I_1,I_2,I_3}\epsilon_{I_1I_2I_3}\langle G_{I_1}(\bm{x})  \pi_{I_2}(\bm{x})\Pi_{I_3}(\bm{x}) \rangle_{\eta}
    &\geq -2 \sum_{\zeta}|b_\zeta|^2 \sum_{\bm{x}\in S_\zeta}\sum_{I_1,I_2,I_3}
    |\epsilon_{I_1I_2I_3}|\times|\langle \Pi_{I_1}(\bm{x}) \rangle_\eta| \times |\langle \pi_{I_2}(\bm{x}) \rangle_\eta|\\
    &\geq 
    -\sum_{\zeta}|b_\zeta|^2\sum_{\bm{x}\in S_\zeta}\sum_{I_1,I_2,I_3}
    |\epsilon_{I_1I_2I_3}|\left(a_L^{-1}\langle \pi_{I_2}(\bm{x}) \rangle_\eta^2  +a_L\langle \Pi_{I_1}(\bm{x}) \rangle_\eta^2 \right)
    \label{Eq:WT-bounding-step-I}\\
    &\geq -2\sum_\zeta |b_\zeta|^2 \sum_{I,\bm{x}\in S_\zeta} \left( 
    a_L\langle \Pi^2_I(\bm{x}) \rangle_\eta +  
     a_L^{-1}\langle \pi^2_I(\bm{x}) \rangle_\eta \right),
    \label{Eq:WT-bounding-step-II}
\end{align}
where the first inequality is obtained by noting that $|\langle G_I(\bm{x})\rangle_{\eta,\zeta}| \geq 2$ and the second inequality follows by observing that
\begin{align}
    \left(\frac{1}{\sqrt{a_L}}|\langle \pi_{I_2}(\bm{x}) \rangle_\eta| - \sqrt{a_L}|\langle \Pi_{I_1}(\bm{x}) \rangle_\eta|\right)^2 &\geq 0,
\end{align}
which gives
\begin{align}
    -2 |\langle \pi_{I_2}(\bm{x}) \rangle_\eta|\times|\langle \Pi_{I_1}(\bm{x}) \rangle_\eta |  \geq -\left(a_L^{-1}\langle \pi_{I_2}(\bm{x}) \rangle_\eta^2  +a_L\langle \Pi_{I_1}(\bm{x})  \rangle_\eta^2 \right).
\end{align}
Finally, the third inequality is obtained by noting that, first, $\langle \pi_{I}(\bm{x}) \rangle_\eta^2 \leq \langle \pi_{I}^2(\bm{x}) \rangle_\eta$ and $\langle \Pi_{I}(\bm{x}) \rangle_\eta^2 \leq \langle \Pi_{I}^2(\bm{x}) \rangle_\eta$, and, second, in the sum over $I_1,I_2,I_3$ in \cref{Eq:WT-bounding-step-I}, each $\pi_{I_2}$ or $\Pi_{I_1}$ only appears in two terms. 

Now putting everything together, and using the complete-the-square method to deal with the axial-vector term, we have
\begin{align}
    E+8\eta|C|+4\eta|C_{I^2}| \geq \sum_\zeta &|b_\zeta|^2 \bigg[\frac{a_L^3}{2}\sum_{I,\bm{x}\in \overline{S}_\zeta} \left( \langle\Pi^2_I(\bm{x}) \rangle_\eta  + m_{\pi}^2\langle \pi_I^2 (\bm{x})\rangle_\eta \right) 
     \nonumber  \\
    & + \sum_{I,\bm{x}\in S_\zeta} \left( B\langle \Pi^2_I(\bm{x}) \rangle_\eta +  A\langle \pi^2_I(\bm{x}) \rangle_\eta\right)- \frac{6g_A}{f_\pi a_L}\sum_{\bm{x}\in S_{\zeta}'} 
    \sum_{I,S}  
    \sqrt{\langle\pi^2_I(\bm{x})\rangle_\eta} \bigg],
    \label{Eq:Full_Expectation_Value}
\end{align}
where $A$ and $B$ are defined in \cref{Eq:A,Eq:B}, respectively. We assume that $A,B>0$ (by appropriately choosing the value of lattice spacing $a_L$). Now noting that $\langle\Pi^2_I(\bm{x}) \rangle_\eta\geq 0$, we have
\begin{align}
   E+8\eta|C|+4\eta|C_{I^2}| \geq \sum_\zeta |b_\zeta|^2 \bigg[ &\frac{a_L^3m_{\pi}^2}{2}\sum_{I,\bm{x}\in \overline{S}_\zeta^{\prime}} \langle \pi_I^2 (\bm{x})\rangle 
     +  A\sum_{I,\bm{x}\in S_\zeta} \langle \pi^2_I(\bm{x}) \rangle - \frac{6g_A}{f_\pi a_L}\sum_{\bm{x}\in S_{\zeta}} \sum_{I,S}  \sqrt{\langle\pi^2_I(\bm{x})\rangle}  \nonumber\\
    &+\frac{a_L^3m_{\pi}^2}{2}\sum_{I,\bm{x}\in (S_{\zeta}'-S_\zeta)}\langle\piI^2(\bm{x})\rangle 
    - \frac{6g_A}{f_\pi a_L}\sum_{\bm{x}\in S_{\zeta}'-S_\zeta} \sum_{I,S}  \sqrt{\langle\pi^2_I(\bm{x})\rangle} \bigg].
\end{align}
Recognizing that $ \langle \pi_I^2 (\bm{x})\rangle$ terms for $\bm{x}\in \overline{S}_\zeta^{\prime}$ can be removed from the expression due to their non-negativity, and then completing the square, gives
\begin{align}
    E+8\eta|C|+4\eta|C_{I^2}| \geq \sum_\zeta |b_\zeta|^2  &  \left\{ A\sum_{I,\bm{x}\in S_\zeta} \left[\left( \sqrt{\langle \pi^2_I(\bm{x}) \rangle} - \frac{3g_A}{f_\pi a_L A}  \right)^2 - \left(\frac{3g_A}{f_\pi a_L A}\right)^2\right] \right .  \nonumber\\
    & \left .+\frac{a_L^3m_{\pi}^2}{2}\sum_{I,\bm{x}\in (S_{\zeta}'-S_\zeta)} \left[\left( \sqrt{\langle\piI^2(\bm{x})\rangle} 
    - \frac{6g_A}{m_\pi^2f_\pi a_L^4}\right)^2 - \left( \frac{6g_A}{m_\pi^2f_\pi a_L^4} \right)^2\right] \right\}.
\end{align}
Note that the expression in the curly bracket is now independent of the fermionic occupation configuration. Then, since $\sum_\zeta|b_\zeta^2|=1$ (for a properly normalized state), the bound must apply to every term in the sum over $\zeta$ as well. Therefore, we can consider a single $\zeta\in \{0,1\}^{|\Lambda|}$ to proceed. 
Since $A<\frac{m_{\pi}^2a_L^3}{2}$, we have
\begin{align}
    E+8\eta|C|+4\eta|C_{I^2}|
    \geq & A\left( \sqrt{\langle \pi^2_I(\bm{x}) \rangle} - \frac{3g_A}{f_\pi a_L A}  \right)^2 - 3\eta A\left(\frac{3g_A}{f_\pi a_L A}\right)^2 - \frac{m_{\pi}^2a_L^3}{2}\times 18\eta \left( \frac{6g_A}{m_\pi^2f_\pi a_L^4} \right)^2,
\end{align}
where we have used the bound on the number of fermions, i.e., $|S_\zeta| \leq \eta$ and $|S_{\zeta}'-S_\zeta| \leq 6\eta$. The statement of the lemma then follows.
\end{proof}

Finally, we can use a similar technique for the conjugate momentum.
\begin{lemma}\label{Lemma:Dynamical_Momentum_Cutoff}
Let $\ket{\psi_\eta}$ be any state with $\eta$ fermions such that $\langle H \reta \coloneqq  \langle \psi_\eta | H | \psi_\eta \rangle \leq E$.
Then,
\begin{align}
    \langle \Pi^2_I(\bm{x}) \rangle \leq \frac{E+8\eta|C|+4\eta|C_{I^2}|}{B}+ \frac{3 \eta}{AB}  \left(\frac{3g_A}{f_\pi a_L}\right)^2 + \frac{9\eta m_{\pi}^2a_L^3}{B} \left( \frac{6g_A}{m_\pi^2f_\pi a_L^4} \right)^2.
\end{align}
where 
\begin{align}
   B\coloneqq \frac{a_L^3}{2}-\frac{a_L}{2f^2_{\pi}}, \quad \quad 
   \label{Eq:B}
\end{align}
and $A$ is defined in \cref{Eq:A}, for lattice spacings $a_L$ such that $A,B>0$.
\end{lemma}

\begin{proof}
Using \cref{Eq:Full_Expectation_Value}, we complete the square for the $\pi_I$-field terms in the same way to get
\begin{align}
    E+8\eta|C|+4\eta|C_{I^2}| \geq &\frac{a_L^3}{2}\sum_{I,\bm{x}\in \overline{S}_\zeta} 
    \langle\Pi^2_I(\bm{x}) \rangle 
     + B\sum_{I,\bm{x}\in S_\zeta}  
    \langle \Pi^2_I(\bm{x}) \rangle  
     \nonumber\\  
     & +\frac{a_L^3m_{\pi}^2}{2}\sum_{I,\bm{x}\in \overline{S}_\zeta^{\prime}} \langle \pi_I^2 (\bm{x})\rangle 
     + A\sum_{I,\bm{x}\in S_\zeta} \left[\left( \sqrt{\langle \pi^2_I(\bm{x}) \rangle} - \frac{3g_A}{f_\pi a_L A}  \right)^2 - \left(\frac{3g_A}{f_\pi a_L A}\right)^2 \right]
     \nonumber\\
     & +\frac{a_L^3m_{\pi}^2}{2}\sum_{I,\bm{x}\in (S_{\zeta}'-S_\zeta)} \left[\left( \sqrt{\langle\piI^2(\bm{x})\rangle} 
    - \frac{6g_A}{m_\pi^2f_\pi a_L^4}\right)^2 - \left( \frac{6g_A}{m_\pi^2f_\pi a_L^4} \right)^2\right].
\end{align}
Removing many of the terms that can only be non-negative,
we find
\begin{align}
    E+8\eta|C|+4\eta|C_{I^2}| \geq  
    B\sum_{I,\bm{x}\in S_\zeta} \langle \Pi^2_I(\bm{x}) \rangle
     - 3A 
     \eta\left(\frac{3g_A}{f_\pi a_L A}\right)^2  
    -\frac{a_L^3m_{\pi}^2}{2}\times 18
    \eta\left( \frac{6g_A}{m_\pi^2f_\pi a_L^4} \right)^2,
\end{align}
which gives the statement of the lemma.
\end{proof}

\begin{corollary} \label{Corollary:Bosonic_Cutoffs}
To achieve $\bra{\psi}\psi_{\rm cut}\rangle\geq 1 -\epscut$ with $3L^3$ bosonic degrees of freedom, it is sufficient to choose 
\begin{align}
    \pimax &=\left(  \sqrt{\frac{3L^3}{\epscut}}+1\right)\left[ \frac{6g_A}{4f_\pi a_L A} + \sqrt{\frac{E+8\eta|C|+4\eta|C_{I^2}|}{A}+3\eta 
    \left(\frac{3g_A}{f_\pi a_L A}\right)^2 + \frac{9\eta m_{\pi}^2a_L^3}{A}\left( \frac{6g_A}{m_\pi^2f_\pi a_L^4} \right)^2} \right] , \\
    \Pimax &=\left(  \sqrt{\frac{3L^3}{\epscut}}+1\right)
    \sqrt{ \frac{E+8\eta|C|+4\eta|C_{I^2}|}{B}+ \frac{3 \eta}{AB}  \left(\frac{3g_A}{f_\pi a_L}\right)^2 + \frac{9\eta m_{\pi}^2a_L^3}{B} \left( \frac{6g_A}{m_\pi^2f_\pi a_L^4} \right)^2 },
\end{align}
where $A$ and $B$ are defined in \cref{Eq:A,Eq:B}, respectively.
\end{corollary}

This proves \cref{Lem:Dyn-Pions} of the main text.

\section{Summary Tables of Simulation Costs}

The costs of simulating one time step of Trotter evolution in various nuclear EFTs are detailed in a number of Lemmas in \cref{sec:circuits} along with their derivations. 
In this Appendix, we summarize all those simulation costs. These include the pionless-EFT circuit depth with VC and compact encodings in \cref{Tab:Depth-Pionless-VC,Tab:Depth-Pionless-Compact}, respectively; the pionless-EFT $R_z$-gate count for both encodings in \cref{Tab:RZ-Pionless}; the OPE-EFT circuit depth and $R_z$-gate count in \cref{Tab:Depth-OPE,Tab:RZ-OPE}, respectively; and the dynamical-pion EFT circuit depth and $R_z$-gate count in \cref{Tab:Depth-Dpi,Tab:RZ-Dpi}, respectively.

\begin{table}
\begin{center}
\setlength{\tabcolsep}{3pt}
\resizebox{\textwidth}{!}{\begin{tabular}{c|c|c|c|c|c|c|c}
\hline 
\multicolumn{8}{c}{\textbf{Pionless-EFT Circuit Depths (VC Encoding)}}\tabularnewline
\hline 
\textbf{Term(s)} & $e^{-it\tilde{h}_{\sigma}^{x}(i,j)}$ & $e^{-it\tilde{h}_{\sigma}^{y}(i,j)}$ & $e^{-it\tilde{h}_{\sigma}^{z}(i,j)}$ & $e^{-it\big(H_{\Cpi}(i)+H_{\Dpi}(i)\big)}$ & $e^{-itH_{\mathrm{free}}}$ & $e^{-it(H_{\Cpi}+H_{\Dpi})}$ & $\mathcal{P}_{1}^{(\canpi)}(t)$\tabularnewline
\hline 
\textbf{Uncontrolled} & 16 & 22 & 26 & 8 & 512 & 8 & 520\tabularnewline
\hline 
\textbf{Controlled} & 20 & 26 & 30 & 22 & 608 & 22 & 630\tabularnewline
\hline 
\end{tabular}}
\end{center}
\caption{The contributions to the 2-qubit circuit depth for simulating the pionless-EFT Hamiltonian and its controlled version with the VC encoding, according to \cref{Lemma:Pionless_Hopping_Term_VC,Lemma:Pionless_Contact_Term_VC,Lemma:Total_Pionless_EFT_Depth_VC}.
\label{Tab:Depth-Pionless-VC}}
\end{table}
\begin{table}
\begin{center}
\setlength{\tabcolsep}{3pt}
\resizebox{\textwidth}{!}{\begin{tabular}{c|c|c|c|c|c|c|c}
\hline 
\multicolumn{8}{c}{\textbf{Pionless-EFT Circuit Depths (Compact Encoding)}}\tabularnewline
\hline 
\textbf{Term(s)} & $e^{-it\tilde{h}_{\sigma}^{x}(i,j)}$ & $e^{-it\tilde{h}_{\sigma}^{y}(i,j)}$ & $e^{-it\tilde{h}_{\sigma}^{z}(i,j)}$ & $e^{-it\big(H_{\Cpi}(i)+H_{\Dpi}(i)\big)}$ & $e^{-itH_{\mathrm{free}}}$ & $e^{-it(H_{\Cpi}+H_{\Dpi})}$ & $\mathcal{P}_{1}^{(\canpi)}(t)$\tabularnewline
\hline 
\textbf{Uncontrolled} & 10 & 10 & 10 & 8 & 60 & 8 & 68\tabularnewline
\hline 
\textbf{Controlled} & 14 & 14 & 14 & 22 & 84 & 22 & 106\tabularnewline
\hline 
\end{tabular}}
\end{center}
\caption{The contributions to the 2-qubit circuit depth for simulating the pionless-EFT Hamiltonian and its controlled version with the compact encoding, according to \cref{Lemma:Pionless_Hopping_Term_Compact,Lemma:Pionless_Contact_Term_Compact,Lemma:Total_Pionless_EFT_Depth_Compact}.
\label{Tab:Depth-Pionless-Compact}}
\end{table}
\begin{table}
\begin{center}
\setlength{\tabcolsep}{3pt}
\resizebox{\textwidth}{!}{\begin{tabular}{c|c|c|c}
\hline 
\multicolumn{4}{c}{\textbf{Pionless-EFT $R_z$-Gate Count}}\tabularnewline
\hline 
\textbf{Term(s)} & $e^{-itH_{\mathrm{free}}}$ & $e^{-it(H_{\Cpi}+H_{\Dpi})}$ & $\mathcal{P}_{1}^{(\canpi)}(t)$\tabularnewline
\hline 
\textbf{Uncontrolled} & $28\,L^{3}$ & $14\,L^{3}$ & $42\,L^{3}$\tabularnewline
\hline 
\textbf{Controlled} & $56\,L^{3}$ & $28\,L^{3}$ & $84\,L^{3}$\tabularnewline
\hline 
\end{tabular}}
\end{center}
\caption{The number of $R_z$ gates used to simulate the pionless-EFT Hamiltonian and its controlled version with both the VC and compact encodings, according to \cref{Lemma:Pionless_T-Gate_Count}. $L$ denotes the number of sites along each Cartesian direction on the 3D lattice.
\label{Tab:RZ-Pionless}}
\end{table}
\begin{table}
\begin{center}
\setlength{\tabcolsep}{3pt}
\resizebox{\textwidth}{!}{\begin{tabular}{c|c|c|c|c|c|c|c|c}
\hline 
\multicolumn{9}{c}{\textbf{One-Pion-Exchange EFT Circuit Depth}}\tabularnewline
\hline 
\textbf{Term(s)} & $e^{-it\widetilde{h}_{\sigma}(i,j)}$ & $e^{-itH_{C}(i)}$ & $e^{-itH_{C_{I^{2}}}(i)}$ & $e^{-itH_{\mathrm{LR}}(i,j)}$ & $e^{-itH_{\mathrm{free}}}$ & $e^{-it(H_{C}+H_{C_{I^{2}}})}$ & $e^{-itH_{\mathrm{LR}}}$ & $\mathcal{P}_{1}^{({\rm OPE})}(t)$ \tabularnewline
\hline 
\textbf{Uncontrolled} & 64 & 6 & 54 & 14,336 & 512 & 60 & $14,336R_{\ell}$ & $572+14,336R_{\ell}$\tabularnewline
\hline 
\textbf{Controlled} & 76 & 26 & 98 & 16,384 & 608 & 124 & $16,384R_{\ell}$ & $732+16,384R_{\ell}$\tabularnewline
\hline 
\end{tabular}}
\end{center}
\caption{The contributions to the 2-qubit circuit depth for simulating the OPE-EFT Hamiltonian and its controlled version, according to \cref{Lemma:Contact-Depth-OPE,Lemma:Long-Range-Depth-OPE,Lemma:Total-Depth-OPE}. $R_\ell$ is defined in \cref{Lemma:Long-Range-Depth-OPE}.
\label{Tab:Depth-OPE}}
\end{table}
\begin{table}
\begin{center}
\setlength{\tabcolsep}{3pt}
\resizebox{\textwidth}{!}{\begin{tabular}{c|c|c|c|c|c}
\hline 
\multicolumn{6}{c}{\textbf{One-Pion-Exchange EFT $R_z$-Gate Count}}\tabularnewline
\hline 
\textbf{Term(s)} & $e^{-itH_{\mathrm{free}}}$ & $e^{-itH_{C}}$ & $e^{-itH_{C_{I^{2}}}}$ & $e^{-itH_{\mathrm{LR}}}$ & \textbf{$\mathcal{P}_{1}^{({\rm OPE})}(t)$}\tabularnewline
\hline 
\textbf{Uncontrolled} & $28\,L^{3}$ & $10\,L^{3}$ & $18\,L^{3}$ & 1,024$\,R_{\ell}L^{3}$ & $\big(52+1,024\,R_{\ell}\big)L^{3}$\tabularnewline
\hline 
\textbf{Controlled} & $56\,L^{3}$ & $20\,L^{3}$ & $36\,L^{3}$ & 2,048$\,R_{\ell}L^{3}$ & $\big(104+2,048\,R_{\ell}\big)L^{3}$\tabularnewline
\hline 
\end{tabular}}
\end{center}
\caption{The number of $R_z$ gates used to simulate the OPE-EFT Hamiltonian and its controlled version, according to \cref{Lemma:T-gate-OPE}. $R_\ell$ is defined in \cref{Lemma:Long-Range-Depth-OPE}. $L$ denotes the number of sites along each Cartesian direction on the 3D lattice.
\label{Tab:RZ-OPE}}
\end{table}
\begin{table}
\begin{center}
\setlength{\tabcolsep}{1.5pt}
\resizebox{\textwidth}{!}{\begin{tabular}{c|c|c|c|c|c|c}
\hline 
\multicolumn{7}{c}{\textbf{Dynamical-Pion EFT Circuit Depth}}\tabularnewline
\hline 
\textbf{Term(s)} & $e^{-itH_{\pi^{2}}}$ & $e^{-itH_{(\nabla\pi)^{2}}}$ & $e^{-itH_{\Pi^{2}}}$ & $e^{-itH_{\mathrm{AV}}}$ & $e^{-itH_{\mathrm{WT}}}$ & \textbf{$\mathcal{P}_{1}^{(D\pi)}(t)$}\tabularnewline
\hline 
\multirow{2}{*}{\textbf{Uncontrolled}} & $2\left\lceil \frac{n_{b}}{2}\right\rceil +$ & $12\left\lceil \frac{n_{b}}{2}\right\rceil +$ & $2n_{b}^{2}+$ & $1296+$ & $98n_{b}^{2}+$ & $\max\left\{572,2n_b^2+16 \left\lceil \frac{n_b}{2} \right\rceil+26n_b-32\right\}$ \tabularnewline
 & $2n_{b}-4$ & $24n_{b}-24$ & $2\left\lceil \frac{n_{b}}{2}\right\rceil -4$ & $864n_{b}$ & $94n_{b}+96$ & $+98n_b^2+958n_b+1392$ 
 \tabularnewline
\hline 
\multirow{2}{*}{\textbf{Controlled}} & $n_{b}^{2}+2\left\lceil \frac{n_{b}}{2}\right\rceil $ & $24n_{b}^{2}+12\left\lceil \frac{n_{b}}{2}\right\rceil $ & $3n_{b}^{2}+2\left\lceil \frac{n_{b}}{2}\right\rceil $ & $1296+$ & $146n_{b}^{2}+$ & $\max\{732,28n_b^2+16 \left\lceil \frac{n_b}{2} \right\rceil+40n_b-32\}$ \tabularnewline
 & $+3n_{b}-4$ & $+36n_{b}-24$ & $+n_{b}-4$ & $1728n_{b}$ & $190n_{b}+144$ & $+146n_b^2+1918n_b+1440$ \tabularnewline
\hline 
\end{tabular}}
\end{center}
\caption{The contributions to the 2-qubit circuit depth for simulating the dynamical-pion EFT Hamiltonian and its controlled version, according to 
\cref{Lemma:piI-Squared-Depth,Lemma:Del-piI-Squared-Depth,Lemma:Conjugate-PiI-Squared-Depth,Lemma:HAV_Circuit_Depth,Lemma:HWT_Circuit_Depth,Lemma:Total-Depth-Dyn-Pions}. 
Here, $n_b$ denotes the number of qubits holding the value of each $\pi_I(\bm{x})$. Entries corresponding to $e^{-it\Hfree}$, $e^{-itH_C}$, and $e^{-itH_{C_{I^2}}}$ are the same as in \cref{Tab:Depth-OPE} and are left out, but their contributions are accounted for in the total count in the last column.
\label{Tab:Depth-Dpi}}
\end{table}
\begin{table}
\begin{center}
\setlength{\tabcolsep}{1.2pt}
\resizebox{\textwidth}{!}{\begin{tabular}{c|c|c|c|c|c|c}
\hline 
\multicolumn{7}{c}{\textbf{Dynamical-Pion EFT $R_z$-Gate Count}}\tabularnewline
\hline 
\textbf{Term(s)} & $e^{-itH_{\pi^{2}}}$ & $e^{-itH_{(\nabla\pi)^{2}}}$ & $e^{-itH_{\Pi^{2}}}$ & $e^{-itH_{\mathrm{AV}}}$ & $e^{-itH_{\mathrm{WT}}}$ & \textbf{$\mathcal{P}_{1}^{(D\pi)}(t)$}\tabularnewline
\hline 
\textbf{Uncontrolled} & $\frac{3}{2}(n_{b}^{2}+n_{b})L^{3}$ & $3(2n_{b}^{2}+n_{b})L^{3}$ & $\frac{3}{2}(5n_{b}^{2}-3n_{b})L^{3}$ & $72n_{b}L^{3}$ & $6(3n_{b}^{2}+3n_{b}+2)L^{3}$ & $(33n_{b}^{2}+90n_{b}+64)L^{3}$\tabularnewline
\hline 
\textbf{Controlled} & $3(n_{b}^{2}+n_{b})L^{3}$ & $6(2n_{b}^{2}+n_{b})L^{3}$ & $3(5n_{b}^{2}-3n_{b})L^{3}$ & $144n_{b}L^{3}$ & $12(3n_{b}^{2}+3n_{b}+2)L^{3}$ & $2(33n_{b}^{2}+90n_{b}+64)L^{3}$\tabularnewline
\hline 
\end{tabular}}
\end{center}
\caption{The number of $R_z$ gates used to simulate the dynamical-pion EFT Hamiltonian and its controlled version, according to \cref{Lemma:Total-T-Gate-Dyn-Pions}. Here, $n_b$ denotes the number of qubits holding the value of each $\pi_I(\bm{x})$, and $L$ denotes the number of sites along each Cartesian direction on the 3D lattice. Entries corresponding to $e^{-it\Hfree}$, $e^{-itH_C}$, and $e^{-itH_{C_{I^2}}}$ are the same as in \cref{Tab:RZ-OPE} and are left out, but their contributions are accounted for in the total count in the last column.
\label{Tab:RZ-Dpi}}
\end{table}

\section{Higher-Order Trotter Error Bounds for Translation-Invariant Fermionic Hamiltonians} \label{Sec:Higher_Order_Trotter_Error}

In this Appendix, we present Trotter error bounds for a general class of fermionic Hamiltonians, which includes the nuclear-EFT Hamiltonians considered in this work.
Then in \cref{Sec:Analytic_Trotter_Bounds_Proof}, we present bounds for specific EFT Hamiltonians by computing the prefactors explicitly, which are typically much better as they exploit the structure of the Hamiltonians, rather than resorting to general assumptions about their form, as is done in this Appendix.

The individual terms that make up the Hamiltonian, namely the number-preserving fermionic operators (NPFO), are introduced in \cref{Def:NPFO_Def} of the main text. In the following theorem, we bound the semi-norm of such operators. The indices $i_1, i_2,\dots$ should be thought of as fermionic modes on a lattice, $\vec{i}$ is a subset of fermionic modes on a lattice, and $\Omega$ denotes sets of subsets of fermionic modes.

\begin{theorem} \label{Theorem:NPFO_Norm}
Consider a set of fermionic modes, $M$.
Let $\vec{i}=(i_1,i_2,\dots, i_{ k_{\vec i} })$ denote a tuple of $k_{\vec i}$ indices 
for some constant $k_{\vec i}$, and let $\Omega = \{\vec{i}_1,\vec{i}_2\dots \}$ be a set of such tuples such that no tuple shares indices with any other tuple: $\forall\, \vec{i}_a,\vec{i}_b\in \Omega$ with $a \neq b$, $\vec{i}_a \cap \vec{i}_b = \emptyset$.
Define the fermionic operator
\begin{align}
    X_{\Omega} = \sum_{\vec{i}\in \Omega} J_{\vec{i}} h_{\vec{i}},
\end{align}
such that each $h_{\vec{i}}$ is a NPFO acting on the fermionic modes in $\vec{i}\subset M$. Then, the fermionic semi-norm can be bounded as
\begin{align}
    \norm{X_\Omega}_{\eta} \leq J_{\rm max} \min \left\{ \left\lceil \frac{\eta}{ \lceil k_{\min}/2 \rceil} \right\rceil, |\Omega| \right\},
\end{align}
where $k_\mathrm{min}$ is the minimum locality of $h_{\vec{i}}$ and $J_\mathrm{max} = \max_{\vec{i}\in \Omega}\{|J_{\vec{i}}|\} $.
\end{theorem}

\begin{proof}
Without loss of generality, consider the case where $J_{\vec{i}}=1$.
Note that $\lambda(h_{\vec{i}}) \in \{0, 1\}$ for all $\vec{i}$, where $\lambda(h_{\vec{i}})$ denotes the eigenvalue set of $h_{\vec{i}}$.
Since $X_\Omega$, $h_{\vec{i}}$, and $
N\coloneqq  \sum_{j\in M} N(j)$ commute, they can be simultaneously diagonalized. 

For a contradiction, suppose there exists a normalized state $\ket{\psi}$ such that $N\ket{\psi} = \eta \ket{\psi}$ and that $X_{\Omega}\ket{\psi} = \lambda \ket{\psi}$ where $|\lambda|> \min\{\left\lceil\eta/\lceil k_{\mathrm{min}}/2 \rceil\right\rceil, |\Omega| \}$. 
 Since $X_\Omega$, $h_{\vec{i}}$, and $
N$ are mutually commuting for all $\vec{i}$, and  $X_\Omega= \sum_{\vec{i}\in \Omega} h_{\vec{i}}$ where $\{h_{\vec{i}}\}_{\vec{i}\in \Omega}$ do not act on any of the same modes, then we can choose to work with an eigenstate of all $\{h_{\vec{i}}\}_{\vec{i}\in \Omega}$ simultaneously (note that since any state can be written as a superposition of eigenstates and,  
by convexity, the maximum value of this superposition is always achieved for a single eigenstate, without loss of generality, we can consider an eigenstate). 
Then $h_{\vec{i}}\ket{\psi} = \ket{\psi}$ for at least $|\lambda|$ such terms $h_{\vec{i}}$. For any given $k_{\vec i}$-local NPFO $h_{\vec{i}}$ to be non-zero on state $\ket{\psi}$, there must be at least  $\lceil  k_{\vec i} /2 \rceil$ fermions on the subset of indices $\vec{i}$.
Hence for at least $|\lambda|$ tuples $\vec{i}\in \Omega$, we have
\begin{align}
    \bra{\psi}\sum_{j\in \vec{i}} N(j)\ket{\psi} \geq \left \lceil \frac{ k_{\vec i} }{2} \right \rceil,
    \end{align}
    where the exact value depends on the form of the NPFO (in particular, the number of hopping versus number operators present).
    Since all tuples in $\Omega$ are disjoint, we have
\begin{align}
    \bra{\psi} N\ket{\psi} &> |\lambda| \left\lceil  \frac{k_{\min}}{2}\right\rceil. 
\end{align}
Hence, using our assumption, we have
\begin{align}
    \bra{\psi} N\ket{\psi} &> \min \left\{ \left\lceil\frac{ \eta}{\lceil k_{\min}/2 \rceil}\right\rceil, |\Omega| \right\} \left\lceil  \frac{k_{\min}}{2}\right\rceil.
\end{align}
Now if $\min\left\{ \left\lceil\frac{ \eta}{\lceil k_{\min}/2 \rceil}\right\rceil,  |\Omega| \right\} = \left\lceil\frac{ \eta}{\lceil k_{\min}/2 \rceil}\right\rceil $, then this implies $\bra{\psi}N\ket{\psi}> \eta$ which is a contradiction. 
On the other hand, if $\min\left\{ \left\lceil\frac{ \eta}{\left\lceil k_{\min}/2\right\rceil}\right\rceil,  |\Omega| \right\} = |\Omega|$, then $|\lambda|> |\Omega|$, which is trivially a contradiction as there are only $|\Omega|$ terms in the sum for $
X_{\Omega}$.
Since $\|X_{\Omega}\|_\eta=|\lambda|$, this proves that $\|X_{\Omega}\|_\eta \leq \min\{ \left\lceil\frac{ \eta}{\left\lceil k_{\min}/2\right\rceil}\right\rceil,  |\Omega| \}$. Finally, since $|J_{\vec{i}}|<J_{\max}$, the bound claimed in the theorem statement follows.
\end{proof}

Although \cref{Theorem:NPFO_Norm} is based on Theorem 23 of Ref.~\cite{Clinton_Bausch_Cubitt2021}, \cref{Theorem:NPFO_Norm} is more general as it also applies to NPFOs that contain number operators and terms of locality greater than or equal to $2$.

\subsection{Bounding the Commutator with Disjoint Operators}

We now investigate how many NPFOs are generated when one takes the commutator of two local NPFOs.

\begin{lemma}\label{Lemma:NPFO_Commutator}
Let $h_{\vec{i}}$ and $ h_{\vec{j}}$ be two non-commuting  NPFOs with locality $k_{\vec i}$ and $ k_{\vec j}$, respectively. 
Then, $[h_{\vec{i}}, h_{\vec{j}}]$ is a sum of at most $2^{1+\min\{k_{\vec i},k_{\vec j}\}/2}$ NPFOs, each of which has locality of at most $k_{\vec i}+k_{\vec j}-1$ and at least $\max\{k_{\vec i}, k_{\vec j}\}$.
\end{lemma}
\begin{proof}
Consider $[h_{\vec{i}},h_{\vec{j}}] = h_{\vec{i}}h_{\vec{j}} - h_{\vec{j}}h_{\vec{i}} $ and explicitly write out the term
\begin{align}
    {h_{\vec{i}}}
    {h_{\vec{j}}} =
    \adag(i_1)\dots \adag(i_m)a(i_{m+1})\dots a(i_{
    2m})&N(i_{
    2m+1})\dots N(i_{k_{ \vec i}}) \nonumber \\
    &\times
    {\adag(j_1)\dots \adag(j_
    l)a(j_{
    l+1})\dots a(j_{
    2l})N(j_{
    2l+1})\dots N(j_{k_{\vec j}})}.
\end{align}
To put this in the NPFO form, we move all $a^\dagger(j)$ operators to the left. 
Using the relations
\begin{align}
    \begin{split}
    &a(i)\adag(j) = \delta_{ij} - \adag(j)a(i), \\
    &N(i)\adag(j) = a^\dag(j)N(i), \quad i \neq j \\
    &N(j)\adag(j) = \adag(j), \\
    &N(i)a(j) = a(j)N(i), \quad i \neq j \\
    &N(j)a(j) = 0,
    \end{split}
\end{align}
and the NPFO property that ensures each $a(i)$ only intersects with at most one $a^\dagger(j)$, we observe that by pushing all the $\adag(j)$ operators to the left, at most $2^{\min\{k_{\vec i},k_{\vec j}\}/2}$ terms are generated.
This can be understood by assuming that all the $a^\dagger(j_1),\ldots,a^\dagger(j_l)$ operators intersect with one of the $a(i_{m+1}),\ldots,a(i_{2m})$ operators and $k_{\vec j} \leq k_{\vec i}$, in which case at most $k_{\vec j}/2$ terms of the form $1-a^\dagger(r) a(r)$ get generated within the $h_{\vec{i}}h_{\vec{j}}$ string. 
The factor of $1/2$ arises as $h_{\vec{j}}$ is $k_{\vec j}$-local, hence it can have at most $k_{\vec j}/2$ creation operators. Similarly, if $k_{\vec j} \geq k_{\vec i}$, at most $k_{\vec i}/2$ creation operators within $h_{\vec{j}}$ overlap with the creation operators within $h_{\vec{i}}$. This gives rise to at most $2^{{\rm min}\{k_{\vec i},k_{\vec j}\}/2}$ NPFOs in $h_{\vec{i}}h_{\vec{j}}$. 

The above scenario is not the only possibility, as some of $a^\dagger(j_1),\ldots,a^\dagger(j_l)$ may instead intersect with some of the $N(i_{2m+1}),\ldots,N(i_{k_{\vec i}})$, but that eliminates the number operator from the string, resulting in fewer NPFOs. Using the same arguments, all the number operators belonging to $h_{\vec{i}}$ can be moved to the far left of the individual creation and annihilation operators at the cost of a smaller number of terms. 

The overall conclusion is that at most $2^{{\rm min}\{k_{ \vec i},k_{\vec j}\}/2}$ NPFOs are generated for $h_{\vec{i}}h_{\vec{j}}$.

Thus $[h_{\vec{i}}, h_{\vec{j}}]$ can be written as a sum of at most $2\times 2^{\min\{k_{\vec i},k_{ \vec j}\}/2}$ NPFOs.
The locality is then i) no more than $k_{\vec i}+k_{\vec j}-1$ (where the $-1$ arises from the fact that the operators must overlap on at least one site to have nonzero commutator), and ii) no less than that of the maximum of the locality of the original operators, since the definition of an NPFO precludes cancellations. 
\end{proof}

Next, given two disjoint, translation-invariant operators $X$ and $Y$, we upper bound the number of disjoint sets of terms their commutator generates. 

\begin{lemma}\label{Lemma:NPFO_Comm}
Let $ X$ and $Y$ be two translation-invariant operators, each defined as a sum of disjoint NPFOs, with interactions with locality no more than $k_X$ and $k_Y$, respectively.
Then the operator $[Y,X]$ can be written as a sum of at most $2k_Xk_Y(k_X-1)(k_Y-1) 2^{1+\min\{k_X,k_Y\}/2}$ translation-invariant, disjoint operators which are sums of NPFOs.
The individual NPFOs have locality of at most $k_X+k_Y-1$ and at least $\max\{k_X,k_Y\}$.
\end{lemma}
\begin{proof}

We can write the operators as
\begin{align}
    X &= J_X \sum_{\vec{i}\in \Omega_X} X_{\vec{i}}, \\
    Y &= J_Y \sum_{\vec{i}\in \Omega_Y} Y_{\vec{i}},
\end{align}
where $\Omega_X$ and $\Omega_Y$ are sets of tuples with no more than $k_X$ and $k_Y$ indices in each tuple, respectively.
Furthermore, because $X$ and $Y$ are each a sum of disjoint NPFOs, $\vec{i}\cap \vec{j} = \emptyset$ for any $\vec{i}, \vec{j} \in \Omega_X$, and similarly for $\Omega_Y$.

Both $Y$ and $X$ are translation-invariant, and we wish to write $[Y,X]$ as a sum of translation-invariant terms.
First note that $[Y,X]$ can be decomposed into a sum of terms of the form $[Y_{\vec j},X_{\vec i }]$, where $\vec{i}\cap \vec{j} \neq \emptyset$ (otherwise this commutator is zero).
There are at most $k_Xk_Y$ possible ways of translating a term of the form $X_{\vec i}$ to intersect with $Y_{\vec j}$.
For each of these possible translations, we label the corresponding terms $w^{(a)}_{\vec k} = [Y_{\vec j},X_{\vec i }]$, where $a\in \{1,\dots, k_Xk_Y\}$ and $\vec{k}=\vec{i}\cup  \vec{j}$.
For a fixed $a$, every term of the form $w^{(a)}_{\vec k}$ is a translation of every other term of this form.
Since $X$ and $Y$ are sums of translation-invariant NPFOs, we can write
\begin{align}
    [Y,X] = J_XJ_Y \sum_{a=1}^{k_Xk_Y}\sum_{\vec{k}\in \Omega_a} w^{(a)}_{\vec{k}}.
\end{align}
To summarize, each term $w^{(a)}_{\vec{k}}$ corresponds to a particular $[Y_{\vec{j}},X_{\vec{i}}]$ with $\vec{k}=\vec{i}\cup  \vec{j}$, and $\Omega_a$ is the translation-invariant set of tuples the $w^{(a)}_{\vec{k}}$ have support on for a given $a$.
We give examples in \cref{Fig:Commutator_Decomposition} and \cref{Fig:Translation_Overlaps}.

\begin{figure}
    \centering
    \includegraphics[scale=0.475]{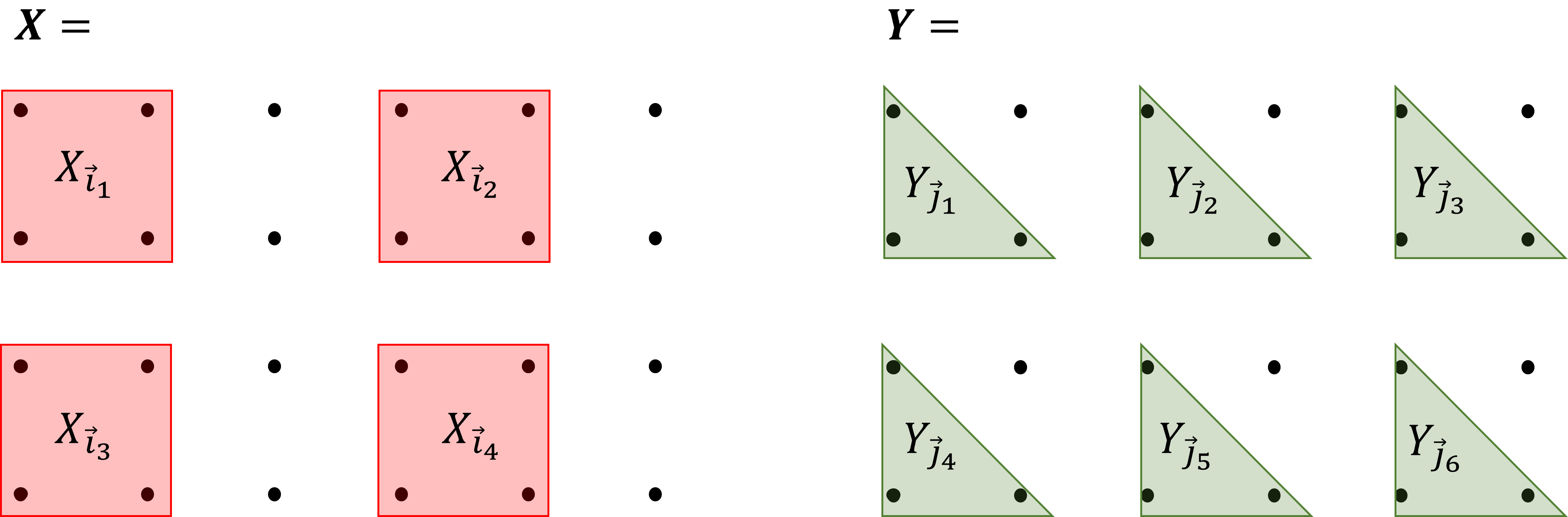}
    \caption{Examples of $X$ and $Y$ and their decompositions into local, disjoint, translationally invariant NPFOs. 
    The colored regions represent where the operators act non-trivially.
    }
    \label{Fig:Commutator_Decomposition}
\end{figure}

\begin{figure}
    \centering
    \includegraphics[scale=0.475]{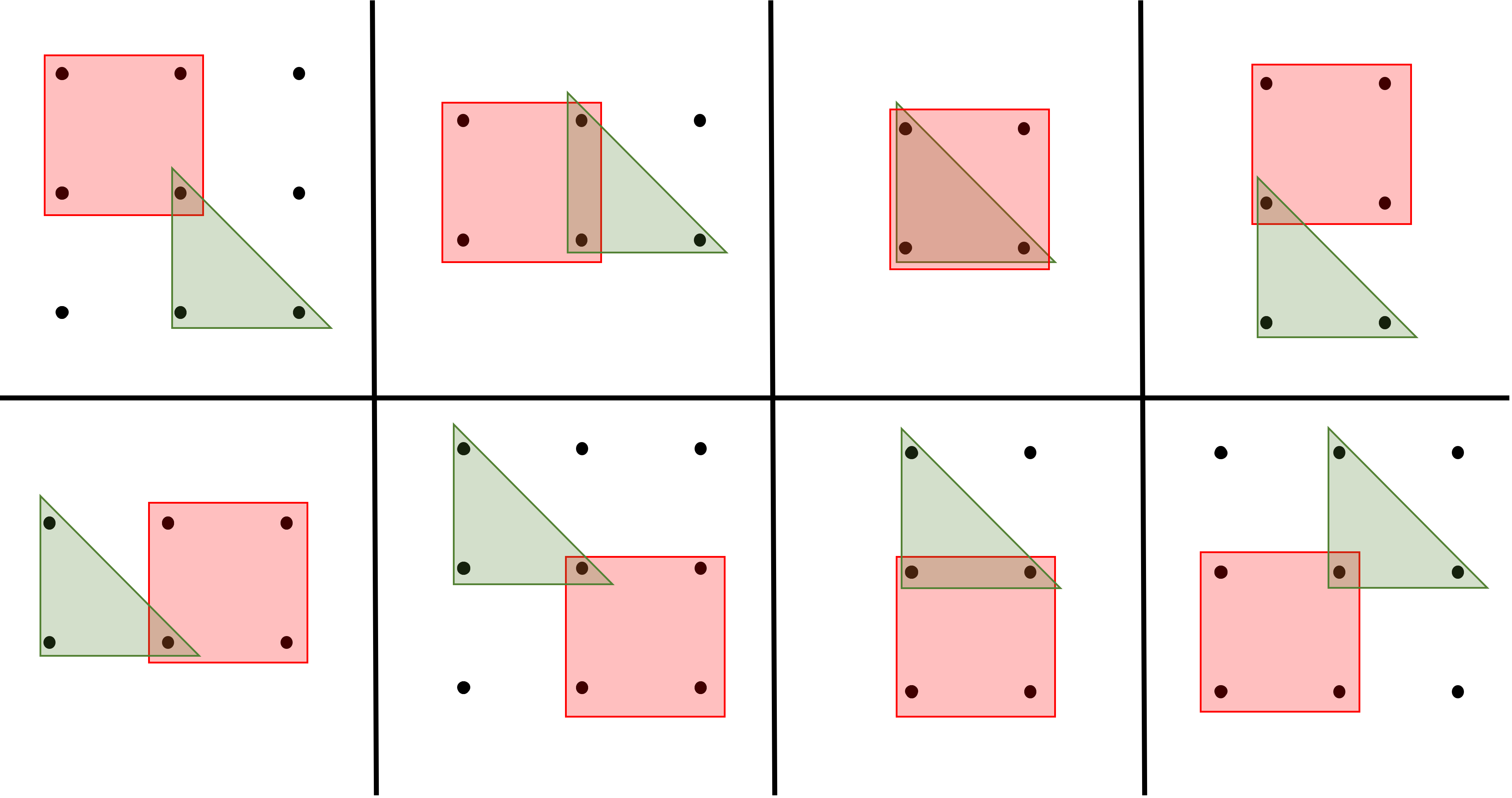}
    \caption{All possible overlapping translations of $X_{\vec i }$ and $Y_{\vec j }$.
    There are $8 \leq 4\times 3 = k_Xk_Y$ such translations.
    These form the set of operators $\{w^{(a)}_{\vec k} \}_a$.
    The commutator $[Y,X]$ can be written as a sum of such operators.
    }
    \label{Fig:Translation_Overlaps}
\end{figure}

\begin{figure}
    \centering
    \includegraphics[scale=0.575]{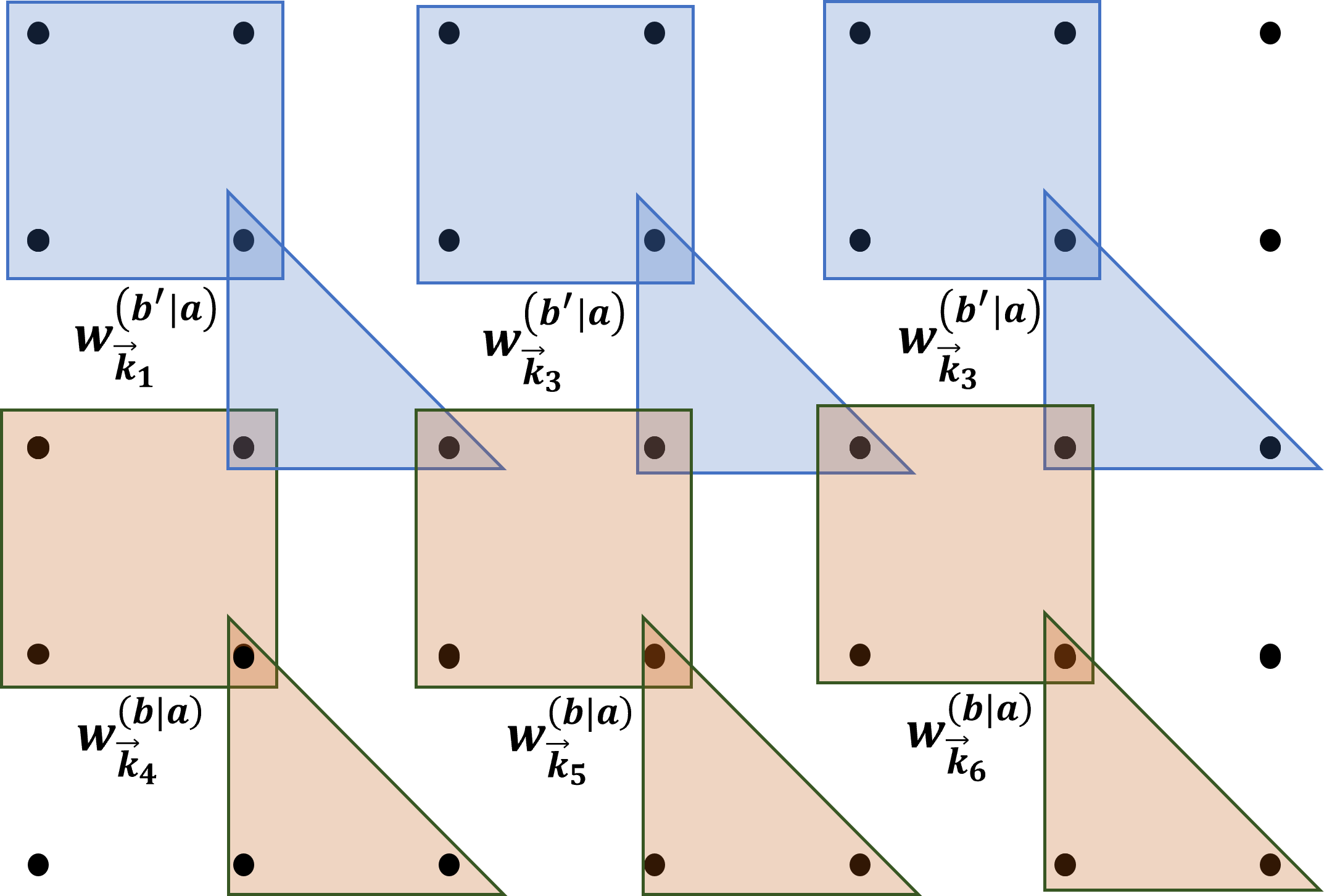}
    \caption{Consider a set of terms which are translations of the top-left operator in \cref{Fig:Translation_Overlaps}.
    When grouping terms $\{w^{(a)}_{\vec k }\}_{\vec k \in \Omega_a}$, we wish to split them into non-overlapping sets 
    $\{w^{(b|a)}_{\vec i }\}_{\vec i \in \Omega_{b|a}}$ 
    and $\{w^{(b'|a)}_{\vec j }\}_{\vec j \in \Omega_{b'|a}}$, denoted here by the blue and brown terms, such that the new sets are now disjoint.   }
    \label{Fig:Possible_Decompositions}
\end{figure}

\begin{figure}
    \centering
    \includegraphics[scale=0.475]{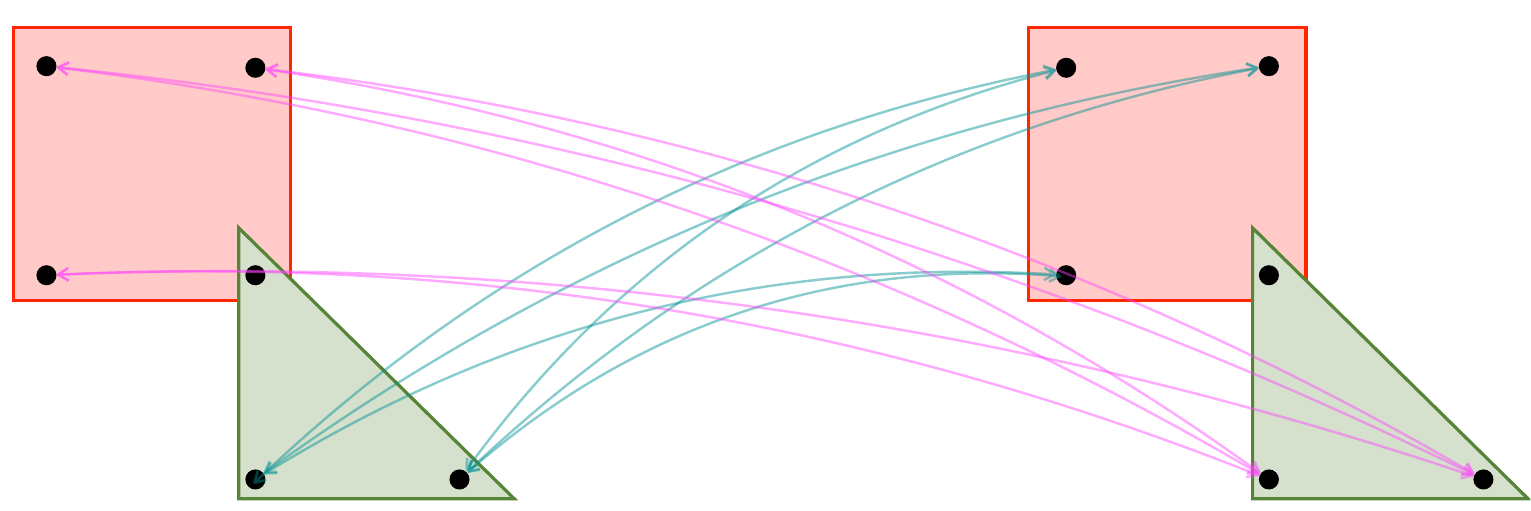}
    \caption{  The lines represent the possible places the operators can overlap.   Generally, the maximum number of disjoint sets of $\{ w^{(b|a)} \}_b$ can be obtained by noting that at most $k_X-1$ vertices from the original $\Omega_X$ set can overlap with at most $k_Y-1$ vertices from the original $\Omega_Y$ set and vice versa, giving an upper bound on the number of disjoint sets of $2(k_X-1)(k_Y-1)$.}
    \label{Fig:Possible_Overlaps}
\end{figure}

So far we have written $[Y,X]$ as a sum of translation-invariant terms. 
We now split these into sets of terms which only contain disjoint operators.
Each commutator $w^{(a)}_{\vec{k}} = [Y_{\vec{j}},X_{\vec{i}}]$ may have locality at most $k_X+k_Y-1$ since $X_{\vec i}$ and $Y_{\vec j }$ individually have locality $k_X$ and $k_Y$, respectively, but must intersect on at least one mode---if they do not intersect on at least one mode, the commutator is zero.
Since $X$ and $Y$ are each a sum of disjoint operators,
a given $w^{(a)}_{\vec{k}}=[Y_{\vec{j}},X_{\vec{i}}]$ can overlap with at most $2(k_X-1)(k_Y-1)$ other terms of the form $w^{(a)}_{\vec{k}'}$. 
To see this, recall that all $\{Y_{\vec j}\}_{\vec{j}\in \Omega_Y}$ are disjoint, and similarly for $\{X_{\vec i}\}_{\vec{i}\in \Omega_X}$.
Furthermore, consider $w^{(a)}_{\vec{k}}=[Y_{\vec{j}},X_{\vec{i}}]$ and one of its translations, $w^{(a)}_{\vec{k'}}=[Y_{\vec{j'}},X_{\vec{i'}}]$, such that they intersect on at least one mode.
Then $X_{\vec{i'}}$ cannot intersect $X_{\vec{i}}$ anywhere, and so can only intersect $Y_{\vec{j}}$, and vice versa.
As a result, for a particular $a$, $w^{(a)}_{\vec{k}}$ can only intersect up to $2(k_X-1)(k_Y-1)$ terms that are translations of 
itself (see \cref{Fig:Possible_Overlaps} for an illustration).
Note that we do not need to consider intersections between $w^{(a)}$ and $w^{(a')}$ for $a\neq a'$ as we immediately group them into different sets.

Thus, to decompose $[Y,X]$ into disjoint sets of terms such that none of the terms have support on the same fermionic modes, one can partition the terms $w^{(a)}_{\vec{k}}$ by taking each $\sum_{\vec{k}\in \Omega_a} w^{(a)}_{\vec{k}}$, and rearranging into $2(k_X-1)(k_Y-1)$ disjoint sets of commutators that are translation invariant. 
Thus $\Omega_a$ decomposes into disjoint, translation-invariant subsets, which we label as $\Omega_{b|a}$: 
\begin{align}
    \sum_{\vec{k}\in \Omega_a} w^{(a)}_{\vec{k}} = \sum_{b=1}^{2(k_X-1)(k_Y-1)}\sum_{\vec{k}\in \Omega_{b|a}} w^{(b,a)}_{\vec{k}},
\end{align}
where, for given $a,b$, none of the $w^{(b,a)}_{\vec{k}}$ have support on each other.
We given an example of how this could be done in \cref{Fig:Possible_Decompositions} and \cref{Fig:Possible_Overlaps}.
Thus so far, we have
\begin{align}
    [Y,X] = J_XJ_Y \sum_{a=1}^{k_Xk_Y}\sum_{b=1}^{2(k_X-1)(k_Y-1)}\sum_{\vec{k}\in \Omega_{b|a}} w^{(b,a)}_{\vec{k}}.
\end{align}

Now each term $w^{(b,a)}_{\vec{k}}$ corresponds to a commutator $[Y_{\vec{j}},X_{\vec{i}}]$ that, by \cref{Lemma:NPFO_Commutator}, generates at most $2^{1+\min\{k_X,k_Y\}/2}$ NPFO terms.
Since, for a fixed $a,b$ pair, each term $w^{(b,a)}_{\vec{k}}$ is a translation of all other $w^{(b,a)}_{\vec{l}}$, then $\sum_{\vec{k}\in \Omega_{b|a}} w^{(b,a)}_{\vec{k}}$ can be further decomposed into at most $2^{1+\min\{k_X,k_Y\}/2}$ translation-invariant, disjoint sums of NPFOs:
\begin{align}
    \sum_{\vec{k}\in \Omega_{b|a}} w^{(b,a)}_{\vec{k}} = \sum_{c=1}^{2^{1+\min\{k_X,k_Y\}/2}}\sum_{\vec{k}\in \Omega_{c|b|a}} v^{(c,b,a)}_{\vec{k}},
\end{align}
where for fixed $a,b,c$, $\sum_{\vec{k}\in \Omega_{c|b|a}} v^{(c,b,a)}_{\vec{k}}$ is a translation-invariant, disjoint sum of NPFOs $v^{(c,b,a)}_{\vec{k}}$.
Therefore,
\begin{align}
    [Y,X] = J_XJ_Y \sum_{a=1}^{k_Xk_Y}\sum_{b=1}^{2(k_X-1)(k_Y-1)}\sum_{c=1}^{2^{1+\min\{k_X,k_Y\}/2}}\sum_{\vec{k}\in \Omega_{c|b|a}} v^{(c,b,a)}_{\vec{k}}.
\end{align}
The lemma statement then follows.
\end{proof}

We now use the above lemmas to bound the (semi-)norm of a nested commutator.

\begin{theorem} [Restatement of \cref{Theorem:General_Order_Fermionic_Error_2} of the main text]\label{Theorem:General_Order_Fermionic_Error}
Let $\{H_{\gamma_i}\}_i$ be a set of translation-invariant, disjoint Hamiltonians such that
\begin{align}
    H_{\gamma_i} = J^{(\gamma_i)}\sum_{\vec{j}} h_{\vec{j}}^{(\gamma_i)},
\end{align}
and each $h_{\vec{j}}^{(\gamma_i)}$ is a NPFO with locality $k^{(\gamma_i)}$.
Then,
\begin{align}
    &\norm{\left[ H_{\gamma_{p+1}},\dots ,[H_{\gamma_2},H_{\gamma_1}] \right]}_\eta \leq \left(\prod_{n=1}^{p+1} \big| J^{(\gamma_n)} \big| \right) 
    \prod_{m=2}^{p+1}\bigg[ 2k^{(\gamma_m)} (k^{(\gamma_m)} -1) \left(\sum_{n=1}^{m-1}k^{(\gamma_n)} - (m-2) \right) 
    \nonumber\\
    & \hspace{4.25 cm} \times \left(\sum_{n=1}^{m-1}k^{(\gamma_n)} - (m-1) \right) 2^{1+\min\{ k^{(\gamma_m)}, \sum_{n=1}^{m-1}k^{(\gamma_n)} - (m-2) \}/2  }\bigg] \left\lceil \frac{\eta}{\left\lceil k_{\min}/2\right\rceil} \right\rceil ,
\end{align}
where $k_{\min}\coloneqq\min_{1\leq i\leq p+1}\left\{ k^{(\gamma_i)} \right\}$.
\end{theorem}

\begin{proof}
We proceed by induction, starting with the $p=1$ case.
\paragraph{Case $p=1$.}
Using \cref{Lemma:NPFO_Comm}, we can write $[H_{\gamma_2},H_{\gamma_1}]$ as a sum of translation-invariant, disjoint NPFO terms: 
\begin{align}
    [H_{\gamma_2},H_{\gamma_1}] = J^{(\gamma_1)}J^{(\gamma_2)} \sum_{m=1} C^{(m)}_{\gamma_1 \gamma_2},
\end{align}
where the sum goes up to $2k^{(\gamma_1)}k^{(\gamma_2)}(k^{(\gamma_1)}-1)(k^{(\gamma_2)}-1) 2^{1+\min\{(k^{(\gamma_1)}),(k^{(\gamma_2)})\}/2}$.
Each of the terms has the form
\begin{align}
    C^{(m)}_{\gamma_1 \gamma_2} = \sum_{\vec{i}\in \Omega_{m}} c^{(m)}_{\gamma_1\gamma_2,\vec{i}}\,,
\end{align}
where the $c^{(m)}_{\gamma_1\gamma_2,\vec{i}}$ are NPFOs with locality at most $k^{(\gamma_1)}+k^{(\gamma_2)}-1$ and at least $\mathrm{max}\{k^{(\gamma_1)},k^{(\gamma_2)}\}$.
Then, using the triangle inequality and \cref{Theorem:NPFO_Norm}, the fermionic semi-norm can be bounded for each of the sets associated with $\Omega_m$ for each $m$, giving
\begin{align}
    \norm{[H_{\gamma_2},H_{\gamma_1}]}_\eta &\leq |J^{(\gamma_1)}J^{(\gamma_2)}| \sum_{m=1} \norm{C^{(m)}_{\gamma_1 \gamma_2}}_\eta \\
    &\leq |J^{(\gamma_1)}|\times|J^{(\gamma_2)}| 2\times k^{(\gamma_1)}k^{(\gamma_2)}(k^{(\gamma_1)}-1)(k^{(\gamma_2)}-1) 2^{\min\{(k^{(\gamma_1)}),(k^{(\gamma_2)})\}+1}  \left \lceil \frac{\eta}{\lceil k_{\min}/2\rceil} \right\rceil,
\end{align}
where $k_{{\min}}=\min\{k^{(\gamma_1)},k^{(\gamma_2)}\}$, and where we have ignored the dependence of $|\Omega|$ when using \cref{Theorem:NPFO_Norm}.
This proves the $p=1$ case.

\paragraph{Case of general $p>2$.}

Assume that 
\begin{align}
    \left[ H_{\gamma_{p}},\dots ,[H_{\gamma_2},H_{\gamma_1}] \right] = \prod_{n=1}^{p}J^{(\gamma_n)} \sum_{m_1, m_2, \dots, m_p} C^{(m_1,m_2,\dots, m_p)}_{\gamma_1 \gamma_2\dots \gamma_p},
\end{align}
where each $m_n$ sum goes up to 
\begin{align}
2k^{(\gamma_n)}(k^{(\gamma_n)}-1) \left(\sum_{m=1}^{n-1}k^{(\gamma_m)} - (n-2) \right) \left(\sum_{m=1}^{n-1}k^{(\gamma_m)} - (n-1) \right)
    2^{1+\mathrm{min}\{ k^{(\gamma_n)}, \sum_{m=1}^{n-1}k^{(\gamma_m)} - (n-2) \}/2  },
\end{align}
and each of the $ C^{(m_1,m_2,\dots, m_p)}_{\gamma_1 \gamma_2\dots \gamma_p}$
is a translation-invariant sum of disjoint terms of the form
\begin{align}
     C^{(m_1,m_2,\dots, m_p)}_{\gamma_1 \gamma_2\dots \gamma_p} = \sum_{\vec{i}\in \Omega_{m_1,\dots, m_p}} c^{(m_1,\dots, m_p)}_{\gamma_1\gamma_2\dots \gamma_p,\vec{i}},
\end{align}
where each $c^{(m_1,\dots, m_p)}_{\gamma_1\gamma_2\dots \gamma_p,\vec{i}}$ is an NPFO.
Then,
\begin{align}
    \left[ H_{\gamma_{p+1}},\dots ,[H_{\gamma_2},H_{\gamma_1}] \right] &= \left[H_{\gamma_{p+1}},  
    \prod_{n=1}^{p}J^{(\gamma_n)} \sum_{m_1, m_2, \dots, m_p} 
    C^{(m_1,m_2,\dots, m_p)}_{\gamma_1 \gamma_2\dots \gamma_p} \right] \\
    &= \prod_{n=1}^{p}J^{(\gamma_n)} \sum_{m_1, m_2, \dots, m_p}  \left[ H_{\gamma_{p+1}}, C^{(m_1,m_2,\dots, m_p)}_{\gamma_1 \gamma_2\dots \gamma_p} \right].
\end{align}
Since both $H_{\gamma_{p+1}}$ and $C^{(m_1,m_2,\dots, m_p)}_{\gamma_1 \gamma_2\dots \gamma_p}$ are translation-invariant, disjoint sums of NPFOs, \cref{Lemma:NPFO_Comm} can be applied to write
\begin{align}
    \left[ H_{\gamma_{p+1}}, C^{(m_1,m_2,\dots, m_p)}_{\gamma_1 \gamma_2\dots \gamma_p} \right] = J^{(\gamma_{p+1})}\sum_{m_{p+1}} C^{(m_1,m_2,\dots, m_p,m_{p+1})}_{\gamma_1 \gamma_2\dots \gamma_p \gamma_{p+1}},
\end{align}
where the sum of $m_{p+1}$ goes up to at most 
\begin{align}
2k^{(\gamma_{p+1})} (k^{(\gamma_{p+1})} 
 - 1)\left(\sum_{m=1}^{p}k^{(\gamma_m)} - (p-1) \right) \left(\sum_{m=1}^{p}k^{(\gamma_m)} - p \right)
    2^{1+\min\{ k^{(\gamma_{p+1})}, \sum_{m=1}^{p}k^{(\gamma_m)} - (p-1) \}/2  },
\end{align}
and the terms $C^{(m_1,m_2,\dots, m_p,m_{p+1})}_{\gamma_1 \gamma_2\dots \gamma_p \gamma_{p+1}}$ are translation-invariant, disjoint sums of NPFOs.
Thus, one can see that 
\begin{align}
    \left[ H_{\gamma_{p+1}},\dots ,[H_{\gamma_2},H_{\gamma_1}] \right] &= \prod_{n=1}^{p+1}J^{(\gamma_n)} \sum_{m_1, m_2, \dots, m_p, m_{p+1}} C^{(m_1,m_2,\dots, m_p,m_{p+1})}_{\gamma_1 \gamma_2\dots \gamma_p\gamma_{p+1}},
\end{align}
so
\begin{align}
   \norm{ \left[ H_{\gamma_{p+1}},\dots ,[H_{\gamma_2},H_{\gamma_1}] \right]}_{\eta} &\leq \prod_{n=1}^{p+1}|J^{(\gamma_n)}| \sum_{m_1, m_2, \dots, m_p, m_{p+1}} \norm{ C^{(m_1,m_2,\dots, m_p,m_{p+1})}_{\gamma_1 \gamma_2\dots \gamma_p\gamma_{p+1}} }_{\eta}.
   \label{Eq:Total_Commutator}
\end{align}
Using \cref{Theorem:NPFO_Norm} to bound $\norm{ C^{(m_1,m_2,\dots, m_p,m_{p+1})}_{\gamma_1 \gamma_2\dots \gamma_p\gamma_{p+1}} }_{\eta} $, we have
\begin{align}
    \sum_{m_1, \dots, m_{p+1}} \norm{ C^{(m_1,\dots,m_{p+1})}_{\gamma_1\dots \gamma_{p+1}} }_{\eta} \leq&   \prod_{m=2}^{p+1}\bigg[ 2k^{(\gamma_m)}(k^{(\gamma_m)} -1) \left(\sum_{n=1}^{m-1}k^{(\gamma_n)} - (m-2) \right) \left(\sum_{n=1}^{m-1}k^{(\gamma_n)} - (m-1) \right)
    \nonumber\\
    & \hspace{2.5 cm} \times 2^{1+\min\{ k^{(\gamma_m)}, \sum_{n=1}^{m-1}k^{(\gamma_n)} - (m-2) \}/2 }\bigg]  \max_{m_1, \dots, m_{p+1}}\norm{ C^{(m_1,\dots,m_{p+1})}_{\gamma_1\dots \gamma_{p+1}} }_{\eta} \\
    \leq&  \prod_{m=2}^{p+1}\bigg[ 2k^{(\gamma_m)} (k^{(\gamma_m)} -1) \left(\sum_{n=1}^{m-1}k^{(\gamma_n)} - (m-2) \right) \left(\sum_{n=1}^{m-1}k^{(\gamma_n)} - (m-1) \right) 
    \nonumber\\
    & \hspace{4.5 cm} \times 2^{1+\min\{ k^{(\gamma_m)}, \sum_{n=1}^{m-1}k^{(\gamma_n)} - (m-2) \}/2  }\bigg] \left\lceil \frac{\eta}{\left\lceil k_{\min}/2\right\rceil} \right\rceil ,
\end{align}
where we have used the fact that $C^{(m_1,m_2,\dots, m_p)}_{\gamma_1 \gamma_2\dots \gamma_p\gamma_{p+1}}$ has locality of at least $k_{\min}=\min_{1\leq i\leq p+1}\left\{ k^{(\gamma_i)} \right\}$, and have ignored the bound depending on $|\Omega|$ when using \cref{Theorem:NPFO_Norm}.
Substituting this into \cref{Eq:Total_Commutator} gives the statement in the theorem.
\end{proof}

\subsection{Asymptotic Scaling of Bounds for Fermionic-Bosonic Hamiltonians}

A simple analog to \cref{Theorem:NPFO_Norm} can be obtained when the fermionic terms are coupled to a bosonic term. 
Provided the bosonic Hilbert space is truncated, the magnitude of the coupled term can be bounded by taking the maximum value of the bosonic operator and then treating the bosonic part as a coefficient of the fermionic terms.
We apply this strategy to the nuclear EFTs that we consider.

\begin{theorem}[Asymptotic Dynamical-Pion EFT Bound on Nested Commutators]
\label{Theorem:Pionful_Asymptotics}
Let $H$ be the dynamical-pion Hamiltonian described in \cref{Sec:Dyn_Pions}, using the decomposition $\{H_{\gamma_i}\}_i$ given in \cref{Sec:Pionful_Circuit_Costs}.
Then, assuming $\eta < L^3/2$, we have
\begin{align}
    \norm{\left[ H_{\gamma_{p+1}},\dots ,[H_{\gamma_2},H_{\gamma_1}] \right]}_\eta \leq 
    O\left(\pimax^{p+1} \Pimax^{p+1}L^{3}\right).
\end{align}
\end{theorem}
\begin{proof}

It is not hard to see that the largest contribution to the commutator of two Hamiltonian terms arises from $[H_{\rm WT},H_{\rm WT}]$. This is because $H_{\rm WT}$ consists of a $\pi \Pi$ operator times a sum of fermionic operators, and these fermionic operators may not commute at each site. Therefore, $[H_{\rm WT},H_{\rm WT}]=O(\pi_{\max}^2\Pi_{\rm max}^2 L^3)$, where the factor of $L^3 \sim \eta$ arises because the two spatial sums from each term in the commutator turn into one sum after the effect of the fermionic commutation. The final sum over fermionic operators can be upper bounded by the number of fermions present, $\eta$. The explicit computation of this commutation bound is provided in \cref{Lemma:HWT_HWT_Commutator}. Any other commutators of two Hamiltonian terms in the dynamical-pion EFT is suppressed compared to this scaling, as verified in \cref{Sec:Trotter_Error_Dyn_p1}, since i) $H_{\rm WT}$ is a sum of the largest factor of bosonic and fermionic operators compared to other Hamiltonian terms, ii) the commutation among bosonic terms reduces one factor of $\pi \Pi$ in the product due to \cref{eq:commutations-pi}, and iii) commutation among purely fermionic terms scales at most as $O(\eta)=O(L^3)$, and can be ignored compared with the dominant one identified.

Now for the $p$th-order nested commutator involving Hamiltonian terms, it is evident that the largest scaling arises from $[H_{\rm WT},[H_{\rm WT},...,[H_{\rm WT},H_{\rm WT}]]]$, which by the above argument is $O(\pi^{p+1}_{\max}\Pi^{p+1}_{\rm max} L^3)$. Overall, we have
    \begin{align}
         \norm{\left[ H_{\gamma_{p+1}},\dots ,[H_{\gamma_2},H_{\gamma_1}] \right]}_\eta 
         &=O\left(\pimax^{p+1} \Pimax^{p+1}L^{3}\right).
    \end{align}
\end{proof}

\section{Analytic Trotter Error Bounds for Nuclear-EFT Hamiltonians} \label{Sec:Analytic_Trotter_Bounds_Proof}

This Appendix contains the full proofs on the Trotter error bounds in \cref{s:trotter_error_bnds}. 
We begin by introducing some useful notation, which simplifies the following calculations. In particular, we define a hopping term as
\begin{equation}
    \hop_{ij}^\pm \coloneqq a^\dagger(i)a(j)\pm a^\dagger(j)a(i)
\end{equation}
for $i\neq j$. When calculating fermionic norms of such an operator, the sign in the hopping term is irrelevant. Simple computation shows that
\begin{equation}
\label{eq:norm_hop}\norm{\hop_{ij}^\pm}_\eta=\max_{\ket{\psi},\ket{\phi}}\bra{\psi}\hop_{ij}^\pm\ket{\phi}= 1,
\end{equation} 
with saturation occurring for the Fock state $\ket{\psi}=\ket{\phi}=(\ket{10}\pm\ket{01})/\sqrt{2}$ over sites $i,j$. Consequently, for the purpose of analyzing Trotter error bounds, we do not have to distinguish between the hopping terms $\pm\hop_{ij}^+$ and $\pm\hop_{ij}^-$, so we refer to all four such terms in this equivalence class as $\hop_{ij}$. When the flavor $\sigma$ of the particle is relevant, we denote this with a superscript as $\hop_{ij}^\sigma$. In the rest of this Appendix, an equals sign indicates equality up to this equivalence class of operators as, for all the quantities, we ultimately care about various fermionic semi-norms for which the fine-grained sign information is irrelevant.

In terms of this equivalence class of hopping operators, we have the following useful commutation relations:
\begin{align}
&[\hop_{ij}^\sigma,{\hop_{kl}^\sigma}']=\delta_{\sigma\sigma'}\left[\hop_{jl}^\sigma\delta_{ik}+\hop_{ik}^\sigma\delta_{jl}+\hop_{il}^\sigma\delta_{jk}+\hop_{jk}^\sigma\delta_{il}\right],
\label{Eq:Comm-Hop-Hop}
\\
&[\hop_{ij}^\sigma,N_{\sigma'}(l)]=\delta_{\sigma\sigma'}(\delta_{il}+\delta_{jl})\hop_{ij}^\sigma,
\label{Eq:Comm-Hop-N}
\end{align}
where $\delta_{uv}$ are Kronecker deltas.

The irrelevance of the sign information in the hopping terms also allows us to introduce a simple diagrammatic notation for representing hopping terms and their commutators on a lattice. For instance, within a small 2D sublattice: 
\begin{figure}[ht]
	\centering
	\includegraphics[width=0.4\textwidth]{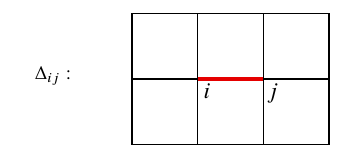}
\end{figure}
\begin{samepage}
The commutation relation in \cref{Eq:Comm-Hop-N} can be conveniently represented as a diagram where, if the two hopping terms (marked via red edges) share a vertex, then the resulting commutator is a hopping term joining the two ``free'' vertices (marked via a dashed blue edge). For instance, in the case where $i=k$, we consider \cref{eq:hopping_comm_diagram}.
\begin{figure}[ht]
	\centering
	\includegraphics[width=0.5\textwidth]{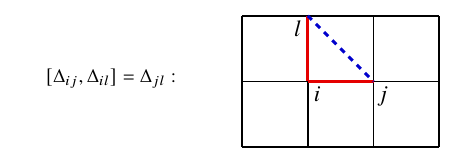}
	\label{eq:hopping_comm_diagram}
\end{figure}\label{eq:nested_comm_diagram}
If the hopping terms do not share a vertex or if they are identical, the commutator vanishes. 

	The diagrammatic notation, combined with dropping the sign information, makes computations even easier when evaluating nested commutators. For instance, we have the diagram in \cref{Fig:hopping_term_nested_commutator_diagram}
\begin{figure}[ht]
	\centering
	\includegraphics[width=0.53\textwidth]{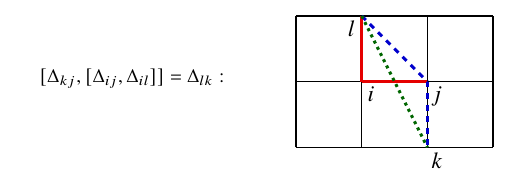}
	
	\label{Fig:hopping_term_nested_commutator_diagram}
\end{figure}
where red lines indicate the inner commutator, dashed blue lines the outer commutator, and dotted green lines the final result.
\end{samepage}

To finish our setup, we have the following useful lemma.

\begin{lemma}\label{Lemma:fermion_norm}
Let $\Omega$ be a set of completely disjoint ordered pairs of lattice sites $(i,j)$ with $i< j$. That is, for any $(i,j)\in\Omega$, there exists no distinct element $(k, l)\in\Omega$ with $k\in\{i,j\}$ or $l\in\{i,j\}$. Let $S$ be a set of particle flavors, and define
\begin{equation}\label{eq:lemHhop}
    H_{\Omega,S}=\sum_{\sigma\in S}\sum_{(i,j)\in\Omega} \hop_{ij}^\sigma.
\end{equation}
Then, assuming $\eta<|\Omega|$,
\begin{equation}
    \norm{H_{\Omega,S}}_\eta\leq \eta,
\end{equation}
where $\norm{\cdot}_\eta$ is the fermionic semi-norm for $\eta$-fermion states.
\end{lemma}
\begin{proof}
From \cref{eq:norm_hop}, the fermionic semi-norm for any individual hopping term in $H_{\Omega, S}$ is $1$. As each term in \cref{eq:lemHhop} acts on a disjoint set of fermionic modes, we can consider states that individually saturate as many of these terms as possible. The statement of the lemma follows immediately, with saturation occurring for a state with $\eta$ of the ordered pairs $(i,j)\in\Omega$ being in a superposition state that achieves the semi-norm in \cref{eq:norm_hop}.
\end{proof}

With these facts and definitions in hand, we are now ready to prove the bounds on the Trotter errors for the various EFTs of this work.

\subsection{Pionless-EFT Bounds
\label{App:Pionless-Bounds}}

\subsubsection{\texorpdfstring{$p=1$}{p=1}} \label{Sec:p=1_Pionless_Trotter_Error}

To prove \cref{thm:pionless_trotter_error_p1}, we split the Hamiltonian in \cref{Eq:pionless_Hamiltonian} into seven terms $\{H_\gamma\}$ that are then used in the $p=1$ Trotter error bound of \cref{eq:P-first-order}. In particular, we consider a set $\{H_{1}, \ldots, H_{6}\}$ with each term corresponding to a maximal set of kinetic-type non-overlapping terms in $H_\mathrm{free}$ 
plus an additional contact term $V \coloneqq H_{7}=H_{\Cpi}+H_{\Dpi}$. In particular, each kinetic-type term can be written as 
\begin{align}\label{Eq:kinetic_term}
   H_{\gamma_\xi}&=-h\sum_{\sigma}\sum_{(i,j)\in \Omega_{\gamma_\xi}} a_\sigma^\dagger(i)a_\sigma(j)+a_\sigma^\dagger(j)a_\sigma(i)\quad \left(+h\sum_\sigma\sum_i N_\sigma(i)\right)\\
   &=-h\sum_{\sigma}\sum_{(i,j)\in \Omega_{\gamma_\xi}} \hop_{ij}^\sigma \quad \left(+h\sum_\sigma\sum_i N_\sigma(i)\right)
\end{align}
for $\gamma_\xi\in\{1,\ldots, 6\}$, where $\Omega_{\gamma_\xi}$ consists of a maximum-sized set of completely disjoint ordered pairs of fermionic sites such that for each $(i,j)\in\Omega_{\gamma_\xi}$, the sites $i<j$ are nearest neighbors on the cubic lattice $\Lambda(L)$. The term in parentheses can be ignored henceforth as it cancels when computing the commutator of two kinetic-like terms and commutes with the contact terms. Given a choice of such sets with $\bigcap_{k=1}^6 \Omega_{\gamma_\xi}=\emptyset$, it is clear that
\begin{equation}
    H_\mathrm{free}=\sum_{\gamma=1}^6 H_{\gamma}.
\end{equation}
Diagrammatically, this splitting of the terms in $H_\mathrm{free}$ can be described (for a representative 3D sub-lattice) by the diagram \cref{Fig:hopping_term_splitting_diagram}, where solid red, blue, and green, and dashed red, blue, and green lines each represent hopping terms in a distinct set $H_{\gamma_\xi}$.
\begin{figure}[ht]
	\centering
	\includegraphics[width=0.3\textwidth]{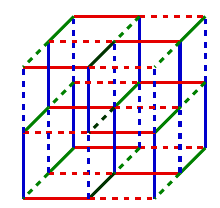}
	\caption{Solid red, blue, and green, and dashed red, blue, and green lines each represent hopping terms in a distinct set $H_{\gamma_\xi}$.}
	\label{Fig:hopping_term_splitting_diagram}
\end{figure}
 These are also the distinct sets introduced in \cref{Sec:Total-Pionless-EFT-Circuit-Depth} to parallelize the circuit implementation of the hopping operators, and \cref{Fig:Kinetic_Lattice_Decomposition} in the main text gives an equivalent representation.

Therefore, commutators of operators in $\{H_\gamma\}$ come in two types: kinetic-kinetic (that is, commutators of the form $[H_{\gamma_\mu}, H_{\gamma_\nu}]$ for $\gamma_\mu,\gamma_\nu\neq 7)$ and kinetic-potential (that is, commutators of the form $[V, H_{\gamma_\mu}]$ for $\gamma_\mu\neq 7$). Their fermionic semi-norms are bounded in the following lemmas. The parameter $h$ is assumed to be positive in this appendix; otherwise its instances must be changed to $|h|$.

\begin{lemma}\label{Lemma:kinetic-kinetic}
Kinetic-kinetic commutators are bounded as
\begin{equation}
    \norm{[H_{\gamma_\mu}, H_{\gamma_\nu}]}_\eta \leq 2h^2\eta,
\end{equation}
where $\gamma_\mu, \gamma_\nu \in \{1,\ldots, 6\}$ and $\gamma_\mu\neq \gamma_\nu$.
\end{lemma}
\begin{proof}
We proceed by direct computation. In particular,
\begin{align}
    \norm{[H_{\gamma_\mu}, H_{\gamma_{\nu}}]}_\eta&=h^2\norm{\bigg[\sum_{\sigma}\sum_{(i,j)\in \Omega_{\gamma_\mu}} \hop_{ij}^\sigma, \sum_{\sigma'}\sum_{(i',j')\in \Omega_{\gamma_{\nu}}} \hop_{i'j'}^{\sigma'}\bigg]}_\eta \nonumber\\
    &=h^2\norm{\sum_\sigma \bigg[\sum_{(i,j)\in \Omega_{\gamma_\mu}} \hop_{ij}^{\sigma}, \sum_{(i',j')\in \Omega_{\gamma_{\nu}}} \hop_{i'j'}^{\sigma}\bigg]}_\eta\nonumber\\
    &=h^2\norm{\sum_\sigma\sum_{(i,j)\in \Omega_{\gamma_\mu}} \sum_{\substack{(i',j')\in \Omega_{\gamma_{\nu}}\\ (i,j)\cap(i',j')\neq\emptyset}}\bigg[ \hop_{ij}^{\sigma}, \hop_{i'j'}^{\sigma}\bigg]}_\eta,
\end{align}
where in the second line, we used the fact that the operators acting on different species commute, and in the third line, we used the fact that the operators acting on entirely disjoint pairs of sites commute. 

Observe that this last line is simply a sum of commutators of hopping terms that share only a single vertex. From \cref{Eq:Comm-Hop-Hop}, these commutators each evaluate to a new hopping term. 
Diagrammatically, it is easy to see that these new hopping terms can be split into two completely disjoint sets. In particular, for each particle type $\sigma$, the terms in $H_{\gamma_\mu}$ and $H_{\gamma_\nu}$ form a 2D sub-lattice of hopping terms, whose commutators can be split into two completely disjoint sets as \cref{Fig:kinetic-kinetic-diagram}.
\begin{figure}[ht]
	\centering
	\includegraphics[width=0.6\textwidth]{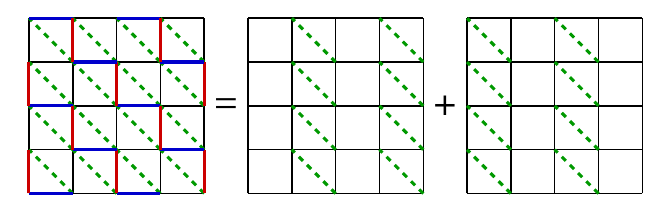}
	\caption{A decomposition of the kinetic-kinetic commutator terms.}
	\label{Fig:kinetic-kinetic-diagram}
\end{figure}
Here, blue lines correspond to hopping terms in $H_{\gamma_\mu}$, red lines correspond to hopping terms in $H_{\gamma_\nu}$, and dotted green lines correspond to their commutator.
Applying \cref{Lemma:fermion_norm} to each of these disjoint sets of hopping terms gives
\begin{equation}
     \norm{[H_{\gamma_\mu}, H_{\gamma_{\nu}}]}_\eta \leq 2h^2\eta,
\end{equation}
proving the result.
\end{proof}

\begin{lemma}\label{Lemma:kinetic-potential}
Kinetic-potential commutators are bounded as
\begin{equation}
    \norm{[V,H_{\gamma_\mu}]}_\eta \leq 2h\left(2 |\Cpi| \left\lfloor\frac{\eta}{2}\right\rfloor + \left(2|3\Cpi+\Dpi|+|\Dpi|\right)\left\lfloor\frac{\eta}{3}\right\rfloor + \left(2|6\Cpi+4\Dpi|+4|\Dpi|\right)\left\lfloor\frac{\eta}{4}\right\rfloor \right) ,
\end{equation}
where $i\in\{1,\ldots, 6\}$.
\end{lemma}
\begin{proof}
Using \cref{Eq:kinetic_term} for the kinetic term, \cref{Eq:H_contact_1,Eq:H_contact_2} for the contact terms, and \cref{Eq:Comm-Hop-N} for the commutator between hopping and number operators, we have
\begin{align}
    \norm{[V,H_{\gamma_\mu}]}_\eta&=h\norm{\bigg[H_{\Cpi}+H_{\Dpi},\sum_{\sigma}\sum_{(i,j)\in\Omega_{\gamma_\mu}}\hop_{ij}^\sigma\bigg]}_\eta\\
    &= h\norm{\sum_{( i,j)\in\Omega_{\gamma_\mu}}\bigg[H_{\Cpi}+H_{\Dpi},\sum_{\sigma}\hop_{ij}^\sigma\bigg]}_\eta\\
    &=h\Bigg|\Bigg|\underbrace{\sum_{(i,j)\in\Omega_{\gamma_\mu}}\sum_{\substack{\sigma, \sigma'\\\sigma'\neq \sigma}}\hop_{ij}^\sigma N_{\sigma'}(i)\left(\frac{\Cpi}{2}+\frac{\Dpi}{6}\sum_{\substack{\sigma''\\\sigma''\neq\sigma'\neq\sigma}}N_{\sigma''}(i)\right)}_{\text{A}}\nonumber \\
    &\qquad\qquad\qquad+\quad \text{A}(i\leftrightarrow j)\quad +\quad \text{A}(\sigma\leftrightarrow \sigma')\quad +\quad \text{A}(i\leftrightarrow j, \sigma\leftrightarrow\sigma')\nonumber\\
    &\qquad\qquad\qquad + \underbrace{\sum_{(i, j)\in\Omega_{\gamma_\mu}}\sum_{\substack{\sigma, \sigma', \sigma'' \\ \sigma\neq\sigma'\neq\sigma''}}N_\sigma(i)N_{\sigma'}(i)\frac{\Dpi}{6}\hop_{ij}^{\sigma''}}_{B} \quad + \quad B(i\leftrightarrow j) \, \Bigg|\Bigg|_\eta \label{eq:kinetic-potential-six}\\
    &\leq 4 h \norm{A}_\eta+ 2h\norm{B}_\eta \label{Eq:triangle_kinetic_contact}.
\end{align}
We have used the fact that operators acting on different particle types (labeled by different $\sigma$) commute. The notation $A(i\leftrightarrow j)$ denotes term $A$ with indices $i$ and $j$ swapped, and similarly for other terms expressed via the same notation. \Cref{Eq:triangle_kinetic_contact} comes from the triangle inequality and the fact that the different $A$ terms are identical under the given swaps, as are the different $B$ terms.

The semi-norm of each of these terms can be bounded by considering each individual sub-term in the semi-norms and applying the triangle inequality again. To bound the semi-norm of each individual sub-term, recall that there are four (fermionic) species labeled by $\sigma\in\{p, n\}\times \{\uparrow, \downarrow\}$, so the maximum occupation number on a given site is four.

We start by bounding the semi-norm of term $A$ by considering a single sub-term in the sum over $(i, j)\in\Omega_{\gamma_\mu}$:
\begin{equation}
    A_{ij} \coloneqq \sum_{\substack{\sigma, \sigma'\\\sigma'\neq \sigma}}\hop_{ij}^\sigma N_{\sigma'}(i)\left(\frac{\Cpi}{2}+\frac{\Dpi}{6}\sum_{\substack{\sigma''\\\sigma''\neq\sigma'\neq\sigma}}N_{\sigma''}(i)\right).
\end{equation}
To bound the norm of this operator, we consider a sum over bounds restricted to the possible particle numbers per site. Trivially, if the fermion number per site is zero, $\norm{A_{ij}}_\eta=0$. If the fermion number per site is one, then again $\norm{A_{ij}}_\eta=0$ as for any state, there are no fermions of one of the types $\sigma$ and $\sigma'$, causing either the $\hop_{ij}^\sigma$ operator or the $N_{\sigma'}(i)$ to vanish. 

For fermion number two on a site, the analysis becomes non-trivial. While the part of $A_{ij}$ proportional to $\Dpi$ still vanishes (given that among two-fermion states, necessarily one of $\Delta_{ij}^\sigma$, $N_{\sigma'}(i)$, and $N_{\sigma''}(i)$ causes the semi-norm to vanish), the term proportional to $\Cpi$ is non-zero. In particular, we have 
\begin{align}
  \norm{\hop_{ij}^\sigma N_{\sigma'}(i)}_\eta\leq \norm{\hop_{ij}^\sigma}_\eta \norm{N_{\sigma'}(i)}_\eta = 1.
\end{align}
Summing over $\sigma$ and $\sigma'$ yields a factor of two, so that $\norm{A_{ij}}_\eta\leq 2\times |\Cpi|/2=|\Cpi|$. To bound the full term in the two-fermion subspace, observe that the number of pairs $(i,j)$ with fermion number two is upper bounded by $\lfloor\eta/2\rfloor$. 
Thus, we have $\norm{A}_\eta\leq \lfloor\eta/2\rfloor |\Cpi|$. 

For fermion number three per site, both the $\Cpi$ and $\Dpi$ terms survive, and we have $\norm{A_{ij}}_\eta\leq 3\times 2\times\left|\Cpi/2+\Dpi/6\right|=6\left|\Cpi/2+\Dpi/6\right|$. The factor of $3$ comes from the sum over $\sigma$, and the factor of $2$ comes from the sum over $\sigma'\neq\sigma$. There are at most $\lfloor\eta/3\rfloor$ pairs of sites with fermion number three, so in this sector (via the triangle inequality), $\norm{A}_\eta\leq \lfloor\eta/3\rfloor |3\Cpi+\Dpi|$.

Finally, if the occupation number per site is four, then $\norm{A_{ij}}_\eta\leq 4\times 3\times\left|\Cpi/2+2\times\Dpi/6\right|=12\left|\Cpi/2+\Dpi/6\right|$. 
The factor of $4$ comes from the sum over $\sigma$, and the factor of $3$ comes from the sum over $\sigma'\neq\sigma$. The factor of $2$ in the $\Dpi$ term comes from the sum over $\sigma''$ with $\sigma''\neq\sigma'\neq\sigma$. 
There are at most $\lfloor\eta/4\rfloor$ pairs of sites with fermion number four, so in this sector, $\norm{A}_\eta\leq \lfloor\eta/4\rfloor |6\Cpi+4\Dpi|$.

Putting this all together, we have
\begin{equation}\label{eq:normA}
    \norm{A}_\eta\leq
    \left\lfloor\frac{\eta}{2}\right\rfloor |\Cpi|+ \left\lfloor\frac{\eta}{3}\right\rfloor |3\Cpi+\Dpi|+\left\lfloor\frac{\eta}{4}\right\rfloor |6\Cpi+4\Dpi|. 
\end{equation}

Now to bound $\norm{B}_\eta$, consider the individual sub-terms for a particular pair of sites $(i,j)\in\Omega_{\gamma_\mu}$. That is, we first bound the fermionic semi-norm of
\begin{equation}
    B_{ij} \coloneqq \sum_{\substack{\sigma, \sigma', \sigma'' \\ \sigma\neq\sigma'\neq\sigma''}}N_\sigma(i)N_{\sigma'}(i)\frac{\Dpi}{6}\hop_{ij}^{\sigma''}.
\end{equation}

The semi-norm for this term is zero for fermion-number subspaces less than three. 
With fermion number three on site $i$, we have $\norm{B_{ij}}_\eta=3 \times 2(|\Dpi|/6)$, where we use
\begin{align}
\norm{N_{\sigma}(i)N_{\sigma'}(i)\hop_{ij}^{\sigma''}}_\eta\leq \norm{N_{\sigma}(i)}_\eta\norm{N_{\sigma'}(i)}_\eta\norm{\hop_{ij}^{\sigma''}}_\eta  = 1
\end{align}
and, as for the previous term, the factor of of $3$ comes from the sum over $\sigma$ and the factor of $2$ comes from the sum over $\sigma'\neq\sigma$. 
Again, there are at most $\lfloor\eta/3\rfloor$ pairs of sites with fermion number three, so in this sector, $\norm{B}_\eta\leq \lfloor\eta/3\rfloor |\Dpi|$. 

For the fermion number four subspace,  $\norm{B_{ij}}_\eta=4\times 3\times 2(|\Dpi|/6)$. The factors of 4, 3, and 2 come from the sum over $\sigma$, $\sigma'\neq\sigma$, and $\sigma''\neq\sigma'\neq\sigma$, respectively. So in this fermion-number sector, $\norm{B}_\eta\leq 4\lfloor\eta/4\rfloor |\Dpi|$.
Putting this all together, we have
\begin{equation}\label{eq:normB}
    \norm{B}_\eta\leq
    \left\lfloor\frac{\eta}{3}\right\rfloor |\Dpi|+\left\lfloor\frac{\eta}{4}\right\rfloor 4|\Dpi|. 
\end{equation}
Combining \cref{Eq:triangle_kinetic_contact} with \cref{eq:normA,eq:normB} yields the lemma statement.
\end{proof}

With \cref{Lemma:kinetic-kinetic,Lemma:kinetic-potential}, \cref{thm:pionless_trotter_error_p1}, restated here for convenience, can be proved.
\begin{theorem}[$p=1$ Pionless-EFT Trotter Error (\cref{thm:pionless_trotter_error_p1} of the main text)]\label{thm:pionless_trotter_error_p1_2}
For the pionless-EFT Hamiltonian described in \cref{Sec:Pionless_EFT_Hamiltonian},
\begin{align}
    \norm{e^{-iHt} - \calP_1^{(\canpi)}(t)}_\eta \leq t^2\left( 15h^2\eta + 6h\left(A_1\left\lfloor\frac{\eta}{2}\right\rfloor + A_2\left\lfloor\frac{\eta}{3}\right\rfloor+A_3\left\lfloor\frac{\eta}{4}\right\rfloor   \right)  \right),
\end{align}
where $h=\frac{1}{2M{a_L^2}}$ is the coefficient of the hopping term, and
\begin{align}
    A_1 = 2|\Cpi|, \quad \quad A_2= 2|3\Cpi+\Dpi|+|\Dpi|, \quad \quad  A_3=2|6\Cpi + 4\Dpi| + 4|\Dpi|,
\end{align}
Here, $\Cpi$ and $\Dpi$ are the low-energy constants of two- and three-nucleon contact terms.
\end{theorem}

\begin{proof}
Consider the splitting of the pionless-EFT Hamiltonian $H$ into seven terms $H_{\gamma}$ for $\gamma\in\{1,\ldots, 7\}$ as described above. Then directly apply the $p=1$ Trotter error bound in \cref{eq:P-first-order}. In this bound, there are $\binom{6}{2}=15$ kinetic-kinetic commutators, each bounded as in \cref{Lemma:kinetic-kinetic}. In addition, there are 6 kinetic-potential commutators, each bounded as in \cref{Lemma:kinetic-potential}. Applying the triangle inequality in the expression in \cref{eq:P-first-order} for the $p=1$ pionless EFT, the Trotter error bound immediately yields the result. 
\end{proof}

\subsubsection{\texorpdfstring{$p=2$}{p=2}} \label{Sec:Pionless_EFT_p2}

We now present the proof of \cref{thm:pionless_trotter_error_p2}, which bounds the $p=2$ Trotter error for the pionless EFT. We consider the same splitting of the Hamiltonian in \cref{Eq:pionless_Hamiltonian} into seven terms $\{H_\gamma\}$ as in the previous subsection. To evaluate the $p=2$ Trotter error formula, we evaluate semi-norms of commutators of the forms $[H_{\gamma_\mu}, [H_{\gamma_\nu}, H_{\gamma_\xi}]]$, $[V, [H_{\gamma_\mu}, H_{\gamma_\nu}]]$, $[H_{\gamma_\mu}, [V, H_{\gamma_\xi}]]$, and $[V, [V, H_{\gamma_\mu}]]$, where from here on we assume that $\gamma_\mu,\gamma_\nu,\gamma_\xi\in\{1,\ldots, 6\}$ unless explicitly stated otherwise. Also recall that $V \coloneqq H_{7}$. The proofs are presented in the following lemmas. 

\begin{lemma}\label{Lemma:kinetic-kinetic-kinetic}
\begin{equation}
    \norm{[H_{\gamma_\mu}, [H_{\gamma_\nu}, H_{\gamma_\xi}]]}_\eta \leq 2h^3\eta,
\end{equation}
for $\nu \neq \xi$.
\end{lemma}

\begin{proof}
We proceed similarly to \cref{Lemma:kinetic-kinetic}. Direct computation yields
\begin{align}
    &[H_{\gamma_\mu}, [H_{\gamma_\nu}, H_{\gamma_\xi}]]= \bigg[H_{\gamma_\mu},\sum_\sigma\sum_{(i,j)\in \Omega_{\gamma_\nu}} \sum_{\substack{(i',j')\in \Omega_{\gamma_{\xi}} \\ (i,j)\cap(i',j')\neq\emptyset}}\bigg[ \hop_{ij}^\sigma, \hop_{i'j'}^\sigma\bigg]\bigg] \nonumber \\
    &=h^3\sum_\sigma\sum_{(i,j)\in \Omega_{\gamma_\nu}} \sum_{\substack{(i',j')\in \Omega_{\gamma_{\xi}}\\ (i,j)\cap(i',j')\neq\emptyset}}  \sum_{\substack{(i'',j'')\in \Omega_{\gamma_{\mu}}\\ (i'',j'')\cap(i',j')\cap(i,j)\neq\emptyset}}\bigg[\hop_{i''j''}^\sigma,\bigg[ \hop_{ij}^\sigma, \hop_{i'j'}^\sigma\bigg]\bigg].
\end{align}
Each term in this sum is a nested commutator of three ``connected'' hopping terms. Using a diagrammatic approach, it is straightforward to show that the hopping terms that result from this nested commutator can be split into two sums over disjoint sets of hopping terms. In particular, for each $\sigma$, the splitting goes as \cref{Fig:kinetic-kinetic-kinetic-diagram}, where we have considered a representative 2D sub-lattice.
\begin{figure}[ht]
	\centering
	\includegraphics[width=0.6\textwidth]{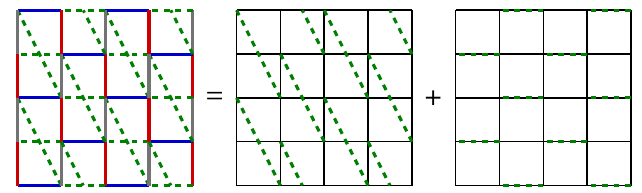}
	\caption{A decomposition of the kinetic-kinetic-kinetic commutator terms.}
	\label{Fig:kinetic-kinetic-kinetic-diagram}
\end{figure}
 The inner commutators are between hopping terms, represented by solid red and blue lines, and the outer commutator is between the result of these commutators and hopping terms represented by solid gray lines. Dashed green lines represent the result of the full nested commutator. \Cref{Lemma:fermion_norm} can then be applied to each of these sums over disjoint hopping terms, proving the result.
\end{proof}

\begin{lemma}\label{Lemma:potential-kinetic-kinetic}
\begin{equation}
    \norm{[V, [H_{\gamma_\mu}, H_{\gamma_\nu}]]}_\eta \leq 4h^2 (w_2+w_3+w_4) ,
\end{equation}
where
\begin{align}
 w_2=2|\Cpi|\left\lfloor\frac{\eta}{2}\right\rfloor, \qquad w_3=\left(2|3\Cpi+\Dpi|+|\Dpi|\right)\left\lfloor\frac{\eta}{3}\right\rfloor, \qquad
    w_4=\left(2|6\Cpi+4\Dpi|+4|\Dpi|\right)\left\lfloor\frac{\eta}{4}\right\rfloor,
\end{align}
for $\mu\neq \nu$.
\end{lemma}

\begin{proof}
Using the diagram in \cref{Fig:kinetic-kinetic-diagram} of \cref{Lemma:kinetic-kinetic}, the inner commutator amounts to a sum over two disjoint sets of hopping terms that cover the lattice.
Consequently, \cref{Lemma:kinetic-potential} can be applied directly to each term. The result follows immediately.
\end{proof}

\begin{lemma}\label{Lemma:kinetic-potential-kinetic}
\begin{align}
    \norm{[H_{\gamma_\mu},[V, H_{\gamma_\nu}]]}_\eta \leq 12h^2&(n_2+ n_3+ n_4) +12h^2(c_3+ c_4),
\end{align}
where
\begin{align}
    n_2=|\Cpi|\left\lfloor\frac{\eta}{2}\right\rfloor, \qquad\qquad n_3=|3\Cpi+\Dpi|\left\lfloor\frac{\eta}{3}\right\rfloor, \qquad\qquad
    n_4=|6\Cpi+4\Dpi|\left\lfloor\frac{\eta}{4}\right\rfloor,
\end{align}
and 
\begin{align}
    c_3=|\Dpi|\left\lfloor\frac{\eta}{3}\right\rfloor, \qquad\qquad
    c_4=|4\Dpi|\left\lfloor\frac{\eta}{4}\right\rfloor.
\end{align}
\end{lemma}
\begin{proof}
In the notation of \cref{eq:kinetic-potential-six}, we have
\begin{align}\label{eq:kinetic-potential-kinetic-init}
     \norm{[H_{\gamma_\mu},[V, H_{\gamma_\nu}]]}_\eta&=h\norm{[H_{\gamma_\mu}, A+A(i\leftrightarrow j)+A(\sigma\leftrightarrow \sigma')+A(i\leftrightarrow j, \sigma\leftrightarrow \sigma')+B + B(i\leftrightarrow j)]}_\eta,
\end{align}
recalling that each of the terms $A$ and $B$ consists of a sum of hopping terms on disjoint indices multiplied by some constants and some number operators. Applying the triangle inequality, we can consider each term separately. First,
\begin{align}
    [H_{\gamma_\mu},A]&=\Bigg[H_{\gamma_\mu},\sum_{(i,j)\in\Omega_{\gamma_\nu}}\sum_{\substack{\sigma, \sigma'\\\sigma'\neq \sigma}}\hop_{ij}^\sigma N_{\sigma'}(i)\left(\frac{\Cpi}{2}+\frac{\Dpi}{6}\sum_{\substack{\sigma''\\\sigma''\neq\sigma'\neq\sigma}}N_{\sigma''}(i)\right)\Bigg]. 
\end{align}
For the hopping terms in $H_{\gamma_\mu}$, we use the commutation relation $[\hop_{ij}, \hop_{jl}]=\hop_{il}$, yielding two terms of type $A$, as depicted in \cref{Fig:kinetic-kinetic-diagram}. For the hopping terms of species $\sigma'$ and $\sigma''$, we use the commutation relation $[\hop_{ij}, N(i)]=\hop_{ij}$. Putting these together, we obtain
\begin{align}
    [H_{\gamma_\mu},A]= 2 A \,+ &\sum_{(i,j)\in\Omega_{\gamma_\nu}}\sum_{(i,k)\in\Omega_{\gamma_\mu}}\sum_{\substack{\sigma, \sigma'\\\sigma'\neq \sigma}}\hop_{ij}^\sigma \hop_{ik}^{\sigma'}\left(\frac{\Cpi}{2}+\frac{\Dpi}{6}\sum_{\substack{\sigma''\\\sigma''\neq\sigma'\neq\sigma}}N_{\sigma''}(i)\right) \nonumber \\
    &+ \sum_{(i,j)\in\Omega_{\gamma_\nu}}\sum_{(i,k)\in\Omega_{\gamma_\mu}}\sum_{\substack{\sigma, \sigma'\\\sigma'\neq \sigma}}\hop_{ij}^\sigma N_{\sigma'}(i)\left(\frac{\Dpi}{6}\sum_{\substack{\sigma''\\\sigma''\neq\sigma'\neq\sigma}}\hop_{ik}^{\sigma''}\right),
\end{align}
where $A$ denotes a term of type $A$ (since the particular indices are irrelevant when evaluating the norm). To bound the semi-norms of the new terms, we use that 
\begin{align}
\norm{\hop_{ij}\hop_{ik}N(i)}_\eta\leq \norm{\hop_{ij}}_\eta\norm{\hop_{ik}}_\eta\norm{N(i)}_\eta\leq 1,
\end{align}
and then bound the terms for different fermion numbers per site using arguments nearly identical to those in \cref{Lemma:kinetic-potential}. We find that
\begin{align}
    \norm{[H_{\gamma_\mu},A]}_\eta\leq 3h\, (n_2+ n_3+ n_4)+h\, (c_3+ c_4),
\end{align}
where
\begin{align}
    n_2=|\Cpi|\left\lfloor\frac{\eta}{2}\right\rfloor, \qquad\qquad n_3=|3\Cpi+\Dpi|\left\lfloor\frac{\eta}{3}\right\rfloor, \qquad\qquad
    n_4=|6\Cpi+4\Dpi|\left\lfloor\frac{\eta}{4}\right\rfloor,
\end{align}
and 
\begin{align}
    c_3=|\Dpi|\left\lfloor\frac{\eta}{3}\right\rfloor, \qquad\qquad
    c_4=|4\Dpi|\left\lfloor\frac{\eta}{4}\right\rfloor.
\end{align}

Now consider the terms of type $B$. In this case,
\begin{align}
    &[H_{\gamma_\mu},B]= \Big[H_{\gamma_\mu},\sum_{(i, j)\in\Omega_{\gamma_\mu}}\sum_{\substack{\sigma, \sigma', \sigma'' \\ \sigma\neq\sigma'\neq\sigma''}}N_\sigma(i)N_{\sigma'}(i)\frac{\Dpi}{6}\hop_{ij}^{\sigma''}\Big] \\
    &\quad=2 B+\sum_{(i, j)\in\Omega_{\gamma_\nu}}\sum_{(i, k)\in\Omega_{\gamma_\nu}}\sum_{\substack{\sigma, \sigma', \sigma'' \\ \sigma\neq\sigma'\neq\sigma''}}\hop_{ik}^\sigma N_{\sigma'}(j)\frac{\Dpi}{6}\hop_{ij}^{\sigma''} +\sum_{(i, j)\in\Omega_{\gamma_\nu}}\sum_{(i, k)\in\Omega_{\gamma_\nu}}\sum_{\substack{\sigma, \sigma', \sigma'' \\ \sigma\neq\sigma'\neq\sigma''}}N_{\sigma}(i)\hop_{jk}^{\sigma'}\frac{\Dpi}{6}\hop_{ij}^{\sigma''}.
\end{align}
Again, the semi-norms of the latter two terms can be found for each fermion-number subspace using arguments identical to those of \cref{Lemma:kinetic-potential}, yielding a semi-norm equivalent to terms of type $B$. Therefore,
\begin{align}
    \norm{[H_{\gamma_\mu},B]}_\eta\leq 4 h(c_3+c_4).
\end{align}
Using the fact that there are four terms of the form $[H_{\gamma_\mu},A]$ and two terms of the form $[H_{\gamma_\mu},B]$ in \cref{eq:kinetic-potential-kinetic-init}, and applying the triangle inequality, the result follows.
\end{proof}

\begin{lemma}\label{Lemma:potential-potential-kinetic}
\begin{equation}
    \norm{[V, [V, H_{\gamma_\mu}]]}_\eta \leq 4h (q_2+q_3+q_4) +2h(q_3'+q_4'),
\end{equation}
where
\begin{align}
    &q_2=2|\Cpi|^2\left\lfloor\frac{\eta}{2}\right\rfloor, \qquad \qquad
    q_3=\left|\frac{\Cpi}{2}+\frac{\Dpi}{6}\right|\left(12 \left |\frac{\Cpi}{2}+\frac{\Dpi}{6}\right|+\left|\Dpi\right|\right)\left\lfloor\frac{\eta}{3}\right\rfloor, \nonumber \\
    &\qquad\qquad q_4=24 \left |\frac{\Cpi}{2}+\frac{\Dpi}{3}\right|\left(6 \left |\frac{\Cpi}{2}+\frac{\Dpi}{3}\right| +\left|\Dpi\right|\right)\left\lfloor\frac{\eta}{4}\right\rfloor, 
\end{align}
and 
\begin{align}
    q_3'=\left(8|\Dpi|\left|\frac{\Cpi}{2}+\frac{\Dpi}{6}\right|+\frac{2}{3}\Dpi^2\right)\left\lfloor\frac{\eta}{3}\right\rfloor, \qquad\qquad
    q_4'=8|D_\pi|\left(6\left|\frac{\Cpi}{2}+\frac{\Dpi}{3}\right|+|\Dpi|\right)\left\lfloor\frac{\eta}{4}\right\rfloor.
\end{align}
\end{lemma}
\begin{proof}
The inner commutator was already evaluated in \cref{eq:kinetic-potential-six} in \cref{Lemma:kinetic-potential}. In the notation defined there,
\begin{equation}\label{eq:kinetic-kinetic-potential-comm}
    \norm{[V, [V, H_{\gamma_\mu}]]}_\eta=h\norm{[V, A+A(i\leftrightarrow j)+A(\sigma\leftrightarrow \sigma')+A(i\leftrightarrow j, \sigma\leftrightarrow \sigma')+B + B(i\leftrightarrow j)]}_\eta.
\end{equation} 
Therefore, it suffices to bound two types of terms, $[V,A]$ and $[V,B]$. First,
\begin{equation}
\begin{aligned}
    [V,A]= 
    \Bigg[&\sum_k\sum_{\substack{\xi, \xi'\\ \xi\neq\xi'}}N_\xi(k)N_{\xi'}(k)\Bigg(\frac{\Cpi}{2}+\frac{\Dpi}{6}\sum_{\substack{\xi''\\ \xi\neq\xi'\neq\xi''}}N_{\xi''}(k)\Bigg),
    \\
    &\sum_{(i,j)\in\Omega_{\gamma_\mu}}\sum_{\substack{\sigma, \sigma'\\ \sigma\neq\sigma'}}\Delta_{ij}^\sigma N_{\sigma'}(i)\Bigg(\frac{\Cpi}{2}+\frac{\Dpi}{6}\sum_{\substack{\sigma''\\ \sigma\neq\sigma'\neq\sigma''}}N_{\sigma''}(i)\Bigg)\Bigg]  ,
    \end{aligned}
\end{equation}
where we have used $\xi,\xi',\xi''$ to label particle types in the first operator in the commutator. The semi-norms of these terms for each possible on-site fermion number can be bounded just as in \cref{Lemma:kinetic-potential}. There is no contribution to the semi-norm from states with no or a single fermion per site, so we begin with the case of two fermions per site, for which the commutator simplifies to
\begin{align}
[V,A]
&= \Bigg[\frac{\Cpi}{2}\sum_k\sum_{\substack{\xi, \xi'\\ \xi\neq\xi'}}N_\xi(k)N_{\xi'}(k),\sum_{(i,j)\in\Omega_{\gamma_\mu}}\sum_{\substack{\sigma, \sigma'\\ \sigma\neq\sigma'}}\Delta_{ij}^\sigma N_{\sigma'}(i)\frac{\Cpi}{2}\Bigg] + \text{(terms with semi-norm 0)}.
\label{Eq:V-A-Comm}
\end{align}
This is because terms proportional to $\Dpi$ have zero semi-norm in the sector with two fermions per site. This commutator is non-zero for terms where $k=i$ or $k=j$ and $\xi=\sigma$ or $\xi'=\sigma$. In addition, among two-particle states, there are two possibilities for $\sigma$ in the sum over $\sigma$ and one possibility for $\sigma'\neq\sigma$. This makes for a total of $2^3=8$ terms with non-zero fermionic semi-norm. In total, using that there are at most $\lfloor\eta/2\rfloor$ pairs of sites $(i,j)$ with fermion number two, this gives a bound on $\norm{[V,A]}_\eta$ terms with the assumption of two particles per site of
\begin{align}
q_2 \coloneqq 8\frac{|\Cpi|^2}{4}\left\lfloor\frac{\eta}{2}\right\rfloor=2|\Cpi|^2\left\lfloor\frac{\eta}{2}\right\rfloor.
\end{align}

Now consider the case of three fermions per site. Here, all terms in \cref{Eq:V-A-Comm} contribute. 
First consider the set of non-zero terms where $k=i$ or $k=j$ and $\xi=\sigma$ or $\xi'=\sigma$. Such terms have a fermionic semi-norm bounded by $\left |\Cpi/2+\Dpi/6\right|^2$. When summing over $\sigma$ for the case of three fermions per site, there are three options for $\sigma$, two options for $\sigma'\neq\sigma$, and two options for $\xi'\neq\sigma$ or $\xi\neq\sigma$ (when $\xi=\sigma$ and $\xi'=\sigma$, respectively). This makes for a total of $2\times 2\times 3\times 2\times 2=48$ non-zero terms and yields a fermionic semi-norm for all of these terms bounded by 
\begin{equation}
2\times 2\times 3\times 2\times 2 \times\left |\frac{\Cpi}{2}+\frac{\Dpi}{6}\right|^2\left\lfloor\frac{\eta}{3}\right\rfloor=
48 \left |\frac{\Cpi}{2}+\frac{\Dpi}{6}\right|^2\left\lfloor\frac{\eta}{3}\right\rfloor.
\end{equation}

With three particles per site, we also have a set of terms where $k=i$ or $k=j$ and $\xi''=\sigma$ that have a non-zero semi-norm. Here, there are three options for $\sigma$, two options for $\sigma'\neq\sigma$, two options for $\xi\neq\sigma$, and one option for $\xi'\neq\xi\neq\sigma$. 
This yields an upper bound on the semi-norm on of all such terms of
\begin{equation}
\label{eq:CDterms-3}
2\times 3\times 2\times 2 \times 1 \times  \left|\frac{\Dpi}{6}\right|\left|\frac{\Cpi}{2}+\frac{\Dpi}{6}\right|\left\lfloor\frac{\eta}{3}\right\rfloor
= 4 \left|\Dpi\right|\left|\frac{\Cpi}{2}+\frac{\Dpi}{6}\right|\left\lfloor\frac{\eta}{3}\right\rfloor.
\end{equation}

Putting these together, we define
\begin{equation}
q_3 \coloneqq 4\left|\frac{\Cpi}{2}+\frac{\Dpi}{6}\right|\left(12 \left |\frac{\Cpi}{2}+\frac{\Dpi}{6}\right|+\left|\Dpi\right|\right)\left\lfloor\frac{\eta}{3}\right\rfloor.
\end{equation}

When there are four fermions per site, the same terms are non-zero as in the case of three fermions, but now, when summing over $\sigma$, there are four options for $\sigma$ with non-zero semi-norm, three options for $\sigma'\neq\sigma$, and three options for $\xi'\neq\sigma$ or $\xi\neq\sigma$ (when $\xi=\sigma$ and $\xi'=\sigma$, respectively). Furthermore, the term proportional to $\Dpi$ comes with two more options for index $\sigma'' \neq \sigma' \neq \sigma$. These all yield a bound of
\begin{equation}
2\times 2\times 4 \times 3\times 3 \times\left |\frac{\Cpi}{2}+2\times\frac{\Dpi}{6}\right|^2\left\lfloor\frac{\eta}{4}\right\rfloor= 144 \left |\frac{\Cpi}{2}+\frac{\Dpi}{3}\right|^2\left\lfloor\frac{\eta}{4}\right\rfloor.
\end{equation}
With four particles per site, we also have a set of terms where $k=i$ or $k=j$ and $\xi''=\sigma$ that have a non-zero semi-norm. Here, there are four options for $\sigma$, three options for $\sigma'\neq\sigma$, three options for $\xi\neq\sigma$, and two options for $\xi'\neq\xi\neq\sigma$. We also still have two options for $\sigma''\neq\sigma'\neq\sigma$. This yields an upper bound on the semi-norm on all such terms of
\begin{equation}
\label{eq:D2terms}
2\times 4\times 3\times 3 \times 2 \times \left|\frac{\Dpi}{6}\right|\left|\frac{\Cpi}{2}+2\times\frac{\Dpi}{6}\right|\left\lfloor\frac{\eta}{4}\right\rfloor
=24\left|\Dpi\right|\left|\frac{\Cpi}{2}+\frac{\Dpi}{3}\right|\left\lfloor\frac{\eta}{4}\right\rfloor.
\end{equation}
Together, these terms are bounded by
\begin{equation}
q_4 \coloneqq 24 \left |\frac{\Cpi}{2}+\frac{\Dpi}{3}\right|\left(6 \left |\frac{\Cpi}{2}+\frac{\Dpi}{3}\right| +\left|\Dpi\right|\right)\left\lfloor\frac{\eta}{4}\right\rfloor.
\end{equation}

Now consider
\begin{align}
&[V,B]= \Bigg[\sum_k\sum_{\substack{\xi, \xi'\\ \xi\neq\xi'}}N_\xi(k)N_{\xi'}(k)\Bigg(\frac{\Cpi}{2}+\frac{\Dpi}{6}\sum_{\substack{\xi''\\ \xi\neq\xi'\neq\xi''}}N_{\xi''}(k)\Bigg),\sum_{(i, j)}\sum_{\substack{\sigma, \sigma', \sigma'' \\ \sigma\neq\sigma'\neq\sigma''}}N_\sigma(i)N_{\sigma'}(i)\frac{\Dpi}{6}\hop_{ij}^{\sigma''}\Bigg].
\end{align}
The semi-norm of this commutator is zero for states with less than three fermions per site, so we begin with the case of three fermions per site. Non-zero terms occur when $k=i$ or $k=j$ and $\xi=\sigma''$ or $\xi'=\sigma''$.
Similar to the case of the $[V,A]$ commutator, there are three options for $\sigma''$ with non-zero semi-norm, two options for $\sigma$, one option for $\sigma'$, and two options for $\xi'\neq\sigma''$ or $\xi\neq\sigma''$ (when $\xi=\sigma''$ and $\xi'=\sigma''$, respectively). 
This yields a bound on the fermionic semi-norm over states with three fermion per site of 
\begin{equation}
2\times 2\times3\times 2\times 2\times\left|\frac{\Dpi}{6}\right|\left|\frac{\Cpi}{2}+\frac{\Dpi}{6}\right|\left\lfloor\frac{\eta}{3}\right\rfloor =
8|\Dpi|\left|\frac{\Cpi}{2}+\frac{\Dpi}{6}\right|\left\lfloor\frac{\eta}{3}\right\rfloor.
\end{equation}
Furthermore, if $\xi''=\sigma''$, we get a bound of
\begin{equation}
2 \times 3 \times 2 \times 2 \times \left(\frac{\Dpi}{6}\right)^2\left\lfloor\frac{\eta}{3}\right\rfloor=\frac{2}{3}\Dpi^2\left\lfloor\frac{\eta}{3}\right\rfloor.
\end{equation}
Together, this yields a bound of
\begin{equation}
q_3' \coloneqq \left(8|\Dpi|\left|\frac{\Cpi}{2}+\frac{\Dpi}{6}\right|+\frac{2}{3}\Dpi^2\right)\left\lfloor\frac{\eta}{3}\right\rfloor.
\end{equation}
Similar reasoning yields a bound for the case of four fermions per site of
\begin{equation}
2\times 2\times 4\times 3\times 3\times 2\times\left|\frac{\Dpi}{6}\right|\left|\frac{\Cpi}{2}+2\times\frac{\Dpi}{6}\right|\left\lfloor\frac{\eta}{4}\right\rfloor=
48|\Dpi|\left|\frac{\Cpi}{2}+\frac{\Dpi}{3}\right|\left\lfloor\frac{\eta}{4}\right\rfloor
\end{equation}
for $\xi=\sigma''$ or $\xi'=\sigma''$, and
\begin{equation}
2 \times 4 \times 3 \times 2 \times 3 \times 2 \times \left(\frac{\Dpi}{6}\right)^2\left\lfloor\frac{\eta}{4}\right\rfloor=8\Dpi^2\left\lfloor\frac{\eta}{4}\right\rfloor
\end{equation}
for $\xi''=\sigma''$. Putting these together yields a bound of 
\begin{equation}
q_4' \coloneqq 8|D_\pi|\left(6\left|\frac{\Cpi}{2}+\frac{\Dpi}{3}\right|+|\Dpi|\right)\left\lfloor\frac{\eta}{4}\right\rfloor.
\end{equation}
Finally, using the fact that there are four commutators of type $[V,A]$ and two commutators of type $[V,B]$ in \cref{eq:kinetic-kinetic-potential-comm}, and summing over the different fermion-per-site sectors, the result follows.
\end{proof}

Using \crefrange{Lemma:kinetic-kinetic-kinetic}{Lemma:potential-potential-kinetic} and \cref{Eq:Second-Order-PF-Def}, we prove \cref{thm:pionless_trotter_error_p2}, which is restated here for convenience.

\begin{theorem}[$p=2$ Pionless-EFT Trotter Error (\cref{thm:pionless_trotter_error_p2} from the main text)]\label{thm:pionless_trotter_error_p2_copy}
For the pionless-EFT Hamiltonian described in \cref{Sec:Pionless_EFT_Hamiltonian},
\begin{align}
    \norm{e^{-iH_\canpi t} - \calP_2^{(\canpi)}(t)}_\eta &\leq \frac{t^3}{12} \bigg( 125h^3\eta + 216h^2\left( (n_2+n_3+n_4) + c_3+c_4\right) \\ &+60h^2(w_1+w_2+w_3) + 12h\left(2 (q_2+q_3+q_4) + q_3'+q_4' \right) \bigg),
\end{align}
where $h=\frac{1}{2M{a_L^2}}$ is the coefficient of the hopping term, and
\begin{align}
    n_2 &=|\Cpi|\left\lfloor\frac{\eta}{2}\right\rfloor, \quad n_3 =|3\Cpi+\Dpi|\left\lfloor\frac{\eta}{3}\right\rfloor, \quad n_4 =|6\Cpi+4\Dpi|\left\lfloor\frac{\eta}{4}\right\rfloor, \\
    c_3 &= |\Dpi|\left\lfloor\frac{\eta}{3}\right\rfloor, \quad c_4 =4|\Dpi|\left\lfloor\frac{\eta}{4}\right\rfloor, \\
    w_2 &=2|\Cpi|\left\lfloor\frac{\eta}{2}\right\rfloor, \quad  
    w_3 =\left(|\Dpi|+ 2|3\Cpi+\Dpi|\right)\left\lfloor\frac{\eta}{3}\right\rfloor, \quad w_4 =\left(4|\Dpi|+ 2|6\Cpi+4\Dpi|\right)\left\lfloor\frac{\eta}{4}\right\rfloor, \\
    q_2&=2|\Cpi|^2\left\lfloor\frac{\eta}{2}\right\rfloor, \quad
    q_3=4\left|\frac{\Cpi}{2}+\frac{\Dpi}{6}\right|\left(12 \left |\frac{\Cpi}{2}+\frac{\Dpi}{6}\right|+\left|\Dpi\right|\right)\left\lfloor\frac{\eta}{3}\right\rfloor, \\
    q_4&=24 \left |\frac{\Cpi}{2}+\frac{\Dpi}{3}\right|\left(6 \left |\frac{\Cpi}{2}+\frac{\Dpi}{3}\right| +\left|\Dpi\right|\right)\left\lfloor\frac{\eta}{4}\right\rfloor, \\
    q_3'&=\left(8|\Dpi|\left|\frac{\Cpi}{2}+\frac{\Dpi}{6}\right|+\frac{2}{3}\Dpi^2\right)\left\lfloor\frac{\eta}{3}\right\rfloor, \qquad
    q_4'=8|D_\pi|\left(6\left|\frac{\Cpi}{2}+\frac{\Dpi}{3}\right|+|\Dpi|\right)\left\lfloor\frac{\eta}{4}\right\rfloor,
\end{align}
Here, $\Cpi$ and $\Dpi$ are the low-energy constants of two- and three-nucleon contact terms. 
\end{theorem}

\begin{proof}
Directly applying the $p=2$ Trotter error bound in \cref{Eq:p-2-commutator-bound}, the proof follows from counting the number of terms of each type considered in the above lemmas and a single application of the triangle inequality. In particular, from the first term in the bound in \cref{Eq:p-2-commutator-bound}, we obtain $55$ commutators of the form $[H_{\gamma_\mu}, [H_{\gamma_\nu}, H_{\gamma_\xi}]]$ bounded as in \cref{Lemma:kinetic-kinetic-kinetic}. We also obtain $15$ commutators of the form $[H_{\gamma_\mu},[V,H_{\gamma_\nu}]]$ bounded as in \cref{Lemma:kinetic-potential-kinetic}, $15$ commutators of the form $[V,[H_{\gamma_\mu},H_{\gamma_\nu}]]$ bounded as in \cref{Lemma:potential-kinetic-kinetic}, and $6$ commutators of the form $[V,[V,H_{\gamma_\mu}]]$ bounded as in \cref{Lemma:potential-potential-kinetic}. From the second term of the bound in \cref{Eq:p-2-commutator-bound}, we obtain $15$ commutators of the form $[H_{\gamma_\nu}, [H_{\gamma_\nu}, H_{\gamma_\xi}]]$ bounded as in \cref{Lemma:kinetic-kinetic-kinetic} and $6$ commutators of the form $[H_{\gamma_\mu},[H_{\gamma_\mu},V]]$ bounded as in \cref{Lemma:kinetic-potential-kinetic} (using the Jacobi identity and the fact that $[H_{\gamma_\mu},H_{\gamma_\mu}]=0$). Summing all these contributions using the triangle inequality yields the theorem statement.
\end{proof}

\subsection{One-Pion-Exchange Bounds} \label{Sec:Trotter_Error_OPE_p1}
In this Appendix, we derive the $p=1$ Trotter error bounds for the OPE EFT Hamiltonian.

\begin{theorem}[One-Pion-Exchange Trotter Error Bound]
For the time evolution of the OPE EFT with a first-order product formula, $\norm{ \calP_1^{({\rm OPE})}(t) - e^{-itH_{\rm OPE}} }_\eta \leq \frac{t^2}{2}\zeta$, where $\zeta$ is the sum of the bounds which are reported in the Lemmas noted in \cref{Table:Table_of_Commutators_OPE}. 
\end{theorem}

\begin{table}[ht]
    \centering
   \resizebox{\textwidth}{!}{ \begin{tabular}{c|c|c|c|c|c}
       & $H_\mathrm{free}$  & $H_{C}$  & $H_{C_{I^2}}$ & $H_{\rm LR}(0)$  & $H_{\rm LR}$   \\ \hline 
    $H_\mathrm{free}$  &  \cref{Lemma:kinetic-kinetic}   & \cref{Lemma:Hfree_HC1_Commutator}   & 
 \cref{Lemma:Hfree_HC2_Commutator} &  \cref{Lemma:Hfree_HKR(0)_Commutator}  & \cref{Lemma:HFree_HLR_Commutator} \\ \hline
    $H_{C}$   &  -  & \cref{Lemma:HC1_HC1_Commutator}
    &  \cref{Lemma:HC1_HC2_Commutator}  & \cref{Lemma:HC1_HKR(0)_Commutator} & \cref{Lemma:HC1_HLR_Commutator}    \\ \hline
    $H_{C_{I^2}}$   &  -  & -  & \cref{Lemma:HC2_HC2_Commutator}  &  \cref{Lemma:HC2_HKR(0)_Commutator}&\cref{Lemma:HC2_HLR_Commutator}   \\ \hline     
    $H_{\rm LR}(0)$   &  -  & -  & -  & \cref{Lemma:HLR(0)_HKR(0)_Commutator} & \cref{Lemma:HLR(0)_HKR(r)_Commutator}  \\ \hline
    $H_{\rm LR}$   &  -  & -  & -  & - & \cref{Lemma:HLR_HLR_Commutator} ,  \cref{Lemma:HLR_HLR_Commutator_2}
    \end{tabular}}
    \caption{Commutators for the OPE EFT Hamiltonian and the lemma in which a bound on the value of the commutator is computed. 
    }
    \label{Table:Table_of_Commutators_OPE}
\end{table}

\begin{proof}
We use the expression for Trotter error in \cref{eq:P-first-order}, restated here for convenience:
\begin{align}
    \norm{e^{iH_{\rm OPE}t} - \calP^{({\rm OPE})}_1(t)}_\eta \leq \frac{t^2}{2}\sum_{\gamma_1=1}^{\Gamma} \norm{ \sum_{\gamma_2 \geq \gamma_1+1} \left[H_{\gamma_2}, H_{\gamma_1}  \right] }_\eta.
\end{align}
To proceed, we assign the $\gamma_i$ labels to the local terms in the Hamiltonian for their use in the commutator bound above.
The labeling of the kinetic and contact interaction terms has been discussed previously. However, the contact and long-range terms are new and are discussed below.

For $H_C$ given in \cref{eq:HC}, all the summands commute, so they can be implemented with no error with a fixed circuit, so we consider $H_C$ to consist of a single term whose time evolution involves zero Trotter error. 

For $H_{C_{I^2}}$ given in \cref{Eq:explicit-HCI2}, there are 11 types of Hermitian terms, and the time evolution of each can be implemented via straightforward circuits, as discussed in \cref{Sec:OPE-circuits}.

Finally, as per \cref{Sec:Long_Ranged_Pauli}, when considering long-range terms in $H_{\rm LR}(r)$ (i.e., $r=|\bm{x}-\bm{y}|>0$), we must implement terms of the form given in \cref{Eq:rhoSI-rhoSI}.
Therefore, the most general possible term has two fermionic creation and two fermionic annihilation operators, which has $(2^2)^4=256$ possible combinations.
This puts an upper bound on the number of terms to be implemented for fixed $\bm{x}$ and $\bm{y}$ at a given $r$.
We denote the number of lattice sites at distance $r$ of a given lattice site as $q(r)$. For instance, for $H_{\rm LR}(\sqrt{2}a_L)$, part of the decomposition is shown in \cref{Fig:HLR_Decomposition}, while the complete decomposition yields $q(\sqrt{2}a_L)=12$ disjoint sets.

\begin{table}[h!]
\begin{center}
\resizebox{\textwidth}{!}{\begin{tabular}{ c |c |c }
Hamiltonian Term & Set of Terms  & Number of Layers Upper Bound \\
\hline
$\Hfree$ & $\Gamma_{\rm free}$ & 6 \\ 
$H_C$ & $\Gamma_{C}$ & 1 \\  
 $H_{C_{I^2}}$ & $\Gamma_{C_{I^2}}$ & 12  \\
 $H_{\rm LR}(0)$ & $\Gamma_{LR0}$ & 256\\
 $H_{\rm LR}(r)$ & $\Gamma_{\rm LR}$ & $256\,q(r)$
\end{tabular}}
\end{center}
\caption{Decomposition of $H_{\rm OPE}$ in \cref{Eq:H-LR} into layers for the application of the first-order Trotter error bound. Here, $q(r)$ is the number of lattice points at distance $r$ of any other given lattice point.}
\label{Table:Table_of_Number_OPE_Terms}
\end{table}

 \begin{figure}
     \centering
     \begin{subfigure}[]
         \centering
         \includegraphics[width=0.25\textwidth]{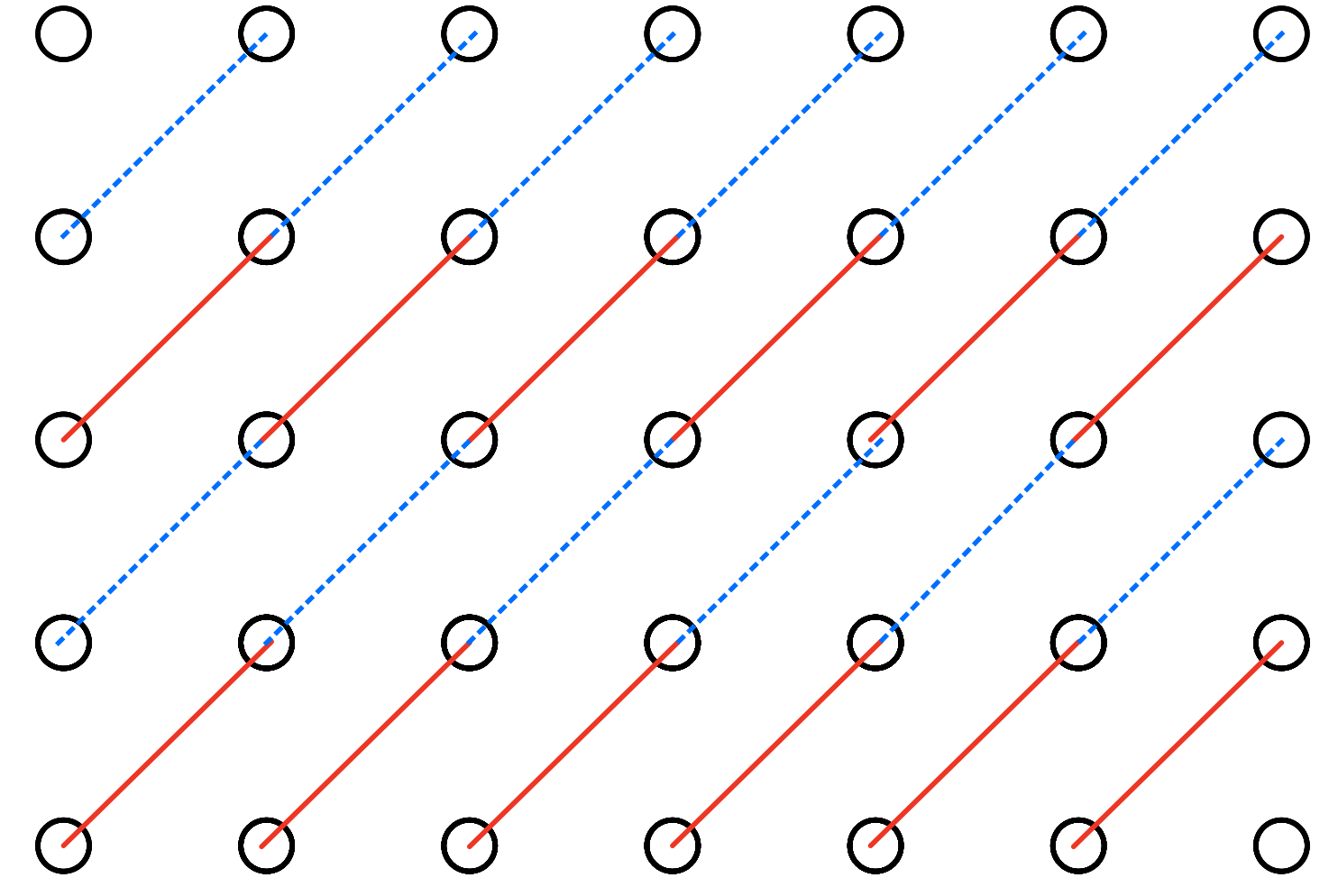}
         \label{fig:y equals x}
     \end{subfigure}
     \begin{subfigure}[]
         \centering
         \includegraphics[width=0.25\textwidth]{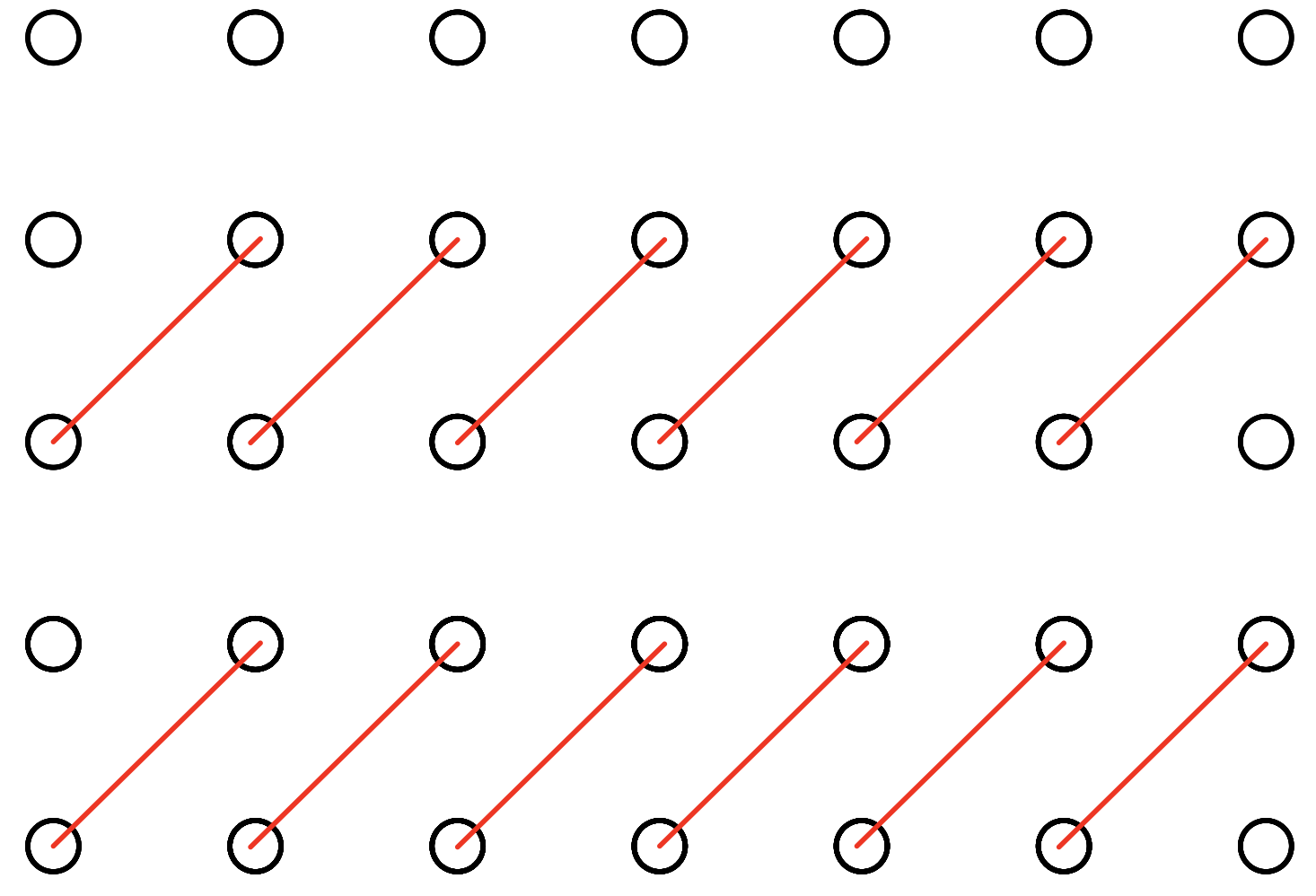}
         \label{fig:three sin x}
     \end{subfigure}
     \begin{subfigure}[]
         \centering
         \includegraphics[width=0.25\textwidth]{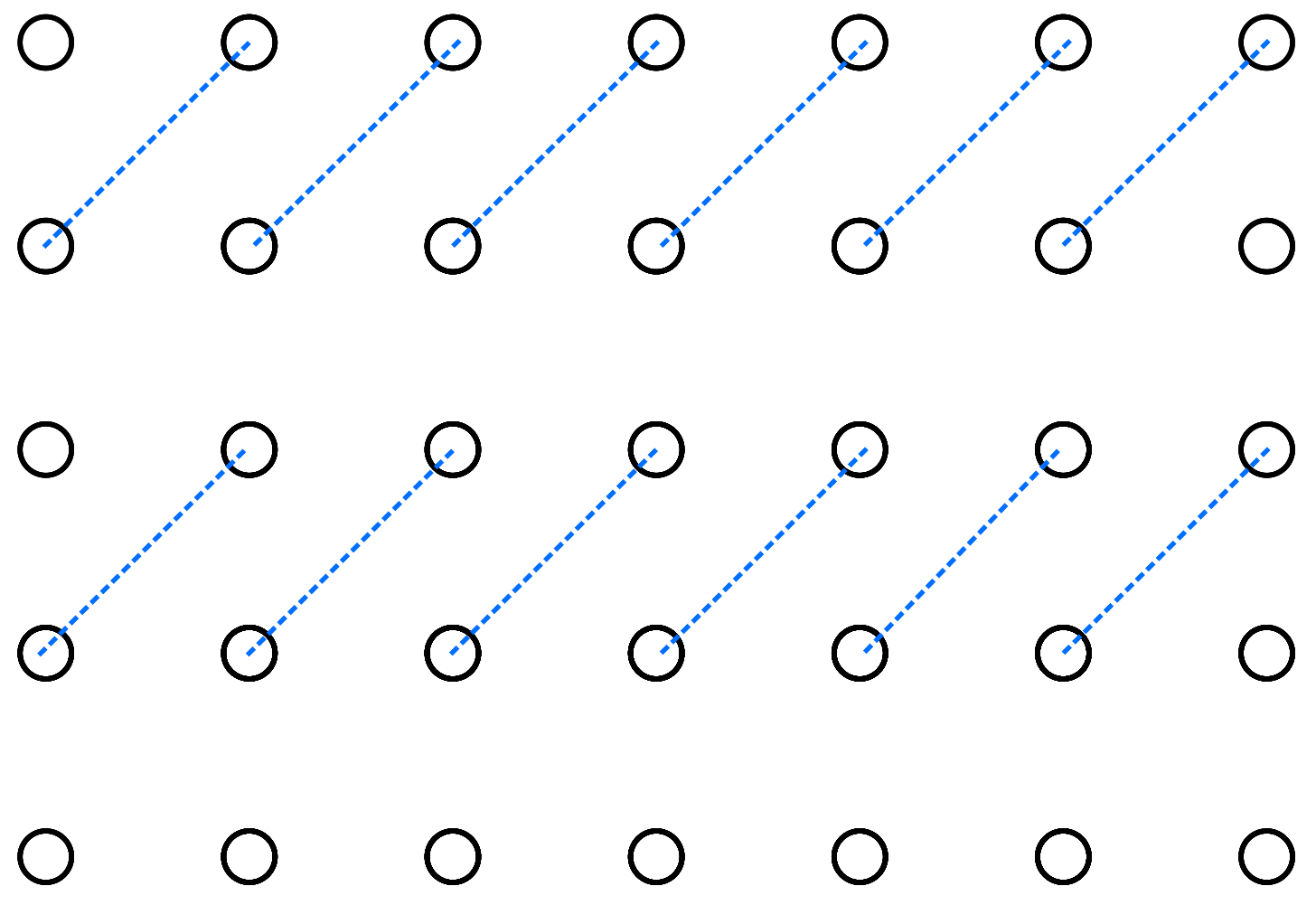}
         \label{fig:five over x}
     \end{subfigure}
        \caption{(a) All diagonal north-west facing terms in $H_{\rm LR}$. (b) and (c) Decomposition into two sets of terms, such that the terms within each set have zero support on other terms in that same set. Other examples of interaction types and their decomposition to disjoint sets are shown in \cref{Fig:LR_Layers} of the main text.}
        \label{Fig:HLR_Decomposition}
\end{figure}

With this identification of the number of Hamiltonian sets, as summarized in \cref{Table:Table_of_Number_OPE_Terms}, we bound the commutator as
\begin{align}
    &\sum_{\gamma_1=1}^{\Gamma} \norm{ \sum_{\gamma_2 \geq \gamma_1+1} \left[H_{\gamma_1}, H_{\gamma_2}  \right] } 
    \leq \sum_{\gamma_1}\bigg(  \norm{\sum_{\gamma_2\geq \gamma_1+1}[T_{\gamma_1},T_{\gamma_2}  ] }_\eta + \norm{[T_{\gamma_1},H_{C}  ] }_\eta + \norm{[T_{\gamma_1},H_{C_{I^2}}  ] }_\eta\nonumber\\ 
    &\hspace{4em}+  \norm{\sum_{\gamma_2}[T_{\gamma_1},H_{\rm LR}^{(\gamma_2)}(0)  ]}_\eta + \norm{\sum_{\gamma_2}[T_{\gamma_1},H_{\rm LR}^{(\gamma_2)}(a_L)  ]}_\eta\nonumber \\
    &\hspace{4em}+\norm{\sum_{\gamma_2}[T_{\gamma_1},H_{\rm LR}^{(\gamma_2)}(\sqrt{2}a_L)  ] }_\eta + \cdots +\norm{\sum_{\gamma_2}[T_{\gamma_1},H_{\rm LR}^{(\gamma_2)}(\ell a_L)  ] }_\eta \bigg) \nonumber\\
    &\hspace{4em}+ \bigg( \norm{\sum_{
    \gamma_2}[H_{C},  H^{(\gamma_2)}_{C_{I^2}}  ] }_\eta +\norm{ \sum_{\gamma_2}[H_{C}, H_{\rm LR}^{(\gamma_2)}(0)  ]}_\eta+ \cdots
    + \norm{\sum_{
    \gamma_2}[H_{C},H_{\rm LR}^{(\gamma_2)}(\ell a_L) }_\eta \bigg) \nonumber\\
    &\hspace{4em}+ \sum_{\gamma_1}\bigg(\norm{\sum_{\gamma_2}[H_{C_{I^2}}^{(\gamma_1)},H_{\rm LR}^{(\gamma_2)}(0)  ]}_\eta+ \dots + \norm{\sum_{\gamma_2}[H_{C_{I^2}}^{(\gamma_1)},H_{\rm LR}^{(\gamma_2)}(\ell a_L) }_\eta \bigg)\nonumber\\
    &\hspace{4em}+ \sum_{\gamma_1}\bigg(\norm{\sum_{\gamma_2\geq \gamma_1+1}[H_{\rm LR}^{(\gamma_1)}(0),H_{\rm LR}^{(\gamma_2)}(0)  ]}_\eta + \dots + \norm{\sum_{
    \gamma_2}[H_{\rm LR}^{(\gamma_1)}(0),H_{\rm LR}^{(\gamma_2)}(\ell a_L)  ]}_\eta \bigg) \nonumber\\
    & \hspace{4em}\quad \quad \quad \vdots \nonumber\\
    &\hspace{4em}+ \sum_{\gamma_1} \norm{\sum_{\gamma_2\geq \gamma_1+1}[H_{\rm LR}^{(\gamma_1)}(\ell a_L ),H_{\rm LR}^{(\gamma_2)}(\ell a_L )  ]}_\eta,
\end{align}
where $\ell$ is the cutoff length of the long-range interaction.
The different types of commutators appearing in this expression, and the lemmas that bound them, are summarized in \cref{Table:Table_of_Commutators_OPE}. 
\end{proof}

To prove these bounds on the commutators, we make extensive use of \cref{Theorem:NPFO_Norm,Lemma:fermion_norm}, which bound the fermionic semi-norms of NPFOs in terms of the number of fermions rather than the number of fermionic modes.

\begin{lemma} \label{Lemma:Hfree_HC1_Commutator}
\begin{align}
    \norm{ [H_\mathrm{free},H_{C}] }_\eta  \leq 18h|C|\eta.
\end{align}
\end{lemma}
\begin{proof}
Starting with \cref{Eq:H_C_OPE_Paulis}, which states that
\begin{align}
    H_{C} = \frac{C}{2}  \sum_k \sum_{\sigma,\sigma'} N_\sigma(k) N_{\sigma'}(k),
\end{align}
we consider the commutator
\begin{align}
    \bigg[ \sum_\sigma \sum_{\langle i,j \rangle} \hop_{ij}^\sigma,  \sum_{\sigma',\sigma''}\sum_{k} N_{\sigma'}(k) N_{\sigma''}(k)\bigg] 
    &= \sum_{\sigma,\sigma'} \sum_{\langle i,j \rangle,k} \bigg(N_{\sigma'}(k) [\hop_{ij}^\sigma, N_{\sigma}(k)  ] + [\hop_{ij}^\sigma, N_{\sigma}(k)] N_{\sigma'}(k)\bigg).
\end{align}
Using that $[N_{\sigma}(i), \hop_{ij}^\sigma) ] = \hop_{ij}^\sigma$, this can be written as
\begin{align}
    & -\sum_{\sigma,\sigma'} \sum_{\langle i,j \rangle} \bigg( N_{\sigma'}(i)\hop_{ij}^\sigma +  N_{\sigma'}(j)\hop_{ij}^\sigma + \hop_{ij}^\sigma N_{\sigma'}(i)+\hop_{ij}^\sigma N_{\sigma'}(j)\bigg) \nonumber\\ 
    &=-\sum_{\sigma\neq\sigma'} \sum_{\langle i,j \rangle} \bigg(N_{\sigma'}(i)\hop_{ij}^\sigma +  N_{\sigma'}(j)\hop_{ij}^\sigma + \hop_{ij}^\sigma N_{\sigma'}(i)+\hop_{ij}^\sigma N_{\sigma'}(j)\bigg) -2\sum_\sigma\sum_{\langle i,j \rangle}\hop_{ij}^\sigma.
\end{align}
In the last line, we used the identity $\hop_{ij}^\sigma N_\sigma(i) = \adag_\sigma(j)a_\sigma(i)$.
Consequently,
\begin{align}
    \norm{ [H_\mathrm{free},H_{C}] }_\eta 
    &\leq \frac{h|C|}{2}\norm{\sum_{\sigma\neq\sigma'} \sum_{\langle i,j \rangle} \left(N_{\sigma'}(i)\hop_{ij}^\sigma +  N_{\sigma'}(j)\hop_{ij}^\sigma + \hop_{ij}^\sigma N_{\sigma'}(i)+\hop_{ij}^\sigma N_{\sigma'}(j)\right)}_\eta\nonumber \\
    &\hspace{8.5cm}+\frac{h|C|}{2}\norm{2\sum_\sigma \sum_{\langle i,j \rangle} \hop_{ij}^\sigma}_\eta \nonumber \\
    &\leq 2h|C|\sum_{\gamma_k=1}^6\norm{\sum_{\sigma\neq\sigma'} \sum_{\langle i,j \rangle\in \Omega_{\gamma_k}} N_{\sigma'}(i)\hop_{ij}^\sigma}_\eta+h|C|\sum_{\gamma_k=1}^6\norm{\sum_\sigma \sum_{\langle i,j \rangle\in \Omega_{\gamma_k}} \hop_{ij}^\sigma}_\eta \nonumber \\
    &\leq 18h|C|\eta.
    \label{Eq:C1_Free_FSNorm}
\end{align}
Here, $\Omega_{\gamma_k}$ are one of the 6 disjoint sets of kinetic terms as explained in \cref{Sec:p=1_Pionless_Trotter_Error}.
\end{proof}

Note that above and in the following lemmas, we loosely bound the semi-norm of the product of fermionic operators over disjoint sets by $\eta$. A more fine-grained approach would make considerations similar to those presented in \cref{App:Pionless-Bounds} so that, for example, the first term in the second line of \cref{Eq:C1_Free_FSNorm} will be bounded as $12h|C| \left\lfloor\eta/2\right\rfloor$. However, keeping track of such distinctions will prove difficult in later cases, so we simplify the analysis at the cost of slightly worse bounds.

\begin{lemma} \label{Lemma:HC1_HC1_Commutator}
\begin{align}
    \sum_{\gamma_1}\norm{ \Big[H^{(\gamma_1)}_{C},\sum_{\gamma_2=\gamma_1+1}H^{(\gamma_2)}_{C}\Big] }_\eta  =0.
\end{align}
\end{lemma}

\begin{proof}
    All commutators that arise take the form
    \begin{align}
        \Big[N_\sigma(i) N_{\sigma'}(i), N_{\sigma''}(j) N_{\sigma'''}(j)\Big] =0,
    \end{align}
    so the total commutator vanishes.
\end{proof}

\begin{lemma} \label{Lemma:Hfree_HC2_Commutator}
\begin{align}
    \norm{ \Big[H_\mathrm{free},\sum_{\gamma_2}H_{C_{I^2}}^{(\gamma_2)}\Big] }_\eta \leq
     528\,h|C_{I^2}|\eta.
\end{align}
\end{lemma}

\begin{proof}
Starting with \cref{Eq:explicit-HCI2}, which states that
\begin{align}
&H_{C_{I^2}} = 
\frac{C_{I^2}}{2}\sum_{\bm{x}} :\left\{ N_{\uparrow p}^2+N_{\downarrow p}^2+N_{\uparrow n}^2+N_{\downarrow n}^2-6N_{\uparrow p}N_{\uparrow n}+2N_{\uparrow p}N_{\downarrow p}-2N_{\uparrow p}N_{\downarrow n}-2N_{\downarrow p}N_{\uparrow n}\right . \nonumber\\
& \hspace{6.0 cm} \left . +2N_{\uparrow n}N_{\downarrow n}-6N_{\downarrow p}N_{\downarrow n} -\, 4\left(a^\dagger_{\uparrow p}a_{\downarrow p}a^\dagger_{\downarrow n}a_{\uparrow n}+{\rm h.c.}\right) \right\}: \,,
\end{align}
we consider the commutator
\begin{align}
\sum_\sigma \sum_k \sum_{\langle i,j \rangle}&\bigg(\left[\hop^\sigma_{ij},N_{\uparrow p}^2(k)+N_{\downarrow p}^2(k)+N_{\uparrow n}^2(k)+N_{\downarrow n}^2(k)  \right] \nonumber \\
    &+\Big[\hop^\sigma_{ij},-6N_{\uparrow p}(k)N_{\uparrow n}(k)+2N_{\uparrow p}(k)N_{\downarrow p}(k)-2N_{\uparrow p}(k)N_{\downarrow n}(k)-2N_{\downarrow p}(k)N_{\uparrow n}(k) \nonumber\\
      &\left . +2N_{\uparrow n}(k)N_{\downarrow n}(k)-6N_{\downarrow p}(k)N_{\downarrow n}(k) \right]+4\left[\hop^\sigma_{ij},:\left(a^\dagger_{\uparrow p}(k)a_{\downarrow p}(k)a^\dagger_{\downarrow n}(k)a_{\uparrow n}(k)+{\rm h.c.}\right): \right]
    \bigg).
\end{align}

First, note that one can either have $k=i$ or $k=j$. Then, there are 4 possible commutators of the type $[\Delta^\sigma,N_\sigma^2]$, and each commutator generates a term of the form $N_\sigma\Delta^\sigma + \Delta^\sigma N_\sigma$, which itself generates 4 NPFOs.
There are at most 10 commutators of the type $[\Delta^{\uparrow p},N_{\uparrow p}N_{\sigma'}]$ for $\sigma' \neq \, \uparrow p$, which 
generates a term of the form $\Delta^{\uparrow p}N_{\sigma'}$, which itself consists of 2 NPFOs.
Similarly, for each $\sigma = \, \downarrow p,\uparrow n$, and $\downarrow n$, at most $20$ NPFOs are generated.
Finally, there are 4 commutators of the form $[\Delta^\sigma,\adag_{\uparrow p} \adag_{\downarrow n} a_{\uparrow n} a_{\downarrow p} ]$ (and 4 commutators from Hermitian-conjugate term). Each $\sigma$ coincides with one of the species indices in the four-fermion operator to give non-zero commutation, and there are four possible $\sigma$ values. Each such commutator generates an operator semi-norm $1$.\footnote{  
Consider, for example, the commutator $[a^\dagger b+b^\dagger a , c^\dagger d^\dagger e a]=[a^\dagger b , c^\dagger d^\dagger e a]+[b^\dagger a , c^\dagger d^\dagger e a]$, where for simplicity, we have denoted distinct fermionic operators with different letters. The second commutator gives zero since it involves an $a^2$ operator, while the first commutator gives
\begin{align*}
[a^\dagger b , c^\dagger d^\dagger e a]&=a^\dagger b  c^\dagger d^\dagger e a - c^\dagger d^\dagger e a a^\dagger b
=a^\dagger c^\dagger d^\dagger b e a - c^\dagger d^\dagger e b a a^\dagger \\
&=c^\dagger d^\dagger b e a^\dagger a - c^\dagger d^\dagger e b (-a^\dagger a+1)
= -c^\dagger d^\dagger e b,
\end{align*}
which has semi-norm of at most 1. The other three possible types of commutators similarly have an operator semi-norm of at most 1.}
Finally, to ensure that no overlapping spatial lattice sites are present, we break the kinetic hopping terms into 6 disjoint sets in the typical way. 
Using this information and applying \cref{Theorem:NPFO_Norm}, we find
\begin{align}
\norm{ \Big[H_\mathrm{free},\sum_{\gamma_2}H_{C_{I^2}}^{(\gamma_2)}\Big] }_\eta &\leq \frac{|C_{I^2}|h}{2} \times 6\times (32\eta + 80\eta  +
64  \eta)
= 
528\,h|C_{I^2}|\eta.
\end{align}
\end{proof}

\begin{lemma}\label{Lemma:HC1_HC2_Commutator}
\begin{align}
   \norm{ \Big[H_{C}, \sum_{\gamma_2}  H_{C_{I^2}}^{(\gamma_2)}\Big] }_\eta =  0.
\end{align}
\end{lemma}
\begin{proof}
Dividing $H_C$ and $H_{C_{I^2}}$ into the subterms acting on individual lattice sites, we have
\begin{align}
    \Big[H_{C}, \sum_{\gamma_2}  H_{C_{I^2}}^{(\gamma_2)}\Big]
    &=\sum_{i}\Big[H_{C}(i),H_{C_{I^2}}(i)\Big] \\
    &= \sum_i \sum_{\sigma,\sigma'} \Big[N_\sigma(i)N_{\sigma'}(i),H_{C_{I^2}}(i)\Big] \\
    &= \sum_i \sum_{\sigma} N_\sigma(i)\Big[\sum_{\sigma'}N_{\sigma'}(i),H_{C_{I^2}}(i)\Big] + \Big[\sum_{\sigma}N_{\sigma}(i),H_{C_{I^2}}(i)\Big]\sum_{\sigma'}N_{\sigma'},
\end{align}
where the first equality arises from the fact that the only potentially non-zero commutators are among the terms on the same sites. Here, $i$ refers to the qubit index of site $\bm{x}$. Now,  since $H_{C_{I^2}}(i)$ is a number-preserving operator on each spatial lattice site separately, $[\sum_{\sigma}N_{\sigma}(i),H_{C_{I^2}}(i)]=0$. Thus the entire commutator is zero.
\end{proof}

\begin{lemma}\label{Lemma:HC2_HC2_Commutator}
\begin{align}
    \sum_{\gamma_1}\norm{ \Big[H^{(\gamma_1)}_{C_{I^2}},\sum_{\gamma_2=\gamma_1+1}H^{(\gamma_2)}_{C_{I^2}}\Big] }_\eta  \leq 
   60\, C^2_{I^2}
    \eta.
\end{align}
\end{lemma}
\begin{proof}
All the number-operator terms commute, so we are left with only commutators of the type $[N_\sigma^2,\adag_{\uparrow p} \adag_{\downarrow n} a_{\uparrow p} a_{\downarrow p} ]$ or $[N_\sigma N_{\sigma'},\adag_{\uparrow p} \adag_{\downarrow n} a_{\uparrow n} a_{\downarrow p} ]$ with $\sigma \neq \sigma'$ (as well as those with Hermitian conjugate terms). Then, it is easy to show that i) for $[N_\sigma^2,\adag_{\uparrow p} \adag_{\downarrow n} a_{\uparrow n} a_{\downarrow p} ]$, one gets operators of at most semi-norm 3 if $\sigma=\,\uparrow p$ or $\sigma = \, \downarrow n $, and of at most semi-norm 1 if $\sigma=\,\uparrow n$ or $\sigma =\, \downarrow p $, ii) for $[N_\sigma N_{\sigma'},\adag_{\uparrow p} \adag_{\downarrow n} a_{\uparrow n} a_{\downarrow p} ]$, one gets operators of at most semi-norm 3 if $\sigma=\,\uparrow p,~\sigma'=\,\downarrow n$ or $\sigma =\, \downarrow n,~\sigma'=\,\uparrow p$, and of at most semi-norm 1 if $\sigma=\,\uparrow n,~\sigma'=\,\downarrow p$ or $\sigma = \,\downarrow p,~\sigma' =\, \uparrow n $.\footnote{The derivation is similar to that presented in the previous footnote, but with the realization that $Na^\dagger=-a^\dagger N +a^\dagger$ and $Na=0$.} Note that the last term in $H_{C_{I^2}}$, of the form $a^\dagger a a^\dagger a+\rm{h.c.}$, does not need to be decomposed since the Hermitian-conjugate pair can be written as a sum of commuting Pauli strings that can be implemented together, as in \cref{Eq:H_CI2_Paulis}.
Now accounting for the coefficients of each operator in $H_{C_{I^2}}$, we find
\begin{align}
    \sum_{\gamma_1}\norm{ \Big[H^{(\gamma_1)}_{C_{I^2}},\sum_{\gamma_2=\gamma_1+1}H^{(\gamma_2)}_{C_{I^2}}\Big] }_\eta  &\leq \frac{C^2_{I}}{4}\times
     240 \,
    \eta.
\end{align}
\end{proof}

\begin{lemma}\label{Lemma:Hfree_HKR(0)_Commutator}
\begin{align}
    \norm{ \Big[H
    _{\rm free},\sum_{\gamma_2
    }H^{(\gamma_2)}_{\rm LR}(0)\Big] }_\eta  \leq 
     \frac{131072}{3} 
     a_L^{-3}h\left( \frac{g_A}{2f_\pi} \right)^2 
    \eta.
\end{align}
\end{lemma}
\begin{proof}
From \cref{Eq:OPE_Long-Range_Terms}, we notice that
    \begin{align}
     H_{\rm LR}(0) \coloneqq -\frac{1}{9a_L^3} \left(\frac{g_A}{2f_\pi}\right)^2 \sum_{I,S}
      \sum_{\bm{x}}
       [\tau_I(\bm{x})]_{\beta' \delta'}[\tau_I(\bm{x}) ]_{\beta \delta}
    &[\bm{\sigma}_S(\bm{x})]_{\alpha'\gamma'}[\bm{\sigma}_S(\bm{x})]_{\alpha\gamma}\nonumber \\
    &\times :\adag_{\alpha' \beta'}(\bm{x})a_{\gamma' \delta'}(\bm{x})\adag_{\alpha\beta}(\bm{x})a_{\gamma\delta}(\bm{x}):. 
      \label{Eq:OPE_Long-Range_Terms_0_App}
\end{align}
Thus, each commutator is of the general form $[\adag_\xi(i)a_\xi(j)+\adag_\xi(j)a_\xi(i),a^\dagger_\sigma(k) a^\dagger_{\sigma'}(k) a_{\sigma''}(k) a_{\sigma'''}(k)]$. The non-vanishing commutators arise from $i=k$ or $j=k$. Let us inspect one of these options:
\begin{align}
[\adag_\xi(k)a_\xi(j)+\adag_\xi(j)a_\xi(k),a^\dagger_\sigma(k) a^\dagger_{\sigma'}(k) a_{\sigma''}(k) a_{\sigma'''}(k)]
&=[\adag_\xi(k),a^\dagger_\sigma(k) a^\dagger_{\sigma'}(k) a_{\sigma''}(k) a_{\sigma'''}(k)]a_\xi(j)
\nonumber\\
&+\adag_\xi(j)[a_\xi(k),a^\dagger_\sigma(k) a^\dagger_{\sigma'}(k) a_{\sigma''}(k) a_{\sigma'''}(k)].
\end{align}
Bounding the number of NPFOs contributing to the resulting commutators can be cumbersome if we attempt to specify all the possibilities for species indices, so we resort to finding a rather loose bound. The maximum number of NPFOs is generated when as many operators as possible are of similar type, so moving them around to make normal-ordered operators according to \cref{Def:NPFO_Def} could give rise to additional terms arising from their non-trivial anti-commutation. So for this purpose, we consider a commutator of the form $[a^\dagger,a^\dagger a^\dagger a a]$ or $[a,a^\dagger a^\dagger a a]$, which each generate at most 4 NPFOs. This is of course a loose bound since, if all these operators were the same, the semi-norm of some operators would have been zero as $a^n_\sigma=(\adag_\sigma)^n=0$ for $n>1$. Nonetheless, we proceed with this upper bound.

Next, note that there are 4 hopping terms for each $\{i,j\}$ pair associated with each fermion species, and there are up to 256 terms in $H_{\rm LR}(0)$ for all combinations of spin-isospin indices in the four-fermion operator. 

Finally, we count the number of disjoint sets arising from the commutator before applying \cref{Theorem:NPFO_Norm}. If one takes a commutator of a $T^{(\gamma_1)}$ term and a $H_{\rm LR}(0)$, the resulting terms will not necessarily be disjoint, similarly to the kinetic-kinetic commutators in \cref{Lemma:kinetic-kinetic}. However, unlike the kinetic-kinetic case, which only acts on a single fermionic species, $H_{\rm LR}(0)$ can act on two fermionic species per site. So we split the terms into 4 disjoint sets (i.e.\ where each set is composed of disjoint operators) instead of 2. This is because the commutator can mix two species and it will no longer be the case that it can be split into 2 sets for each species.
Applying the triangle inequality, we obtain a factor of 4.

Combining these bounds, we find 
\begin{align}
    \sum_{\gamma_1}\norm{ \Big[T^{(\gamma_1)},\sum_{\gamma_2}H^{(\gamma_2)}_{\rm LR}(0)\Big] }_\eta
    &\leq h \times \frac{1}{9a_L^3} \left( \frac{g_A}{2f_\pi} \right)^2 \times
    6\times 4 \times 2 \times  4  \times 2 \times 4 \times 256 \, \eta,
\end{align}
where the factor $6$ comes from $\gamma_1=6$ disjoint sets of terms in $\Hfree$.
\end{proof}

\begin{lemma}\label{Lemma:HC1_HKR(0)_Commutator} 
\begin{align}
    \norm{ \Big[H_{C},\sum_{\gamma_2
    }H^{(\gamma_2)}_{\rm LR}(0)\Big] }_\eta  = \frac{7168}{3} a_L^{-3} |C|\left( \frac{g_A}{2f_\pi } \right)^2\eta.
\end{align}
\end{lemma}

\begin{proof}
On each site, $H_{LR}(0)$ has a total of $256$ terms, while
$H_C$ has 6 terms.
Furthermore, each commutator of the form $[\adag \adag a a,\adag \adag a a]$ decomposes into a sum of at most $14$ NPFOs, assuming all the operators are of the same type and resorting to a loose bound, as discussed in the proof of \cref{Lemma:Hfree_HKR(0)_Commutator}. 
Then, using the triangle inequality and the fermionic semi-norm, we find
\begin{align}
    \norm{ \Big[H_{C},\sum_{\gamma_2
    }H^{(\gamma_2)}_{\rm LR}(0)\Big] }_\eta  \leq \frac{1}{9a_L^3}
    \frac{|C|}{2}\left( \frac{g_A}{2f_\pi } \right)^2 \times  6\times 256\times 
     14  \, \eta.
\end{align}
\end{proof}

\begin{lemma}\label{Lemma:HC2_HKR(0)_Commutator}
\begin{align}
    \sum_{\gamma_1}\norm{ \Big[H^{(\gamma_1)}_{C_{I^2}},\sum_{\gamma_2
    }H^{(\gamma_2)}_{\rm LR}(0)\Big] }_\eta  \leq \frac{50176 }{9} a_L^{-3} |C_{I^2}| \left( \frac{g_A}{2f_\pi } \right)^2
    \eta.
\end{align}
\end{lemma}

\begin{proof}
There are at most $256 $ terms in each $H^{(\gamma_2)}_{\rm LR}(0)$, and 
the total weight of the operators in $H_{C_{I^2}}$ is 28.
The terms then take commutators of the form $[\adag \adag a a , \adag \adag a a]$.
So following the same argument as in \cref{Lemma:HC2_HC2_Commutator}, we find
\begin{align}
    \norm{ \sum_{\gamma_1}\Big[H^{(\gamma_1)}_{C_{I^2}},\sum_{\gamma_2}H^{(\gamma_2)}_{\rm LR}(0)\Big] }_\eta &\leq \frac{1}{9a_L^3}\left( \frac{g_A}{2f_\pi } \right)^2 \frac{|C_{I^2}|}{2} \times 256\times 28 \times 14 \, \eta.
\end{align}
\end{proof}

\begin{lemma}\label{Lemma:HLR(0)_HKR(0)_Commutator}
\begin{align}
    \sum_{\gamma_1}\norm{ \Big[H^{(\gamma_1)}_{\rm LR}(0),\sum_{\gamma_2=\gamma_1+1}H^{(\gamma_2)}_{\rm LR}(0)\Big] }_\eta  \leq
     \frac{152320}{27}a_L^{-6}\left( \frac{g_A}{2f_\pi } \right)^4\eta.
\end{align}
\end{lemma}
\begin{proof}
   The proof is similar to that of \cref{Lemma:HC2_HKR(0)_Commutator}, except the total weight of the $H_{C_{I^2}}$ term is replaced with the total weight of the $H_{LR}(0)$ term, which is bounded by $256$. However, the sum over $\gamma_2>\gamma_1+1$ forbids more than half of the terms.
   The same calculation gives
   \begin{align}
       \sum_{\gamma_1}\norm{ \Big[H^{(\gamma_1)}_{\rm LR}(0),\sum_{\gamma_2=\gamma_1+1}H^{(\gamma_2)}_{\rm LR}(0)\Big] }_\eta  
       &\leq \left(\frac{1}{9a_L^3}\right)^2\left( \frac{g_A}{2f_\pi } \right)^2 \times \frac{256\times (256-1)}{2} \times 14\, \eta.
   \end{align}
\end{proof}

\begin{lemma}\label{Lemma:HFree_HLR_Commutator}
\begin{align}
\norm{ \Big[
H_{\rm free},\sum_{\gamma_2}H_{\rm LR}^{(\gamma_2)}(r)  \Big] }_\eta \leq 
 \frac{
  98304 }{\pi} 
h\left(\frac{g_A}{2f_\pi}\right)^2q(r)f(r)\left( g(r) + 1 \right)
\eta,
\end{align}
where the $\gamma_2$ summation runs over all $H_{\rm LR}$ terms acting between lattice sites distance $r$ apart. Here, $q(r)$ is the number of lattice sites at distance $r$ away from any given lattice site. Furthermore, we have defined
\begin{align}
    f(r) &\coloneqq \frac{m_\pi^2e^{-m_\pi r}}{r},
    \label{Eq:f_function_def}\\
    g(r) &\coloneqq 1 + \frac{3}{m_{\pi}r} + \frac{3}{m_{\pi}^2r^2}.
    \label{Eq:g_function_def}
\end{align}
\end{lemma}

\begin{proof}
Recall that, according to \cref{Eq:OPE_Long-Range_Terms}, the long-range OPE Hamiltonian takes the form
\begin{align}
     H_{\rm LR}(r) \coloneqq \frac{1}{12\pi} 
      \left(\frac{g_A}{2f_\pi}\right)^2
      \sum_{\bm{x},\bm{y}}
      & \sum_{I} [\tau_I(\bm{x})]_{\beta' \delta'}[\tau_I(\bm{y}) ]_{\beta \delta}
     f(r) \Big(g(r)[S_{12}]_{\alpha'\gamma'\alpha\gamma}+\sum_S[\bm{\sigma}_S(\bm{x})_{\alpha'\gamma'}[\bm{\sigma}_S(\bm{y})]_{\alpha\gamma} \Big)
    \nonumber\\
    &\hspace{6 cm}:\adag_{\alpha' \beta'}(\bm{x})a_{\gamma' \delta'}(\bm{x})\adag_{\alpha\beta}(\bm{y})a_{\gamma\delta}(\bm{y}):,
      \label{Eq:OPE_Long-Range_Terms_App}
\end{align}
where $S_{12}$ is defined in \cref{Eq:S12-def} and $\bm{x}$ and $\bm{y}$ are at distance $r$ from each other. Therefore, the commutators for both the radial and tensor parts of the long-range Hamiltonian are of the general form $[\adag_\xi(k)a_\xi(l)+\adag_\xi(l)a_\xi(k),\adag_\sigma(i)\adag_{\sigma'}(j)a_{\sigma''}(i)a_{\sigma'''}(j)+\adag_{\sigma'''}(j)\adag_{\sigma''}(i)a_{\sigma'}(j)a_{\sigma}(i)]$, and the non-vanishing commutators occur for $k=i$, $k=j$, $l=i$, or $l=j$. Consider one such option:
\begin{align}
[\adag_\xi(i)a_\xi(l)+\adag_\xi(l)a_\xi(i),\adag_\sigma(i)\adag_{\sigma'}(j)a_{\sigma''}(i)a_{\sigma'''}(j)&+\adag_{\sigma'''}(j)\adag_{\sigma''}(i)a_{\sigma'}(j)a_{\sigma}(i)]=
\nonumber\\
&~~~~a_\xi(l)[\adag_\xi(i),\adag_\sigma(i)a_{\sigma''}(i)]\adag_{\sigma'}(j)a_{\sigma'''}(j)
\nonumber\\
&+a_\xi(l)[\adag_\xi(i),\adag_{\sigma''}(i)a_{\sigma}(i)]\adag_{\sigma'''}(j)a_{\sigma'}(j)
\nonumber\\
&-\adag_\xi(l)[a_\xi(i),\adag_\sigma(i)a_{\sigma''}(i)]\adag_{\sigma'}(j)a_{\sigma'''}(j)
\nonumber\\
&-\adag_\xi(l)[a_\xi(i),\adag_{\sigma''}(i)a_{\sigma}(i)]\adag_{\sigma'''}(j)a_{\sigma'}(j).
\label{Eq:commutator-adag-a-with-a}
\end{align}
Each resulting commutator is of the form $[a,\adag a]$ or $[\adag,\adag a]$, which each can be written as at most 3 NPFOs, resorting to our loose bound. So overall, \cref{Eq:commutator-adag-a-with-a} is a sum of $4 \times 3=12$ NPFOs.

Continuing, there are 4 hopping terms for each $\{k,l\}$ pair, there are $256$ terms in $H_{\rm LR}$ for each $\{i,j\}$ pair, and there are $q(r)$ terms for each $\{i,j\}$ corresponding to $\{\bm{x},\bm{y}\}$ at distance $|\bm{x}-\bm{y}|=r$.
Finally, there are 4 disjoint sets of NPFOs arising from the commutators as argued in \cref{Lemma:Hfree_HKR(0)_Commutator}.
Putting all these together gives
\begin{align}
&\sum_{\gamma_1} \norm{ \Big[T^{(\gamma_1)},\sum_{\gamma_2}H_{\rm LR}^{(\gamma_2)}(r)  \Big] }_\eta 
  \leq  \frac{h}{12\pi} \left( \frac{g_A}{2f_\pi} \right)^2 f(r)\left(g(r) + 1 \right) \times 6 \times 4  \times 4  \times 12 \times 4 \times 256 \, q(r) \, \eta. 
\end{align}
Here, the factor $6$ comes from $\gamma_1=6$ disjoint sets of terms in $\Hfree$.
\end{proof}

\begin{lemma}\label{Lemma:HC1_HLR_Commutator}
\begin{align}
\norm{ \Big[H_{C},\sum_{\gamma_2}H_{\rm LR}^{(\gamma_2)}(r)  \Big] }_\eta \leq \frac{1024|C|}{\pi}\left( \frac{g_A}{2f_\pi}\right)^2q(r)f(r)[g(r) + 1] \eta,
\end{align}
where the $\gamma_2$ summation runs over all $H_{\rm LR}$ terms acting between lattice sites distance $r$ apart.
Here, $q(r)$ is the number of lattice sites distance $r$ away from any given lattice site, and the $f$ and $g$ functions are defined in \cref{Eq:f_function_def,Eq:g_function_def}, respectively.
\end{lemma}
\begin{proof}
The analysis of $[H_{C},H_{\rm LR}^{(\gamma_1)}(r)  ]$ amounts to computing commutators of the form
\begin{align}
    [N_\xi(i), \adag_{\sigma}(i)\adag_{\sigma'}(j) a_{\sigma''}(i) a_{\sigma'''}(j)],
    \label{Eq:n_adag_a_Commutator}
\end{align}
which involves the following non-vanishing commutators:
\begin{align}
    &[N_\sigma(i), \adag_\sigma(i)a_{\sigma'}(i)] = \adag_\sigma(i)a_{\sigma'}(i), \\
    &[N_\sigma(i), \adag_{\sigma'}(i)a_{\sigma}(i)] = -\adag_{\sigma'}(i)a_{\sigma}(i).
\end{align}
Therefore, each term of the form in \cref{Eq:n_adag_a_Commutator} breaks down into at most one term of the form $\adag_{\sigma}(i)\adag_{\sigma'}(j) a_{\sigma''}(i) a_{\sigma'''}(j)$. Thus,
\begin{align}
    \norm{ \Big[H_{C},\sum_{\gamma_2}H_{\rm LR}^{(\gamma_2)}(r) \Big] }_\eta 
    &\leq \frac{|C|}{2}\sum_{\gamma_2}\sum_{\xi, \xi'} \norm{\sum_i \Big[N_\xi(i)N_{\xi'}(i),H_{\rm LR}^{(\gamma_2)}(r) \Big] }_\eta \nonumber\\
    &\leq \frac{|C|}{2}\sum_{\gamma_2}\sum_{\xi, \xi'} \sum_i \norm{N_\xi(i)\Big[N_{\xi'}(i),H_{\rm LR}^{(\gamma_2)}(r) \Big] + \Big[ N_\xi(i) ,H_{\rm LR}^{(\gamma_2)}(r) \Big] N_{\xi'}(i) }_\eta \nonumber\\
    &\leq \frac{|C|}{2} \times \frac{1}{12\pi}\left(\frac{g_A}{2f_\pi}\right)^2 f(r)\left(g(r)+1\right) \sum_{\sigma \sigma' \sigma'' \sigma'''}\sum_{\xi,\xi'}  \sum_i\sum_{i,j}\bigg |\bigg| N_\xi(i)\Big[N_{\xi'}(i), \nonumber\\
    &\hspace{2.5em}  \adag_{\sigma}(i)\adag_{\sigma'}(j)a_{\sigma''}(i) a_{\sigma'''}(j)  \Big]+ \Big[ N_\alpha(i) , \adag_{\sigma}(i)\adag_{\sigma'}(j) a_{\sigma''}(i) a_{\sigma'''}(j) \Big] N_{\xi'}(i) \bigg |\bigg|_\eta \nonumber\\
    &\leq \frac{|C|}{2} \times \frac{1}{12\pi}\left( \frac{g_A}{2f_\pi}\right)^2 \times 12 \times 2 \times  2 \times 2 \times 256\,q(r)f(r)\left(g(r) + 1\right)
    \eta.
\end{align}
Here, $i$ and $j$ are the qubit indices of sites $\bm{x}$ and $\bm{y}$, respectively, at distance $r$ from each other.
The factors above arise from counting terms and applying the triangle inequality as follows. The factor of $12$ comes from the sum over $\xi$ and $ \xi'$ with $\xi \neq \xi'$. 
One of the factors of 2 comes from the two terms present inside the semi-norm (of the form $N[N,H_{\rm LR}]$ and $[N,H_{\rm LR}]N$). Another factor of 2 comes from normal ordering the creation and annihilation operators in these two terms, at most doubling the number of terms.
A final factor of two arises since one generates terms of the form above for each end of the $H_{\rm LR}(i,j)$ term.
As before, the factor of $256\,q(r)$ comes from bounding the sum over $\gamma_3$ by a sum over $\sigma, \sigma', \sigma'', \sigma'''$, of which there are at most $2^8$ possible terms. Note that unlike the $[H_{\rm LR}, H_\mathrm{free}]$ or $[H_{\rm LR}, H_{\rm LR}]$ cases, the commutators here do not need to be split into further sets of disjoint operators because the $H_{C}$ are constrained to a single lattice site.
\end{proof}

\begin{lemma}\label{Lemma:HC2_HLR_Commutator}
\begin{align}
\norm{ \Big[H_{C_{I^2}},\sum_{\gamma_2}H_{\rm LR}^{(\gamma_2)}(r) \Big] }_\eta \leq \frac{43008|C_{I^2}|}{12\pi}\left( \frac{g_A}{2f_\pi}\right)^2q(r)f(r)\left(g(r) +1 \right)
\eta,
\end{align}
where the $\gamma_2$ summation runs over all $H_{\rm LR}$ terms acting between lattice sites distance $r$ apart.
Here, $q(r)$ is the number of lattice sites distance $r$ away from any given lattice site, and the $f$ and $g$ functions are defined in \cref{Eq:f_function_def,Eq:g_function_def}, respectively.
\end{lemma}

\begin{proof}
Expanding $H_{C_{I^2}}$ as a sum of weight 28 NPFOs 
and $H_{\rm LR}(r)$ as a sum of at most 256 NPFOs, the commutators to be evaluated are of the form $[\adag_\sigma(k)\adag_{\sigma'}(k) a_{\sigma''}(k) a_{\sigma'''}(k), \adag_\xi(i)\adag_{\xi'}(j) a_{\xi''}(i) a_{\xi'''}(j)]$. The non-vanishing commutators arise from $k=i$ or $k=j$. Each of these possibilities can be broken down to
\begin{align}
    &[\adag_\sigma(i)\adag_{\sigma'}(i) a_{\sigma''}(i) a_{\sigma'''}(i), \adag_\xi(i)\adag_{\xi'}(j) a_{\xi''}(i) a_{\xi'''}(j)] \nonumber \\
    &\qquad\qquad= -[\adag_\sigma(i)\adag_{\sigma'}(i) a_{\sigma''}(i) a_{\sigma'''}(i), \adag_\xi(i) a_{\xi''}(i) ]\adag_{\xi'}(j)a_{\xi'''}(j).
\end{align}
The remaining commutator of the form $[\adag \adag aa,\adag a]$ can generate at most 6 NPFOs. 
Putting everything together gives
\begin{align}
    \norm{\left[H_{C_{I^2}}, \sum_{\gamma_2}H^{(\gamma_2)}_{\rm LR}(r)\right] }_\eta
    &\leq \frac{|C_{I^2}|}{2} \times \frac{1}{12\pi}\left( \frac{g_A}{2f_\pi}\right)^2 \times 28 \times 256 \,q(r) \times  6  \times  2 \times f(r)\left(g(r) + 1\right)
    \eta.
\end{align}
\end{proof} 

\begin{lemma}\label{Lemma:HLR(0)_HKR(r)_Commutator}
\begin{align}
    \sum_{\gamma_1}\norm{ \Big[H^{(\gamma_1)}_{\rm LR}(0),\sum_{\gamma_2}H^{(\gamma_2)}_{\rm LR}(r)\Big] }_\eta  \leq  \frac{458752}{27\pi}
    a_L^{-3}\left( \frac{g_A}{2f_\pi}\right)^4q(r)f(r)\left(g(r) +1 \right) \eta.
\end{align}
Here, $q(r)$ is the number of lattice sites at distance $r$ away from any given lattice site, and $f$ and $g$ functions are defined in \cref{Eq:f_function_def,Eq:g_function_def}, respectively.
\end{lemma}

\begin{proof}
The proof proceeds in the same way as \cref{Lemma:HC2_HLR_Commutator}, except $H_{\rm{LR}}(0)$ is counted as $256$ NPFOs. Therefore,
\begin{align}
     \sum_{\gamma_1}\norm{ \Big[H^{(\gamma_1)}_{\rm LR}(0),\sum_{\gamma_2}H^{(\gamma_2)}_{\rm LR}(r)\Big] }_\eta  
     \leq \frac{1}{9a_L^3}\left( \frac{g_A}{2f_\pi}\right)^2&\times \frac{1}{12\pi}\left( \frac{g_A}{2f_\pi}\right)^2\times 256 \times 256 \,q(r) \nonumber \\
     &\times 14\times  2 \times f(r)\left(g(r) +1 \right) \eta.
\end{align}
\end{proof}

\begin{lemma} \label{Lemma:HLR_HLR_Commutator}
\begin{align}
\norm{\sum_{\gamma_1} \Big[H_{\rm LR}^{(\gamma_1)}(r),\sum_{\gamma_2}H_{\rm LR}^{(\gamma_2)}(r') \Big] }_\eta \leq 3670016 \left( \frac{1}{12\pi} \right)^2 \left( \frac{g_A}{2f_\pi}\right)^4 q(r)q(r') f(r)f(r')\left(g(r) + 1\right)\left(g(r') + 1\right) 
\eta,
\end{align}
where $r \neq r'$, and the summation over $\gamma_1$ and $\gamma_2$ is over all $H_{\rm LR}$ terms of length $r$ and $r'$, respectively.
Here, $q(r)$ is the number of lattice sites distance $r$ away from a given lattice site, and $f$ and $g$ functions are defined in \cref{Eq:f_function_def,Eq:g_function_def}, respectively.
\end{lemma}

\begin{proof}
 All commutators here take the form $[\adag_{\sigma}(i)\adag_{\sigma'}(j) a_{\sigma''}(i) a_{\sigma'''}(j), \adag_\xi(k)\adag_{\xi'}(l) a_{\xi''}(k) a_{\xi'''}(l)]$, with four possibilities for qubit indices to coincide to give non-vanishing commutations.
 Let us inspect one of those possibilities:
\begin{align}
    &[\adag_{\sigma}(i)\adag_{\sigma'}(j) a_{\sigma''}(i) a_{\sigma'''}(j), \adag_\xi(k)\adag_{\xi'}(i) a_{\xi''}(k) a_{\xi'''}(i)]\nonumber \\
    &\qquad\qquad= \adag_{\sigma'}(j)a_{\sigma'''}(j)  [\adag_\sigma(i)a_{\sigma''}(i) , \adag_{\xi'}(i) a_{\xi'''}(i) ]\adag_{\xi}(k) a_{\xi''}(k).
\end{align}
The internal commutator $[\adag_\sigma(i)a_{\sigma''}(i) , \adag_{\xi'}(i) a_{\xi'''}(i) ]$ can consist of at most 2 NPFOs (corresponding to when $\sigma=\xi'''$ and $\sigma''=\xi'$).  

Now in order to apply \cref{Theorem:NPFO_Norm}, we find distinct sets of commutators when summing the Hamiltonian terms over all lattice points. Let us define two vectors $\vec{r}, \vec{r}'$ starting on $\bm{x}$, such that $|\vec{r}|=r, \ |\vec{r}'|=r'$, where $\bm{x}$ is the lattice site associated with qubit index $i$.
Let $T(\vec{r},\vec{r}')$ be the set of translations of this pair by lattice vectors.
Since $\vec{r}$ and $\vec{r}'$ together form a triangle, we can partition $T(\vec{r},\vec{r}')$ into translation-invariant sets $T_q(\vec{r},\vec{r}')$ such that for $q\neq q'$, $T_q(\vec{r},\vec{r}')$ and $T_{q'}(\vec{r},\vec{r}')$ do not have vectors that intersect with each other on any vertex.
Given $T(\vec{r},\vec{r}')$, the minimum number of subsets needed is $7$.
This is because any given triangle can only intersect translations of itself at its 3 vertices.
Then at these intersections, the triangle can intersect 2 of the translated triangle's vertices, giving $3\times 2$ possible sets.
Including the set defined by itself, this gives $6+1$ possible sets.
See \Cref{Fig:translationally_invariant} for a visual illustration of this.
Since the commutators associated with each of these sets are guaranteed to be disjoint, we can now apply \cref{Theorem:NPFO_Norm}.

Putting everything together gives
\begin{align}
 \sum_{\gamma_1} \norm{ \sum_{\gamma_2} \Big[H^{(\gamma_1)}_{\rm LR}(r), H^{(\gamma_2)}_{\rm LR}(r')\Big] }_\eta     
  \leq \left( \frac{1}{12\pi} \right)^2 \left( \frac{g_A}{2f_\pi}\right)^4 256^2q(r)q(r') f(r)f(r')&\left(g(r) + 1 \right)\left(g(r') + 1\right) \nonumber \\
  \times 4 \times  7 \times  2  \, \eta,
  \end{align}
where the total number of possible terms in each $H_{\rm LR}(r)$ is also accounted for, in accordance with previous lemmas.
\end{proof}

\begin{figure}[t!] 
\centering
\includegraphics[width=0.35\textwidth]{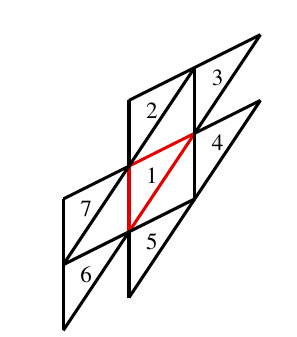}
\caption{The red triangle represents the triangle formed by two vectors $\vec{r}$ and $\vec{r}'$. The black triangles are translations of the original, which share a vertex with the red triangle.}
\label{Fig:translationally_invariant}
\end{figure}

\begin{lemma} \label{Lemma:HLR_HLR_Commutator_2}
\begin{align}
 \sum_{\gamma_1} \norm{\sum_{\gamma_2} [H_{\rm LR}^{(\gamma_1)}(r),H_{\rm LR}^{(\gamma_2)}(r)  ] }_\eta &\leq 3670016 \left( \frac{1}{12\pi} \right)^2 \left( \frac{g_A}{2f_\pi}\right)^4 q(r)(q(r)-1) f^2(r)\left(g(r) + 1\right)^2 
 \eta
 \nonumber\\
 &\hspace{1.6 cm}+ 
 524288 \left( \frac{1}{12\pi} \right)^2 \left( \frac{g_A}{2f_\pi}\right)^4 q(r) f^2(r)\left(g(r) + 1\right)^2 \eta,
\end{align}
where the summations over $\gamma_1$ and $\gamma_2$ are over all $H_{\rm LR}$ terms of length $r$.
Here, $q(r)$ is the number of lattice sites distance $r$ away from a given lattice site, and $f$ and $g$ functions are defined in \cref{Eq:f_function_def,Eq:g_function_def}, respectively.
\end{lemma}

\begin{proof}
The proof is almost identical to that of \cref{Lemma:HLR_HLR_Commutator}, except that to get non-zero commutations here, there are two types of contributions.
The first has exactly the same from as discussed in \cref{Lemma:HLR_HLR_Commutator}---that is when the terms only coincides on one end, which are of the form
\begin{equation}
[\adag_{\sigma}(i)\adag_{\sigma'}(j) a_{\sigma''}(i) a_{\sigma'''}(j), \adag_\xi(k) \adag_{\xi'}(l) a_{\xi''}(k) a_{\xi'''}(l)].
\end{equation} 
So overall we have
\begin{align}
    \left( \frac{1}{12\pi} \right)^2 \left( \frac{g_A}{2f_\pi}\right)^4 256^2q(r)(q(r)-1) f^2(r)\left(g(r) + 1\right)^2 \times 7 \times 2 \, \eta,
\end{align}
where the factor of $q(r)(q(r)-1)$ comes about by excluding the terms where $\vec r=\vec r'$, with $\vec r, \vec r'$ defined as in \cref{Lemma:HLR_HLR_Commutator}.

The second type of commutators are those where the terms coincides on both ends. These are commutators of the form
\begin{align}
    &[\adag_{\sigma}(i)\adag_{\sigma'}(j) a_{\sigma''}(i) a_{\sigma'''}(j), \adag_\xi(i)\adag_{\xi'}(j) a_{\xi''}(i) a_{\xi'''}(j)]\nonumber \\
    &\qquad\qquad= \adag_{\sigma}(i)a_{\sigma''}(i) \adag_{\xi}(i)a_{\xi''}(i) [ \adag_{\sigma'}(j) a_{\sigma'''}(j), \adag_{\xi'}(j) a_{\xi'''}(j)]
    \nonumber\\
    &\qquad\qquad\qquad+ [ \adag_{\sigma}(i) a_{\sigma''}(i), \adag_{\xi}(i) a_{\xi''}(i)]\adag_{\sigma'}(j)a_{\sigma'''}(j) \adag_{\xi'}(j)a_{\xi'''}(j).
\end{align}
Now each of the commutators of type $[\adag a,\adag a]$ can be decomposed into at most 2 NPFOs, and each of the accompanying operators $\adag a \adag a$ consists of 2 NPFOs, giving a total of at most $2 \times 2  \times 2 = 8$ NPFOs.
The same calculation as in \cref{Lemma:HLR_HLR_Commutator} then gives a contribution to the commutator bound of the form
\begin{align}
    &\left( \frac{1}{12\pi} \right)^2 \left( \frac{g_A}{2f_\pi}\right)^4 256^2 q(r) f^2(r)\left(g(r) + 1\right)^2 
    \times 8  \, \eta.
\end{align}
\end{proof}

\subsection{Dynamical-Pion Bounds} \label{Sec:Trotter_Error_Dyn_p1}

The dynamical-pion case requires us to deal with the explicit representation of the pions. We then partition the terms as per the Trotter decomposition used. To this end, $\Hpp$ defined in \cref{Eq:Pion_Only_Hamiltonian} is split into two separate contributions:
\begin{align}
    \Hpp^{(1)} &\coloneqq H_{\Pi^2} = \frac{a_L^D}{2}\sum_{\bm{x}}\sum_I \Pi_I^2(\bm{x}), \label{Eq:Hpp1}
    \\
    \Hpp^{(2)} &\coloneqq H_{(\nabla \pi)^2} + H_{\pi^2} = \frac{a_L^{D-2}}{2}\sum_{\bm{x},j}\sum_{I}
    \left(\pi_I(\bm{x}+a_L \hat{\bm{n}}_j)-\pi_I(\bm{x})\right)^2 + \frac{a_L^Dm_\pi^2}{2}\sum_{\bm{x}}\sum_{I}  \pi_I(\bm{x})^2. \label{Eq:Hpp2}
\end{align}
where $D=3$.

In the following, we make use of the various nucleonic bilinear operators introduced in \crefrange{Eq:Particle_Density_1}{Eq:Particle_Density_4}. For convenience, these are repeated below:
\begin{align}
    &\rho(\bm{x}) = \sum_{\alpha} \sum_{\beta } a^{\dagger}_{\alpha \beta}(\bm{x})a_{\alpha \beta}(\bm{x}), \\
    &\rho_S(\bm{x}) = \sum_{\alpha,\gamma} \sum_\beta a^{\dagger}_{\alpha\beta}(\bm{x})[\sigma_S]_{\alpha \gamma}a_{\gamma\beta}(\bm{x}),
    \\
    &\rho_I(\bm{x}) = \sum_\alpha \sum_{\beta,\delta} a^{\dagger}_{\alpha\beta}(\bm{x}) [\tau_I]_{\beta\delta}a_{\alpha\delta}(\bm{x}),
    \\
    &\rho_{S,I}(\bm{x}) = \sum_{\alpha,\gamma} \sum_{\beta,\delta} a^{\dagger}_{\alpha\beta}(\bm{x}) [\sigma_S]_{\alpha\gamma}[\tau_I]_{\beta\delta}a_{\gamma\delta}(\bm{x}),
\end{align}

\begin{theorem}[Dynamical-Pion Trotter Error Bound] For the time evolution of the dynamical-pion EFT with a first-order product formula, $\norm{ \calP_1^{\dyn}(t) - e^{-itH_{\dyn}} }_\eta \leq \frac{t^2}{2} \Xi$, where $\Xi$ is the sum of the bounds which are reported in the Lemmas noted in \cref{Table:Table_of_Commutators_Dynamical_Pions}. 
\begin{table}[ht]
    \centering
   \resizebox{\textwidth}{!}{ \begin{tabular}{c|c|c|c|c|c|c|c}
      &$\Hfree$  & $H_C$  & $H_{C_{I^2}}$  & $H_{\pi}^{(1)}$  & $H_{\pi}^{(2)}$  & $H_{\rm AV}$ & $H_{\rm WT}$  \\ \hline 
    $\Hfree$ &  \cref{Lemma:kinetic-kinetic}   & \cref{Lemma:Hfree_HC1_Commutator}   & \cref{Lemma:Hfree_HC2_Commutator}    & 0 & 0  &    \cref{Lemma:HFree_HAV_Commutator}  & \cref{Lemma:HFree_HWT_Commutator} \\ \hline
    $H_C$   &  -  &  \cref{Lemma:HC1_HC1_Commutator} &  \cref{Lemma:HC1_HC2_Commutator}  &  0  &   0  &    \cref{Lemma:HAV_HC1_Commutator}  & \cref{Lemma:H_WT_H_C1_Commutator} \\ \hline
    $H_{C_{I^2}}$   &  -  & -  &  \cref{Lemma:HC2_HC2_Commutator}  &  0  &  0    &   \cref{Lemma:HAV_HC2_Commutator}   & \cref{Lemma:C2_HWT_Commutator} \\ \hline
    $H_{\pi}^{(1)}$   &  -  & -  & -  & 0  &  \cref{Lemma:Hpp1_Hpp2_Commutator}   &   \cref{Lemma:Hpp1_HAV}   & \cref{Lemma:Hpp_HWT_Commutator} \\ \hline
    $H_{\pi}^{(2)}$  &  -  & -  & -  & -  & 0   &   \cref{Lemma:Hpp2_HAV}   & \cref{Lemma:Hpp2_HWT_Commutator} \\ \hline
    $H_{\rm AV}$  &  -  & -  & -  & -  & -   & \cref{Lemma:HAV_HAV}    & \cref{Lemma:HAV_HWT_Commutator} \\ \hline
    $H_{\rm WT}$  &  -  & -  & -  & -  & -   & -    & \cref{Lemma:HWT_HWT_Commutator}
    \end{tabular}}
    \caption{Commutators for the dynamical-pion EFT Hamiltonian and the lemmas in which bounds on the value of the commutators are computed. Zeros indicate when the commutators are trivially zero.}   \label{Table:Table_of_Commutators_Dynamical_Pions}
\end{table}
\end{theorem}

\begin{proof}
According to the expression for Trotter error  in \cref{Eq:p-1-commutator-bound}, i.e.,
\begin{align}
    \norm{e^{-itH} - \mathcal{P}_1(t)}
    &\leq \frac{t^2}{2}\sum_{\gamma_1=1}^\Gamma\norm{ \left[H_{\gamma_1},\sum_{\gamma_2=\gamma_1+1}^\Gamma H_{\gamma_2} \right]  },
\end{align}
we decompose the Hamiltonian into terms or `layers' and assign a $\gamma_i$ labeling.
The decomposition we choose is summarized in \cref{Tab:dynamic_pions_layers}. 
For the terms $\Hfree$, $H_C$, and $H_{C_{I^2}}$, we make an identical decomposition as in the OPE-EFT case, so a number of commutators can be used from \cref{Sec:Trotter_Error_OPE_p1}. For the new terms, $H_{\pi}^{(1)}$ can be seen in \cref{Eq:Hpp1} to consist of only local terms and hence can be decomposed into a single layer. For $H_{\pi}^{(2)}$, according to \cref{Eq:Hpp2}, all terms act between neighbors on the lattice. Consequently, this Hamiltonian can be broken down into 6 layers similarly to the fermionic hopping terms. For $H_{\rm AV}$, the fermionic parts of all terms take the form $\adag_{\alpha\beta}a_{\gamma\delta}$, so summing over $\alpha$, $\beta$, $\gamma$, and $\delta$ gives $2^4=16$ terms. These terms all involve interactions between adjacent sites, so again we account for a factor of $6$ to divide these into disjoint sets. This yields a total of $96$ terms. For $H_{\rm WT}$, all terms take the form $\pi_I\Pi_J \adag_{\alpha\beta }a_{\gamma\delta}$.
Since $I\neq J$, there are at most $6\times 2^4 =96$ terms which appear here. This term is local, hence no further disjoint sets need to be realized.

\begin{table}
\begin{center}
\begin{tabular}{ c |c |c }
Hamiltonian Term & Set of Terms  & Number of Layers Upper Bound \\ 
\hline
$\Hfree$ & $\Gamma_{\rm free}$ & 6 \\ 
$H_C$ & $\Gamma_{C}$ & 1 \\  
 $H_{C_{I^2}}$ & $\Gamma_{C_{I^2}}$ & 12  \\
 $H_{\pi}^{(1)}$ & $\Gamma_{\pi\pi1}$ & 1 \\
 $H_{\pi}^{(2)}$ & $\Gamma_{\pi\pi2}$ & 6 \\
  $H_{\rm AV}$ & $\Gamma_{\rm AV}$ & 96 \\
   $H_{\rm WT}$ & $\Gamma_{\rm WT}$ & 96
 \end{tabular}
\end{center}
\caption{Decomposition of $H_{D\pi}$ in \cref{Eq:H-dyn} into layers for the application of the first-order Trotter error bound. \label{Tab:dynamic_pions_layers}}
\end{table}

We now compute bounds on the commutators of the terms above, ordering their $\gamma_i$ labels according to the first column in \cref{Tab:dynamic_pions_layers} from top to bottom. Several of the requisite commutators are either trivially zero or already computed in \cref{Sec:Trotter_Error_OPE_p1}, and we proceed with analyzing the remainder in \crefrange{Lemma:Hpp1_HAV}{Lemma:HWT_HWT_Commutator}. The right-hand side of \cref{Eq:p-1-commutator-bound} can then be computed using the sum of the commutators listed in \cref{Table:Table_of_Commutators_Dynamical_Pions}.
\end{proof}

\begin{lemma} \label{Lemma:Hpp1_Hpp2_Commutator}
\begin{align}
    \norm{ \Big[\Hpp^{(1)}, \Hpp^{(2)}\Big] }
    \leq \left( \frac{36
    }{a_L^2} +3
    m_\pi^2 \right)a_L^D\pimax \Pimax L. 
\end{align}
\end{lemma}

\begin{proof}
First, consider the commutator of $\Hpp^{(1)}$ with the gradient part of $\Hpp^{(2)}$. Note that
\begin{align}
    \sum_{\bm{x},J,j}\Big[\Pi_I(\bm{y}), \left(\pi_J(\bm{x}+a_L\hat{n}_j)-\pi_J(\bm{x})\right)^2\Big] 
    &= \sum_{\bm{x},j}\Big[\Pi_I(\bm{y}), (\pi_I(\bm{x}+a_L\hat{\bm{n}}_j)-\pi_I(\bm{x}))\Big]\left(\pi_I(\bm{x}+a_L\hat{\bm{n}}_j)-\pi_I(\bm{x})\right) \nonumber \\
    &+\left(\pi_I(\bm{x}+a_L\hat{\bm{n}}_j)-\pi_I(\bm{x})\right)\Big[\Pi_I(\bm{y}), (\pi_I(\bm{x}+a_L\hat{\bm{n}}_j)-\pi_I(\bm{x}))\Big] \nonumber \\
    &= \frac{i}{a_L^D} \sum_{\bm{x},j} \left( \delta(\bm{x}+a_L\hat{\bm{n}}_j-\bm{y}) - \delta(\bm{x}-\bm{y}) \right)\left(\pi_I(\bm{x}+a\hat{\bm{n}}_j)-\pi_I(\bm{x})\right) \nonumber \\
    &+ \frac{i}{a_L^D}\sum_{\bm{x},j}\left(\pi_I(\bm{x}+a_L\hat{\bm{n}}_j)-\pi_I(\bm{x})\right)\left( \delta(\bm{x}+a_L\hat{\bm{n}}_j-\bm{y}) - \delta(\bm{x}-\bm{y})  \right) \nonumber \\
    &= \frac{2i}{a_L^D}\sum_j\Big( \pi_I(\bm{y})-\pi_I(\bm{y}-a_L\hat{\bm{n}}_j) + \pi_I(\bm{y}) - \pi_I(\bm{y}+a_L\hat{\bm{n}}_j) \Big),
    \label{Eq:First_Pion_Commutator}
\end{align}
which gives
\begin{align}
    \sum_{\bm{x},J,j}\Big[\Pi^2_I(\bm{y}),(\pi_J(\bm{x}+a\hat{\bm{n}}_j)-\pi_J(\bm{x}))^2 \Big] &=\frac{2i}{a_L^D}\sum_j\Pi_I(\bm{y})\Big( \pi_I(\bm{y})-\pi_I(\bm{y}-a_L\hat{\bm{n}}_j) + \pi_I(\bm{y}) - \pi_I(\bm{y}+a_L\hat{\bm{n}}_j) \Big)\nonumber\\
    &+\frac{2i}{a_L^D}\sum_j\Big( \pi_I(\bm{y})-\pi_I(\bm{y}-a_L\hat{\bm{n}}_j) + \pi_I(\bm{y}) - \pi_I(\bm{y}+a_L\hat{\bm{n}}_j) \Big)\Pi_I(\bm{y}).
\end{align}
Therefore,
\begin{align}
    &\norm{ \sum_{\bm{x},\bm{y},I,J,j}\Big[\frac{a_L^D}{2}\Pi^2_I(\bm{y}),\frac{a_L^{D-2}}{2}\left(\pi_i(\bm{x}+a\hat{\bm{n}}_j)-\pi_i(\bm{x})\right)^2 \Big] }\nonumber\\ 
    &\hspace{2.5 cm}\leq \frac{a_L^{2D-2}}{4} \times 2\times  \frac{2}{a_L^D}\norm{ \sum_{\bm{y},I} \Pi_I(\bm{y}) \sum_j\left( \pi_I(\bm{y})-\pi_I(y-a_L\hat{\bm{n}}_j) + \pi_I(\bm{y}) - \pi_I(y+a_L\hat{\bm{n}}_j) \right)} \nonumber\\
    &\hspace{2.5 cm} \leq \frac{a_L^{2D-2}}{4} \times 2\times  \frac{2}{a_L^D} \times  \Pimax \times 4 \pimax \times 3 \times 3L \nonumber\\
    &\hspace{2.5 cm}=36 \, a_L^{D-2}\pimax \Pimax L.
\end{align}
In the third line, we have used $\norm{\Pi_I(\bm{y})}\leq \Pimax$ and $\norm{\pi_I(\bm{y})}\leq \pimax$, and have taken advantage of the Cauchy-Schwarz inequality and triangle inequality. The factor of $3$ results from the sum of directions $j$, and the factor of $3L$ arises from the sum over $I$ and $\bm{y}$. 

Let us now consider the commutator with the mass part of $\Hpp^{(2)}$. First note that
\begin{align}
   \Big[\Pi_I(\bm{x}), \pi_J^2(\bm{y})\Big] = \frac{2i}{a_L^D}\pi_J(\bm{y})\delta_{IJ}\delta(\bm{x}-\bm{y}).
\end{align}
Therefore,
\begin{align}
    \norm{\sum_{\bm{x},\bm{y}, I, J} \Big[ \frac{a_L^D}{2}\Pi_I(\bm{x})^2, \frac{a_L^Dm_\pi^2}{2}\pi_J^2(\bm{\bm{y}})  \Big] }
    &= \frac{a_L^{2D}m_\pi^2}{4} \norm{\sum_{\bm{x},\bm{y}, I, J} \left(\Pi_I(\bm{x}) \Big[ \Pi_I(\bm{x}), \pi_J^2(\bm{y})  \Big]  + \Big[ \Pi_I(x), \pi_J^2(\bm{y})  \big]\Pi_I(x)\right)} \nonumber\\
    &= \frac{2a_L^Dm_\pi^2}{4} \norm{\sum_{\bm{x}, I} \left(\Pi_I(\bm{x})\pi_I(\bm{x}) + \pi_I(\bm{x})\Pi_I(\bm{x})\right)} \nonumber\\
    &\leq \frac{2a_L^Dm_\pi^2}{4} \times 2 \times \pimax \times \Pimax \times 3L \nonumber\\
    &= 3\,a_L^Dm_\pi^2\pimax\Pimax L,
\end{align}
where the factor of $3$ arises from the sum over pion species.

Adding the above two results gives the statement of the lemma.
\end{proof}

\begin{lemma}\label{Lemma:HFree_HAV_Commutator}
\begin{align}
    \norm{ [H_{\rm free }, H_{\rm AV}] }_\eta \leq  2592  \left(\frac{g_A}{2 f_\pi}\right)a_L^{-1}h\pimax\eta  .
\end{align}
\end{lemma}
\begin{proof}

First note that, for given $I$ and $S$,
\begin{align}
   \norm{ \Big[ \rho_{S,I}(k)
   ,\sum_\sigma \Delta_{ij}^\sigma \Big]} \leq 2\times 4  \norm{\Big[\adag(i) a(i), \adag(j) a(i) + \adag(i)a(j)\Big]},
\end{align}
where the factor of 2 comes from the two possibilities $k=i$ or $k=j$ and the factor of 4 comes from the sum over $\sigma$ (hence suppressing the species indices on the right-hand side).
Now the $[\adag(i) a(i), \adag(j) a(i) + \adag(i)a(j)]$ term can generate NPFOs with a weight of at most $2 \times 3 = 6$ . Summing over $S$ and $I$ gives $3\times 3 \times 2\times 4 \times  6 = 432$ terms. Finally, to apply the bound on the number of fermions, we group the $\Delta_{ij}^\sigma$ terms into 6 sets of commuting terms (as with the $T^{(\gamma_1)}$ term analyzed previously), giving a total of $ 432  \times 6=  2592 $ terms.
Thus,
\begin{align}
    \norm{ [H_{\rm free}, H_{\rm AV}] }_\eta &\leq \sum_{\gamma_1}\norm{ \left(\frac{g_A}{2f_\pi}\right)\sum_{\bm{x}}\sum_{S,I}\frac{\pi_I(\bm{x}+a_L\hat{\bm{n}}_S)-\pi_I(\bm{x})}{a_L} \Big[T^{(\gamma_1)}, 
    \rho_{S,I}(\bm{x})\Big]} \nonumber\\
    &\leq \left(\frac{g_A}{2f_\pi}\right) \times h \times \frac{2\pimax}{a_L}\times  2592  \, \eta.
\end{align}
\end{proof}

\begin{lemma}\label{Lemma:HAV_HC1_Commutator}
\begin{align}
    \norm{ [H_{\rm AV}, H_{C}] }_\eta =0.
\end{align}
\end{lemma}
\begin{proof}
Note that $H_{\rm AV}$ preserves the number of nucleons on a particular lattice site. 
The proof is then identical to \cref{Lemma:H_WT_H_C1_Commutator} below.
\end{proof}

\begin{lemma}\label{Lemma:HAV_HC2_Commutator}
\begin{align}
    \norm{ [H_{\rm AV}, H_{C_{I^2}}] }_\eta \leq  6048  \left(\frac{g_A}{2f_\pi}\right)a_L^{-1}|\CI|\pimax\eta.
\end{align}
\end{lemma}

\begin{proof}
Consider the commutator $[\rho_{S,I},H_{C_{I^2}}]$. For every $S$ and $I$, the term $\rho_{S,I}$ contains 4 terms of the form $a^\dagger a$, and $H_{C_{I^2}}$ gives 28 terms of the form $a^\dagger a^\dagger a a$. Thus, for each $S$ and $I$, there are only commutators of the form $[\adag a, \adag \adag a a]$. Each generates at most  6  NPFOs, giving a total of at most $4\times 28\times  6  \times 9 = 6048 $ NPFOs. Summing over $S$ and $I$, we have
\begin{align}
    \norm{ [H_{\rm AV}, H_{C_{I^2}}] }_\eta &\leq \left(\frac{g_A}{2f_\pi}\right) \frac{|C_{I^2}|}{2} \sum_{\bm{x}} \sum_{S,I} \norm{ \frac{\pi_I(\bm{x}+a_L\hat{\bm{n}}_S) - \pi_I(\bm{x})}{a_L}[
    \rho_{S,I}(\bm{x}), H_{C_{I^2}}] }_\eta \nonumber\\
    &\leq \left(\frac{g_A}{2f_\pi}\right) \frac{|C_{I^2}|}{2}\times \frac{2\pimax}{a_L} \times  6048  \, \eta. 
\end{align}
\end{proof}

\begin{lemma}\label{Lemma:Hpp1_HAV}
\begin{align}
    \norm{ [\Hpp^{(1)},H_{\rm AV}] }_\eta \leq  36  \left(\frac{g_A}{2f_\pi}\right)a_L^{-1}\Pimax\eta.
\end{align}
\end{lemma}
\begin{proof}

First note that
\begin{align}
    &\sum_{\bm{y},S,I}\left[ \Pi_J^2(\bm{y}), 
    \rho_{S,I}(\bm{x})\left(\pi_I(\bm{x}+a_L \hat{\bm{n}}_S)-\pi_I(\bm{\bm{x}})  \right)  \right]\nonumber \\  
    &\qquad= \sum_{\bm{y},S,I}\Pi_J(\bm{y})\left[ \Pi_J(\bm{y}), 
    \rho_{S,I}(\bm{x})\left(\pi_I(\bm{x}+a_L \hat{\bm{n}}_S)-\pi_I(\bm{\bm{x}})  \right)  \right] \nonumber\\
    &\qquad\qquad\qquad    +\sum_{\bm{y},S,I}\left[ \Pi_J(\bm{y}), 
    \rho_{S,I}(\bm{x})\left(\pi_I(\bm{x}+a_L \hat{\bm{n}}_S)-\pi_I(\bm{\bm{x}})  \right)  \right]\Pi_J(\bm{y})
    \nonumber\\
  &\qquad= \frac{2i}{a_L^D}\sum_{\bm{y},S,I}\delta_{I,J}   \left(\delta(\bm{y}-\bm{x}-a_L\hat{\bm{n}}_S)-\delta(\bm{y}-\bm{x})\right) 
  \rho_{S,I}(\bm{x})\Pi_J(\bm{y})  \nonumber\\
    &\qquad= \frac{2i}{a_L^D}\sum_{S,I}
    \rho_{S,I}(\bm{x})\left(\Pi_I(\bm{x}+a_L\hat{\bm{n}}_S)-\Pi_J(\bm{x})\right).
\end{align}
Using this, the bound on the semi-norm of the full commutator is
\begin{align}
    \norm{  \sum_{\bm{x},\bm{y},S,I,J}\Big[ \frac{a_L^D}{2}\Pi_J^2(\bm{y}),\frac{g_A}{2a_Lf_\pi}
    \rho_{S,I}(\bm{x})\left(\pi_I(\bm{x}+a_L \hat{\bm{n}}_S)-\pi_I(\bm{\bm{x}})  \right)  \Big]}_\eta
    &\leq \frac{g_A}{4a_Lf_\pi} \times 4 \norm{\sum_{\bm{x},S,I}   
    \rho_{S,I}(\bm{x})\Pi_I(\bm{x}) }_\eta
    \nonumber\\
    &\leq \frac{g_A}{a_Lf_\pi} \times 4 
    \times 9 \times \Pimax \times \eta,
\end{align}
where the factor of 4 in the second line comes from the fact that each $\rho_{S,I}$ is a sum of at most 4 NPFOs, and the factor of 9 is the result of summing over $S$ and $I$.
\end{proof}

\begin{lemma}\label{Lemma:Hpp2_HAV}
\begin{align}
    \norm{ [\Hpp^{(2)},H_{\rm AV}] }_\eta = 0.
\end{align}
\end{lemma}
\begin{proof}
Since $\Hpp^{(2)}$ only depends on factors of $\pi_I(\bm{x})$ while $H_{\rm AV}$ does not contain any  $\Pi_I(\bm{x})$, these Hamiltonian terms commute.
\end{proof}

\begin{lemma}\label{Lemma:HAV_HAV}
\begin{align}
    \sum_{\gamma_1} \norm{\Big[H^{(\gamma_1)}_{\rm AV}, \sum_{\gamma_2 \geq \gamma_1+1}H^{(\gamma_2)}_{\rm AV}\Big] }_\eta \leq 20736\left( \frac{g_A}{2f_\pi}\right)^2a_L^{-2}\pimax^2\eta.
\end{align}
\end{lemma}

\begin{proof}
    The commutators that arise are of the form $[
    \rho_{S,I},\rho_{S',I'}]$, since the $\pi_I(\bm{x})$ terms commute with each other.
    For given $S$ and $I$, each $\rho_{S,I}$ is a sum of 4 terms of the form $\adag a$.
    Thus, for given $S$, $I$, $S$, and $I'$, there are 16 commutators of the form $[\adag a, \adag a]$.
    Each of these can further be written as a sum of 4 NPFOs. Now using the triangle and Cauchy-Schwarz inequalities, we have
    \begin{align}
        \left( \frac{g_A}{2f_\pi}\right)^2 \sum_{\bm{x},S,I,S',I'}\frac{1}{a_L^2}\Big|\Big| (\pi_I(\bm{x}+a_L \hat{\bm{n}}_S)-\pi_I(\bm{x})  )(\pi_{I'}(\bm{x}+a_L &\hat{\bm{n}}_{S'})-\pi_{I'}(\bm{x})  )   \Big|\Big|\norm{ \Big[\rho_{S,I}(\bm{x}),\rho_{S',I'}(\bm{x})\Big]  }_{\eta} \nonumber\\
        &\leq  \left( \frac{g_A}{2f_\pi}\right)^2 \times 3^4 \times \frac{1}{a_L^2} \times  4\pimax^2 \times 16\times 4
        \, \eta,
    \end{align}
where the factor of $3^4$ results from the sum over $S$, $I$, $S$, and $I'$.
\end{proof}

\begin{lemma}\label{Lemma:HFree_HWT_Commutator}
\begin{align}
    \norm{ [H_{\rm free}, H_{\rm WT}] }_\eta \leq  
     \frac{432h}{f_\pi^2} \pimax\Pimax
     \eta.
\end{align}
\end{lemma}
\begin{proof}

Here, the commutators are of the form $\Big[\rho_{I}(k),\Delta_{ij}^\sigma \Big]$,
so the non-vanishing commutators occur for $k=i$ or $k=j$. Then, for each $I$ and $\sigma$, one such commutator is of the form $[\adag(i) a(i), \adag(i) a(j) + \adag(j) a(i)]$, which is a sum of 6 NPFOs. 
Taking into account the sum over $\sigma$ yields an extra factor of 4, giving a total of $2 \times 6 \times 4 =48$ NPFOs from non-vanishing commutators.

Returning to the full commutator, and letting $k$ denote the qubit index of site $\bm{x}$, we have
\begin{align}
    \norm{ [H_{\rm free}, H_{\rm WT}] }_\eta 
    &\leq \frac{h}{4f_\pi^2} \sum_{\gamma_1} \norm{\sum_{I_1,I_2,I_3}\sum_{\bm{x}} \epsilon_{I_1I_2I_3}\pi_{I_2}(\bm{x})\Pi_{I_3}(\bm{x}) \Big[\rho_{I_1}(\bm{x}), T^{(\gamma_1)} \Big]}_\eta \nonumber\\
    &\leq  \frac{h}{4f_\pi^2} \times 6\times 6 \times \pimax\Pimax \times 48 \, \eta,
\end{align}
where one of the factors of 6 results because $\epsilon_{I_1I_2I_3}$ is non-zero for exactly $6$ terms in the sum of $I_1$, $I_2$, and $I_3$, and the other factor accounts for 6 non-commuting layers in the hopping operator.
\end{proof}

\begin{lemma}\label{Lemma:H_WT_H_C1_Commutator}
\begin{align}
    \norm{ [ H_{C},H_{\rm WT}] }_\eta =0.
\end{align}
\end{lemma}
\begin{proof}
Explicitly, this commutator has the form
\begin{align}
    \sum_{\bm{x}}\Big[ 
    \rho(\bm{x})\rho(\bm{x}), \sum_{I_1,I_2,I_3}\epsilon_{I_1I_2I_3}\pi_{I_2}(\bm{x})\Pi_{I_3}(\bm{x})
    \rho_{I_1}(\bm{x})\Big].
\end{align}
Note that $H_{\rm WT}$ preserves the number of nucleons on a particular site, so $H_{\rm WT}$ must commute with the sum of number operators for all $\bm{x}$. 
Note also that $\rho(\bm{x})=\sum_\sigma N_\sigma(\bm{x})$ is just the sum of number operators. 
Hence, we have the decomposition
\begin{align}
    [H_{C},H_{\rm WT}] 
    &= \sum_{\bm{x},\bm{y}} \Big[ \sum_{\sigma,\sigma'} 
    N_\sigma(\bm{x})N_{\sigma'}(\bm{x}),H_{\rm WT}(\bm{y})\Big] \nonumber\\
    &=  \sum_{\bm{x},\bm{y}} \sum_{\sigma,\sigma'} \left(
    N_\sigma(\bm{x})\Big[ 
    N_{\sigma'}(\bm{x}),H_{\rm WT}(\bm{y})\Big] + \Big[ 
    N_\sigma(\bm{x}),H_{\rm WT}(\bm{y})\Big] 
    N_{\sigma'}(\bm{x})\right).
\end{align}
Now each sub-commutator vanishes considering the number-preserving property of $H_{\rm WT}$.
\end{proof}

\begin{lemma}\label{Lemma:C2_HWT_Commutator}
\begin{align}
    \norm{ [H_{C_{I^2}},H_{\rm WT}] }_\eta \leq \frac{
     504 |C_{I^2}|}{f_\pi^2} \pimax\Pimax\eta.
\end{align}
\end{lemma}

\begin{proof}
Here, we take a cruder approach to bounding the commutator. The key commutator to compute is $[H_{C_{I^2}},
\rho_I]$.
For each $I$, $\rho_I$ generates $4$ terms of the form $a^\dagger a$, while $H_{C_{I^2}}$ consists of $a^\dagger a^\dagger a  a$ operators with weight  $28$. 
Then, each of the commutators of the form $[\adag \adag a a,\adag a]$ generates at most  6  NPFOs.
Thus, the term decomposes into weight $4\times 28 \times  6 =672 $ NPFOs (see the proof of \cref{Lemma:HAV_HC2_Commutator} where a similar analysis was used).
Thus, the full commutator can be bounded as 
\begin{align}
    \norm{ [H_{C_{I^2}},H_{\rm WT}] }_\eta &\leq  \frac{1}{4f_\pi^2} \frac{ |C_{I^2}|}{2}\sum_{I_1,I_2,I_3}\epsilon_{I_1I_2I_3} \norm{ \sum_{\bm{x}} [H_{C_{I^2}}, \pi_{I_2}(\bm{x}) \Pi_{I_3}(\bm{x}) 
    \rho_{I_1}(\bm{x})]}_\eta \nonumber\\
    &\leq \frac{1}{4f_\pi^2} \frac{|C_{I^2}|}{2} \times 6 \times \pimax \Pimax \times 672 \, \eta,
\end{align}
where the factor of 6 arises from the summation of $I_1$, $I_2$, and $I_3$ in the presence of the Levi-Civita tensor.
\end{proof}

\begin{lemma}\label{Lemma:Hpp_HWT_Commutator}
\begin{align}
    \norm{ [\Hpp^{(1)},H_{\rm WT}] }_\eta = 0.
\end{align}
\end{lemma}

\begin{proof}
We start by considering
\begin{align}
    \sum_{J,I_1,I_2,I_3} \Big[\Pi_J^2(\bm{x}), \epsilon_{I_1I_2I_3}\pi_{I_2}(\bm{x})\Pi_{I_3}(\bm{x}) 
    \rho_{I_1}(\bm{x})\Big] 
    &=  \sum_{J,I_1,I_2,I_3} \Pi_J(\bm{x})\Big[\Pi_J(\bm{x}), \epsilon_{I_1I_2I_3}\pi_{I_2}(\bm{x})\Pi_{I_3}(\bm{x}) 
    \rho_{I_1}(\bm{x})\Big] \nonumber\\
    &+\sum_{J,I_1,I_2,I_3} \Big[\Pi_J(\bm{x}), \epsilon_{I_1I_2I_3}\pi_{I_2}(\bm{x})\Pi_{I_3}(\bm{x}) 
    \rho_{I_1}(\bm{x})\Big] \Pi_J(\bm{x}) \nonumber\\
    &= \frac{i}{a_L^D}\sum_{I_1,I_2,I_3} \epsilon_{I_1I_2I_3} \Pi_{I_2}(\bm{x})\Pi_{I_3}(\bm{x}) 
    \rho_{I_1}(\bm{x})\nonumber\\
    &+\frac{i}{a_L^D}\sum_{I_1,I_2,I_3} \epsilon_{I_1I_2I_3} \Pi_{I_3}(\bm{x})\Pi_{I_2}(\bm{x}) 
    \rho_{I_1}(\bm{x}).
\end{align}
This vanishes since $\epsilon_{I_1I_2I_3}= - \epsilon_{I_2I_1I_3}$ is anti-symmetric under the exchange of $I_2$ and $I_3$ indices while $\Pi_{I_2}\Pi_{I_3}=\Pi_{I_3}\Pi_{I_2}$ is symmetric. Therefore, each term sums to zero.
\end{proof}

\begin{lemma} \label{Lemma:Hpp2_HWT_Commutator}
\begin{align}
    \norm{ [\Hpp^{(2)},H_{\rm WT}] }_\eta \leq  \frac{   72}{f_\pi^2}a_L^{-2}\pimax^2\eta.
\end{align}
\end{lemma}
\begin{proof}
By the same reasoning as in \cref{Lemma:Hpp_HWT_Commutator}, we have
\begin{align}
    \sum_{J,I_1,I_2,I_3} \Big[\pi_J^2(\bm{x}), \epsilon_{I_1I_2I_3}\pi_{I_2}(\bm{x})\Pi_{I_3}(\bm{x}) 
    \rho_{I_1}(\bm{x})\Big]  = 0.
\end{align}
The only terms in $\Hpp^{(2)}$ that are not of the form $\pi_J^2(\bm{x})$ are terms of the form $\pi_J(\bm{x})\pi_J(\bm{y})$ appearing in the discretized derivative. Their commutator with $H_{\rm WT}$ gives
\begin{align}
   &\sum_{\bm{x}}\sum_{\langle \bm{y},\bm{z} \rangle} \sum_{J,I_1,I_2,I_3} \Big[\pi_J(\bm{z})\pi_J(\bm{y}), \epsilon_{I_1I_2I_3}\pi_{I_2}(\bm{x})\Pi_{I_3}(\bm{x}) 
   \rho_{I_1}(\bm{x})
   \Big] \nonumber\\
   &\hspace{2 cm} = \sum_{\bm{x}}\sum_{\langle \bm{y},\bm{z} \rangle} \sum_{J,I_1,I_2,I_3} \epsilon_{I_1I_2I_3}  \pi_{I_2}(\bm{x})\Big(\pi_J(\bm{z})\Big[\pi_J(\bm{y}),\Pi_{I_3}(\bm{x})\Big]
   +\Big[\pi_J(\bm{z}),\Pi_{I_3}(\bm{x}) \Big]\pi_J(\bm{y}) \Big) 
   \rho_{I_1}(\bm{x}) \nonumber\\
   &\hspace{2 cm} = \frac{i}{a_L^D}\sum_{\langle \bm{y},\bm{z} \rangle} \sum_{I_1,I_2,I_3}  \epsilon_{I_1I_2I_3}\Big(\pi_{I_2}(\bm{y})\pi_{I_3}(\bm{z}) 
   \rho_{I_1}(\bm{y})+\pi_{I_2}(\bm{z})\pi_{I_3}(\bm{y}) 
   \rho_{I_1}(\bm{z})\Big),
   \label{eq:comm_27}
\end{align}
which can be non-zero. Therefore, for the full commutator,
\begin{align}
    \norm{ [\Hpp^{(2)},H_{\rm WT}] }_\eta &= 2\frac{a_L^{D-2}}{2}\frac{1}{4f_\pi^2}\norm{ \sum_{\bm{x}}\sum_{\langle \bm{y},\bm{z} \rangle} \sum_{J,I_1,I_2,I_3} \Big[\pi_J(\bm{z})\pi_J(\bm{y}), \epsilon_{I_1I_2I_3}\pi_{I_2}(\bm{x})\Pi_{I_3}(\bm{x}) 
    \rho_{I_1}(\bm{x})\Big] }_\eta \nonumber\\
    &\leq \frac{1}{4a_L^2f_\pi^2}\times2\norm{ \sum_{\langle \bm{y},\bm{z} \rangle} \sum_{I_1,I_2,I_3}  \epsilon_{I_1I_2I_3}\pi_{I_2}(\bm{y})\pi_{I_3}(\bm{z}) 
    \rho_{I_1}(\bm{y})}_\eta \nonumber\\
    &\leq \frac{1}{4a_L^2f_\pi^2}\times 2 \times 6 \times 6 \times  \pimax^2\times 4 \, \eta,
\end{align}
where the factor of 2 comes from accounting for two terms of equal semi-norm in the last line of \cref{eq:comm_27}, one of the factors of 6 results from the summation over $I_1$, $I_2$, and $I_3$ with the Levi-Civita tensor, another factor of 6 accounts for 6 non-commuting sets when implementing nearest-neighbor pairs $\langle \bm{y}, \bm{z} \rangle$, and finally the factor of 4 counts the maximum number of NPFOs arising from $\rho_{I_1}$ for each $I_1$.
\end{proof}

\begin{lemma}\label{Lemma:HAV_HWT_Commutator}
\begin{align}
    \sum_{\gamma_1}\norm{ \Big[H^{(\gamma_1)}_{\rm AV}, \sum_{\gamma_2}H^{(\gamma_2)}_{\rm WT}\Big] }_\eta \leq \frac{g_A}{f_\pi^3a_L}
    \left(72 \, a_L^{-D}+216 \, \pimax\Pimax\right)\pimax\eta.
\end{align}
\end{lemma}

\begin{proof}
Consider the term
\begin{align}
    &\Big[\big(\pi_J(\bm{x}+a_L\hat{\bm{n}}_S)-\pi_J(\bm{x}) \big)  
    \rho_{J,S}(\bm{x}), \epsilon_{I_1,I_2,I_3} 
    \pi_{I_2}(\bm{y})\Pi_{I_3}(\bm{y})
    \rho_{I_1}(\bm{y})\Big] \nonumber\\ 
    & \hspace{3 cm}= 
    \rho_{J,S}(\bm{x})  
    \rho_{I_1}(\bm{y})\Big[\big(\pi_J(\bm{x}+a_L\hat{\bm{n}}_S)-\pi_J(\bm{x}) \big) , \epsilon_{I_1I_2I_3}\pi_{I_2}(\bm{y})\Pi_{I_3}(\bm{y})\Big]  \nonumber\\
    &\hspace{3 cm} +
    \epsilon_{I_1I_2I_3}\pi_{I_2}(\bm{y})\Pi_{I_3}(\bm{y})\big(\pi_J(\bm{x}+a_L\hat{\bm{n}}_S)-\pi_J(\bm{x}) \big)\Big[
    \rho_{J,S}(\bm{x}), 
    \rho_{I_1}(\bm{y})\Big].
\end{align}
Let us treat these two commutators separately.
First consider
\begin{align}
  \Big[\big(\pi_J(\bm{x}+a_L\hat{\bm{n}}_S)-\pi_J(\bm{x}) \big) , \epsilon_{I_1I_2I_3}\pi_{I_2}(\bm{y})\Pi_{I_3}(\bm{y})\Big] 
  &= 
  \frac{i}{a_L^D}\epsilon_{I_1I_2I_3} \delta_{I_3,J} \pi_{I_2}(\bm{y})\Big(  \delta(\bm{y}-\bm{x}-a_L\hat{\bm{n}}_S) - \delta(\bm{y}-\bm{x})  \Big).
\end{align}
This commutator is accompanied by the term $
\rho_{J,S}$, which decomposes into at most 16 product terms of the form $\adag a \adag a$ for given values of $J$, $S$, and $I_1$. Each such term breaks into at most $2$ terms in the NPFO form, yielding a factor of 32 overall.

Now consider
\begin{align}\label{eq:comm_21}
    \Big[\rho_{J,S}(\bm{x}),\rho_{I_1}(\bm{y})\Big].
\end{align}
For every $J$ and $S$ value, there are 4 terms of the form $\adag a$, and for each $I_1$ value there are $4$ terms of the form $\adag a$.
Thus, there are 16 commutators of the form $[\adag a, \adag a]$, each of which generates at most 2 NPFOs (for when at least one of the creation and annihilation operators are associated with different types). 
Thus for each $J$, $S$, and $I_1$, there are 32 NPFOs.

Finally, consider the full commutator,
\begin{align}
     \bigg |\bigg |&\sum_{\bm{x},\bm{y}}\sum_{I_1,I_2,I_3}\sum_{J,S} \Big[\frac{g_A}{2f_\pi a_L}\big(\pi_J(\bm{x}+a_L\hat{\bm{n}}_S)-\pi_J(\bm{x}) \big) 
     \rho_{J,S}(\bm{x}),\frac{1}{4f_\pi^2} \epsilon_{I_1,I_2,I_3} 
    \pi_{I_2}(\bm{y})\Pi_{I_3}(\bm{y})
    \rho_{I_1}(\bm{y})\Big]\bigg |\bigg |_\eta \nonumber\\ 
    &\leq \frac{g_A}{8f_\pi^3a_L} \times \frac{1}{a_L^D}\sum_{I_1,I_2,I_3}  \sum_{J,S}\delta_{I_3,J}\norm{\epsilon_{I_1I_2I_3}   \sum_{\bm{x},\bm{y}} 
    \rho_{J,S}(\bm{x})
    \rho_{I_1}(\bm{y})\pi_{I_2}(\bm{y}) \Big( \delta(\bm{y}-\bm{x}-a_L\hat{\bm{n}}_S) -   \delta(\bm{y}-\bm{x})\Big)   }_\eta \nonumber\\
    &+ \frac{g_A}{8f_\pi^3a_L} \sum_{I_1,I_2,I_3}\sum_{J,S}  \max_{\bm{x}}\norm{\epsilon_{I_1I_2I_3}
    \pi_{I_2}(\bm{x})\Pi_{I_3}(\bm{x})\big(\pi_{I_3}(\bm{x}+a_L\hat{\bm{n}}_S) \big)} \times 32  \,\eta \nonumber\\
    &\leq \frac{g_A}{8f_\pi^3a_L} \times\frac{1}{a_L^D} \times 6 \times 3 \times 32 \,\eta \times 2\pimax +\frac{g_A}{8f_\pi^3a_L} \times 6 \times 3^2 \times \pimax^2 \Pimax \times 32 \,\eta,
\end{align}
where the triangle and Cauchy-Schwarz inequalities are applied as usual. Here, in the first term the factor of 6 comes from summing over $I_1$, $I_2$, and $I_3$ with the Levi-Civita tensor, and the factor of $3$ comes from summing over $S$ while the sum over $J$ does not produce any additional factor because of the Kronecker delta. In the second term, the factor of 6 has the same origin as in the first term, and the factor of $3^2$ accounts for the sum over $J$ and $S$. 
Note that, to get to the third line here, we have used the fact that one of the terms inside the parentheses in the third line of the equation does not contribute to the semi-norm as $\sum_{I_2,I_3}\epsilon_{I_1I_2I_3}\pi_{I_2}(\bm{x})\pi_{I_3}(\bm{x})=0$.
Finally, simplifying the expression gives the result.
\end{proof}

\begin{lemma}\label{Lemma:HWT_HWT_Commutator}
\begin{align}
    \sum_{\gamma_1}\norm{\Big[ H^{(\gamma_1)}_{\rm WT},\sum_{\gamma_2 \geq \gamma_1+1} H^{(\gamma_2)}_{\rm WT}\Big] }_\eta \leq 384 \left(\frac{1}{4f_\pi^2}\right)^2   \bigg(  3  \,\pimax \Pimax +\frac{ 2 }{a_L^D}  \bigg)\Pimax \pimax\eta.
\end{align}
\end{lemma}

\begin{proof}
    Suppressing the spatial arguments, the commutators are of the form
    \begin{align}
        \Big[\pi_{J_2} \Pi_{J_3} 
        \rho_{J_1}, \pi_{I_2} \Pi_{I_3}
        \rho_{I_1}\Big] 
        &= \pi_{J_2} \Pi_{J_3} \Big[
        \rho_{J_1}, \pi_{I_2} \Pi_{I_3} 
        \rho_{I_1}\Big] + \Big[\pi_{J_2} \Pi_{J_3} , \pi_{I_2} \Pi_{I_3} 
        \rho_{I_1}\Big]
        \rho_{J_1} \nonumber\\
        &= \pi_{J_2} \Pi_{J_3}  \pi_{I_2} \Pi_{I_3} \Big[
        \rho_{J_1},
        \rho_{I_1}\Big] + \Big[\pi_{J_2} \Pi_{J_3} , \pi_{I_2} \Pi_{I_3}  \Big]
        \rho_{I_1}
        \rho_{J_1}.
        \label{Eq:WT-WT-sub-analysis}
    \end{align}
Let us inspect these two terms separately,  recovering the spatial arguments and summing over the lattice volume. For the semi-norm of the first term,
\begin{align}
  &\sum_{\bm{x},\bm{y}}\norm{  \pi_{J_2}(\bm{y}) \Pi_{J_3}(\bm{y})  \pi_{I_2}(\bm{x}) \Pi_{I_3}(\bm{x}) \Big[\rho_{J_1}(\bm{y}), 
  \rho_{I_1}(\bm{x})\Big] }_\eta \nonumber\\
  &\quad\leq \sum_{\bm{x},\bm{y}}\pimax^2 \Pimax^2 \norm{\Big[
  \rho_{J_1}(\bm{y}), 
  \rho_{I_1}(\bm{x})\Big]}_\eta
  \nonumber\\
  &\quad\leq  \pimax^2 \Pimax^2 \times  32  \, \eta.
\end{align}
Note that for each $J_1$ or $I_1$, $\rho_{J_1}$ or $\rho_{I_1}$ is a sum of at most $4 $ NPFOs of the form $\adag a$, leading to at most 16 terms of the form $[\adag a, \adag a]$, which generate up to 2 NPFOs each, hence the factor of 32.

For the semi-norm of the second term in \cref{Eq:WT-WT-sub-analysis},
\begin{align}
    \sum_{\bm{x},\bm{y}}\norm{\Big[\pi_{J_2}(\bm{y}) \Pi_{J_3}(\bm{y}) , \pi_{I_2}(\bm{x}) \Pi_{I_3}(\bm{x})  \Big]
    \rho_{I_1}(\bm{x})
    \rho_{J_1}(\bm{y})}_\eta 
    \leq \frac{1}{a_L^D}(\delta_{J_2I_3}  +  \delta_{J_3I_2}) \Pimax \pimax \times 32 \,\eta.
\end{align}
Here, we have applied the canonical commutation of $\pi$ and $\Pi$ fields twice. Furthermore, for each $I_1$ and $J_1$, $\rho_{I_1}\rho_{J_1}$ generates 16 terms of the form $\adag a \adag a$, which each can be further broken to 2 NPFOs, giving a total of 32 NPFOs.

Putting all these together, the semi-norm of the full commutator is bounded as
    \begin{align}
        &\sum_{\gamma_1}\norm{ \Big[ H^{(\gamma_1)}_{\rm WT}, \sum_{\gamma_2 \geq \gamma_1+1}H^{(\gamma_2)}_{\rm WT}\Big] }_\eta \nonumber \\
        &\qquad\qquad\leq \left(\frac{1}{4f_\pi^2}\right)^2  \sum_{J_1,J_2,J_3,I_1,I_2,I_3}\epsilon_{I_1I_2I_3}\epsilon_{J_1J_2J_3} \bigg(  32  \,\pimax^2 \Pimax^2\eta + \frac{32}{a_L^D}(\delta_{J_2I_3}  +  \delta_{J_3I_2})  \Pimax \pimax \eta\bigg)
        \nonumber\\
        &\qquad\qquad\leq \left(\frac{1}{4f_\pi^2}\right)^2   \bigg(  36 \times 32  \,\pimax^2 \Pimax^2 +\frac{ 24 \times 32 }{a_L^D} \Pimax \pimax \bigg)\eta,
    \end{align}
where in the last line, we have used the fact that $\epsilon_{I_1I_2I_3}$ and $\epsilon_{J_1J_2J_3}$ are non-zero for exactly 6 terms each, contributing an additional factor of 36 in the first term in parentheses. On the other hand, there is a factor of $\delta_{J_2I_3}  +  \delta_{J_3I_2}$ in the second term multiplying $\epsilon_{I_1I_2I_3}\epsilon_{J_1J_2J_3}$. This limits the possibilities for non-zero contributions to $12+12=24$. Simplifying the expression gives the bound claimed in the statement of the lemma.
\end{proof}


%

\end{document}